\documentclass[11pt]{article}

\usepackage[letterpaper,margin=1in]{geometry}
\usepackage[utf8]{inputenc}
\usepackage{microtype}
\usepackage{graphicx}
\graphicspath{{./}{Fig/}}
\usepackage{amsmath,amssymb,amsthm}
\usepackage{natbib}
\usepackage{url,bm}
\usepackage{adjustbox,varwidth}
\usepackage{setspace,booktabs}
\usepackage{tabularx,multirow,array,makecell,siunitx}
\usepackage{pifont}
\usepackage[colorlinks=true,citecolor=blue,linkcolor=blue,urlcolor=blue]{hyperref}
\usepackage{setspace}
\usepackage{caption}
\usepackage{threeparttable}
\allowdisplaybreaks
\theoremstyle{plain}
\newtheorem{theorem}{Theorem}

\newtheorem{corollary}[theorem]{Corollary}
\theoremstyle{definition}

\newtheorem{assumption}{Assumption}

\theoremstyle{remark}

\theoremstyle{plain}
\newtheorem*{thm*}{Theorem}
\newtheorem*{fact*}{Fact}
\newtheorem*{prop*}{Proposition}

\newtheorem*{lem*}{Lemma}

\newtheorem*{definition*}{Definition}

\newtheorem*{defn*}{Definition}
\newtheorem*{exercise*}{Exercise}
\newtheorem*{remark*}{Remark}

\newtheorem*{rmk*}{Remark}

\newtheorem{apxlemma}{Lemma}
\newtheorem{apxass}{Assumption}
\newtheorem{apxthm}{Theorem}
\newtheorem{apxcor}{Corollary}
\newtheorem{apxprop}{Proposition}

\newcommand{\cmark}{\ding{51}}
\newcommand{\xmark}{\ding{55}}

\DeclareMathOperator*{\argmin}{arg\,min}
\DeclareMathOperator*{\argmax}{arg\,max}
\newcommand{\convd}{\overset{d}{\longrightarrow}}
\newcommand{\iid}{\overset{\mathrm{iid}}{\sim}}

\def \Pb {\mathbb{P}}
\def \Eb {\mathbb{E}}
\def \Rb {\mathbb{R}}

\def \q {\mathrm{Quantile}}

\def \Xb {\mathbf{X}}
\def \Zb {\mathbf{Z}}
\def \Gb {\mathbf{G}}
\def \Zbt {\mathbf{\tilde{Z}}}

\def \xb {\mathbf{x}}
\def \zb {\mathbf{z}}
\def \gb {\mathbf{g}}
\def \zbt {\mathbf{\tilde{z}}}

\def \k {[K_1] \times \cdots \times [K_l]}
\def \kg {[K_1] \times \cdots \times [K_s]}
\def \kz {[K_{s+1}] \times \cdots \times [K_l]}

\def \rcov {R^{XZ}}
\def \rfull {R^Y}
\def \wcov {W^{XZ}}
\def \wfull {W^{Y}}

\def \Dc {\mathcal{D}}
\def \done {\mathcal{D}_1}
\def \donet {\tilde{\mathcal{D}}_1}
\def \dtwo {\mathcal{D}_2}
\def \dtwot {\tilde{\mathcal{D}}_2}
\def \doneh {\mathcal{D}_{1h}}

\def \dtwoh {\mathcal{D}_{2h}}

\def \dtra {\mathcal{D}_{\mathrm{tra}}}

\def \dttra {\tilde{\mathcal{D}}_{\mathrm{tra}}}

\def \dcal {\mathcal{D}_{\mathrm{cal}}}

\def \dtrah {\mathcal{D}_{\mathrm{tra}, h}}

\def \dcalh {\mathcal{D}_{\mathrm{cal}, h}}

\def \dtest {\mathcal{D}_{\mathrm{test}}}
\def \dtesth {\mathcal{D}_{\mathrm{test}, h}}

\def \doneg {\mathcal{D}_{1\gb}}
\def \donegt {\tilde{\mathcal{D}}_{1\gb}}
\def \dtwog {\mathcal{D}_{2\gb}}
\def \dttwog {\tilde{\mathcal{D}}_{2\gb}}
\def \dtrag {\mathcal{D}_{\mathrm{tra}, \gb}}
\def \dttrag {\tilde{\mathcal{D}}_{\mathrm{tra}, \gb}}
\def \dcalg {\mathcal{D}_{\mathrm{cal}, \gb}}

\def \dtestg {\mathcal{D}_{\mathrm{test}, \gb}}

\def \ione {\mathcal{I}_1}
\def \ionet {\tilde{\mathcal{I}}_1}
\def \itwo {\mathcal{I}_2}
\def \itwot {\tilde{\mathcal{I}}_2}
\def \ioneh {\mathcal{I}_{1h}}

\def \itwoh {\mathcal{I}_{2h}}

\def \itwoz {\mathcal{I}_{2\zbt}}
\def \itra {\mathcal{I}_{\mathrm{tra}}}

\def \ittra {\tilde{\mathcal{I}}_{\mathrm{tra}}}
\def \ittraz {\tilde{\mathcal{I}}_{\mathrm{tra}, \zbt}}
\def \ical {\mathcal{I}_{\mathrm{cal}}}

\def \itrah {\mathcal{I}_{\mathrm{tra}, h}}

\def \icalh {\mathcal{I}_{\mathrm{cal}, h}}

\def \itest {\mathcal{I}_{\mathrm{test}}}
\def \itesth {\mathcal{I}_{\mathrm{test}, h}}

\def \ioneg {\mathcal{I}_{1\gb}}
\def \ionegt {\tilde{\mathcal{I}}_{1\gb}}
\def \itwog {\mathcal{I}_{2\gb}}
\def \ittwog {\tilde{\mathcal{I}}_{2\gb}}
\def \itrag {\mathcal{I}_{\mathrm{tra}, \gb}}
\def \ittrag {\tilde{\mathcal{I}}_{\mathrm{tra}, \gb}}
\def \icalg {\mathcal{I}_{\mathrm{cal}, \gb}}

\def \itestg {\mathcal{I}_{\mathrm{test}, \gb}}

\def \ntrah {n_{\mathrm{tra}, h}}
\def \ncalh {n_{\mathrm{cal}, h}}
\def \ntesth {n_{\mathrm{test}, h}}
\def \nz {\mathbf{n}_\zbt}
\def \nb {\mathbf{n}}

\def \rnz {\widehat{r}_\zbt}
\def \rng {\widehat{r}_\gb}
\def \munk {\widehat{\mu}_{n, \zbt}}

\def \drawone {\tilde{\mathcal{D}}_1}
\def \nrawone {\tilde{n}_1}

\def \ptrz {f_{1\zbt}}
\def \ptez {f_{2\zbt}}

\def \hatptez {\widehat{f}_{2\zbt}}

\def \ntrz {n_{\mathrm{tra}, \zbt}}
\def \nttrz {\tilde{n}_{\mathrm{tra}, \zbt}}

\def \ntez {n_{2\zbt}}
\def \fnz {\mathcal{F}_{\zbt}(\nz)}
\def \gnz {\mathcal{G}_{\zbt}(\nz)}
\def \snz {S_{\zbt}(\nz)}

\def \gnzmbar {\mathcal{G}_{\zbt}^{\bar M}(\nz)}
\def \hatgnz {\widehat{\mathcal {G}}_{\zbt}(\nz)}
\def \ntra {n_{\mathrm{tra}}}

\def \alo {\alpha_{\mathrm{lo}}}
\def \ahi {\alpha_{\mathrm{hi}}}
\def \gammab {\boldsymbol{\gamma}}
\def \ab {\boldsymbol{\alpha}}
\def \bb {\boldsymbol{\beta}}
\def \hatgb {\widehat{\boldsymbol{\gamma}}}
\def \hatbb {\widehat{\boldsymbol{\beta}}}
\def \mub {\boldsymbol{\mu}}
\def \Sigmab {\boldsymbol{\Sigma}}
\def \Tb {\mathbf{T}}
\def \Sl {\mathcal S_\ell}
\def \fc {\mathcal F_{\ell, h}}
\def \tfc {\tilde{\mathcal{F}}_{\ell, h}}
\def \edc {\mathcal{E}(\Dc)}
\def \hatsl {\mathcal S_{\hat \ell}}
\def \hast {h^\ast}

\def \irawone {\tilde{\mathcal{I}}_1}

\def \dttwo {\tilde{\mathcal{D}}_2}
\def \ittwo {\tilde{\mathcal{I}}_2}

\def \tgwcp {\widehat{t}_{\alpha}^{\mathrm{GWCP}}}
\def \tdawcp {\widehat{t}_{\alpha}^{\mathrm{DA-WCP}}}

\newcommand{\indep}{\perp\!\!\!\perp}

\def \qlo {q_{\alo}}
\def \qhi {q_{\ahi}}
\def \hatqlo {\widehat q_{\alo}}
\def \hatqhi {\widehat q_{\ahi}}
\def \tl {\widehat t_{\alpha, \ell}}
\def \thatl {\widehat t_{\alpha, \hat{\ell}}}

\newcommand{\appropto}{\mathrel{\vcenter{
  \offinterlineskip\halign{\hfil$##$\cr
    \propto\cr\noalign{\kern2pt}\sim\cr\noalign{\kern-2pt}}}}}

\providecommand{\indep}{\perp\!\!\!\perp}

\title{Predicting Current Outcomes From Historical Survey \\ Data
With Weighted Conformal Prediction}

\author{
Chihoon Lee$^{*,1}$,
Sungkyu Jung$^{\dagger,1}$,
and Hyokyoung G. Hong$^{\ddagger,2}$
\\[0.5em]
{\small $^{1}$Department of Statistics, Seoul National University, Seoul, South Korea}\\
{\small $^{2}$Biostatistics Branch, Division of Cancer Epidemiology and Genetics,}\\
{\small National Cancer Institute, Bethesda, MD, USA}
}

\date{\today}

\begin{document}

\maketitle

\begingroup
\renewcommand{\thefootnote}{\fnsymbol{footnote}}
\footnotetext[1]{Email: chihuni0922@snu.ac.kr}
\footnotetext[2]{Email: sungkyu@snu.ac.kr}
\footnotetext[3]{Email: grace.hong@nih.gov}
\endgroup

\begin{abstract}
In large-scale complex surveys such as the National Health and Nutrition Examination Survey (NHANES), some outcomes are measured only in selected years, leaving incomplete records across survey waves. We develop a weighted conformal prediction framework that enables valid population-level prediction of unobserved outcomes using information from earlier surveys. The method accommodates covariate shift, where both continuous and categorical covariate distributions evolve over time while survey design affects representativeness. It integrates subgroup-specific density ratio and subgroup-proportion estimation to approximate likelihood ratios between the historical and target covariate distributions, and we establish coverage guarantees for the resulting prediction sets. Simulation studies and an application predicting low-density lipoprotein cholesterol (LDL-C) for the current U.S. population show that the proposed approach achieves coverage close to the nominal level and improved efficiency over existing methods, particularly when covariate distributions are complex or unknown.
\end{abstract}


\vspace{1.5em}

\section{Introduction}

Large-scale complex surveys are a primary source of population-level evidence in public health, economics, and social science. In many such surveys, however, key outcomes are not measured in every wave, even when they remain scientifically or policy relevant. Variables may be discontinued because of cost, logistical constraints, or changing priorities; they may also be unavailable in periods when in-person examinations are suspended. As a result, researchers are often faced with the following problem: how can one make valid population-level predictions for an outcome that was observed in the past but is unobserved in the present?

The prediction problem described above involves two major sources of difficulty. First, the covariate distribution may change across survey waves as the population evolves over time. Demographic changes such as population aging, migration, and shifts in socioeconomic composition can substantially alter the covariate distribution between the past and current populations \citep{racialshift}. This phenomenon, often referred to as covariate shift, requires models trained on past data to be transported to current populations. Second, survey data are typically collected under complex sampling designs involving stratification, clustering, and unequal selection probabilities, and subject to nonresponse \citep{nhanesestimation}. Standard regression approaches yield only point predictions and typically ignore both distributional shift and sampling design, providing no formal coverage guarantees for the target population.

Conformal prediction \citep{vovk} provides a general, distribution-free framework for constructing prediction sets with finite-sample coverage guarantees under an exchangeability assumption. Extensions such as weighted conformal prediction (WCP) address covariate shift by incorporating importance weights that account for differences between the training and target covariate distributions \citep{wcp,gwcp,yang2024doubly}. However, existing WCP methods assume i.i.d. training data and do not account for design weights or structured subgroup oversampling. Consequently, they cannot be directly applied to repeated survey data under complex sampling designs. Moreover, when importance weights must be estimated rather than known, coverage guarantees can deteriorate, and theoretical results are limited to simplified settings such as purely group-based shifts \citep{gwcp}. At the same time, the survey sampling literature has long addressed the challenges posed by unequal probability sampling and nonresponse in population inference \citep{pfeffermann1993role}. Although design-based inference methods account for sampling weights when estimating population quantities, they focus primarily on point estimation and do not provide prediction sets with finite-sample coverage guarantees. \citet{designcp} applied WCP to survey data collected under unequal-probability sampling for population-level prediction, but their approach is restricted to a single time point.

This paper develops a weighted conformal prediction framework tailored to complex survey data observed across time. To our knowledge, this is the first approach that provides prediction sets with finite-sample coverage guarantees for current population outcomes using historical survey data, while explicitly accounting for (i) temporal covariate shift in both continuous and categorical predictors, and (ii) survey design features including stratification, oversampling, and nonresponse. Our main contributions are as follows. First, we formulate a unified probabilistic framework that embeds complex survey sampling within a superpopulation model while preserving key design features. By imposing transparent assumptions on sampling and nonresponse mechanisms, we show that the resulting training and target distributions satisfy a covariate shift structure amenable to conformal inference. Second, we develop a design-adapted weighted conformal prediction (DA-WCP) procedure in which the weight function depends jointly on continuous predictors and categorical subgroup variables, extending beyond the group-only setting of GWCP \citep{gwcp}. The weight is constructed to simultaneously adjust for temporal distributional change and survey design imbalance. Third, we establish a nonasymptotic lower bound on the coverage probability of the proposed prediction sets. The bound decomposes the miscoverage error into components arising from (i) data-dependent threshold selection and (ii) weight estimation error, and yields explicit convergence rates depending on the dimension of the continuous predictors.

More broadly, our work bridges conformal prediction and design-based survey inference, providing a framework for uncertainty quantification in prediction problems involving complex survey data and evolving populations. Simulation studies demonstrate that the proposed method achieves coverage close to the nominal level across a wide range of scenarios involving temporal covariate shift and heteroscedasticity, while existing group-based approaches can substantially under- or over-cover. We further apply the method to repeated waves of the National Health and Nutrition Examination Survey (NHANES) to construct prediction intervals for both continuous and categorical low-density lipoprotein cholesterol levels in the current U.S. population using historical data, illustrating the practical relevance of the framework.

The remainder of the paper is organized as follows. Section~\ref{sec-background} reviews the WCP framework and complex survey designs. Section~\ref{sec-survey} develops the survey-based probabilistic framework and formal problem setup. Section~\ref{sec-method} presents the proposed DA-WCP framework. Section~\ref{sec-theoretical} establishes theoretical coverage guarantees. Section~\ref{sec-simul} reports simulation results, and Section~\ref{sec-ldl} presents the NHANES application. Additional details and proofs are provided in the appendix.

\section{Background} \label{sec-background}

\subsection{Weighted Conformal Prediction (WCP)} \label{sec-background-wcp}

Suppose that the dataset has the form $\mathcal D = \{(X_i, Y_i)\}_{i=1}^n \subset \mathcal X \times \mathcal Y$, 
and let $(X_{n+1}, Y_{n+1}) \in \mathcal X \times \mathcal Y$ be a test data point.
We adopt the \textit{covariate shift} setting, where the training and test data are distributed as
\begin{align} \label{eqn:background-wcp-dist}
(X_i, Y_i) \overset{\text{i.i.d.}}{\sim} P_X \times P_{Y \mid X}, \quad i = 1, \ldots, n, 
\qquad (X_{n+1}, Y_{n+1}) \sim Q_X \times P_{Y \mid X}.
\end{align}
To address covariate shift, \citet{wcp} proposed a variant of CP known as the weighted conformal prediction (WCP) framework. 
WCP incorporates importance weights between the test and training covariate distributions, defined as 
$w := dQ_X/dP_X$, which is well defined whenever $Q_X$ is absolutely continuous with respect to $P_X$. 
If the true weight function $w$ is known, WCP proceeds as follows. 
First, randomly split $\mathcal D$ into a training set $\dtra$ and a calibration set $\dcal$ with index sets $\itra$ and $\ical$, typically with $|\itra| = |\ical| = n/2$. 
Then, fit a \textit{score function} $\mathcal S : \mathcal X \times \mathcal Y \rightarrow \mathbb R$ on $\dtra$ using a learning algorithm $\mathcal A$, and compute the scores $s_i = \mathcal S(X_i, Y_i)$ for $i \in \ical$. 
The function $\mathcal S(x, y)$ measures how incompatible $y$ is with the conditional distribution of $Y$ given $X = x$, and is therefore often called a \textit{nonconformity score}. 
In practice, $\mathcal A$ may be a regression model for continuous $Y$ or a classification model for categorical $Y$; for example, when $Y$ is continuous one often uses $\mathcal S(x, y) := |y - \mathcal A(x)|$. 
Using the scores $\{s_i\}_{i \in \ical}$, the WCP threshold under miscoverage level $\alpha \in (0, 1)$ at $X_{n+1} = x$ is given by
\begin{equation} \label{eqn:background-wcp-threshold}
    \widehat t_\alpha(x) = \mathrm{Quantile}_{1-\alpha}\left(\sum_{i \in \ical'} \frac{w(X_i)}{\sum_{j \in \ical'} w(X_j)}\delta_{s_i}\right),
\end{equation}
where $\ical' := \ical \cup \{n+1\}$, $s_{n+1} := +\infty$, and $\delta_{s_i}$ denotes the Dirac measure at $s_i$.
The WCP prediction set at $X_{n+1} = x$ is then given by $\widehat C_{n, \alpha}(x) = \{y \in \mathcal Y \mid \mathcal S(x, y) \le \widehat{t}_\alpha(x)\}$, and it satisfies the finite-sample coverage guarantee
$\Pb(Y_{n+1} \in \widehat C_{n, \alpha}(X_{n+1})) \ge 1-\alpha$, where the probability is over $\Dc$ and the independent test data $(X_{n+1}, Y_{n+1}) \sim Q_X \times P_{Y \mid X}$. 
The WCP has been extended by \cite{gwcp,yang2024doubly} to allow for unknown weights.

\subsection{Sampling Designs in Complex Surveys} \label{sec-background-design}

We summarize the sampling design of complex surveys that target a self-weighting sample within each sampling domain, using NHANES \citep{nhanesestimation} as a representative example. Similar self-weighting design principles also arise in other major surveys, including the National Survey on Drug Use and Health (NSDUH) \citep{nsduh} and the Programme for the International Assessment of Adult Competencies (PIAAC) \citep{piaac}.

NHANES uses a four-stage sampling design that incorporates stratification and clustering, with oversampling of minority subgroups. Specifically, the design targets a self-weighting sample within each sampling domain defined by demographic characteristics: individuals in the same domain have approximately equal inclusion probabilities, with higher inclusion probabilities assigned to minority subgroups.
At the first stage, primary sampling units (PSUs), which are counties or groups of contiguous counties, are selected by probability proportional to size sampling without replacement (PPSWOR) within each geographic stratum. At the second stage, within each selected PSU, segments, which are collections of census blocks, are selected again by PPSWOR. Within each selected segment, dwelling units (DUs) are sampled with equal probability, using subsampling rates designed to produce a national, approximately equal-probability sample of DUs. Finally, within each selected DU, persons are sampled with domain-specific subsampling rates.

\section{Weighted Conformal Prediction for Complex Survey Data} \label{sec-survey}

\subsection{Survey Framework and Objectives} \label{sec-survey-1}

\subsubsection{Survey Framework and Simplified Designs} \label{sec-survey-1-1}

While our goal is to analyze complex surveys, it is theoretically intractable to fully incorporate all design features described in Section~\ref{sec-background-design}. We therefore adopt simplified designs that retain the main features of complex surveys, including stratification, clustering, and the within-domain self-weighting property, while abstracting away from lower-level stages such as segment and DU selection which would unnecessarily complicate the analysis. These simplified designs provide the sampling framework under which we apply WCP and establish theoretical guarantees. Motivated by NHANES, we focus on surveys that yield self-weighting samples within sampling domains defined by demographics. Some large-scale surveys, such as the U.S.\ Current Population Survey (CPS) \citep{cps}, instead adopt geographic domains; our analysis applies to such settings with only minor modifications, with details provided in appendix~\ref{app:geo}.

Table~\ref{tab:simplified-designs} summarizes the simplified sampling designs considered in this paper. We consider two individual selection schemes, Poisson sampling and PPSWOR. For each scheme, we consider three sampling structures of increasing complexity: sampling from the entire population, within strata, and within selected PSUs. In Poisson sampling, individuals are sampled independently with their own (pre-specified) inclusion probabilities, whereas in PPSWOR, individuals are drawn without replacement with probabilities proportional to their inclusion probabilities. The key distinction is that Poisson sampling yields a random sample size, whereas PPSWOR yields a fixed sample size (before nonresponse). For details on Poisson sampling and PPSWOR implementations, see \citet{poisson} and \citet{ppswor}, respectively.

\begin{table}[t!]
\centering
\caption{Simplified sampling designs and their equivalent representations under the assumptions in Section~\ref{sec-survey-ass}.}
\label{tab:simplified-designs}
\vspace{-2mm}
\setlength{\tabcolsep}{7pt}
\renewcommand{\arraystretch}{1.1}
\begin{tabularx}{\textwidth}{c l X}
\toprule
\makecell[c]{\textbf{Individual} \\ \textbf{selection scheme}}
& \multicolumn{1}{c}{\textbf{Sampling structure}}
& \multicolumn{1}{c}{\textbf{Equivalent representation}} \\
\midrule
\multirow{3}{*}{\centering \makecell[c]{Poisson \\ sampling}}
& From the entire population & PPSWOR within a single stratum \\
& Within strata & PPSWOR within a single stratum \\
& Within selected PSUs & \makecell[l]{PPSWOR within a single stratum \\ (conditional on PSU selection)} \\
\midrule
\multirow{3}{*}{\centering PPSWOR}
& From the entire population & PPSWOR within a single stratum \\
& Within strata & PPSWOR within strata \\
& Within selected PSUs & \makecell[l]{PPSWOR within strata \\ (conditional on PSU selection)} \\
\bottomrule
\end{tabularx}
\end{table}


Table~\ref{tab:simplified-designs} also shows, for the purpose of our analysis, the equivalent representation to which each simplified design reduces under the assumptions in Section~\ref{sec-survey-ass}. 
When subgroup-wise population sizes are sufficiently large, samples from Poisson sampling can be regarded as samples from PPSWOR, conditional on the realized sample sizes. 
Since PPSWOR within a single stratum can be viewed as a special case of PPSWOR within strata, all six designs reduce to PPSWOR within strata. Accordingly, we focus on the PPSWOR-within-strata design in what follows, and our proposed methodology and corresponding theoretical guarantees apply directly to the other five designs. Further details are provided in appendix~\ref{app:reduction}.

Although motivated by NHANES, these sampling designs are not specific to NHANES and are applicable to a range of complex surveys that share similar design features. In general, one should choose the representation corresponding to the true sampling design; however, when stratum or cluster memberships are unavailable in practice, as is often the case, one can choose the representation that reflects the true design as closely as possible.

\subsubsection{Notation and Objectives} \label{sec-survey-1-2}

We aim to make predictive inferences for an unobserved outcome at the population level at the present time ($t = 2$), using data from a past survey cycle ($t = 1$) in which the outcome is observed. For each $t \in \{1, 2\}$, let $N_t$ be the population size. Indexing individuals in the population from $1$ to $N_t$, the survey sample can be viewed as a subset of $[N_t]$. For each individual $i \in [N_t]$, let $Y_{ti}$ denote the outcome (either continuous or categorical), $\Xb_{ti} \in \Rb^d$ the $d$-dimensional vector of continuous covariates, and $\Zbt_{ti} \in [K_1] \times \cdots \times [K_l]$ the $l$-dimensional vector of categorical covariates, including demographics.
Since we focus on complex surveys designed to yield self-weighting samples within demographic sampling domains, we decompose $\Zbt_{ti}$ into $s$ demographic variables $\Gb_{ti} \in [K_1] \times \cdots \times [K_s]$ and $(l-s)$ non-demographic variables $\Zb_{ti} \in [K_{s+1}] \times \cdots \times [K_l]$, and write $\Zbt_{ti} = (\Gb_{ti}^\top, \Zb_{ti}^\top)^\top$.

In large-scale surveys, data collection typically consists of an interview followed by a separate in-person component. However, because some sampled units do not participate in all stages and some variables are subject to additional eligibility or protocol requirements, missingness patterns may differ across variables. Since we require sampled units with fully observed covariates, for each $t \in \{1, 2\}$, we introduce two response indicators for each sampled unit $i \in [N_t]$: $\rcov_{ti}=1$ if and only if all covariates $(\Xb_{ti}, \Zbt_{ti})$ are observed, and $\rfull_{ti}=1$ if and only if both the covariates and the outcome, $(\Xb_{ti}, Y_{ti}, \Zbt_{ti})$, are fully observed. Moreover, each sampled unit is assigned a sampling weight, indicating how many individuals in the population it represents. Accordingly, we define a sampling weight for each response indicator: $\wcov_{ti}$ for units with $\rcov_{ti}=1$ and $\wfull_{ti}$ for those with $\rfull_{ti}=1$.

If the outcomes $Y_{2i}$ are partially observed at $t = 2$, we can use all units with fully observed covariates at $t = 2$ and construct prediction sets in essentially the same way as when $Y_{2i}$ are completely unobserved. Moreover, when a sufficient number of units have both covariates and outcomes observed at $t = 2$, stronger
coverage guarantees 
for our method 
can be obtained; see appendix~\ref{app:partial} for details. 
In what follows we focus on the setting where $Y_{2i}$ are completely unobserved, so that $\rfull_{2i} \equiv 0$.

The survey data with complete covariates at $t = 1, 2$ are denoted by
\begin{equation} \label{eqn:setting-form-rawdata}
\donet = \{(\Xb_{1i}, Y_{1i}, \Zbt_{1i}, \rfull_{1i}, \wcov_{1i}, \wfull_{1i})\}_{i \in \ionet}, 
\quad 
\dtwo = \{(\Xb_{2i}, \Zbt_{2i}, \wcov_{2i})\}_{i \in \itwo},
\end{equation}
where $\ionet \subset [N_1]$ and $\itwo \subset [N_2]$ collect indices with fully observed covariates  (i.e., $\rcov_{ti} = 1$). 
Since $Y_{2i}$ are completely unobserved, the association between covariates and outcomes must be inferred from complete cases at $t=1$ (i.e., $\rfull_{1i}=1$).  
Therefore, the survey dataset of interest at $t = 1$, which serves as the training data, is given by
\begin{equation} \label{eqn:setting-form-data}
    \done = \{(\Xb_{1i}, Y_{1i}, \Zbt_{1i}, \wfull_{1i})\}_{i \in \ione}, \qquad \ione = \{i \in \tilde{\mathcal I}_1 : \rfull_{1i}=1\}.
\end{equation}

Our goal is to construct \emph{valid} prediction sets for $Y_2$ at the population level using survey datasets $\drawone$ and $\dtwo$. Here, the training data is $\done$, a past survey dataset in which both covariates and outcomes are observed, and the target distribution is the $t=2$ population. For a miscoverage level $\alpha \in (0, 1)$, let $\widehat C_\alpha(\xb, \zbt) \equiv \widehat C_\alpha(\xb, \zbt; \drawone, \dtwo)$ denote the prediction set for $Y_2$ at covariate $(\Xb_2, \Zbt_2) = (\xb, \zbt)$.
We consider two notions of validity: (i) \emph{marginal validity} and (ii) \emph{group-conditional validity}.
For marginal validity, we require $\Pb(Y_2 \in \widehat C_\alpha(\Xb_2, \Zbt_2)) \ge 1 - \alpha$, where the probability is over $\drawone$, $\dtwo$, and an independent test point $(\Xb_2, Y_2, \Zbt_2)$ drawn at random from the $t=2$ population.
For group-conditional validity, we require $\Pb(Y_2 \in \widehat C_\alpha(\Xb_2, \Zbt_2) \mid \Gb_2 = \gb) \ge 1 - \alpha$ for each demographic subgroup $\gb$; see appendix~\ref{app:group} for a formal definition.

As illustrated in Figure~\ref{fig:setting-prob}, the training and target covariate distributions may differ due to (i) population shift between $t=1$ and $t=2$ and (ii) the survey sampling design at $t=1$. 
A natural approach is therefore to apply the WCP framework, which accounts for covariate shift.
Here, the current survey dataset $\mathcal{D}_2$ is combined with $\drawone$, including not only the complete cases in $\done$ but also those with missing outcomes, to capture covariate shift and estimate the WCP weight function.

\begin{figure}[!t]
\centering
\vspace{2mm}
\includegraphics[width = \textwidth]{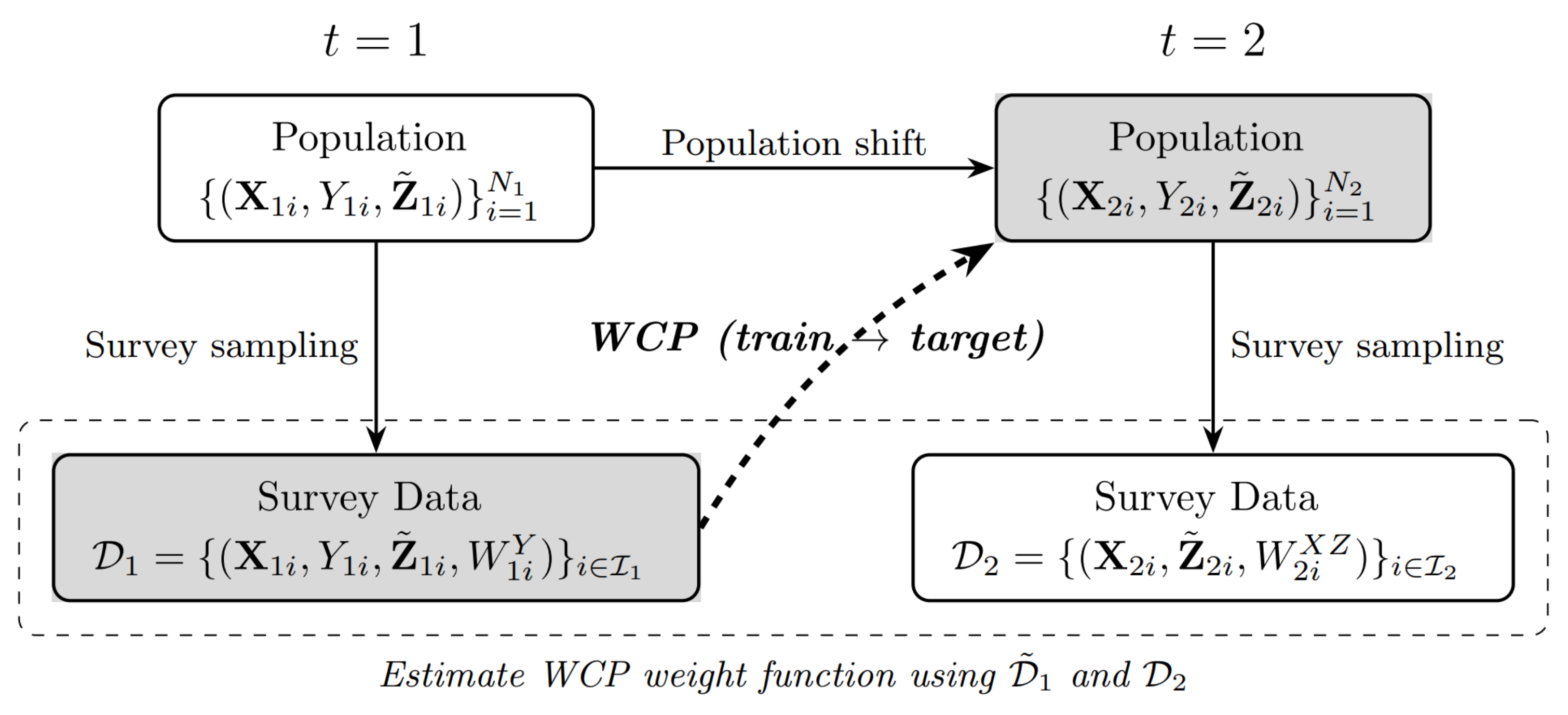}
\caption{Overview of the problem setting. At each time point ($t=1,2$), survey samples are drawn from the populations according to the survey sampling design. WCP is trained on $\done$ and targets the $t=2$ population.
}
\label{fig:setting-prob}
\end{figure}

However, WCP cannot be directly applied in our setting. WCP requires the training data to be an i.i.d.\ sample from an underlying distribution, as assumed in~\eqref{eqn:background-wcp-dist}. In contrast, the complex surveys we consider involve intricate sampling designs, including stratification, clustering, and oversampling, and are further subject to nonresponse. Consequently, the training data $\mathcal{D}_1$ are not i.i.d., and WCP cannot be applied without modification. This motivates the need for explicit assumptions on the population, nonresponse mechanism, and sampling design in order to apply WCP reliably to complex survey data.

\subsection{Assumptions for Applying WCP to Complex Survey Data} \label{sec-survey-ass}

As shown in Section~\ref{sec-background-wcp}, applying WCP requires the support of the target covariate distribution to be contained in that of the training covariate distribution, so that the importance weights are well defined. In practice, this condition can fail because the population evolves over time and the covariates $(\Xb_{ti}, \Zbt_{ti})$ may change even for the same individual. To address this, we adopt a \textit{superpopulation} framework that views the finite population as a sample from an underlying distribution \citep{deming, barry}, allowing us to define and estimate importance weights between the training and target covariate distributions.
However, we do not adopt the conventional superpopulation framework verbatim. 
Instead, because we focus on large-scale surveys designed to be self-weighting within each demographic subgroup defined by $\Gb_{ti}$, we treat $\Gb_{ti}$ as fixed by design. 
Conditioning on the variables that determine the sampling rates, such as these demographic labels, allows for the application of a superpopulation framework to $(\Xb_{ti}, Y_{ti}, \Zb_{ti})$  within each subgroup. 

\begin{assumption}[Superpopulation framework within demographic subgroups] \label{ass:super}
For each \(t \in \{1,2\}\) and \(\gb \in \kg\), and for each unit \(i \in [N_t]\) with \(\Gb_{ti}=\gb\), the variables \((\Xb_{ti}, Y_{ti}, \Zb_{ti})\) are assumed to be independently drawn from an underlying distribution \(P_{(\Xb_t, Y_t, \Zb_t)\mid \Gb_t=\gb}\), which admits the factorization
\(P_{(\Xb_t, Y_t, \Zb_t)\mid \Gb_t=\gb}
= P_{Y_t \mid \Xb_t, \Zb_t, \Gb_t=\gb}
\times P_{\Xb_t \mid \Zb_t, \Gb_t=\gb}
\times P_{\Zb_t \mid \Gb_t=\gb}\).
\end{assumption}
 
This formulation abstracts away from within-cluster dependence at the population-generating stage. Such an approximation is commonly adopted in superpopulation models for large-scale surveys \citep[e.g.,][]{pfeffermann1993role}, where the number of primary sampling units is typically large.
We also allow distributional differences across strata and clusters, interpreting them as arising from differences in demographic composition.

At $t = 2$, we must rely on $\dtwo$ to infer the population covariate distribution and thereby capture the covariate shift between the training and target distributions.
For this to be valid, we require that, for each $\gb \in \kg$, the distribution of $(\Xb_{2i}, \Zb_{2i}) \mid \Gb_{2i} = \gb$ coincides with that of $(\Xb_{2i}, \Zb_{2i}) \mid \Gb_{2i} = \gb, \rcov_{2i} = 1$. 
This leads us to the following assumption about the nonresponse mechanism.

\begin{assumption}[Nonresponse mechanism] \label{ass:nonresponse}
    The response indicators \((\rcov_{ti}, \rfull_{ti})\) are independent across units and depend only on \(\Gb_{ti}\).
    Specifically, for each $t \in \{1, 2\}$ and each unit \(i \in [N_t]\) with \(\Gb_{ti} = \gb\), $\Pb(\rcov_{ti} = 1) = a_{t\gb}^{xz}$, $\Pb(\rfull_{ti} = 1) = a_{t\gb}^y$, and \((\rcov_{ti}, \rfull_{ti}) \indep (\Xb_{ti}, Y_{ti}, \Zb_{ti})\), where $a_{t\gb}^{xz}, a_{t\gb}^y \in (0, 1]$ are unknown.
\end{assumption}

Assumption~\ref{ass:nonresponse} is a MAR assumption, where nonresponse depends only on \(\Gb_{ti}\). As noted above, this assumption is only needed for the $t=2$ case. While we state the assumption for both $t=1$ and $t=2$, all subsequent development remains valid when the nonresponse mechanism at $t = 1$ depends on $(\Xb_{1i}, \Zbt_{1i})$. We now formally specify the PPSWOR-within-strata sampling design in the next assumption.

\begin{assumption}[PPSWOR within strata] \label{ass:design}
For each $t \in \{1, 2\}$, the survey is designed to yield a self-weighting sample within each subgroup by targeting the pre-specified inclusion probability $\pi_{t\gb} \in (0, 1]$ for units in subgroup $\gb$. The finite population is partitioned into $H$ geographical strata, indexed by $h = 1, \dots, H$. Let $N_{th\gb}$ denote the population size of subgroup $\gb$ in stratum $h$, which is treated as known from the most recent census. The stratum sample size $m_{th}$ is chosen to satisfy the self-weighting assumption, namely $m_{th} = \sum_{\gb} N_{th\gb} \pi_{t\gb}$ for each stratum $h$.
Within each stratum $h$, we sample $m_{th}$ units using PPSWOR, where the single-draw selection probability for an individual in subgroup $\gb$ is  $\pi_{t\gb} / m_{th}$. The sampling is independent of $(\rcov_{ti}, \rfull_{ti})$ and $(\Xb_{ti}, Y_{ti}, \Zb_{ti})$. 
\end{assumption}

Since stratum population sizes are much larger than the number of sampled units, the subgroup composition can be approximated by a with-replacement model with subgroup probabilities $\tilde q_{th\gb} := N_{th\gb}\pi_{t\gb}/\sum_{\gb'} N_{th\gb'}\pi_{t\gb'}$, leading to the multinomial approximation in Assumption~\ref{ass:multinomial}.
Let $\mathrm{Mult}(\cdot, \cdot)$ denote the multinomial distribution.

\begin{assumption}[Multinomial approximation for subgroup-wise sample counts] \label{ass:multinomial}
For each $t \in \{1,2\}$ and $h \in [H]$, letting $n_{th\gb}^0$ denote the number of sampled units in subgroup $\gb$ within stratum $h$ at time $t$ (before nonresponse), we assume $(n_{th\gb}^0)_{\gb} \sim \mathrm{Mult}(m_{th}, (\tilde q_{th\gb})_{\gb})$.
\end{assumption}

However, due to nonresponse, the realized sample size within each stratum is random. The following theorem characterizes the distribution of subgroup-wise sample sizes within each stratum under Assumptions~\ref{ass:super}--\ref{ass:multinomial}.

\begin{theorem}[The distribution of subgroup-wise sample sizes within each stratum] \label{thm:multinomial}
    For each $t \in \{1, 2\}$ and $h \in [H]$, let $\mathcal D_{th}$ denote the subset of $\mathcal D_t$ (defined in \eqref{eqn:setting-form-rawdata} and \eqref{eqn:setting-form-data}) in stratum $h$, and let $n_{th} = |\mathcal D_{th}|$. For each $\gb \in \kg$, let $n_{th\gb}$ denote the sample size of subgroup $\gb$ in stratum $h$ at time $t$. Then, under Assumptions~\ref{ass:super}--\ref{ass:multinomial}, we have $(n_{th\gb})_{\gb} \mid n_{th} \sim \mathrm{Mult}(n_{th}, (q_{th\gb})_{\gb})$, where $q_{1h\gb} = N_{1h\gb}\pi_{1\gb}a_{1\gb}^y/\sum_{\gb'} N_{1h\gb'}\pi_{1\gb'}a_{1\gb'}^y$ and $q_{2h\gb} = N_{2h\gb}\pi_{2\gb}a_{2\gb}^{xz}/\sum_{\gb'} N_{2h\gb'}\pi_{2\gb'}a_{2\gb'}^{xz}$.
\end{theorem}

Finally, we assume that the conditional distribution of the outcome given covariates is invariant over time. 
This is a standard assumption in the covariate shift literature and is thus essential for applying WCP, as assumed in \eqref{eqn:background-wcp-dist}.

\begin{assumption}[Temporal invariance of the conditional distribution] \label{ass:invariant} 
We assume that $P_{Y_1 \mid \Xb_1 = \xb, \Zbt_1 = \zbt} = P_{Y_2 \mid \Xb_2 = \xb, \Zbt_2 = \zbt}$ for all $\xb \in \mathbb R^d$ and $\zbt \in \k$.
\end{assumption}

\subsection{Problem Setup} \label{sec-survey-prob}

Under Assumptions~\ref{ass:super}--\ref{ass:multinomial}, the training and target distributions can be characterized explicitly.  
Although the subgroup labels $\Gb_{ti}$ are treated as nonrandom at the population level, Theorem~\ref{thm:multinomial} implies that, conditional on the realized sample sizes, the sampled subgroup labels may be modeled as i.i.d.\ draws from a categorical distribution within each stratum. 
For each stratum $h \in [H]$, let $\ioneh \subset [N_1]$ and $\itwoh \subset [N_2]$ denote the index sets corresponding to the datasets $\doneh$ and $\dtwoh$, respectively.
Then, under Assumptions~\ref{ass:super}--\ref{ass:multinomial} and conditional on $n_{1h}$ and $n_{2h}$, the survey datasets $\doneh$ and $\dtwoh$ may be regarded as i.i.d.\ samples:
\begin{align} 
(\Xb_{1i}, Y_{1i}, \Zb_{1i}, \Gb_{1i}) &\iid P_{(\Xb_1, Y_1, \Zb_1) \mid \Gb_1} \times P_{\Gb_{1h} \mid \rfull_{1h} = 1}, \quad i \in \ioneh, \ \ h \in [H], \label{eqn:survey-3-doneh} \\
(\Xb_{2i}, \Zb_{2i}, \Gb_{2i}) &\iid P_{(\Xb_2, \Zb_2) \mid \Gb_2} \times P_{\Gb_{2h} \mid \rcov_{2h} = 1}, \quad i \in \itwoh, \ \ h \in [H], \label{eqn:survey-3-dtwoh}
\end{align}
where $P_{(\Xb_1, Y_1, \Zb_1) \mid \Gb_1}$ and $P_{(\Xb_2, \Zb_2) \mid \Gb_2}$ are defined in Assumption~\ref{ass:super}. Moreover, $P_{\Gb_{1h} \mid \rfull_{1h} = 1} := \mathrm{Cat}((q_{1h\gb})_{\gb})$ and $P_{\Gb_{2h} \mid \rcov_{2h} = 1} := \mathrm{Cat}((q_{2h\gb})_{\gb})$, where $\mathrm{Cat}(\cdot)$ denotes the categorical distribution and $q_{1h\gb}$ and $q_{2h\gb}$ are specified in Theorem~\ref{thm:multinomial}.
The target distribution, i.e., the population at $t = 2$, is then given by
\begin{equation} \label{eqn:survey-3-target}
(\Xb_2, Y_2, \Zb_2, \Gb_2) \sim P_{(\Xb_2, Y_2, \Zb_2) \mid \Gb_2} \times P_{\Gb_2}, \quad P_{\Gb_2} := \mathrm{Cat}((p_{2\gb})_{\gb}),
\end{equation}
where $P_{(\Xb_2, Y_2, \Zb_2) \mid \Gb_2}$ is defined in Assumption~\ref{ass:super} and $p_{2\gb} := N_{2\gb}/N_2$.

In the next section, we propose an algorithm that adapts the WCP framework in constructing asymptotically valid prediction sets for the target distribution, using $\donet$ and $\dtwo$. Our approach exploits the fact that the training data are i.i.d.\ within each stratum, conditional on the realized sample sizes, thanks to Assumptions~\ref{ass:super}--\ref{ass:invariant}. 
%
%
%
While these assumptions 
%
capture key design features,
they
merit explicit discussion. Assumption~\ref{ass:super} treats population values as independent draws, ruling out intra-cluster dependence; if within-PSU or within-segment correlation is substantial, 
finite-sample coverage may be affected. Assumption~\ref{ass:nonresponse} imposes a MAR condition and may be violated if nonresponse depends on unobserved covariates or the outcome itself. Although these assumptions are standard in model-assisted survey inference,
they should be assessed carefully in practice.

\section{Proposed Methods} \label{sec-method}
In this section, we present detailed algorithms for constructing prediction sets from $\donet$ and $\dtwo$ under the WCP framework, and refer to the resulting approach as design-adapted weighted conformal prediction (DA-WCP). We focus on marginal prediction sets that ensure marginal validity; the construction of group-conditional prediction sets ensuring group-conditional validity is deferred to appendix~\ref{app:group}.

For each stratum $h \in [H]$, we randomly split the stratum-specific dataset $\doneh$ into a training set $\dtrah$ and a calibration set $\dcalh$, with index sets $\itrah$ and $\icalh$, where $\ntrah = |\itrah|$ and $\ncalh = |\icalh|$ are typically chosen as $\ntrah = \ncalh = n_{1h}/2$. We then define $\dtra := \bigcup_{h=1}^H \dtrah$ and $\dcal := \bigcup_{h=1}^H \dcalh$, with corresponding index sets $\itra := \bigcup_{h=1}^H \itrah$ and $\ical := \bigcup_{h=1}^H \icalh$. 

\subsection{Continuous Outcomes}
We first describe the construction of prediction sets for continuous outcomes, $Y_{t} \in \Rb$.

\paragraph*{\textbf{(a) Score functions and nonconformity scores}}

The score function $\mathcal{S}(\xb, y, \zbt)$ measures how unlikely a candidate value $y$ is under the conditional distribution $P_{Y_t \mid \Xb_t = \xb, \Zbt_t = \zbt}$, and is constructed using $\dtra$.
We consider two commonly used score functions for continuous outcomes. First, given an estimated mean function $\widehat{\mu}$, we use the \textit{absolute residual score} $\mathcal{S}(\xb, y, \zbt) := |y - \widehat{\mu}(\xb, \zbt)|$, as in \citet{lei18}. Second, given estimated quantile functions $\widehat{q}_{\alo}$ and $\widehat{q}_{\ahi}$ at levels $\alo$ and $\ahi$ with $\alo < \ahi$, we use the \textit{conformalized quantile regression (CQR) score} of \citet{cqr}, $\mathcal{S}(\xb, y, \zbt) := \max\{\widehat{q}_{\alo}(\xb, \zbt) - y, \, y - \widehat{q}_{\ahi}(\xb, \zbt)\}$.
Although we focus on these two choices in this paper, other score functions can also be used, such as the standardized residual score of \citet{lei18}.
We then compute the nonconformity scores as $s_i = \mathcal S(\Xb_{1i}, Y_{1i}, \Zbt_{1i})$ for $i \in \ical$.


\paragraph*{\textbf{(b) Estimation of the WCP weight function within each stratum}} 
For each stratum $h \in [H]$, let $w_h$ denote the WCP weight function in stratum $h$ with~\eqref{eqn:survey-3-doneh} as the training distribution and~\eqref{eqn:survey-3-target} as the target distribution. Then, for any $\xb \in \mathbb{R}^d$ and $\zbt = (\gb^\top, \zb^\top)^\top \in \k$, $w_h(\xb, \zbt)$ is given by
\begin{equation} \label{eqn:proposed-construction-w}
w_h(\xb, \zbt) = \left(\frac{dP_{\Xb_2 \mid \Zbt_2 = \zbt}}{dP_{\Xb_1 \mid \Zbt_1 = \zbt}}\right)(\xb) \cdot \frac{P_{\Zb_2 \mid \Gb_2 = \gb}(\zb)}{P_{\Zb_1 \mid \Gb_1 = \gb}(\zb)} \cdot \frac{P_{\Gb_2}(\gb)}{P_{\Gb_{1h} \mid \rfull_{1h} = 1}(\gb)} 
= r_{\zbt}(\xb) \cdot \frac{v_{2\zbt}}{v_{1\zbt}} \cdot \frac{p_{2\gb}}{q_{1h\gb}},
\end{equation}
where $r_\zbt(\cdot) := \left(\tfrac{dP_{\Xb_2 \mid \Zbt_2 = \zbt}}{dP_{\Xb_1 \mid \Zbt_1 = \zbt}}\right)(\cdot)$ denotes the density ratio function for subgroup $\zbt$, $v_{1\zbt} := P_{\Zb_1 \mid \Gb_1 = \gb}(\zb)$ and $v_{2\zbt} := P_{\Zb_2 \mid \Gb_2 = \gb}(\zb)$ are the proportions of $\zb$ within the demographic subgroup $\gb$, and $q_{1h\gb}$ is the proportion of $\gb$ under $P_{\Gb_{1h} \mid \rfull_{1h} = 1}$ in~\eqref{eqn:survey-3-doneh}.

In~\eqref{eqn:proposed-construction-w}, while $p_{2\gb} = N_{2\gb}/N_2$ is known, the remaining quantities are unknown and must be estimated from $\donet$ and $\dtwo$. However, $\dcal$ should be used only for computing scores, since it must be independent of the estimated WCP weight functions. To ensure this, at $t = 1$ we use the remaining data $\dttra := \{(\Xb_{1i}, Y_{1i}, \Zbt_{1i}, \rfull_{1i})\}_{i \in \ittra}$ for estimation, where $\ittra := \ionet \setminus \ical$ and $\ionet$ is defined in~\eqref{eqn:setting-form-rawdata}.

Since estimating the two densities separately and then taking their ratio can be numerically unstable, we directly estimate the subgroup-specific density ratio function $r_\zbt$.
In particular, for each $\zbt \in \k$, we estimate $r_\zbt$ using the KLIEP method of \citet{Sugiyama2007, Sugiyama2008}, applied to $\{\Xb_{1i} \mid i \in \ittra, \Zbt_{1i} = \zbt\}$ and $\{\Xb_{2i} \mid i \in \itwo, \Zbt_{2i} = \zbt\}$. 
By Assumption~\ref{ass:super}, these are i.i.d.\ samples from $P_{\Xb_1 \mid \Zbt_1 = \zbt}$ and $P_{\Xb_2 \mid \Zbt_2 = \zbt}$, respectively. 
We denote the resulting estimator by $\widehat r_\zbt$. KLIEP, detailed in appendix~\ref{app:kliep-1}, provides one concrete route for obtaining explicit rates for the weight estimation error entering the coverage bound.
By the same assumption, for each demographic subgroup $\gb \in \kg$, the sets $\{\Zb_{1i} \mid i \in \ittra, \Gb_{1i} = \gb\}$ and $\{\Zb_{2i} \mid i \in \itwo, \Gb_{2i} = \gb\}$ are i.i.d.\ samples from $P_{\Zb_1 \mid \Gb_1 = \gb}$ and $P_{\Zb_2 \mid \Gb_2 = \gb}$, respectively. Accordingly, for each $\zbt= (\gb^\top, \zb^\top)^\top$, we estimate $\widehat v_{1\zbt} = \sum_{i \in \ittra} \mathbf 1\{\Zbt_{1i} = \zbt\}/\sum_{i \in \ittra} \mathbf 1\{\Gb_{1i} = \gb\}$ and $\widehat v_{2\zbt} = \sum_{i \in \itwo} \mathbf 1\{\Zbt_{2i} = \zbt\}/\sum_{i \in \itwo} \mathbf 1\{\Gb_{2i} = \gb\}$, where $\mathbf 1\{\cdot\}$ denotes the indicator function.


Moreover, for each $h \in [H]$, although $\{\Gb_{1i}\}_{i \in \itrah}$ is not an i.i.d.\ sample from the $t = 1$ population within stratum $h$, conditional on $n_{1h}$ it can 
be modeled as an i.i.d.\ sample from $P_{\Gb_{1h} \mid \rfull_{1h} = 1}$ in~\eqref{eqn:survey-3-doneh}, reflecting key design features such as oversampling and nonresponse. Using $\dtrah$, we therefore estimate $q_{1h\gb}$ by $\widehat q_{1h\gb} = \sum_{i \in \itrah} \mathbf 1\{\Gb_{1i} = \gb\} / \ntrah$ for each $\gb$. Finally, 
for any $\xb$ and $\zbt = (\gb^\top, \zb^\top)^\top$, 
we estimate the WCP weight function in stratum $h$ by
\begin{equation} \label{eqn:proposed-construction-what}
\widehat w_h(\xb, \zbt) = \widehat r_{\zbt}(\xb) \cdot \frac{\widehat v_{2\zbt}}{\widehat v_{1\zbt}} \cdot \frac{p_{2\gb}}{\widehat q_{1h\gb}}.
\end{equation}

\paragraph*{\textbf{(c) Construction of the prediction set}} 
Given the nonconformity scores $\{s_i\}_{i \in \mathcal{I}_{\mathrm{cal}}}$ and the estimated WCP weight functions $\{\widehat{w}_h\}_{h=1}^H$, for a miscoverage level $\alpha \in (0,1)$, we compute the WCP-type threshold as
\begin{equation} \label{eqn:proposed-construction-threshold}
\widehat{t}_\alpha = \text{Quantile}_{1-\alpha}\left(\sum_{h=1}^H \frac{1}{H} \sum_{i \in \icalh} \frac{\widehat{w}_h(\Xb_{1i}, \Zbt_{1i})}{\sum_{j \in \icalh} \widehat{w}_h(\Xb_{1j}, \Zbt_{1j})} \cdot\delta_{s_i} \right).
\end{equation}
In \eqref{eqn:proposed-construction-threshold}, we use $\icalh$ in place of $\icalh'$ (see \eqref{eqn:background-wcp-threshold}) for computational convenience. 
The marginal prediction set $\widehat{C}_{\alpha}(\xb,\zbt)$ for $Y_2$ at $(\Xb_2, \Zbt_2)=(\xb,\zbt)$, with miscoverage level $\alpha$, is constructed as
\begin{equation} \label{eqn:proposed-construction-set}
\widehat{C}_{\alpha}(\xb, \zbt) = \left\{y \in \mathbb{R} : \mathcal{S}(\xb, y, \zbt) \le \widehat{t}_\alpha \right\}, \quad \forall\xb \in \Rb^d,\ \zbt \in \k.
\end{equation}

\paragraph*{\textbf{(d) Selection of the prediction set}}
In the CP literature, a common objective is to minimize prediction interval length subject to the desired coverage guarantee. Because interval length depends on the choice of score function and the best choice is unknown in advance, we construct $L$ prediction sets utilizing different score functions $\mathcal S_1, \dots, \mathcal S_L$. For each $\ell \in [L]$, we apply steps (a)--(c) with $\mathcal S_\ell$ and denote the resulting prediction set in~\eqref{eqn:proposed-construction-set} by $\widehat C_{\alpha,\ell}$. For example, 
we vary the quantile levels $(\alo, \ahi)$ in the CQR score to obtain $L$ candidate prediction sets. We then select the candidate with the smallest average interval length; see appendix~\ref{app:implement} for details. Let $\widehat \ell \in [L]$ denote the selected index; the selected prediction set is $\widehat C_{\alpha,\hat \ell}$.

\subsection{Categorical Outcomes}
For a categorical outcome $Y_t \in [M]$ with $M$ categories, the same procedure applies, with the only difference being the choice of score function. Given a soft classifier $\widehat{f}(\xb,\zbt) = (\widehat{f}_1(\xb,\zbt), \ldots, \widehat{f}_M(\xb,\zbt))^\top$ trained on $\dtra$, we use $\mathcal{S}(\xb,y,\zbt) := 1 - \widehat{f}_y(\xb,\zbt)$ \citep{vovk}. The construction in~\eqref{eqn:proposed-construction-set} then applies with $y \in [M]$. The resulting prediction set $\widehat C_\alpha(\xb, \zbt)$ may contain zero, one, or multiple categories, depending on $(\xb, \zbt)$. In this setting, we may vary the hyperparameters of the soft classifier and select the candidate set with the smallest average set cardinality, analogous to the interval length in the continuous case; see appendix~\ref{app:implement}.

\section{Theoretical Guarantee} \label{sec-theoretical}

In this section, we derive a theoretical coverage guarantee for the proposed marginal prediction sets; an analogous result for group-conditional sets is given in appendix~\ref{app:theoretical-group}. Let $\widehat C_{\alpha,1}, \dots, \widehat C_{\alpha,L}$ denote $L$ candidate prediction sets of the form~\eqref{eqn:proposed-construction-set}, and let $\widehat \ell \in [L]$ be the selected index. The following theorem provides a lower bound on the coverage of the selected set $\widehat C_{\alpha, \hat \ell}$. The proof, which extends the ordinary split CP argument of \citet{aggregation} to our setting, is given in appendix~\ref{app:thm1}.

\begin{theorem}[Lower bound on the coverage probability of $\widehat C_{\alpha, \hat \ell}$] \label{thm:coverage}
    Define $\Dc := \dttra \cup \dtwo$, where \(\dttra\) is given in Section~\ref{sec-method}.
    Denote $\nrawone := |\irawone|$, $n_2 := |\itwo|$, and $n_{1h} := |\ioneh|$ for each $h \in [H]$.
    Under Assumptions~\ref{ass:super}--\ref{ass:invariant} and a set of regularity conditions on the subgroup-specific density ratio functions (Assumption~\ref{ass:kliep-1} in appendix~\ref{app:kliep-2}),
    the selected prediction set $\widehat C_{\alpha, \hat \ell}$ satisfies the following.
    There exists an event $\mathcal E (\Dc)$ that depends solely on $\Dc$ such that
    $\Pb(\mathcal E(\Dc) \mid \tilde n_1, n_2, \{n_{1h}\}_{h=1}^H ) \to 1$
    as $n_{11}, \dots, n_{1H}, n_2 \to \infty$, and on $\mathcal E(\Dc)$,
    \[
    \begin{adjustbox}{max width=\linewidth}
    $\displaystyle
    \begin{aligned}
        &\Pb\left(Y_2 \in \widehat{C}_{\alpha, \hat \ell}(\Xb_2, \Zbt_2) \Bigm| \Dc, \tilde n_1, n_2, \{n_{1h}\}_{h=1}^H \right)  \ge 1-\alpha - \left(2 + C_1\Big(1 +\sqrt{\log(2HL)}\Big)\right) \\ 
        &\times \Bigg( \underbrace{\max_{1 \le h \le H} \Eb\left[\big|\widehat w_h(\Xb_{1h}, \Zbt_{1h}) - w_h(\Xb_{1h}, \Zbt_{1h})\big| \Bigm| \Dc, \tilde n_1, n_2, \{n_{1h}\}_{h=1}^H  \right]}_{(\mathrm{I})} + 
        \underbrace{\max_{\zbt, h} \frac{v_{2\zbt} p_{2\gb}}{v_{1\zbt} q_{1h\gb}} \cdot \frac{C_2}{\sqrt{\min_{h} n_{1h}}}}_{(\mathrm{II})} \Bigg)
    \end{aligned}
    $
    \end{adjustbox}
    \]
    for some universal constants $C_1, C_2> 0$.
    Here, $(\Xb_{1h}, \Zbt_{1h}) \sim P_{(\Xb_1, \Zb_1) \mid \Gb_1} \times P_{\Gb_{1h} \mid \rfull_{1h}=1}$ for each $h \in [H]$ and
    the test point $(\Xb_2, Y_2, \Zbt_2) \sim P_{(\Xb_2, Y_2, \Zb_2) \mid \Gb_2} \times P_{\Gb_2}$ (See \eqref{eqn:survey-3-doneh} and \eqref{eqn:survey-3-target}) are both independent of $\Dc$.
\end{theorem}


Theorem~\ref{thm:coverage} holds for any score function, applies to both continuous and categorical outcomes, and does not depend on the rule used to select $\widehat C_{\alpha, \hat \ell}$.
Since the prefactor grows only logarithmically in $H$, the coverage lower bound remains stable as $H$ increases.
When $H$ is small (e.g., $H = 1$), one can obtain a tighter bound; we provide such a refinement in appendix~\ref{app:smallh}.

The first slack term (I) reflects the use of the estimated weight function $\widehat w_h$ in place of $w_h$. The use of the KLIEP-based estimator $\widehat w_h$ is made for concreteness and analytical tractability; in principle, the same coverage decomposition applies to any density ratio estimator that achieves comparable error bounds. 
The second slack term (II) arises from data-dependent selection via the random index $\widehat \ell \in [L]$, which depends on both $\Dc$ and $\dcal$.
Note that term (II) also depends on the discrepancy between subgroup proportions in the training and target distributions.
Term (I) determines the overall rate, which is in turn determined by the convergence rate of $\widehat r_{\zbt}$ (provided in Theorem~\ref{thm:kliep-1} of appendix). 
Due to the curse of dimensionality, this rate depends on $d$, the dimension of $\Xb_t$.
Combining Theorems~\ref{thm:coverage} and~\ref{thm:kliep-1} then yields the following corollary.

\begin{corollary}[Theoretical coverage guarantee for $\widehat C_{\alpha, \hat \ell}$] \label{cor:coverage}
    Conditional on $\nrawone$, $n_2$, and $\{n_{1h}\}_{h=1}^H$, and under Assumptions~\ref{ass:super}--\ref{ass:invariant} and Assumption~\ref{ass:kliep-1} in appendix, as $n_{11}, \dots, n_{1H}, n_2 \to \infty$, if $\min_{1 \le h \le H} n_{1h} \gtrsim \tilde n_1/H$ and $\min_{1 \le h \le H} n_{1h} \lesssim n_2$, then the selected prediction set $\widehat C_{\alpha, \hat \ell}$ satisfies
    \[
    \begin{adjustbox}{max width=\linewidth}
    $\displaystyle
    \begin{aligned}
        \Pb\Big(Y_2 \in \widehat{C}_{\alpha,\hat\ell}(\Xb_2, \Zbt_2) \Bigm| \Dc, \nrawone, n_2, \{n_{1h}\}_{h=1}^H \Big) 
        \ge 1 - \alpha  - 
        \begin{cases}
            O_p\big((\min_{1 \le h \le H} n_{1h})^{-\frac{1}{2}}\big) & d = 0,\\
            O_p\big(\min(\tilde n_1, n_2)^{-\frac{1}{2+\gamma}}\big) & d = 1, 2, \\
            O_p\Big(\Bigl(\tfrac{\log\min(\tilde n_1, n_2)}{\min(\tilde n_1, n_2)}\Bigr)^{\frac{1}{d}}\Big) & d \ge 3,
        \end{cases}
    \end{aligned}
    $
    \end{adjustbox}
    \]
    for any arbitrarily small $\gamma > 0$. The $O_p(\cdot)$ terms are taken with respect to the randomness in $\Dc$,
    conditional on $\nrawone$, $n_2$, and $\{n_{1h}\}_{h=1}^H$.
    Moreover, $(\Xb_2, Y_2, \Zbt_2) \sim P_{(\Xb_2, Y_2, \Zb_2) \mid \Gb_2} \times P_{\Gb_2}$ is an independent test point drawn from the $t = 2$ population.
\end{corollary}


We conclude by comparing the miscoverage error of GWCP \citep{gwcp} to that of our method in the special case $H = 1$, $d=0$, and $l=s$, where sampling is from a single stratum and $\Gb_{ti}$ is the only covariate. This setting results in the simplified WCP weight function $w(\xb, \zbt) = p_{2\gb}/q_{1\gb}$.
GWCP assumes covariate shift only in $\Gb_{ti}$, estimates $q_{1\gb}$ from $\dcal$, and uses a proof specialized to this setting to obtain the sharper $O(n_1^{-1})$ bound.
In contrast, our method is designed for more general settings where shift may occur in both continuous and categorical variables $(\Xb_t, \Zbt_t)$ and involves set selection, and thus relies on the general bound in Theorem~\ref{thm:coverage}, with $\widehat w$ estimated from $\dtra$.
Consequently, we estimate $q_{1\gb}$ from $\dtra$ rather than $\dcal$, giving $\widehat q_{1\gb} = q_{1\gb} + O_p(\ntra^{-1/2})$ and a miscoverage error of $O_p(n_1^{-1/2})$.

\section{Simulation Study} \label{sec-simul}

In this section, we present a simulation study to evaluate the performance of our proposed method, DA-WCP, under distributional shifts in both $\Xb_{t}$ and $\Zbt_{t}$. For simplicity, we focus on the setting with $H = 1$, dropping the stratum index $h$, and consider categorical covariates consisting only of demographic variables, i.e., $\Zbt_{t} = \Gb_{t}$.
We begin with the most general setting (S1) that mimics complex real survey data, and then conduct an ablation study by sequentially removing sources of distributional shift or modifying the error structure in scenarios (S2)--(S5).
In particular, our simulation study focuses on four elements: (1) $P_{\Gb_1 \mid \rfull_1=1} \neq P_{\Gb_2}$, i.e., the training and target demographic proportions differ; (2) $P_{\Gb_2 \mid \rcov_2=1} \neq P_{\Gb_2}$, i.e., the demographic proportions in the sample at $t=2$ differ from those in the population; (3) $P_{\Xb_1 \mid \Gb_1} \neq P_{\Xb_2 \mid \Gb_2}$, i.e., covariate shift exists in the conditional distribution $P_{\Xb_t \mid \Gb_t}$; and (4) the error distribution in the outcome model is heteroscedastic.
The resulting five scenarios (S1)--(S5) are summarized in Table~\ref{tab:simul}.

\begin{table}[!t]
\centering
\caption{Five scenarios for the simulation study}
\label{tab:simul}
\vspace{-3mm}

\begin{tabular*}{\textwidth}{@{\extracolsep{\fill}}ccccc@{}}
\toprule
\textbf{Scenario} 
& $P_{\Gb_1 \mid \rfull_1=1} \neq P_{\Gb_2}$ 
& $P_{\Gb_2 \mid \rcov_2=1} \neq P_{\Gb_2}$ 
& $P_{\Xb_1 \mid \Gb_1} \neq P_{\Xb_2 \mid \Gb_2}$ 
& Heteroscedastic error \\
\midrule
(S1) & \cmark & \cmark & \cmark & \cmark \\ 
(S2) & \xmark & \cmark & \cmark & \cmark \\ 
(S3) & \cmark & \xmark & \cmark & \cmark \\ 
(S4) & \cmark & \cmark & \xmark & \cmark \\ 
(S5) & \cmark & \cmark & \cmark & \xmark \\
\bottomrule
\end{tabular*}

\end{table}

In this simulation, we set $\Xb_{t} \in \Rb^4$, $Y_{t} \in \Rb$, and $\Gb_{t} \in \{1,\ldots,5\}$.
Datasets $\done$ and $\dtwo$ with sample sizes $n_1=n_2 = 5,000$ are simulated reflecting each of the five scenarios as well as a missingness (nonresponse) pattern; see appendix~\ref{app:simul-setup} for the simulation setup.

We compare DA-WCP, our proposed method that addresses covariate shift in both $\Xb_t$ and $\Gb_t$, with GWCP \citep{gwcp}, which accounts only for shift in $\Gb_t$. Both reduce to ordinary CP in the absence of covariate shift and outperform it otherwise; thus, we exclude ordinary CP from the comparison. To compare these methods across different score functions and notions of validity, we evaluate the eight methods listed in Table~\ref{table:methods}.
For the absolute residual score, we use a fixed covariate set without length minimization, whereas for the CQR score, we select a length-minimizing prediction set over a grid of $(\alo,\ahi)$; in both cases, we use linear models to estimate the mean and quantile functions. See appendix~\ref{app:implement} for details.
For evaluation, we use an i.i.d.\ test sample of size 5,000 from the $t=2$ population. 

\begin{table}[t]
\centering
\caption{Compared methods for continuous outcomes}
\label{table:methods}
\vspace{-3mm}
\begin{tabular*}{\textwidth}{@{\extracolsep{\fill}}ccccc}
\toprule
\multirow{2}{*}{\textbf{Method}} 
& \multicolumn{2}{c}{\textbf{Marginal}} 
& \multicolumn{2}{c}{\textbf{Group-conditional}} \\
\cmidrule{2-3}\cmidrule{4-5}
& Absolute residual & CQR & Absolute residual & CQR \\
\midrule
GWCP  & \texttt{GWCP\_AR\_m}     & \texttt{GWCP\_CQR\_m}     & \texttt{GWCP\_AR\_g}     & \texttt{GWCP\_CQR\_g} \\
DA-WCP& \texttt{DA\_WCP\_AR\_m}  & \texttt{DA\_WCP\_CQR\_m}  & \texttt{DA\_WCP\_AR\_g}  & \texttt{DA\_WCP\_CQR\_g} \\
\bottomrule
\end{tabular*}
\end{table}

First, we compare the four marginal methods listed in Table~\ref{table:methods} across scenarios (S1)--(S5).
The datasets $\mathcal{D}_1$ and $\mathcal{D}_2$ are independently generated 100 times, and empirical coverage is computed using the i.i.d.\ test data; the results are summarized in Figure~\ref{fig:simul-marg}. In scenarios (S1)--(S3), covariate shift occurs in the conditional distribution $P_{\Xb_t \mid \Gb_t}$. As expected, the two DA-WCP methods attain the desired coverage, whereas the two GWCP methods show substantial undercoverage. GWCP does not necessarily exhibit undercoverage, as it undercovers in our setting but may overcover in others. Conditions characterizing undercoverage and overcoverage, along with additional simulations, are provided in appendices~\ref{app:coverage} and~\ref{app:simul}. In scenarios (S4) and (S5), both DA-WCP and GWCP attain the desired coverage. In (S4), where covariate shift exists only in the group variables, GWCP performs well, as expected. We note that in (S5), there is no guarantee that GWCP works well. However, when the errors are homoscedastic and the mean and quantile functions are consistently estimated, the GWCP and DA-WCP thresholds are nearly identical, and both attain the desired coverage. Online supplementary material, Section~\ref{app:coverage}, further explains this phenomenon.

\begin{figure}[t]
    \centering
    \includegraphics[width=\textwidth]{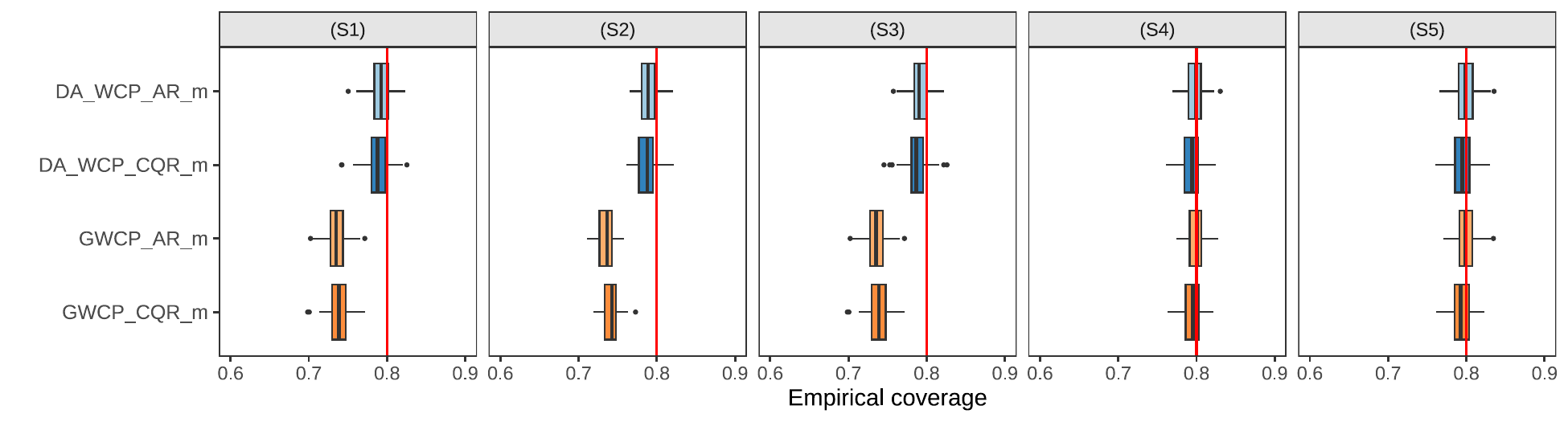}
    \caption{Empirical coverage of marginal methods under the target level \(1-\alpha=0.8\), evaluated across scenarios (S1)--(S5) using 100 independently generated datasets.}
    \label{fig:simul-marg}
\end{figure}

Next, we compare the four group-conditional methods listed in Table~\ref{table:methods} by visualizing their empirical group-wise coverages in Figure~\ref{fig:simul-group}. In scenarios (S1)--(S3), similar to the marginal case, the two DA-WCP methods attain coverage closer to the target level than the two GWCP methods across all groups. In scenarios (S4) and (S5), by contrast, all four methods achieve coverage close to the target level across all groups. Overall, these results suggest that the DA-WCP methods maintain stable coverage across scenarios, making them a reliable choice in practice, especially when the covariate distribution is unknown.

\begin{figure}[t]
    \centering
    \includegraphics[width=\textwidth]{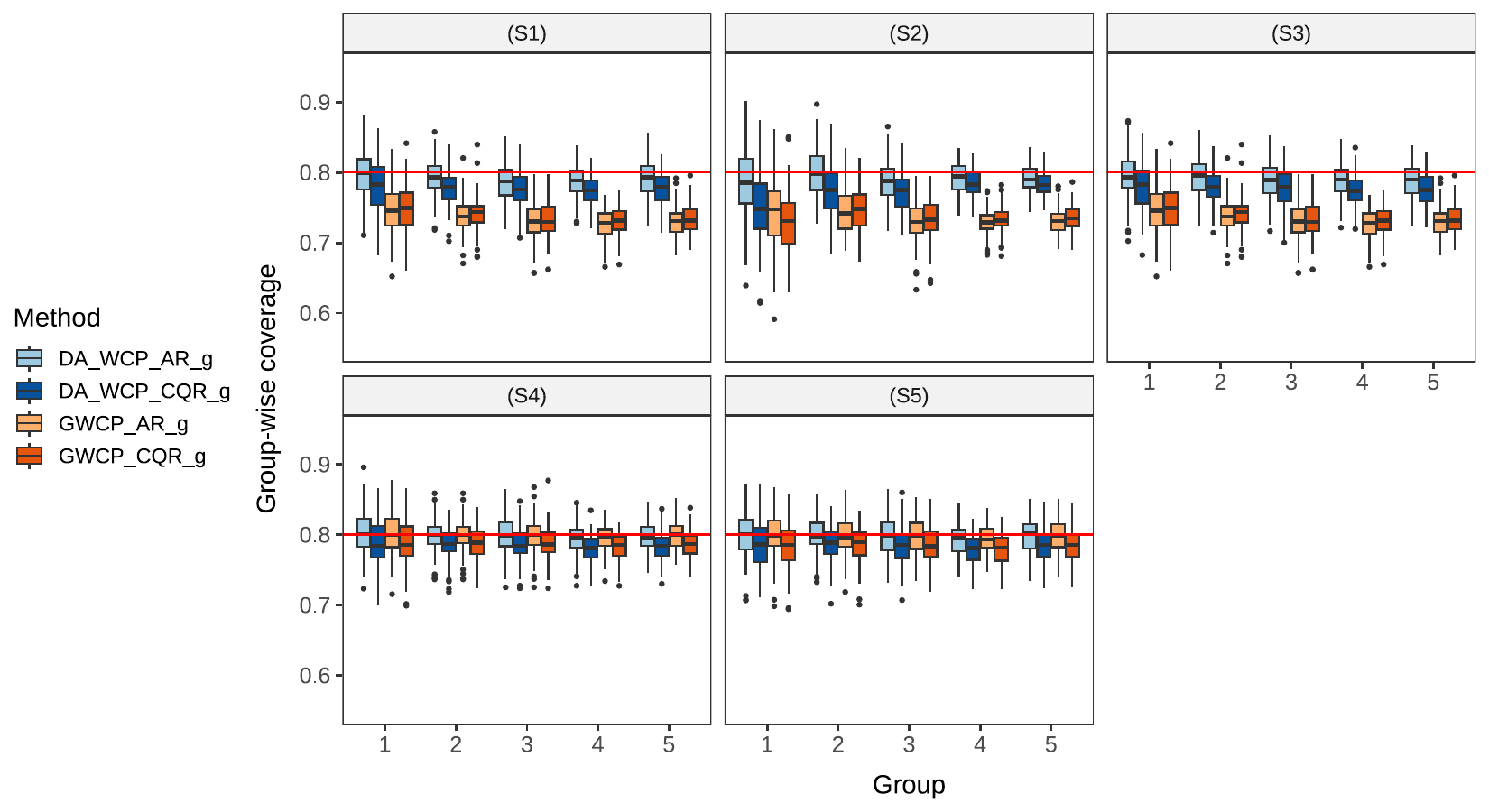}
    \caption{Group-wise coverage of group-conditional methods at the target level \(1 - \alpha = 0.8\), across scenarios (S1)--(S5) using 100 independently generated datasets.}
    \label{fig:simul-group}
\end{figure}

\section{Prediction of LDL-C using repeated NHANES surveys} \label{sec-ldl}

\subsection{Problem Setup and Objectives} \label{sec-ldl-prob}

We apply the proposed DA-WCP framework to construct population-level prediction sets for low-density lipoprotein cholesterol (LDL-C). We use NHANES 1999--2004 wave as the past ($t=1$) dataset and NHANES 2015--2020/03 wave as the present ($t=2$) dataset, with the target distribution corresponding to the $t=2$ population.

The continuous covariates, $\Xb_t \in \Rb^3$, are total cholesterol (TC), high-density lipoprotein cholesterol (HDL-C), and body mass index (BMI). Although triglycerides (TG) are biologically related to LDL-C, we exclude TG because it is measured only for participants who fasted, resulting in substantial missingness. In addition, TG exhibits greater biological variability than TC or HDL-C \citep{tg}, making it less reliable for prediction across survey waves.

The categorical demographic variables are $\Gb_t = (G_{t1}, G_{t2}, G_{t3})^\top \in [2] \times [4] \times [3]$, representing gender, age group, and race or ethnicity. 
Because LDL-C analysis typically targets adult populations with stable lipid profiles, we restrict the sample to individuals aged 20 years or older. 
Moreover, since only the race/ethnicity categories that appear in both waves should be used, we restrict our analysis to three groups: Hispanic, Non-Hispanic White, and Non-Hispanic Black.
These variables define the sampling domains in NHANES and ensure sufficient representation within each subgroup. Non-demographic categorical covariates $\Zb_t$ are excluded to avoid sparsity in the cross-classification of $(\Zb_t, \Gb_t)$, which may otherwise compromise reliable conditional estimation, and thus $\Zbt_t = \Gb_t$.

\begin{table}[!t]
\centering
\renewcommand{\arraystretch}{1.0} 
\caption{Baseline characteristics of variables in the NHANES 1999–2004 and 2015–2020/03 datasets. For demographic variables, percentages of each category are reported, while for continuous covariates, means are presented. Weighted estimates were computed using NHANES sampling weights to approximate the U.S. civilian, noninstitutionalized population.}
\label{table:nhanes}
\vspace{-2mm}
\resizebox{\textwidth}{!}{%
\begin{tabular}{
    ll
    S[table-format=3.1]
    S[table-format=3.1]
    S[table-format=3.1]
    S[table-format=3.1]
}
\toprule
\multicolumn{2}{c}{\textbf{Variable}} &
\multicolumn{2}{c}{\textbf{1999--2004}} &
\multicolumn{2}{c}{\textbf{2015--2020/03}} \\
\cmidrule(lr){3-4} \cmidrule(lr){5-6}
& &
{\makecell[c]{\textbf{U.S. Population}\\(Weighted NHANES)}} &
{\makecell[c]{\textbf{NHANES}\\(Unweighted)}} &
{\makecell[c]{\textbf{U.S. Population}\\(Weighted NHANES)}} &
{\makecell[c]{\textbf{NHANES}\\(Unweighted)}} \\
\midrule
\multirow{2}{*}{Gender (\%)} & Male & 48.8 & 48.8 & 48.9 & 49.4 \\
 & Female & 51.2 & 51.2 & 51.1 & 50.6 \\ 
\midrule
\multirow{5}{*}{Age (\%)} & 20--34 & 21.7 & 13.3 & 20.6 & 14.1 \\
 & 35--49 & 23.0 & 11.9 & 18.8 & 13.9 \\
 & 50--64 & 15.0 & 10.1 & 19.5 & 15.8 \\
 & $\ge$ 65 & 11.5 & 14.0 & 15.4 & 14.8 \\
 & Others & 28.8 & 50.7 & 25.6 & 41.4 \\
\midrule
\multirow{4}{*}{Race (\%)} & Hispanic & 14.8 & 32.6 & 18.0 & 26.5 \\
 & Non-Hispanic White & 67.8 & 38.9 & 60.0 & 32.7 \\
 & Non-Hispanic Black & 12.0 & 24.3 & 12.0 & 24.4 \\
 & Others & 5.4 & 4.2 & 10.1 & 16.5 \\
\midrule
\multicolumn{2}{l}{Total cholesterol (mg/dL)} & 193.4 & 186.3 & 183.9 & 178.6 \\ 
\multicolumn{2}{l}{HDL-C (mg/dL)} & 51.9 & 52.4 & 54.5 & 53.8 \\
\multicolumn{2}{l}{BMI (kg/$m^2$)} & 26.0 & 24.9 & 27.5 & 26.4 \\
\multicolumn{2}{l}{LDL-C (mg/dL)} & 114.3 & 109.2 & 108.2 & 106.1 \\
\bottomrule
\end{tabular}%
}
\end{table}

Table~\ref{table:nhanes} illustrates two sources of covariate shift between the training and target distributions.
First, population-level shifts between NHANES 1999–2004 and 2015–2020/03 are evident. 
For example, the weighted NHANES estimates show an increase in the Hispanic population and a decrease in non-Hispanic Whites. 
Second, NHANES employs a sampling design that intentionally oversamples specific demographic subgroups to improve estimation precision, as reflected in the differences between weighted and unweighted summaries.

Although the proposed method accommodates stratified and clustered sampling designs, the publicly available NHANES data lack the geographic identifiers needed to implement the stratified or PSU-level versions of the method. These variables are restricted for confidentiality and are available only through the NCHS Research Data Center. 
Accordingly, following the guideline in Section~\ref{sec-survey-1-1}, in the NHANES application we use a single-stratum ($H=1$) representation and analyze the data under the PPSWOR-within-a-single-stratum design. 

We apply our method to predict LDL-C at \(t=2\). Our method enables population-level predictive inference even when the \(t=2\) responses are entirely unavailable; however, evaluating empirical coverage requires access to the \(t=2\) responses. Accordingly, although LDL-C is measured in both waves, we treat the \(t=2\) response as unobserved when constructing prediction sets and use it only for evaluation. Detailed procedures on construction and evaluation are provided in appendix~\ref{app:eval}.

\subsection{\label{sec-ldl-conti} Conformal Prediction for Continuous Outcomes}

We apply the proposed DA-WCP method to construct prediction intervals for continuous LDL-C and compare its performance with GWCP. The four marginal methods listed in Table~\ref{table:methods} are evaluated. 
Figure~\ref{fig:conti-marg} presents the empirical coverage and average interval length of these methods at the target coverage level $1-\alpha=0.8$, based on 100 random splits of the training and calibration datasets. 
For implementation, we estimate the mean and quantile functions using weighted linear models with the sampling weights. 
With the absolute residual (AR) score, DA-WCP attains coverage closer to the target level, whereas GWCP tends to overcover; with the CQR score, both methods slightly overcover. Under both scores, DA-WCP yields shorter intervals and is therefore more efficient.
This difference likely arises from how the two procedures handle covariate shift: GWCP adjusts only for shifts in demographic covariates $\Gb_t = (G_{t1}, G_{t2}, G_{t3})$ (gender, age group, and race/ethnicity), while DA-WCP accounts for shifts in both categorical covariates and continuous covariates $\Xb_t \in \mathbb{R}^3$ (total cholesterol, HDL-C, and BMI). 
Because these clinical covariates are strongly associated with LDL-C and vary substantially across NHANES waves, accounting for shifts in both $\Xb_t$ and $\Gb_t$ enables DA-WCP to produce valid and efficient prediction intervals.

\begin{figure}[t]
  \centering
  \includegraphics[width=\linewidth]{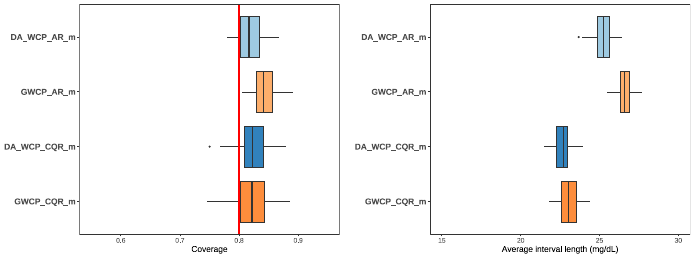}
  \caption{Empirical coverage (left) and average prediction interval length (right) for the four marginal methods, \texttt{DA\_WCP\_AR\_m}, \texttt{GWCP\_AR\_m}, \texttt{DA\_WCP\_CQR\_m}, and \texttt{GWCP\_CQR\_m}, under the target coverage $1 - \alpha = 0.8$, based on 100 random splits.}
  \label{fig:conti-marg}
\end{figure}

Moreover, for both DA-WCP and GWCP, using the CQR score yields shorter prediction intervals than using the AR score. This is because we select the pair $(\alo, \ahi)$ that yields the minimum average interval length, and the CQR score can capture more complex behavior of the interval length. 
When the error distribution is heteroscedastic,  
CQR allows the interval length to vary with $(\xb,\zbt)$, whereas AR yields an interval length that is constant and does not depend on $(\xb,\zbt)$, as summarized in Table~S1 in the appendix. Consequently, AR typically requires longer intervals to attain the desired coverage in such settings.

Results for the group-conditional methods are omitted because stratifying by demographic subgroup yields small within-subgroup sample sizes, leading to unstable estimates of subgroup-wise coverage. The corresponding results are provided in appendix~\ref{app:ldl-group}.

\subsection{\label{sec-ldl-categorical} Conformal Prediction for Categorical Outcomes}

In clinical practice, decisions involving LDL-C, such as when to initiate treatment, adjust medication, or schedule follow-up, are typically based on categorical thresholds rather than exact numerical values. These thresholds align with clinical guidelines and are less affected by assay variability, making categorical prediction more meaningful for decision-making.

We classify LDL-C levels according to the National Cholesterol Education Program (NCEP) guidelines: optimal ($<100$ mg/dL), near optimal ($100 \le$ LDL-C $<130$), borderline high ($130 \le$ LDL-C $<160$), high ($160 \le$ LDL-C $<190$), and very high ($\ge 190$ mg/dL). Conformal sets are constructed within this discrete label space. For an individual with covariates $(\xb, \zbt)$, the conformal set $\widehat{C}_\alpha(\xb, \zbt)$ is a subset of these five categories, and its cardinality $|\widehat{C}_\alpha(\xb, \zbt)|$ represents the number of categories included. This serves as an analogue to interval length in continuous prediction, where smaller sets indicate more precise and informative predictions. A single-category set provides a clear clinical interpretation, while a two-category set such as $\{\text{near optimal}, \text{borderline high}\}$ indicates uncertainty near a decision threshold but remains useful for communication. Larger sets spanning multiple categories offer limited value, as they cover a wide range of potential clinical states. Hence, the frequency distribution of $|\widehat{C}_\alpha(\xb, \zbt)|$ provides a practical measure of efficiency in categorical prediction.

We construct marginal prediction sets using DA-WCP and compare them with GWCP.
The soft classifiers employed are 
XGBoost, Random Forest, and Neural Network. 
For a target coverage level of $1-\alpha=0.8$, we compute the empirical coverage and the distribution of $|\widehat{C}_\alpha(\xb, \zbt)|$ on the test dataset, averaged over 100 random splits. 
The results are summarized in Table~\ref{tab:ldl-categorical}, and those for the group-conditional methods are provided in appendix~\ref{app:ldl-group}.

\begin{table}[t]
\centering
\renewcommand{\arraystretch}{1.0}
\caption{Empirical coverage and distribution of prediction set sizes across six methods, reported as mean (standard deviation). Proportions indicate the fraction of prediction sets containing exactly one, two, or three LDL-C categories.}
\label{tab:ldl-categorical}
\vspace{-2mm}

\begin{minipage}{\textwidth}
\resizebox{\linewidth}{!}{
\begin{tabular}{@{}lcccccc@{}}
\toprule
\textbf{Method} & \texttt{DA\_WCP\_XGB} & \texttt{DA\_WCP\_RF} & \texttt{DA\_WCP\_NN} & \texttt{GWCP\_XGB} & \texttt{GWCP\_RF} & \texttt{GWCP\_NN} \\
\midrule
Coverage & 0.841 (0.022) & 0.834 (0.021) & 0.840 (0.021) & 0.858 (0.022) & 0.864 (0.020) & 0.861 (0.021) \\
\midrule
Size = 1 & 0.872 (0.015) & 0.916 (0.012) & 0.956 (0.014) & 0.831 (0.015) & 0.851 (0.014) & 0.912 (0.015) \\
Size = 2 & 0.125 (0.015) & 0.084 (0.012) & 0.044 (0.014) & 0.163 (0.014) & 0.149 (0.014) & 0.088 (0.015) \\
Size = 3 & 0.003 (0.001) & 0.000 (0.000) & 0.000 (0.000) & 0.005 (0.002) & 0.000 (0.000) & 0.000 (0.000) \\
\bottomrule
\end{tabular}
}
\end{minipage}

\end{table}

Across all methods, both DA-WCP and GWCP slightly exceed the nominal coverage level. 
However, DA-WCP consistently attains coverage closer to the target $1 - \alpha$ across all classifiers, indicating more accurate empirical coverage.
In other words, the observed frequency with which the true category falls within the prediction set aligns more closely with the nominal level.
Consequently, DA-WCP produces smaller and thus more informative prediction sets than GWCP across all classifiers.
Among the soft classifiers, \texttt{nnet} yields the highest proportion of single-category prediction sets, reflecting tighter and more decisive predictions, followed by \texttt{ranger} and \texttt{xgboost}, while all three exhibit comparable empirical coverage.

Overall, these findings empirically demonstrate that DA-WCP yields more accurate predictions than GWCP for all soft classifiers, indicating that accounting for covariate shift in both $\Xb_t$ and $\Zbt_t$ is crucial even when the outcome is categorical.

\section{Conclusion}\label{sec-conc}

In this paper, we proposed a methodology for constructing population-level prediction sets at the current time point using past and present complex survey data, where the outcomes are completely missing in the present dataset.  
Our approach applies the WCP framework \citep{wcp} to complex survey data collected at different time points under minimal assumptions. 

The main novelty of our work lies in applying conformal prediction to complex survey data collected at two different time points.  
Since survey data are samples from finite populations that change over time, prior studies such as \citet{designcp} have considered conformal prediction only within a single survey wave. In contrast, this is the first study to extend conformal prediction to a multi-wave setting.  
In addition, unlike GWCP by \citet{gwcp}, which models covariate shift only in subgroup variables, our method accommodates shifts in both categorical and continuous covariates, resulting in empirical coverage consistently close to the nominal level in applications such as LDL-C prediction.
  
The validity of our approach depends on how well the underlying assumptions hold in practice. In Assumption~\ref{ass:super}, we abstract away from dependence induced by multi-stage cluster sampling; substantial within-cluster dependence may affect the resulting coverage guarantees. The MAR assumption in Assumption~\ref{ass:nonresponse} may also be violated in practice; for example, individuals with high BMI may skip LDL-C examination, leading to an MNAR mechanism. Extending conformal prediction methods to fully general complex survey designs and to settings with MNAR nonresponse remains an important direction for future research.

Beyond population-level predictive inference, the proposed framework can also be applied to missing-data imputation. Specifically, by treating the present data with observed covariates and missing outcomes as the target distribution, one can obtain valid prediction sets via a nearly identical procedure.

\bibliographystyle{abbrvnat}
\bibliography{references}

\clearpage
\begin{center}
{\LARGE\bfseries
\parbox{0.9\textwidth}{\centering
Supplement to ``Predicting Current Outcomes From Historical Survey Data With Weighted Conformal Prediction''
}
}
\end{center}

\vspace{8mm}
\begin{center}
{\Large\bfseries Table of Contents}
\end{center}

\vspace{1mm}
\begin{center}
\begin{minipage}{0.9\textwidth}
\noindent\hyperref[app:kliep]{S1\quad KLIEP: Method and Theoretical Results}\\[0.35em]
\hyperref[app:thm1]{S2\quad Proof of Theorem 1}\\[0.35em]
\hyperref[app:thm2]{S3\quad Proof of Theorem 2}\\[0.35em]
\hyperref[app:smallh]{S4\quad Refined Coverage Bound for Small $H$}\\[0.35em]
\hyperref[app:theoretical-group]{S5\quad Theoretical Results for Group-Conditional Methods}\\[0.35em]
\hyperref[app:proof-2]{S6\quad Useful Lemmas}\\[0.35em]
\hyperref[app:threshold]{S7\quad Details on the WCP-type Threshold}\\[0.35em]
\hyperref[app:implement]{S8\quad Implementation Details}\\[0.35em]
\hyperref[app:reduction]{S9\quad Reduction of Simplified Designs}\\[0.35em]
\hyperref[app:geo]{S10\quad Extension to Geographic Sampling Domains}\\[0.35em]
\hyperref[app:partial]{S11\quad When Outcomes are Partially Observed at $t = 2$}\\[0.35em]
\hyperref[app:eval]{S12\quad Computing Empirical Coverage in Real Data}\\[0.35em]
\hyperref[app:coverage]{S13\quad Coverage Properties of GWCP Relative to DA-WCP}\\[0.35em]
\hyperref[app:simul-setup]{S14\quad Simulation Setup}\\[0.35em]
\hyperref[app:simul]{S15\quad Additional Simulation Results}\\[0.35em]
\hyperref[app:ldl-group]{S16\quad LDL-C Prediction via Group-Conditional Methods}
\end{minipage}
\end{center}

\bigskip
\clearpage

\setcounter{section}{0}
\setcounter{subsection}{0}
\setcounter{subsubsection}{0}
\setcounter{figure}{0}
\setcounter{table}{0}
\setcounter{equation}{0}
\setcounter{apxlemma}{0}
\setcounter{apxass}{0}
\setcounter{apxthm}{0}
\setcounter{apxcor}{0}
\setcounter{apxprop}{0}

\renewcommand{\thesection}{S\arabic{section}}
\renewcommand{\thesubsection}{\thesection.\arabic{subsection}}
\renewcommand{\thesubsubsection}{\thesubsection.\arabic{subsubsection}}
\renewcommand{\thefigure}{S\arabic{figure}}
\renewcommand{\thetable}{S\arabic{table}}
\numberwithin{equation}{section}

\makeatletter
\renewcommand{\theHsection}{supp.\arabic{section}}
\renewcommand{\theHsubsection}{supp.\arabic{section}.\arabic{subsection}}
\renewcommand{\theHsubsubsection}{supp.\arabic{section}.\arabic{subsection}.\arabic{subsubsection}}
\renewcommand{\theHfigure}{supp.\arabic{figure}}
\renewcommand{\theHtable}{supp.\arabic{table}}
\renewcommand{\theHequation}{supp.\arabic{section}.\arabic{equation}}
\makeatother

\section{KLIEP: Method and Theoretical Results} \label{app:kliep}

Here we provide a brief review of the density ratio estimation method proposed by \citet{Sugiyama2007, Sugiyama2008}, and present theoretical results from the original paper along with our extensions to the multivariate setting. The method is called the \textit{Kullback--Leibler Importance Estimation Procedure} (KLIEP).

\subsection{The KLIEP Method} \label{app:kliep-1}
We briefly describe the rationale behind the estimation of the density ratio function \(r_\zbt\) using the KLIEP method.
The notation follows that used throughout the main article.
For each $\zbt \in \k$, let $f_{1\zbt}$ and $f_{2\zbt}$ denote the true densities of $P_{\Xb_1 \mid \Zb_1 = \zbt}$ and $P_{\Xb_2 \mid \Zb_2 = \zbt}$, respectively, and define
\[
\ittraz := \{i \in \ittra \mid \Zbt_{1i} = \zbt\}, \qquad \itwoz := \{i \in \itwo \mid \Zbt_{2i} = \zbt\}.
\]
Then, our goal is to estimate the \textit{importance}, or density ratio function
\[
r_\zbt(\xb) := \frac{f_{2\zbt}(\xb)}{f_{1\zbt}(\xb)}, \qquad  \forall \xb \in \mathrm{supp}(f_{1\zb}),
\]
using $\{\Xb_{1i}\}_{i \in \ittraz}$ and $\{\Xb_{2i}\}_{i \in \itwoz}$, without relying on separate estimation of the density functions.
Denote \(\nttrz = |\ittraz|\) and \(\ntez = |\itwoz|\).
Now we model the importance function $r_\zbt(\cdot)$ using a set of basis functions
$\mathcal{F}_\zbt := \{\varphi_\theta : \theta\in\Theta_\zbt\}$,
where $\Theta_\zbt$ is an index set.
Let $\nz := (\nttrz, \ntez)$.
Given $\nz$, we allow the selected index set
$\Theta_{\zbt}(\nz)\subseteq \Theta_\zbt$ to depend on the data, and define
\[
\mathcal{F}_{\zbt}(\nz)
:= \{\varphi_\theta : \theta\in\Theta_\zbt(\nz)\},
\]
and thus $\mathcal{F}_{\zbt}(\nz)$ may depend on the data. The set of finite linear combinations of the basis functions in $\mathcal F_\zbt$ with nonnegative coefficients and its bounded subset are denoted by
\begin{align*}
    \mathcal G_\zbt := \left\{ \sum_l \alpha_l \varphi_{\theta_l} \Bigm| \alpha_l \ge 0, \varphi_{\theta_l} \in \mathcal F_\zbt \right\}, \qquad \mathcal G^M_\zbt := \{g \in \mathcal{G}_\zbt \mid \lVert g \rVert_\infty \le M\},
\end{align*}
for $M > 0$, and their subsets at $\nz$ samples are denoted by
\begin{align*}
    \mathcal{G}_\zbt(\nz) &:= \left\{ \sum_l \alpha_l \varphi_{\theta_l} \Bigm| \alpha_l \ge 0, \ \varphi_{\theta_l} \in \mathcal F_{\zbt}(\nz) \right\} \subseteq \mathcal G_\zbt, \\
    \mathcal{G}^M_\zbt(\nz) &:= \left\{g \in \mathcal{G}_\zbt(\nz) \mid \lVert g \rVert_\infty \le M\right\} \subseteq \mathcal G^M_\zbt.
\end{align*}
Given a model \(\rnz \in \gnz\) with parameters $\{\widehat \alpha_l\}$,
an estimate of the target density \(\ptez\) is expressed as $\hatptez(\xb) = \rnz(\xb)\ptrz(\xb)$.
The KLIEP method then determines the parameters $\{\widehat \alpha_{l}\}$ by minimizing the Kullback-Leibler (KL) divergence from $\ptez$ to $\hatptez$:
\begin{align*}
D_\text{KL}(\ptez \ || \ \hatptez) &= \int_{\mathcal{X}} \ptez(\xb) \log \frac{\ptez(\xb)}{\rnz(\xb)\ptrz(\xb)} d\xb \\
&= \int_{\mathcal{X}} \ptez(\xb) \log \frac{\ptez(\xb)}{\ptrz(\xb)} d\xb - \int_{\mathcal{X}} \ptez(\xb) \log \rnz(\xb) d\xb.
\end{align*}
Since the first term is independent of $\rnz$, one can simply maximize the following:
\begin{equation} \label{eqn:kliep-1-1}
\int_{\mathcal{X}} \ptez(\xb) \log \rnz(\xb) d\xb \approx \frac{1}{\ntez}\sum_{i \in \itwoz} \log \rnz(\Xb_{2i}).
\end{equation}
Moreover, $\rnz$ should be properly normalized since $\hatptez$ is a probability density function:
\begin{equation} \label{eqn:kliep-1-2}
1 = \int_{\mathcal{X}} \hatptez(\xb)d\xb = \int_{\mathcal{X}} \rnz(\xb)\ptrz(\xb)d\xb \approx \frac{1}{\nttrz}\sum_{i \in \ittraz} \rnz(\Xb_{1i}).
\end{equation}
Combining \eqref{eqn:kliep-1-1} and \eqref{eqn:kliep-1-2}, since $\rnz \in \gnz$, the optimization criterion for the KLIEP method can be summarized as follows:
\begin{align*}
\underset{g \in \gnz}{\text{maximize}}  \sum_{i \in \itwoz} \log g(\Xb_{2i}) \quad  \text{subject to} \quad  \frac{1}{\nttrz}\sum_{i \in \ittraz}g(\Xb_{1i}) = 1.
\end{align*}
Now let $\hatgnz$ be the feasible set of the optimization criterion:
\[
\hatgnz := \Bigg\{g \in \gnz \Bigm| \frac{1}{\nttrz} \sum_{i \in \ittraz} g(\Xb_{1i}) = 1\Bigg\}.
\]
Under this notation, the solution $\rnz$ is given as follows:
\begin{equation} \label{eqn:kliep-1-3}
\rnz := \argmax_{g \in \hatgnz} \frac{1}{\ntez} \sum_{i \in \itwoz} \log g(\Xb_{2i}).
\end{equation}
For KLIEP implementation, \citet{Sugiyama2007} employed a Gaussian kernel model centered at the input points $\{\Xb_{2i}\}_{i \in \itwoz}$, which is also used in our framework:
\[
\gnz = \left\{ \sum_{\ell=1}^b \alpha_\ell K_\sigma(\cdot, \bm{c}_\ell) \Bigm| \alpha_\ell \ge 0 \right\},
\]
where $K_\sigma(\xb, \xb') = \exp(-\lVert \xb - \xb' \rVert^2 / 2\sigma^2)$ denotes a Gaussian kernel with bandwidth $\sigma$, $\{\bm{c}_\ell\}_{\ell = 1}^b$ are template points randomly selected from $\{\Xb_{2i}\}_{i \in \itwoz}$, and $b$ is a fixed number, for which the original paper used $b = \min(100, \ntez)$.
The \textsc{R} code used for implementing the KLIEP method throughout this paper is available at \url{https://github.com/chihoonlee-snu/DA-WCP}.

\subsection{Theoretical Results From the KLIEP Method} \label{app:kliep-2}

Here we present multivariate extensions of the theoretical results for the KLIEP method of \citet{Sugiyama2008}. 
In particular, we adapt their analysis to our setting with $d$-dimensional covariates and provide proofs for the extended results. We begin by introducing a set of assumptions from the original paper, tailored to our context where the variables are $d$-dimensional.

\begin{apxass}[cf. \citet{Sugiyama2008}] \label{ass:kliep-1}
    For each $\zbt \in \k$, 
    \begin{enumerate}
    \item[(A1)] $P_{1\zbt} := P_{\Xb_1 \mid \Zb_1 = \zbt}$ and $P_{2\zbt} := P_{\Xb_2 \mid \Zb_2 = \zbt}$ are mutually absolutely continuous, and there exist universal constants $\eta_0, \eta_1, \eta_2 > 0$ such that
    \[
    0 < \eta_0 \le r_\zbt = \frac{dP_{2\zbt}}{dP_{1\zbt}} \le \eta_1,
    \qquad 
    0 < \eta_2 \le f_{2\zbt},
    \]
    on their common support. 
    In addition, there exists a universal constant $a>0$ such that $\mathrm{supp}(P_{1\zbt}) = \mathrm{supp}(P_{2\zbt})$ is a compact, convex subset of $[-a,a]^d$ with nonempty interior.

    \item[(A2)] $\varphi_\theta \ge 0$ for all $\varphi_\theta \in \mathcal F_\zbt$, and there exist universal constants $\epsilon_0, \xi_0 > 0$ such that
    \[
    \int_{\mathcal{X}} \varphi_\theta(\xb) \ptrz(\xb)d\xb \ge \epsilon_0, \ \ \lVert \varphi_\theta \rVert_\infty \le \xi_0, \qquad \forall \varphi_\theta \in \mathcal F_\zbt.
    \]
    
    \item[(A3)] For the Gaussian RBF kernel with parameter \(\sigma > 0\),
    \[
    K_\sigma(\xb, \xb') = \exp\left(-\frac{\lVert \xb - \xb' \rVert^2}{2\sigma^2}\right), \qquad \forall  \xb, \xb' \in \mathbb{R}^d,
    \]
    \(r_\zbt\) is a mixture of Gaussian RBF kernels:
    \[
    r_\zbt(\xb) = \int_{\mathcal{X}} K_{\sigma_\zbt}(\xb, \xb') \, dF_\zbt(\xb'), \qquad \forall \xb \in \mathrm{supp}(P_{1\zbt}),
    \]
    for some \(\sigma_\zbt > 0\) and a positive finite measure \(F_\zbt\) whose support is contained in \(\mathrm{supp}(P_{1\zbt})\).
    \end{enumerate}
\end{apxass}

In Assumption~\ref{ass:kliep-1}, conditions (A1) and (A2) correspond to Assumption~1 in \citet{Sugiyama2008}, with additional requirements imposed on the common support. Condition (A3) is motivated by Example~1 of the same paper, which assumes that the true density ratio is a location mixture of normals.
We now state Theorem~\ref{thm:kliep-1}, which establishes the convergence rate of the estimated density ratio $\rnz$ to the true ratio $r_\zbt$ in Hellinger distance under $P_{1\zbt}$.  
The proof is inspired by the argument of Example~1 in \citet{Sugiyama2008}, but it requires a nontrivial modification. 
In particular, we extend the result to the multivariate setting and replace the rectangular support assumption in Example~1 with the compact and convex common support condition in (A1) of Assumption~\ref{ass:kliep-1}. 
Accordingly, we adapt the proof strategy to this setting and present the argument in full.

\begin{apxthm} \label{thm:kliep-1}
    Under Assumption~\ref{ass:kliep-1}, for each $\zbt \in \k$, and conditional on 
    $\nttrz = |\ittraz|$ and $\ntez = |\itwoz|$, for $\rnz$ defined in \eqref{eqn:kliep-1-3}, we have
    \[
    h_{P_{1\zbt}}(r_\zbt, \rnz) = \begin{cases}
        O_p(\min(\nttrz, \ntez)^{-\frac{1}{2+\gamma}})  & \text{if} \ \ d = 1, 2,\\
         O_p\left(\left(\frac{\log\min(\nttrz, \ntez)}{\min(\nttrz, \ntez)}\right)^{\frac{1}{d}}\right)     & \text{if} \ \ d \ge 3,
    \end{cases}
    \]
    for any sufficiently small $\gamma > 0$, where $h_{P_{1\zbt}}(r_\zbt, \rnz)$ is a (generalized) Hellinger distance
    \[
    h_{P_{1\zbt}}(r_\zbt, \rnz) = \left(\int \left(\sqrt{r_\zbt(\xb)} - \sqrt{\rnz(\xb)}\right)^2 dP_{1\zbt}(\xb) \right)^{\frac{1}{2}}.
    \]
\end{apxthm}

\begin{proof}[Proof of Theorem~\ref{thm:kliep-1}]
    We begin by verifying that the assumptions required for Theorem~2 of \citet{Sugiyama2008} are satisfied. 
    First, conditions (A1) and (A2) of Assumption~\ref{ass:kliep-1} correspond to the first two requirements of Assumption~1 in \citet{Sugiyama2008}. 
    For the third requirement of the same assumption, a multivariate extension of Theorem~3.1 in \citet{ghosal01}, which follows directly from Lemma~\ref{lem:ghosal}, yields
    \[
    \log N(\epsilon, \mathcal G^M_\zbt, \lVert \cdot \rVert_\infty) 
    \le K \left(\log \frac{M}{\epsilon}\right)^{d+1},
    \]
    for any $0 < \epsilon < 1/2$ and $M > 1$, for some constant $K > 0$, where $\mathcal F_\zbt := \{K_{\sigma_{\zbt}}(\cdot, \xb') \mid \xb'\in\mathrm{supp}(P_{1\zbt})\}$.
    Thus, for sufficiently small $\epsilon > 0$ and any $\gamma > 0$, we have
    \[
    \log N_{[]}\big(\epsilon, \mathcal G^M_\zbt, L_2(P_{1\zbt})\big) 
    \le K\left(\frac{M}{\epsilon}\right)^\gamma, 
    \qquad \forall\zbt \in \k,
    \]
    and therefore the third requirement is also satisfied.
    Using the notation in Appendix~\ref{app:kliep-1}, it then remains to show that there exists \(g_{\zbt}^* \in \hatgnz\) such that \(r_\zbt/g_{\zbt}^* \le c_1^2\) for some constant \(c_1\).
    For any $\sigma > 0$ and any positive finite measure $F'$ on $[-a, a]^d$, define
    \begin{equation} \label{eqn:kliep-2-0}
    g_{F', \sigma}(\xb) := \int K_\sigma(\xb, \xb')dF'(\xb'), \qquad \forall \xb \in \Rb^d.
    \end{equation}
    Then, by (A3) of Assumption \ref{ass:kliep-1}, $r_\zbt = g_{F_\zbt, \sigma_\zbt}$ for each $\zbt \in \k$. By Lemma~\ref{lem:ghosal}, for any $0 < \epsilon_{\nz} < 1/2$, 
    there exists a discrete finite positive measure $F'_\zbt$ satisfying
    \begin{equation} \label{eqn:kliep-2-1}
        \| g_{F_\zbt,\sigma_\zbt} - g_{F'_\zbt,\sigma_\zbt} \|_\infty \lesssim \epsilon_{\nz},
        \quad
        F'_\zbt([-a,a]^d) = F_\zbt([-a,a]^d),
        \quad
        \mathrm{supp}(F'_\zbt) \subseteq \mathrm{supp}(F_\zbt).
    \end{equation}
    Without loss of generality, assume that \(\lVert g_{F_\zbt, \sigma_\zbt} - g_{F'_\zbt, \sigma_\zbt} \rVert_\infty \le \epsilon_{\nz}\) and $\sigma_\zbt = 1$. 
    Now, write $\mathcal X_{\zbt} = \mathrm{supp}(P_{1\zbt}) = \mathrm{supp}(P_{2\zbt})$, and define 
    $M_{\zbt}(\nz) := M(\mathcal X_{\zbt}, \|\cdot\|_2, \epsilon_{\nz}/2)$ as its 
    $\epsilon_{\nz}/2$-packing number, which is finite because $\mathcal X_{\zbt}$ is compact.  
    Let $\{\xb_1,\ldots,\xb_{M_{\zbt}(\nz)}\}$ be a corresponding maximal packing set. 
    For each $j = 1,\ldots, M_{\zbt}(\nz)$, define the Voronoi cell
    \[
    V_j := \left\{
    \xb \in \mathcal X_{\zbt} :
    \|\xb-\xb_j\|_2 \le \|\xb-\xb_s\|_2 \ \text{for all } s \in \{1,\dots,M_\zbt(\nz)\}\setminus\{j\}
    \right\}.
    \]
    Then $\mathrm{diam}(V_j) \le \epsilon_{\nz}$ for all $j$, and $\{V_j\}_{j=1}^{M_{\zbt}(\nz)}$ forms a partition of $\mathcal X_{\zbt}$.
    Next, for each $j = 1,\ldots, M_{\zbt}(\nz)$ and $i \in \itwoz$, define $R_{ij} := 1\{\Xb_{2i} \in V_j\}$. For each fixed $j$, the random variables $\{R_{ij} : i \in \itwoz\}$ are independent Bernoulli. Applying the Chernoff bound yields, for any $\delta>0$,
    \[
    \Pb\left(\sum_{i \in \itwoz} R_{ij} < (1-\delta)\sum_{i \in \itwoz} \Eb[R_{ij}]\right)
    < \exp\left(-\frac{\delta^2}{2} \sum_{i \in \itwoz} \Eb[R_{ij}]\right).
    \]
    By Lemma~\ref{lem:thickness}, there exist $r_0>0$ and $c_0>0$ such that, whenever $\epsilon_{\nz} \le r_0$,
    \begin{equation} \label{eqn:thick}
    \lambda(V_j)
    \ge \lambda(B(\xb_j, \epsilon_{\nz}/4)\cap \mathcal X_{\zbt})
    \ge c_0\,\epsilon_{\nz}^d,
    \end{equation}
    where $\lambda$ denotes Lebesgue measure on $\Rb^d$, and the first inequality follows from the inclusion $B(\xb_j, \epsilon_{\nz}/4)\cap \mathcal X_{\zbt} \subseteq V_j$.
    Setting $\delta = 1/2$, we obtain
    \begin{align} \label{eqn:kliep-2-2}
    \Pb\left(\sum_{i \in \itwoz} R_{ij} < \frac{1}{2}\sum_{i \in \itwoz} \Eb[R_{ij}]\right)
    &< \exp\left(-\frac{1}{8} \sum_{i \in \itwoz} \Eb[R_{ij}]\right)
    \le \exp\left(-\frac{\eta_2 c_0 \min(\nttrz,\ntez)\epsilon_{\nz}^d}{8}\right),
    \end{align}
    since $\ntez \ge \min(\nttrz,\ntez)$ and $\Eb[R_{ij}] \ge \eta_2 \lambda(V_j)\ge \eta_2 c_0 \epsilon_{\nz}^d$ by (A1) of Assumption~\ref{ass:kliep-1} and \eqref{eqn:thick}.
    Now, define the event $W_\zbt(\nz)$ as follows:
    \[
    W_\zbt(\nz) := \left\{\max_{\xb \in \text{supp}(F'_\zbt)} \min_{i \in \itwoz} \ \lVert \xb - \Xb_{2i} \rVert_2 \le \epsilon_{\nz}\right\} \supseteq \left\{\forall 1 \le j \le M_\zbt(\nz), \ \sum_{i \in \itwoz} R_{ij} \ge 1 \right\}.
    \]
    Then, by \eqref{eqn:kliep-2-2}, if $\exp(-\eta_2c_0\min(\nttrz, \ntez)\epsilon_{\nz}^d/8)/\epsilon_{\nz}^d \rightarrow 0$ as $\min(\nttrz, \ntez) \rightarrow \infty$,
    \begin{align*}
    \Pb(W_\zbt(\nz)^\mathsf{c}) &\le  \Pb(\exists 1 \le j \le M_\zbt(\nz) \ \text{such that}  \sum_{i \in \itwoz} R_{ij} = 0) \\
    &\le \sum_{j = 1}^{M_\zbt(\nz)} \Pb\left(\sum_{i \in \itwoz} R_{ij} < \frac{1}{2}\sum_{i \in \itwoz} \Eb[R_{ij}] \right) \\
    &\lesssim \frac{1}{\epsilon_{\nz}^d}  \exp\left(-\frac{\eta_2c_0\min(\nttrz, \ntez)\epsilon_{\nz}^d}{8}\right) \rightarrow 0,
    \end{align*}
    since $M_\zbt(\nz) \lesssim 1/\epsilon_{\nz}^d$.
    Therefore, $\Pb(W_\zbt(\nz)) \rightarrow 1$. Moreover, one can show that for $\forall \xb_1, \xb_2, \xb_3 \in [-a, a]^d$, 
    $| K_1(\xb_1, \xb_2) - K_1(\xb_1, \xb_3)| 
    \lesssim \lVert \xb_2-\xb_3 \rVert_2/\sqrt{e} + \lVert \xb_2-\xb_3 \rVert_2^2/2$ because
    \begin{align*}
        |K_1(\xb_1, \xb_2) - K_1(\xb_1, \xb_3)| 
        &= \exp(-\lVert \xb_1-\xb_2 \rVert_2^2/2)\,
           |1-\exp((\lVert \xb_1-\xb_2 \rVert_2^2 - \lVert \xb_1-\xb_3 \rVert_2^2)/2)| \\
        &\le \exp(-\lVert \xb_1-\xb_2 \rVert_2^2/2)\cdot \exp(2da^2)\cdot 
           \bigl|\lVert \xb_1-\xb_2 \rVert_2^2 - \lVert \xb_1-\xb_3 \rVert_2^2\bigr|/2 \\
        &\lesssim \exp(-\lVert \xb_1-\xb_2 \rVert_2^2/2)\cdot 
           \sum_{i=1}^d (|x_{1i}-x_{2i}|\cdot|x_{2i}-x_{3i}| + |x_{2i}-x_{3i}|^2/2) \\
        &\le \exp(-\lVert \xb_1-\xb_2 \rVert_2^2/2)\cdot 
           [\lVert \xb_1-\xb_2 \rVert_2 \cdot \lVert \xb_2-\xb_3 \rVert_2 
           + \lVert \xb_2-\xb_3 \rVert_2^2/2] \\
        &\le \lVert \xb_2-\xb_3 \rVert_2/\sqrt{e} + \lVert \xb_2-\xb_3 \rVert_2^2/2.
    \end{align*}
    where the mean value theorem was applied in the first inequality, and the Cauchy–Schwarz inequality was used in the third. The last inequality holds since \(x\exp(-x^2/2) \le 1/\sqrt{e}\) for any \(x > 0\).
    On the event $W_{\zbt}(\nz)$, each $\xb' \in \mathrm{supp}(F'_\zbt)$ has a corresponding index 
    $i(\xb') \in \itwoz$ such that $\|\xb' - \Xb_{2i(\xb')}\|_2 \le \epsilon_{\nz}$. 
    Based on this, define a function $\tilde g_{\zbt}^*$ by  
    \[\tilde g_{\zbt}^*(\xb)
    = \sum_{\xb' \in \mathrm{supp}(F'_\zbt)} F'_\zbt(\{\xb'\})\, 
      K_1(\xb, \Xb_{2i(\xb')}), \qquad  \forall \xb \in \mathcal X_{\zbt},
    \]
    since $F'_\zbt$ is a discrete measure. Then, on $W_{\zbt}(\nz)$, we have
    \begin{equation} \label{eqn:kliep-2-3}
    \lVert \tilde g_{\zbt}^* - g_{F'_\zbt,1} \rVert_\infty
    \le 
    F'_\zbt([-a,a]^d)\,(\epsilon_{\nz}/\sqrt e + \epsilon_{\nz}^2/2)
    = O(\epsilon_{\nz}).
    \end{equation}
    Combining \eqref{eqn:kliep-2-1} with \eqref{eqn:kliep-2-3} yields 
    $\lVert \tilde g_{\zbt}^* - r_\zbt \rVert_\infty = O(\epsilon_{\nz})$.  
    Now define $g_{\zbt}^* := \tilde g_{\zbt}^*/(Q_n \tilde g_{\zbt}^*)$, where $Q_n \tilde g_\zbt^* = \sum_{i\in\ittraz} \tilde g_\zbt^*(\Xb_{1i})/\nttrz$.
    Then $g_{\zbt}^* \in \hatgnz$, and observe that 
    \[
    |1 - Q_n \tilde g_{\zbt}^*|
    = 
    |1 - Q_n(\tilde g_{\zbt}^* - g_{F'_\zbt,1} + g_{F'_\zbt,1} - r_\zbt + r_\zbt)|
    \le 
    O(\epsilon_{\nz}) + |1 - Q_n r_\zbt|
    = O_p(\epsilon_{\nz}+\nttrz^{-1/2}).
    \]
    Hence,
    \begin{equation} \label{eqn:kliep-2-4}
    \lVert g_{\zbt}^* - \tilde g_{\zbt}^* \rVert_\infty
    =
    \lVert  g_{\zbt}^* \rVert_\infty\, |1 - Q_n \tilde g_{\zbt}^*|
    =
    O_p(\epsilon_{\nz}+\nttrz^{-1/2}).
    \end{equation}
    Combining \eqref{eqn:kliep-2-3} and \eqref{eqn:kliep-2-4} gives
    \(
    \lVert g_{\zbt}^* - r_\zbt \rVert_\infty
    = O_p(\epsilon_{\nz}+\nttrz^{-1/2}).
    \)
    Now, for any fixed $c_1 > 1$, define the event
    $E_\zbt(\nz) := \{\lVert g_{\zbt}^* - r_\zbt \rVert_\infty \le (1-1/c_1^2)\eta_0\}$.
    Then, we have $\Pb(E_\zbt(\nz)) \to 1$ as $\min(\nttrz, \ntez) \to \infty$, and on $E_\zbt(\nz)$ we have
    \[
    g_{\zbt}^*(\xb) \ge r_\zbt(\xb) - (1-1/c_1^2)\eta_0 \ge r_\zbt(\xb) - (1-1/c_1^2)r_\zbt(\xb) = r_\zbt(\xb)/c_1^2, \qquad \forall \xb \in \mathcal X_{\zbt}.
    \]
    Thus, on $E_\zbt(\nz)$ we have $r_\zbt/g_\zbt^* \le c_1^2$ for any $c_1 > 1$, and hence all conditions of Theorem~2 in \citet{Sugiyama2008} are satisfied. Therefore,
    \begin{equation} \label{eqn:kliep-2-4-1}
        h_{P_{1\zbt}}(r_\zbt, \rnz) =
        O_p\!\left(\min(\nttrz,\ntez)^{-1/(2+\gamma)} 
        + h_{P_{1\zbt}}(r_\zbt, g_{\zbt}^*) \right),
    \end{equation}
    for any sufficiently small $\gamma > 0$.
    Now, let
    \[
    \epsilon_{\nz}
    =
    \left(
    \frac{10 \log\min(\nttrz,\ntez)}
         {\eta_2 c_0 \min(\nttrz,\ntez)}
    \right)^{1/d},
    \]
    so that $\exp(-\eta_2 c_0 \min(\nttrz,\ntez)\epsilon_{\nz}^d/8)
        /\epsilon_{\nz}^d \to 0$.
    Now, since $r_\zbt$ is bounded from below by (A1) of Assumption~\ref{ass:kliep-1} and $\lVert g_{\zbt}^* - r_\zbt \rVert_\infty = O_p(\epsilon_{\nz}+\nttrz^{-1/2})$, we obtain
    \begin{equation} \label{eqn:kliep-2-4-2}
    h_{P_{1\zbt}}(r_\zbt, g_{\zbt}^*)
    =
    O_p(\epsilon_{\nz}+\nttrz^{-1/2})
    =
    O_p\!\left(
    \nttrz^{-1/2}
    +
    \Bigl(\frac{\log\min(\nttrz,\ntez)}{\min(\nttrz,\ntez)}\Bigr)^{1/d}
    \right).
    \end{equation}
    Combining~\eqref{eqn:kliep-2-4-1} and \eqref{eqn:kliep-2-4-2} then yields
    \begin{align*}
    h_{P_{1\zbt}}(r_\zbt, \rnz)
    &=
    O_p\!\left(
    \min(\nttrz,\ntez)^{-1/(2+\gamma)}
    + \nttrz^{-1/2}
    + \Bigl(\frac{\log\min(\nttrz,\ntez)}{\min(\nttrz,\ntez)}\Bigr)^{1/d}
    \right),
    \end{align*}
    for any sufficiently small $\gamma>0$.
    Hence, if $d \ge 3$,
    \[
    h_{P_{1\zbt}}(r_\zbt, \rnz)
    =
    O_p\!\left(
    \Bigl(\frac{\log\min(\nttrz,\ntez)}{\min(\nttrz,\ntez)}\Bigr)^{1/d}
    \right),
    \]
    and if $1 \le d \le 2$, for any small $\gamma>0$,
    \[
    h_{P_{1\zbt}}(r_\zbt, \rnz)
    =
    O_p(\min(\nttrz,\ntez)^{-1/(2+\gamma)}).
    \]
\end{proof}

Using Theorem~\ref{thm:kliep-1}, we obtain the following proposition, which gives the convergence rate of the maximum $L^1$ loss of the estimated WCP weight function $\widehat w_h$ taken over all strata $h \in [H]$. Corollary~\ref{cor:coverage} then follows immediately from Theorem~\ref{thm:coverage} and Proposition~\ref{prop:weight}.

\begin{apxprop} \label{prop:weight}
    Define $\Dc := \dttra \cup \dtwo$. Conditional on $\nrawone$, $n_2$, and $\{n_{1h}\}_{h=1}^H$, under Assumptions~\ref{ass:super}--\ref{ass:invariant} and Assumption~\ref{ass:kliep-1}, if $\min_{1 \le h \le H} n_{1h} \gtrsim \tilde n_1/H$ and $\min_{1 \le h \le H} n_{1h} \lesssim n_2$, then as $n_{11}, \dots, n_{1H}, n_2 \to \infty$, the estimated WCP weight function $\widehat w_h$ satisfies
    \[
    \begin{adjustbox}{max width=\linewidth}
    $\displaystyle
    \max_{1 \le h \le H} \Eb\left[|\widehat w_h(\Xb_{1h}, \Zbt_{1h}) - w_h(\Xb_{1h}, \Zbt_{1h})| \Bigm| \Dc, \tilde n_1, n_2, \{n_{1h}\}_{h=1}^H  \right] = 
    \begin{cases}
        O_p\big((\min_{1 \le h \le H} n_{1h})^{-\frac{1}{2}}\big) & d = 0,\\
        O_p\big(\min(\tilde n_1, n_2)^{-\frac{1}{2+\gamma}}\big) & d = 1, 2, \\
        O_p\Big(\Bigl(\tfrac{\log\min(\tilde n_1, n_2)}{\min(\tilde n_1, n_2)}\Bigr)^{\frac{1}{d}}\Big) & d \ge 3,
    \end{cases}
    $
    \end{adjustbox}
    \]
    for any sufficiently small $\gamma > 0$, where $(\Xb_{1h}, \Zbt_{1h}) \sim P_{(\Xb_1, \Zb_1) \mid \Gb_1} \times P_{\Gb_{1h} \mid \rfull_{1h}=1}$ for each $h \in [H]$.
    The conditions $\min_{1\le h\le H} n_{1h} \gtrsim \tilde n_1/H$ and $\min_{1\le h\le H} n_{1h} \lesssim n_2$ are mild balance requirements ensuring that (i) no stratum at $t=1$ has a vanishingly small sample size relative to the average stratum sample size $\tilde n_1/H$, and (ii) the sample size at $t=2$ is not smaller than the smallest stratum size at $t=1$.
\end{apxprop}

\begin{proof}[Proof of Proposition~\ref{prop:weight}]
    Recall that in Section~4, for each $\zbt = (\gb^\top, \zb^\top)^\top \in \k$ and $h \in [H]$, we defined
    \[
    v_{1\zbt} := P_{\Zb_1 \mid \Gb_1 = \gb}(\zb), 
    \quad 
    v_{2\zbt} := P_{\Zb_2 \mid \Gb_2 = \gb}(\zb), 
    \quad 
    q_{1h\gb} := P_{\Gb_{1h} \mid \rfull_{1h} = 1}(\gb),
    \]
    and $p_{2\gb} = N_{2\gb}/N_2$ is assumed to be known. We estimate the unknown quantities as follows:
    \[
    \widehat v_{1\zbt} 
    = 
    \frac{\sum_{i \in \ittra} \mathbf 1\{\Zbt_{1i} = \zbt\}}
           {\sum_{i \in \ittra} \mathbf 1\{\Gb_{1i} = \gb\}}, 
    \quad 
    \widehat v_{2\zbt} 
    = 
    \frac{\sum_{i \in \itwo} \mathbf 1\{\Zbt_{2i} = \zbt\}}
           {\sum_{i \in \itwo} \mathbf 1\{\Gb_{2i} = \gb\}},
    \quad
    \widehat q_{1h\gb} = \frac{\sum_{i \in \itra} \mathbf 1\{\Gb_{1i} = \gb\}}{\ntra}.
    \]
    With the subgroup-specific density ratio function $\rnz$ estimated via the KLIEP method as in \eqref{eqn:kliep-1-3}, the estimated weight function in stratum $h$ is then given by
    \[
    \widehat w_h(\xb, \zbt) 
    = 
    \rnz(\xb)\cdot 
    \frac{\widehat v_{2\zbt}\, p_{2\gb}}
         {\widehat v_{1\zbt}\, \widehat q_{1h\gb}}, \qquad \forall \xb \in \Rb^d, \ \zbt \in \k.
    \]
    Then, since $\Zbt_{1h} = (\Gb_{1h}^\top, \Zb_{1h}^\top)^\top$ has distribution
    $P_{\Zb_1 \mid \Gb_1}\times P_{\Gb_{1h} \mid \rfull_{1h} = 1}$ for each $h \in [H]$, we have
    \begin{align*}
        &\Eb\left[|\widehat{w}_h(\Xb_{1h}, \Zbt_{1h}) - w_h(\Xb_{1h}, \Zbt_{1h})| \Bigm| \Dc\right] \nonumber \\
        &= \sum_{\zbt}  \Eb\left[ |\widehat{w}(\Xb_{1h}, \Zbt_{1h}) - w_h(\Xb_{1h}, \Zbt_{1h})| \Bigm| \Dc, \Zbt_{1h} = \zbt \right] \cdot v_{1\zbt} q_{1h\gb} \nonumber \\
        &\le \sum_{\zbt}  \Eb\left[ \Bigl|\rnz(\Xb_{1h})\cdot \frac{\widehat v_{2\zbt}p_{2\gb}}{\widehat v_{1\zbt}\widehat q_{1h\gb}} - r_\zbt(\Xb_{1h})\cdot \frac{ v_{2\zbt}p_{2\gb}}{v_{1\zbt} q_{1h\gb}}\Bigr| \Bigm| \Dc, \Zbt_{1h} = \zbt \right]  \nonumber \\
        &\le \sum_{\zbt}  \Eb\left[ \left|(\rnz(\Xb_{1h}) - r_\zbt(\Xb_{1h}))\cdot\frac{\widehat v_{2\zbt}p_{2\gb}}{\widehat v_{1\zbt}\widehat q_{1h\gb}}\right| +\left|r_\zbt(\Xb_{1h})\cdot \left(\frac{\widehat v_{2\zbt}p_{2\gb}}{\widehat v_{1\zbt}\widehat q_{1h\gb}}  - \frac{ v_{2\zbt}p_{2\gb}}{v_{1\zbt} q_{1h\gb}} \right)\right| \Bigm| \Dc, \Zbt_{1h} = \zbt \right]  \nonumber \\
        &= \sum_{\zbt} \Bigl\{\Eb\left[ |\rnz(\Xb_{1h}) - r_\zbt(\Xb_{1h})| \Bigm| \Dc, \Zbt_{1h} = \zbt \right] \cdot \frac{\widehat v_{2\zbt}p_{2\gb}}{\widehat v_{1\zbt}\widehat q_{1h\gb}} +
        p_{2\gb} \cdot \Bigl| \frac{\widehat v_{2\zbt}}{\widehat v_{1\zbt}\widehat q_{1h\gb}} - \frac{v_{2\zbt}}{v_{1\zbt}q_{1h\gb}}\Bigr|\Bigr\},
    \end{align*}
    where the last equality follows from $\Eb[r_\zbt(\Xb_{1h}) \mid \Zbt_{1h} = \zbt] = 1$ for all $\zbt$.
    Although we also condition on $\nrawone$, $n_2$, and $\{n_{1h}\}_{h=1}^H$, we suppress this conditioning in the notation to avoid overly long expressions.
    Moreover, since $\Xb_{1h}\mid \Zbt_{1h}=\zbt$ has the same conditional distribution $P_{\Xb_1\mid \Zbt_1=\zbt}$ for all $h \in [H]$, we obtain, for the maximum $L^1$ loss,
    \begin{equation}\label{eqn:bound-3-1}
    \begin{adjustbox}{max width=\linewidth}
    $\displaystyle
    \begin{aligned}
    &\max_{1 \le h \le H} \Eb\left[|\widehat w_h(\Xb_{1h}, \Zbt_{1h}) - w_h(\Xb_{1h}, \Zbt_{1h})| \Bigm| \Dc \right] \\
    &\le \sum_{\zbt} \Bigl\{ \underbrace{\Eb\left[ |\rnz(\Xb_{1}) - r_\zbt(\Xb_{1})| \Bigm| \Dc, \Zbt_{1} = \zbt \right] \cdot \frac{\widehat v_{2\zbt}p_{2\gb}}{\widehat v_{1\zbt} \min_{1 \le h \le H} \widehat q_{1h\gb}}}_{\mathrm{(I)}} + \underbrace{\max_{1 \le h \le H}  \Bigl| \frac{\widehat v_{2\zbt}}{\widehat v_{1\zbt}\widehat q_{1h\gb}} - \frac{v_{2\zbt}}{v_{1\zbt}q_{1h\gb}}\Bigr|}_{\mathrm{(II)}} \Bigr\},
    \end{aligned}
    $
    \end{adjustbox}
    \end{equation}
    where $\Xb_1\mid \Zbt_1=\zbt$ follows the distribution $P_{\Xb_1\mid \Zbt_1=\zbt}$ for all $\zbt$.
    For (I), the term trivially vanishes when $d = 0$.
    If $d \ge 1$, let $\nb := (\nz)_{\zbt \in \k}$, and define
    \begin{equation} \label{eqn:def-s}
    S_\nb := \bigcap_{\zbt \in \k} \snz,
    \end{equation}
    where $\snz$ is defined in Lemma~\ref{lem:kliep-1}. Then, from the proof of Lemma~\ref{lem:kliep-2}, on the event $S_\nb$,
    we have for each $\zbt \in \k$ that
    \begin{align} \label{eqn:bound-3-2}
    \Eb\left[|\rnz(\Xb_1) - r_\zbt(\Xb_1)|\Bigm| \Dc, \Zbt_1 = \zbt \right]
    &\le \int_{\mathcal X} |\rnz(\xb) - r_\zbt(\xb)| f_{1\zbt}(\xb) \, d\xb \nonumber \\
    &\le (\bar M^{1/2} + \eta_1^{1/2}) \, h_{P_{1\zbt}}(r_\zbt, \rnz),
    \end{align}
    since $\|\rnz\|_\infty \le \bar M$ on the event $\snz$ by Lemma~\ref{lem:kliep-1}.
    Moreover, the same lemma implies that 
    $\Pb(\snz \mid \nttrz, \ntez) \to 1$ as 
    $\nttrz, \ntez \to \infty$. 
    Now, conditional on $\tilde n_1$ and $\{n_{1h}\}_{h=1}^H$, we have
    \begin{align*}
    \nttrz \ge \ntrz 
    &= \sum_{h=1}^H \sum_{i \in \itrah} \mathbf 1\{\Zb_{1i} = \zb, \Gb_{1i} = \gb\} \\
    &> \sum_{h=1}^H \frac{n_{1h}v_{1\zbt}q_{1h\gb}}{4} \\
    &\gtrsim \frac{\tilde n_1 v_{1\zbt} \min_{1 \le h \le H}q_{1h\gb}}{4}
    \end{align*}
    with probability tending to one as $n_{11}, \dots, n_{1H} \to \infty$, since $\min_{1 \le h \le H} n_{1h} \gtrsim \tilde n_1/H$.
    Moreover, conditional on $n_2$, we have $\ntez > n_2v_{2\zbt}(\min_{1\le h \le H} q_{2h\gb})/2$ with probability tending to one as $n_2 \to \infty$.
    Therefore, $\min(\nttrz,\ntez) \gtrsim \min(\tilde n_1, n_2)$ with probability tending to one as $n_{11}, \dots, n_{1H}, n_2 \to \infty$.
    Therefore, we obtain
    \begin{equation} \label{eqn:app:kliep:theoretical-snp}
    \Pb(S_\nb \mid \tilde n_1, n_2, \{n_{1h}\}_{h=1}^H) \to 1
    \quad\text{as } \ n_{11}, \dots, n_{1H}, n_2 \to \infty.
    \end{equation}
    It then follows from Theorem~\ref{thm:kliep-1},  \eqref{eqn:bound-3-2}, and \eqref{eqn:app:kliep:theoretical-snp} that
    \begin{equation} \label{eqn:bound-3-3}
    \Eb\left[|\rnz(\Xb_1) - r_\zbt(\Xb_1)| \Bigm| \Dc, \Zbt_1 = \zbt \right] = 
    \begin{cases}
        0 & \text{if}  \ \ d = 0,\\
        O_p(\min(\tilde n_1, n_2)^{-\frac{1}{2+\gamma}})  & \text{if} \ \  d = 1, 2,\\
         O_p\left(\left(\frac{\log\min(\tilde n_1, n_2)}{\min(\tilde n_1, n_2)}\right)^{\frac{1}{d}}\right)     & \text{if} \ \  d\ge 3,
    \end{cases}
    \end{equation}
    for any sufficiently small $\gamma > 0$, and the same bound holds for (I) since 
    $\widehat v_{2\zbt}/(\widehat v_{1\zbt}  \min_{1 \le h \le H} \widehat q_{1h\gb}) = O_p(1)$ as $n_{11}, \dots, n_{1H}, n_2 \to \infty$.
    For (II), fix $h \in [H]$. If there exists a non-demographic categorical covariate $\Zb_{ti}$ (i.e., $l > s$), then using
    \[
    \widehat v_{1\zbt}\widehat q_{1h\gb} = v_{1\zbt} q_{1h\gb} + O_p(n_{1h}^{-1/2}), \qquad \widehat v_{2\zbt} = v_{2\zbt} + O_p(n_2^{-1/2}),
    \]
    the term is of order $O_p(\min(n_{1h}, n_2)^{-1/2})$.
    If $l = s$, the term simplifies to
    $|1/\widehat q_{1h\gb} - 1/q_{1h\gb}|$,
    which is of order $O_p(n_{1h}^{-1/2})$.
    In either case, since $\min_{1 \le h \le H} n_{1h} \lesssim n_2$, term (II) is of order $O_p((\min_{1 \le h \le H} n_{1h})^{-1/2})$, and it is also $O_p(\tilde n_1^{-1/2})$ since $\min_{1 \le h \le H} n_{1h} \gtrsim \tilde n_1/H$.
    Finally, combining \eqref{eqn:bound-3-1} and \eqref{eqn:bound-3-3} yields the desired result.
\end{proof}

\section{Proof of Theorem 1} \label{app:thm1}

{%
\renewcommand{\thetheorem}{1}
\begin{theorem}[The distribution of subgroup-wise sample sizes within each stratum]
For each $t \in \{1,2\}$ and $h \in [H]$, let $\mathcal D_{th}$ denote the subset of $\mathcal D_t$ (defined in Section~3.1) corresponding to stratum $h$, and let $n_{th} = |\mathcal D_{th}|$. For each $\gb \in \kg$, let $n_{th\gb}$ denote the sample size of subgroup $\gb$ in stratum $h$ at time $t$. Under Assumptions~\ref{ass:super}--\ref{ass:multinomial}, we have
\[
(n_{th\gb})_{\gb} \mid n_{th} \sim \mathrm{Mult}\big(n_{th}, (q_{th\gb})_{\gb}\big),
\]
where $q_{1h\gb} = N_{1h\gb}\pi_{1\gb}a_{1\gb}^y/\sum_{\gb'} N_{1h\gb'}\pi_{1\gb'}a_{1\gb'}^y$ and $q_{2h\gb} = N_{2h\gb}\pi_{2\gb}a_{2\gb}^{xz}/\sum_{\gb'} N_{2h\gb'}\pi_{2\gb'}a_{2\gb'}^{xz}$.
\end{theorem}
}

\begin{proof}[Proof of Theorem 1]
Fix a stratum $h \in [H]$. Recall that $\mathcal D_{1h}$ and $\mathcal D_{2h}$ consist of units with $\rfull_{1i} = 1$ and $\rcov_{2i} = 1$, respectively. Therefore, by Assumption~\ref{ass:nonresponse}, we have
\begin{equation} \label{eqn:app:thm1-ng}
n_{1h\gb} \mid m_{1h\gb} \sim B(m_{1h\gb}, a_{1\gb}^y), \qquad 
n_{2h\gb} \mid m_{2h\gb} \sim B(m_{2h\gb}, a_{2\gb}^{xz}),
\end{equation}
for each $\gb \in \kg$. Combining Assumption~\ref{ass:multinomial} with \eqref{eqn:app:thm1-ng} yields
\[
(n_{1h\gb})_\gb \mid n_{1h} \sim \mathrm{Mult}\left(n_{1h}, \Big(\frac{\tilde q_{1h\gb} a_{1\gb}^y}{\sum_{\gb'} \tilde q_{1h\gb'} a_{1\gb'}^y}\Big)_{\gb} \right), \quad 
(n_{2h\gb})_\gb \mid n_{2h} \sim \mathrm{Mult}\left(n_{2h}, \Big(\frac{\tilde q_{2h\gb} a_{2\gb}^{xz}}{\sum_{\gb'} \tilde q_{2h\gb'} a_{2\gb'}^{xz}}\Big)_\gb \right),
\]
and the theorem follows immediately.
\end{proof}

\section{Proof of Theorem 2} \label{app:thm2}

{%
\renewcommand{\thetheorem}{2}
\begin{theorem}[Lower bound on the coverage probability]
    Define $\Dc := \dttra \cup \dtwo$, where \(\dttra\) is given in Section~\ref{sec-method}.
    Let $\nrawone := |\irawone|$, $n_2 := |\itwo|$, and $n_{1h} := |\ioneh|$ for each $h \in [H]$.
    Under Assumptions~\ref{ass:super}--\ref{ass:invariant} and Assumption~\ref{ass:kliep-1}, the selected prediction set $\widehat C_{\alpha, \hat \ell}$ satisfies the following.
    There exists an event $\mathcal E(\Dc)$ that depends only on $\Dc$ such that
    $\Pb(\mathcal E(\Dc) \mid \tilde n_1, n_2, \{n_{1h}\}_{h=1}^H ) \to 1$
    as $n_{11}, \dots, n_{1H}, n_2 \to \infty$, and, on $\mathcal E(\Dc)$,
    \[
    \begin{adjustbox}{max width=\linewidth}
    $\displaystyle
    \begin{aligned}
        &\Pb\left(Y_2 \in \widehat{C}_{\alpha, \hat \ell}(\Xb_2, \Zbt_2) \Bigm| \Dc, \tilde n_1, n_2, \{n_{1h}\}_{h=1}^H \right)  \ge 1-\alpha - \Big(2 + C_1\big(1 +\sqrt{\log(2HL)}\big)\Big) \\ 
        &\times \Bigg(\max_{1 \le h \le H} \Eb\left[|\widehat w_h(\Xb_{1h}, \Zbt_{1h}) - w_h(\Xb_{1h}, \Zbt_{1h})| \Bigm| \Dc, \tilde n_1, n_2, \{n_{1h}\}_{h=1}^H  \right] + 
    \max_{\zbt, h} \frac{v_{2\zbt} p_{2\gb}}{v_{1\zbt} q_{1h\gb}} \cdot \frac{C_2}{\sqrt{\min_{h} n_{1h}}}  \Bigg)
    \end{aligned}
    $
    \end{adjustbox}
    \]
    for some universal constants $C_1, C_2> 0$, where $q_{1h\gb}$ is defined in Theorem~1 and $v_{1\zbt}$ and $v_{2\zbt}$ are defined in Section~4.
    Here, $(\Xb_{1h}, \Zbt_{1h}) \sim P_{(\Xb_1, \Zb_1) \mid \Gb_1} \times P_{\Gb_{1h} \mid \rfull_{1h}=1}$ for each $h \in [H]$, and
    the test point $(\Xb_2, Y_2, \Zbt_2) \sim P_{(\Xb_2, Y_2, \Zb_2) \mid \Gb_2} \times P_{\Gb_2}$ are both independent of $\Dc$.
\end{theorem}
}

We adopt the proof strategy of Theorem~1 in \citet{aggregation}, but extending the argument to our setting requires additional work. The main complications stem from three features that are absent in the original paper: (i) we study \emph{weighted} conformal prediction, so the relevant conformal scores and quantile arguments must account for non-uniform weights; (ii) the weights are not known a priori, and we therefore use an \emph{estimated} weight function $\widehat w_h$; and (iii) we allow for multiple strata, potentially following different distributions, and construct a single prediction set by aggregating data across strata.

\begin{proof}[Proof of Theorem 2]
For notational convenience, for each stratum $h \in [H]$, we rewrite $\dcalh$ in this proof as
$
\dcalh = \{(\Xb_{1i}, Y_{1i}, \Zbt_{1i})\}_{i \in \icalh} = \{\Tb_{hi}\}_{i=1}^{n_h},
$
where $n_h := \ncalh = n_{1h}/2$, and we denote $\widehat w_h(\Xb_{1i}, \Zbt_{1i})$ by $\widehat w_h(\Tb_{hi})$ and $\mathcal S_\ell(\Xb_{1i}, Y_{1i}, \Zbt_{1i})$ by $\mathcal S_\ell(\Tb_{hi})$.
Then, for each $\ell \in [L]$, the threshold $\tl$ of the $\ell$-th marginal prediction set $\widehat C_{\alpha,\ell}$ can be written as
\begin{align} \label{eqn:app:thm2-threshold}
    \tl &:= \q_{1-\alpha}\left(\frac{1}{H} \sum_{h=1}^H \sum_{i = 1}^{n_h} \frac{\widehat w_h(\Tb_{hi})}{\sum_{j = 1}^{n_h} \widehat w_h(\Tb_{hj})} \cdot \delta_{\Sl(\Tb_{hi})} \right) \nonumber \\
    &= \inf\left\{t \in \mathbb R : \frac{1}{H} \sum_{h=1}^H \sum_{i=1}^{n_h} \frac{\widehat w_h(\Tb_{hi})}{\sum_{j=1}^{n_h} \widehat w_h(\Tb_{hj})} \cdot \mathbf 1\{\Sl(\Tb_{hi}) \le t\} \ge 1-\alpha\right\}.
\end{align}
Define $P_{1h} := P_{(\Xb_1, Y_1, \Zb_1) \mid \Gb_1} \times P_{\Gb_{1h} \mid \rfull_{1h}=1}$ for each $h \in [H]$ and 
$P_2 := P_{(\Xb_2, Y_2, \Zb_2) \mid \Gb_2} \times P_{\Gb_2}$, and let the test point be $\Tb_2 := (\Xb_2, Y_2, \Zbt_2)$.  
As described in Section~3.3, conditional on $\nrawone$, $n_2$, and $\{n_{1h}\}_{h=1}^H$, for each stratum $h \in [H]$, it holds that $\Tb_{hi} \iid P_{1h}, \ i \in [n_h]$ and $\Tb_2 \sim P_2$, so the true WCP weight function can be written as $w_h = dP_2 / dP_{1h}$. 
Then, for the selected index $\widehat\ell \in [L]$, our goal is to derive a lower bound on the conditional coverage probability of $\widehat C_{\alpha,\hat\ell}$:
\[
\Pb\big(Y_2 \in \widehat C_{\alpha,\hat\ell}(\Xb_2, \Zbt_2) \bigm| \Dc \big)
=
\Pb\big(\mathcal S_{\hat\ell}(\Tb_2) \le \widehat t_{\alpha,\hat\ell} \bigm| \Dc \big).
\]
We derive this bound on the event $\edc$, defined as
\begin{equation} \label{eqn:thm2-stilde}
\edc := S_\nb  \cap  \left\{ \min_{1\le h \le H} \Eb[\widehat w_h(\Tb_{1h}) \mid \Dc] > \frac{1}{2} \right\}
\cap
\bigcap_{\substack{
h \in [H]  \\
\zbt \in \k
}}
\left\{
\widehat v_{1\zbt} \widehat q_{1h\gb} > \frac{v_{1\zbt} q_{1h\gb}}{2}, \ \widehat v_{2\zbt} < \frac{3v_{2\zbt}}{2}
\right\}
,
\end{equation}
where $S_\nb$ and the notation $\nb$ are defined in~\eqref{eqn:def-s}.
Note that, by construction, $\edc$ depends only on $\Dc$, since all unknown quantities are estimated solely using $\Dc$.

Since $\Eb[|\widehat w_h(\Tb_{1h}) - w_h(\Tb_{1h})| \mid \Dc] = o_p(1)$ by Proposition~\ref{prop:weight} and $\Eb[w_h(\Tb_{1h})] = 1$ by \eqref{eqn:thm2-term3-2} for all $h \in [H]$, it follows that
$\Eb[\min_{1 \le h \le H} \widehat w_h(\Tb_{1h}) \mid \Dc] > 1/2$ with probability tending to one.
Moreover, we have
$\widehat v_{1\zbt} \widehat q_{1h\gb} = v_{1\zbt} q_{1h\gb} + O_p(n_{1h}^{-1/2})$ and $\widehat v_{2\zbt} = v_{2\zbt} + O_p(n_2^{-1/2})$ for all $\zbt$ and $h$. Combining these with \eqref{eqn:app:kliep:theoretical-snp} yields
$\Pb(\edc \mid \tilde n_1, n_2, \{n_{1h}\}_{h=1}^H) \to 1$ as $n_{11}, \dots, n_{1H}, n_2 \to \infty$.
Furthermore, Lemma~\ref{lem:kliep-1} implies that $\lVert \widehat r_\zbt \rVert_\infty \le \bar M$ for a universal constant $\bar M  = 2\xi_0/\epsilon_0$ for all $\zbt$ on $\edc$.
On the event $\edc$, we therefore have
\begin{equation} \label{eqn:thm2-bound}
    \max_{1 \le h \le H} \lVert \widehat w_h \rVert_\infty 
    \le \max_{\substack{
    1 \le h \le H  \\
    \zbt \in \k
    }}
    \left\{
        \frac{\widehat v_{2\zbt} p_{2\gb}}{\widehat v_{1\zbt} \widehat q_{1h\gb}} 
        \lVert \rnz \rVert_\infty
    \right\}
    \le 3\bar{M} \cdot
    \max_{\substack{
    1 \le h \le H  \\
    \zbt \in \k
    }} \frac{v_{2\zbt} p_{2\gb}}{v_{1\zbt} q_{1h\gb}}
    =: C.
\end{equation}

\noindent
On the event $\edc$, for each $\ell \in [L]$, we aim to derive a concentration bound for
\begin{equation} \label{eqn:app:thm2-wldef}
W_\ell(\Dc, \dcal) := \sup_{t \in \mathbb R} 
\left|
    \frac{1}{H} \sum_{h=1}^H \sum_{i=1}^{n_h} \frac{\widehat w_h(\Tb_{hi})}{\sum_{j=1}^{n_h} \widehat w_h(\Tb_{hj})}
    \cdot \mathbf 1\{\Sl(\Tb_{hi}) \le t\}
    -
    \Pb\!\left(\Sl(\Tb_2) \le t \mid \Dc, \dcal \right)
\right|,
\end{equation}
conditional on~$\Dc$.
To this end, for each stratum $h \in [H]$, we define
\[
W_{\ell, h}(\Dc, \dcal) := \sup_{t \in \mathbb R} 
\left|
    \sum_{i=1}^{n_h} \frac{\widehat w_h(\Tb_{hi})}{\sum_{j=1}^{n_h} \widehat w_h(\Tb_{hj})}
    \cdot \mathbf 1\{\Sl(\Tb_{hi}) \le t\}
    -
    \Pb\!\left(\Sl(\Tb_2) \le t \mid \Dc, \dcal \right)
\right|,
\]
and we first derive a concentration bound for $W_{\ell, h}(\Dc, \dcal)$ conditional on $\Dc$, and then a concentration bound for $W_\ell(\Dc, \dcal)$ using the following relationship:
\begin{equation} \label{eqn:app:thm2-wlwlh}
W_\ell(\Dc, \dcal) \le \frac{1}{H} \sum_{h=1}^H W_{\ell, h}(\Dc, \dcal).
\end{equation}
To bound $W_{\ell, h}(\Dc, \dcal)$, we observe that it admits the following decomposition:
\begin{align} \label{eqn:thm2-decompose}
W_{\ell, h}(\Dc, \dcal) &\le \underbrace{\sup_{t \in \mathbb R} \left| \frac{\sum_{i=1}^{n_h} \widehat w_h(\Tb_{hi}) \mathbf 1\{\Sl(\Tb_{hi}) \le t\}}{\sum_{i=1}^{n_h} \widehat w_h(\Tb_{hi})} - \frac{\sum_{i=1}^{n_h} \widehat w_h(\Tb_{hi}) \mathbf 1\{\Sl(\Tb_{hi}) \le t\}}{\sum_{i=1}^{n_h} \Eb[\widehat w_h(\Tb_{hi}) \mid \Dc]} \right|}_{(\text{I})} \nonumber \\
&\ + \underbrace{\sup_{t \in \mathbb R} \left| \frac{\sum_{i=1}^{n_h} \widehat w_h(\Tb_{hi}) \mathbf 1\{\Sl(\Tb_{hi}) \le t\}}{\sum_{i=1}^{n_h} \Eb[\widehat w_h(\Tb_{hi}) \mid \Dc]} - \frac{\sum_{i=1}^{n_h} \Eb[\widehat w_h(\Tb_{hi}) \mathbf 1\{\Sl(\Tb_{hi}) \le t\} \mid \Dc]}{\sum_{i=1}^{n_h} \Eb[\widehat w_h(\Tb_{hi}) \mid \Dc]} \right|}_{(\text{II})} \nonumber \\
&\ + \underbrace{\sup_{t \in \mathbb R} \left| \frac{\sum_{i=1}^{n_h} \Eb[\widehat w_h(\Tb_{hi}) \mathbf 1\{\Sl(\Tb_{hi}) \le t\} \mid \Dc]}{\sum_{i=1}^{n_h} \Eb[\widehat w_h(\Tb_{hi}) \mid \Dc]} - \frac{\sum_{i=1}^{n_h} \Eb[w_h(\Tb_{hi}) \mathbf 1\{\Sl(\Tb_{hi}) \le t\} \mid \Dc]}{\sum_{i=1}^{n_h} \Eb[w_h(\Tb_{hi})]} \right|}_{(\text{III})} \nonumber \\
&\ + \underbrace{\sup_{t \in \mathbb R} \left| \frac{\sum_{i=1}^{n_h} \Eb[w_h(\Tb_{hi}) \mathbf 1\{\Sl(\Tb_{hi}) \le t\} \mid \Dc]}{\sum_{i=1}^{n_h} \Eb[w_h(\Tb_{hi})]} - \Pb(\Sl(\Tb_2) \le t \mid \Dc, \dcal) \right|}_{(\text{IV})}.
\end{align}

\noindent
Since $\sup_{t \in \Rb} \mathbf 1\{\Sl(\Tb_{hi}) \le t\} = 1$, the term~(I) can be simplified as follows:
\begin{equation} \label{eqn:thm2-term1-1}
(\text{I}) = \left| \frac{1}{\Eb[\widehat w_h(\Tb_{h1}) \mid \Dc]} \cdot \frac{1}{n_h}\sum_{i=1}^{n_h} \widehat w_h(\Tb_{hi}) - 1 \right|.
\end{equation}

\noindent
For each $\ell \in [L]$ and $h \in [H]$, define the following function classes:
\begin{align*}
\fc &:= \left\{\widehat w_h(\cdot) \mathbf 1\{\Sl(\cdot) \le t\} : t \in \mathbb R\right\}, \\
\tfc &:= \left\{\widehat w_h(\cdot) \mathbf 1\{\Sl(\cdot) \le t\} - \Eb\left[\widehat w_h(\Tb_{h1}) \mathbf 1\{\Sl(\Tb_{h1}) \le t\}\mid \Dc\right] : t \in \mathbb R\right\}.
\end{align*}
Then, the term~(II) can be expressed as
\begin{align} \label{eqn:thm2-term2-1}
(\text{II}) &= \frac{1}{n_h \Eb[\widehat w_h(\Tb_{h1}) \mid \Dc]} \cdot \sup_{t \in \mathbb R} \left| \sum_{i=1}^{n_h} \left\{ \widehat w_h(\Tb_{hi}) \mathbf 1\{\Sl(\Tb_{hi}) \le t\} - \Eb[\widehat w_h(\Tb_{hi}) \mathbf 1\{\Sl(\Tb_{hi}) \le t\} \mid \Dc] \right\} \right| \nonumber \\
&= \frac{1}{n_h \Eb[\widehat w_h(\Tb_{h1}) \mid \Dc]} \cdot \sup_{f \in \tfc} \left|\sum_{i=1}^{n_h} f(\Tb_{hi})\right|.
\end{align}

\noindent
To bound \eqref{eqn:thm2-term2-1}, we apply Talagrand's inequality to $\tfc$. 
For this, we first obtain an upper bound on $\Eb\!\left[\sup_{f \in \tfc}|\sum_{i=1}^{n_h} f(\Tb_{hi})| \,\middle|\, \Dc\right]$.
For i.i.d.~Rademacher variables $\varepsilon_i \in \{-1, +1\}$, the symmetrization lemma gives
\begin{align} \label{eqn:thm2-term2-2}
\Eb\left[\sup_{f \in \tfc} \Big|\sum_{i=1}^n f(\Tb_{hi})\Big| \,\middle|\, \Dc\right] 
&= \Eb\left[\sup_{f \in \fc} \Big|\sum_{i=1}^n \left\{f(\Tb_{hi}) - \Eb[f(\Tb_{hi}) \mid \Dc]\right\}\Big| \,\middle|\, \Dc\right]  \nonumber \\
&\le 2\Eb\left[\sup_{f \in \fc} \Big|\sum_{i=1}^n \varepsilon_i f(\Tb_{hi})\Big| \,\middle|\, \Dc\right].
\end{align}
Given $\{\Tb_{hi}\}_{i=1}^{n_h}$, define the stochastic process $\{Z_f\}_{f \in \fc}$ and the empirical $L_2$-norm $\lVert \cdot \rVert_{2, n_h}$ as
\[
Z_f := \frac{1}{\sqrt n_h} \sum_{i=1}^{n_h} \varepsilon_i f(\Tb_{hi}), \qquad 
\lVert f \rVert_{2, n_h} := \sqrt{\frac{1}{n_h} \sum_{i=1}^{n_h} f(\Tb_{hi})^2}, \qquad \forall f \in \fc.
\]
Then, for any $f, g \in \fc$ and $\delta > 0$, Hoeffding's inequality implies
\begin{align*}
    \Pb\left(|Z_f - Z_g| \ge \delta \,\middle|\, \Dc, \{\Tb_{hi}\}_{i=1}^{n_h}\right) 
    &= \Pb\left(\Big|\sum_{i=1}^n \varepsilon_i (f(\Tb_{hi}) - g(\Tb_{hi}))\Big| \ge \sqrt{n_h} \delta \,\middle|\, \Dc, \{\Tb_{hi}\}_{i=1}^{n_h} \right) \\
    &\le 2\exp\left(-\frac{\delta^2}{2\lVert f-g \rVert_{2,n_h}^2}\right).
\end{align*}
Because the process $\{Z_f\}_{f \in \fc}$ satisfies the sub-Gaussian increment condition with respect to the metric $\lVert f - g \rVert_{2,n_h}$, we can apply Dudley's entropy integral bound. Thus, we obtain
\begin{equation*} 
\Eb\left[\sup_{f \in \fc} \Big|\sum_{i=1}^{n_h} \varepsilon_i f(\Tb_{hi})\Big| \,\middle|\, \Dc, \{\Tb_{hi}\}_{i=1}^{n_h} \right] 
\le C_1 \sqrt{n_h} \int_0^{\sup_{f \in \fc} \lVert f \rVert_{2, n_h}} 
\sqrt{\log N(\epsilon, \fc \cup \{0\}, \lVert \cdot \rVert_{2, n_h})} \, d\epsilon,
\end{equation*}
for some universal constant $C_1 > 0$.  
Because the VC subgraph dimension of $\{\mathbf 1\{\Sl(\cdot) \le t\}: t \in \mathbb R\}$ is one and $\widehat w_h \ge 0$, we have $V(\fc)=1$.  
Hence, with $F_{\ell, h} := \sup_{f \in \fc} f$, it follows that
\[
N(\epsilon \lVert F_{\ell, h} \rVert_{2,n_h},\, \fc,\, \lVert \cdot \rVert_{2,n_h})
\;\le\; \frac{C_2}{\epsilon^2},
\]
for some universal constant $C_2 > 0$, for any $\epsilon > 0$. Combining these observations, we obtain
\begin{align} \label{eqn:thm2-term2-3}
&\Eb\left[\sup_{f \in \fc} \Big|\sum_{i=1}^{n_h} \varepsilon_i f(\Tb_{hi})\Big| \,\middle|\, \Dc, \{\Tb_{hi}\}_{i=1}^{n_h} \right]  \nonumber \\
&\le C_1 \sqrt{n_h} \int_0^{\frac{\sup_{f \in \fc} \lVert f \rVert_{2, n_h}}{\lVert F_{\ell, h} \rVert_{2, n_h}}} 
\sqrt{\log N(\epsilon\lVert F_{\ell, h} \rVert_{2, n_h}, \fc \cup \{0\}, \lVert \cdot \rVert_{2, n_h})} \cdot \lVert F_{\ell, h} \rVert_{2, n_h} \ d\epsilon \nonumber \\
&\le CC_1\sqrt{n_h}\int_0^1 \sqrt{\log(1 + N(\epsilon\lVert F_{\ell, h} \rVert_{2, n_h}, \fc, \lVert \cdot \rVert_{2, n_h}))} \ d\epsilon \nonumber \\
&\le CC_1\sqrt{n_h} \cdot \int_0^1 \sqrt{\log\left(1 + \frac{C_2}{\epsilon^2}\right)}  d\epsilon = CC_3\sqrt{n_h}
\end{align}
for all realizations of $\{\Tb_{hi}\}_{i=1}^{n_h}$ and some universal $C_3 > 0$, since 
$\lVert F_{\ell, h} \rVert_{2, n_h} \le \lVert \widehat w_h \rVert_\infty \le C$ on $\edc$ by \eqref{eqn:thm2-bound}
and $\sup_{f \in \fc} \lVert f \rVert_{2, n_h} \le \lVert F_{\ell, h} \rVert_{2, n_h}$.
By \eqref{eqn:thm2-term2-2} and \eqref{eqn:thm2-term2-3}, we have
\begin{align} \label{eqn:thm2-term2-4}
\Eb\left[\sup_{f \in \tfc} \Big|\sum_{i=1}^{n_h} f(\Tb_{hi})\Big| \,\middle|\, \Dc\right] 
&\le 2\Eb\left[\sup_{f \in \fc} \Big|\sum_{i=1}^{n_h} \varepsilon_i f(\Tb_{hi})\Big| \,\middle|\, \Dc\right] \nonumber \\
&= 2 \Eb\left[\Eb\left[\sup_{f \in \fc} \Big|\sum_{i=1}^{n_h} \varepsilon_i f(\Tb_{hi})\Big| \,\middle|\, \Dc, \{\Tb_{hi}\}_{i=1}^{n_h} \right]\,\middle|\, \Dc\right] 
\le 2CC_3\sqrt{n_h}.
\end{align}

\noindent
For term~(III), we can further decompose and bound it as follows:
\begin{align} \label{eqn:thm2-term3-1}
(\text{III}) &\le \sup_{t \in \mathbb R} \left| \frac{\sum_{i=1}^{n_h} \Eb[\widehat w_h(\Tb_{hi}) \mathbf 1\{\Sl(\Tb_{hi}) \le t\} \mid \Dc]}{\sum_{i=1}^{n_h} \Eb[\widehat w_h(\Tb_{hi}) \mid \Dc]} - \frac{\sum_{i=1}^{n_h} \Eb[\widehat w_h(\Tb_{hi}) \mathbf 1\{\Sl(\Tb_{hi}) \le t\} \mid \Dc]}{\sum_{i=1}^{n_h} \Eb[w_h(\Tb_{hi})]} \right| \nonumber \\
&\qquad + \sup_{t \in \mathbb R} \left| \frac{\sum_{i=1}^{n_h} \Eb[\widehat w_h(\Tb_{hi}) \mathbf 1\{\Sl(\Tb_{hi}) \le t\} \mid \Dc]}{\sum_{i=1}^{n_h} \Eb[w_h(\Tb_{hi})]} - \frac{\sum_{i=1}^{n_h} \Eb[w_h(\Tb_{hi}) \mathbf 1\{\Sl(\Tb_{hi}) \le t\} \mid \Dc]}{\sum_{i=1}^{n_h} \Eb[w_h(\Tb_{hi})]} \right| \nonumber \\
&\le \sup_{t \in \Rb} \ \Eb[\widehat w_h(\Tb_{h1}) \mathbf 1\{\Sl(\Tb_{h1}) \le t\} \mid \Dc] \cdot \left| \frac{1}{\Eb[\widehat w_h(\Tb_{h1}) \mid \Dc]} - \frac{1}{\Eb[w_h(\Tb_{h1})]} \right|  \nonumber \\
&\qquad + \sup_{t \in \Rb} \Eb\left[|\widehat w_h(\Tb_{h1}) - w_h(\Tb_{h1})| \cdot \mathbf 1\{\Sl(\Tb_{h1}) \le t\} \mid \Dc\right] \nonumber \\
&= \left| \Eb[\widehat w_h(\Tb_{h1}) \mid \Dc] - 1\right| + \Eb\left[|\widehat w_h(\Tb_{h1}) - w_h(\Tb_{h1})| \mid \Dc\right] \nonumber \\
&\le 2 \Eb\left[|\widehat w_h(\Tb_{h1}) - w_h(\Tb_{h1})| \mid \Dc\right],
\end{align}
since by the definition of the true weight function $w_h$, 
\begin{equation} \label{eqn:thm2-term3-2}
\Eb[w_h(\Tb_{h1})] = \Eb_{\Tb \sim P_{1h}}[w_h(\Tb)] =\Eb_{\Tb \sim P_2}\left[\frac{dP_2}{dP_{1h}}(\Tb) \cdot \frac{dP_{1h}}{dP_2}(\Tb)\right] = 1.
\end{equation}

\noindent
Finally, for term~(IV), since for each $t \in \Rb$ and each $i \in [n_h]$
\begin{align*} 
\Eb[w_h(\Tb_{hi}) \mathbf 1\{\Sl(\Tb_{hi}) \le t\} \mid \Dc] &= \Eb_{\Tb \sim P_{1h}}[w_h(\Tb) \mathbf 1\{\Sl(\Tb) \le t\} \mid \Dc] \nonumber \\
&= \Eb_{\Tb \sim P_2}[\mathbf 1\{\Sl(\Tb) \le t\} \mid \Dc] \nonumber \\
&= \Pb(\Sl(\Tb_2) \le t \mid \Dc, \dcal),
\end{align*}
we conclude that the term (IV) is equal to zero.
Now, for each $h \in [H]$, define 
\[
\Delta_h(\Dc) := \sqrt{n_h} \cdot\Eb\left[|\widehat w_h(\Tb_{h1}) - w_h(\Tb_{h1})| \bigm| \Dc, \nrawone, n_2, \{n_{1h}\}_{h=1}^H \right] + 4CC_3.
\]
Then, combining \eqref{eqn:thm2-decompose} and \eqref{eqn:thm2-term3-1}, for any $u > 2$, we have
\begin{align} \label{eqn:thm2-bound-1}
&\Pb\left(\sqrt{n_h} W_{\ell, h}(\Dc, \dcal) > u \Delta_h(\Dc) \Bigm| \Dc\right) \nonumber \\
&\le \Pb\left((\text{I}) + (\text{II}) + 2 \Eb\left[|\widehat w_h(\Tb_{h1}) - w_h(\Tb_{h1})| \mid \Dc\right] > u\Eb\left[|\widehat w_h(\Tb_{h1}) - w_h(\Tb_{h1})| \mid \Dc\right] + \frac{4CC_3u}{\sqrt{n_h}} \Bigm|  \Dc\right) \nonumber \\
&\le \Pb\left((\text{I}) > \frac{(u-2)\Delta_h(\Dc) + 8CC_3}{2\sqrt{n_h}} \Bigm| \Dc\right) + \Pb\left((\text{II}) > \frac{(u-2)\Delta_h(\Dc) + 8CC_3}{2\sqrt{n_h}} \Bigm| \Dc\right).
\end{align}
On the event $\edc$ defined in \eqref{eqn:thm2-stilde}, since $\lVert \widehat w_h \rVert_\infty \le C$ and $\Eb[\widehat w_h(\Tb_{h1}) \mid \Dc] > 1/2$, 
using \eqref{eqn:thm2-term1-1} together with Hoeffding's inequality yields the following bound:
\begin{align} \label{eqn:thm2-bound-2}
&\Pb\left((\text{I}) > \frac{(u-2)\Delta_h(\Dc) + 8CC_3}{2\sqrt{n_h}} \Bigm| \Dc\right) \nonumber \\ 
&= \Pb\left(\Big|\frac{1}{n_h}\sum_{i=1}^{n_h} \left\{\widehat w_h(\Tb_{hi}) - \Eb[\widehat w_h(\Tb_{hi}) \mid \Dc]\right\}\Big| > \frac{(u-2)\Delta_h(\Dc) + 8CC_3}{2\sqrt{n_h}} \cdot \Eb\left[\widehat w_h(\Tb_{h1}) \mid \Dc\right]  \,\middle|\, \Dc\right) \nonumber \\
&\le \Pb\left(\Big|\frac{1}{n_h}\sum_{i=1}^{n_h} \left\{\widehat w_h(\Tb_{hi}) - \Eb[\widehat w_h(\Tb_{hi}) \mid \Dc]\right\}\Big| > \frac{(u-2)\Delta_h(\Dc)}{4\sqrt{n_h}} \,\middle|\, \Dc\right) \nonumber \\
&\le 2 \exp\left(-\frac{(u-2)^2\Delta_h(\Dc)^2}{32C^2}\right).
\end{align}
Furthermore, by \eqref{eqn:thm2-term2-1} and \eqref{eqn:thm2-term2-4}, we obtain
\begin{align*} 
&\Pb\left((\text{II}) > \frac{(u-2)\Delta_h(\Dc) + 8CC_3}{2\sqrt{n_h}} \,\middle|\, \Dc\right) \nonumber \\
&= \Pb\left( \sup_{f \in \tfc} \Big|\sum_{i=1}^{n_h} f(\Tb_{hi})\Big| > \frac{(u-2)\Delta_h(\Dc) + 8CC_3}{2\sqrt{n_h}} \cdot n_h\Eb[\widehat w_h(\Tb_{h1}) \mid \Dc] \,\middle|\, \Dc\right) \nonumber \\
&\le \Pb\left( \sup_{f \in \tfc} \Big|\sum_{i=1}^{n_h} f(\Tb_{hi})\Big| > \frac{(u-2)\sqrt n_h\Delta_h(\Dc)}{4} + 2CC_3\sqrt{n_h} \,\middle|\, \Dc\right) \nonumber \\
&\le \Pb\left(\sup_{f \in \tfc} \Big|\sum_{i=1}^{n_h} f(\Tb_{hi})\Big| - \Eb\left[\sup_{f \in \tfc} \Big|\sum_{i=1}^{n_h} f(\Tb_{hi})\Big| \,\middle|\, \Dc\right] > \frac{(u-2)\sqrt n_h\Delta_h(\Dc)}{4} \,\middle|\, \Dc \right).
\end{align*}
Note that since $\sup_{f \in \tfc} \left|\sum_{i=1}^{n_h} f(\Tb_{hi}) \right| \le Cn_h$ on $\edc$, 
the above probability becomes zero whenever $(u - 2)\Delta_h(\Dc) > 4C\sqrt n_h$.  
When $(u - 2)\Delta_h(\Dc) \le 4C\sqrt n_h$, note that on $\edc$,
\[
\frac{1}{n_h}\sum_{i=1}^{n_h} \sup_{f \in \tfc} \Eb[f^2(\Tb_{hi}) \mid \Dc] 
= \frac{1}{n_h}\sum_{i=1}^{n_h} \sup_{f \in \fc} \Eb[\{f(\Tb_{hi}) - \Eb[f(\Tb_{hi}) \mid \Dc]\}^2 \mid \Dc] \le C^2.
\]
Applying Bousquet's version of Talagrand's inequality (Theorem~7.3 of \citet{Bousquet03_s}) to the class $\tfc / C$, with $\sigma^2 = 1$ and 
$\nu_{n_h} := n_h + 2\Eb[\sup_{f \in \tfc}|\sum_{i=1}^{n_h} f(\Tb_{hi})| \mid \Dc]/C$, yields
\begin{align} \label{eqn:thm2-bound-3}
\Pb\left((\text{II}) > \frac{(u-2)\Delta_h(\Dc) + 8CC_3}{2\sqrt{n_h}} \Bigm| \Dc\right) 
&\le \exp\left(-\nu_{n_h} h_1\left(\frac{(u-2)\sqrt{n_h}\Delta_h(\Dc)}{4C\nu_{n_h}}\right)\right) \nonumber \\
&\le \exp\left(-\frac{n_h(u-2)^2\Delta_h(\Dc)^2/16C^2}{2\nu_{n_h} + (u-2)\sqrt n_h\Delta_h(\Dc)/6C}\right) \nonumber \\
&\le \exp\left(-\frac{n_h(u-2)^2\Delta_h(\Dc)^2/16}{2C^2 n_h + 4C^2C_3\sqrt{n_h} + C(u-2)\sqrt n_h\Delta_h(\Dc)/6}\right) \nonumber \\
&\le \exp\left(-\frac{n_h(u-2)^2\Delta_h(\Dc)^2/16}{2C^2(1+2C_3)n_h + 2C^2n_h/3}\right) \nonumber \\
&= \exp\left(-\frac{C_4(u-2)^2\Delta_h(\Dc)^2}{C^2}\right)
\end{align}
for any $u$ satisfying $(u - 2)\Delta_h(\Dc) \le 4C\sqrt n_h$ and $h_1(x) := (x+1)\log(1+x) - x$, where $C_4 = 3/(128+192C_3)$.
The second inequality uses the fact that 
$h_1(x) \ge \frac{3x}{4} \log(1 + \frac{2x}{3}) \ge \frac{x^2}{2 + 2x/3}$ for all $x > 0$, 
and the third follows from~\eqref{eqn:thm2-term2-4}.
Combining the two cases, we conclude that \eqref{eqn:thm2-bound-3} holds for all $u > 2$.
Putting together \eqref{eqn:thm2-bound-1}, \eqref{eqn:thm2-bound-2}, and \eqref{eqn:thm2-bound-3}, on $\edc$ we obtain
\begin{equation} \label{eqn:app:thm2-wlh}
\Pb\left(\sqrt{n_h} W_{\ell, h}(\Dc, \dcal) > u \Delta_h(\Dc) \Bigm| \Dc\right) \le 3\exp\left(-\frac{C_4(u-2)^2\Delta_h(\Dc)^2}{C^2}\right), \quad \forall  u > 2,
\end{equation}
for all $\ell \in [L]$ and $h \in [H]$, since $C_4 < 1/32$.
Now, define
\begin{equation*} 
\Delta(\Dc) := \max_{1 \le h \le H} \frac{\Delta_h(\Dc)}{\sqrt{n_{h}}}.
\end{equation*}
By \eqref{eqn:app:thm2-wlwlh} and \eqref{eqn:app:thm2-wlh}, on $\edc$ we have
\begin{align*}
    \Pb\left(W_{\ell}(\Dc, \dcal) > u \Delta(\Dc) \,\middle|\, \Dc\right) &\le \sum_{h=1}^H  \Pb\left(W_{\ell, h}(\Dc, \dcal) > u \Delta(\Dc) \,\middle|\, \Dc\right) \\
    &\le \sum_{h=1}^H  \Pb\left(\sqrt{n_h} W_{\ell, h}(\Dc, \dcal) > u \Delta_h(\Dc) \bigm| \Dc\right) \\ 
    &\le 3 \sum_{h=1}^H \exp\left(-\frac{C_4(u-2)^2\Delta_h(\Dc)^2}{C^2}\right), \quad \forall \ell \in [L], \ u > 2.
\end{align*}
Next, define $W(\Dc, \dcal) := \max_{1 \le \ell \le L} W_{\ell}(\Dc, \dcal)$. 
Then, on $\edc$,
\begin{align} \label{eqn:thm2-bound-4}
\Pb\left(W(\Dc, \dcal) > u \Delta(\Dc) \,\middle|\, \Dc\right) 
&= \Pb\left( \max_{1 \le \ell \le L} W_{\ell}(\Dc, \dcal) > u \Delta(\Dc) \,\middle|\, \Dc\right) \nonumber \\
& \le 3L \sum_{h=1}^H \exp\left(-\frac{C_4(u-2)^2\Delta_h(\Dc)^2}{C^2}\right), \quad \forall u > 2.
\end{align}
Furthermore, by \eqref{eqn:app:thm2-wldef}, it is clear that
\[
\frac{1}{H}\sum_{h=1}^H\sum_{i=1}^{n_h} \frac{\widehat w_h(\Tb_{hi})}{\sum_{j=1}^{n_h} \widehat w_h(\Tb_{hj})} \cdot \mathbf 1\{\hatsl(\Tb_{hi}) \le \thatl\} 
- \Pb(\hatsl(\Tb_2) \le \thatl \mid \Dc, \dcal) \le W(\Dc, \dcal).
\]
Consequently, it follows that
\begin{align*}
\Pb(\hatsl(\Tb_2) \le \thatl \mid \Dc, \dcal) 
&\ge \frac{1}{H}\sum_{h=1}^H\sum_{i=1}^{n_h} \frac{\widehat w_h(\Tb_{hi})}{\sum_{j=1}^{n_h} \widehat w_h(\Tb_{hj})} \cdot \mathbf 1\{\hatsl(\Tb_{hi}) \le \thatl\} - W(\Dc, \dcal) \\
&\ge 1 - \alpha - W(\Dc, \dcal),
\end{align*}
where the last inequality follows from the definition of $\thatl$ in \eqref{eqn:app:thm2-threshold}.
Taking expectations with respect to $\dcal$ on both sides, we obtain
\begin{equation} \label{eqn:thm2-bound-5}
\Pb(\hatsl(\Tb_2) \le \thatl \mid \Dc) \ge 1-\alpha - \Eb\left[W(\Dc, \dcal) \mid \Dc\right]. 
\end{equation}
Now, using \eqref{eqn:thm2-bound-4}, on $\edc$ we can bound $\Eb\left[W(\Dc, \dcal) \mid \Dc\right]$ as
\begin{align} \label{eqn:thm2-bound-6}
&\Eb[W(\Dc, \dcal) \mid \Dc] \nonumber \\
&=\int_0^\infty \Pb\left(W(\Dc, \dcal) > u \Delta(\Dc) \Bigm| \Dc\right) du \cdot \Delta(\Dc) \nonumber \\ 
&\le \left(2 + \frac{\sqrt{\log(2HL)}}{4C_3\sqrt{C_4}} + 3L \sum_{h=1}^H \int_{2 + \frac{\sqrt{\log(2HL)}}{4C_3\sqrt{C_4}}}^\infty \exp\left(-\frac{C_4(u-2)^2\Delta_h(\Dc)^2}{C^2}\right) du\right) \cdot \Delta(\Dc) \nonumber \\ 
&\le \left(2 + \frac{\sqrt{\log(2HL)}}{4C_3\sqrt{C_4}} + 3HL \int_{\sqrt{\log(2HL)}}^\infty \exp(-s^2) ds \cdot \frac{1}{4C_3\sqrt{C_4}}\right) \cdot \Delta(\Dc) \nonumber \\ 
&\le \left(2+ C_5\Big(1 + \sqrt{\log(2HL)}\Big)\right) \cdot \Delta(\Dc),
\end{align}
for $C_5 = 1/(4C_3\sqrt{C_4})$, where the second inequality holds since $\Delta_h(\Dc) \ge 4CC_3$ for all $h$ by its definition, and the last inequality follows from $3HL \int_{\sqrt{\log(2HL)}}^\infty \exp(-s^2) ds \le 1$ for $HL \ge 1$. Finally, combining \eqref{eqn:thm2-bound-5} and \eqref{eqn:thm2-bound-6}, the statement of the theorem follows from \eqref{eqn:thm2-bound}, since
\begin{align} \label{eqn:thm2-bound-7}
\Delta(\Dc) \le \max_{1 \le h \le H} \Eb\left[|\widehat w_h(\Xb_{1h}, \Zbt_{1h}) - w_h(\Xb_{1h}, \Zbt_{1h})| \Bigm| \Dc, \tilde n_1, n_2, \{n_{1h}\}_{h=1}^H  \right] + \frac{4\sqrt{2}CC_3}{\sqrt{\min_{1 \le h \le H} n_{1h}}},
\end{align}
where we used that $n_h = n_{1h}/2$ for all $h \in [H]$.
\end{proof}

\section{Refined Coverage Bound for Small $H$}  \label{app:smallh}

In Theorem~2, the miscoverage error of the selected prediction set $\widehat C_{\alpha, \hat \ell}$, namely
\[
1-\alpha - \Pb\left(Y_2 \in \widehat C_{\alpha, \hat \ell}(\Xb_2, \Zbt_2) \,\middle|\, \Dc\right),
\]
is bounded by a factor of the form $2 + C_1 \big(1 + \sqrt{\log(2HL)}\big)$. Because the dependence on $H$ enters only through a logarithmic term, this bound deteriorates slowly as $H$ increases.
However, Theorem~2 implies that even when $H$ is small, the resulting bound can be conservative, since the factor $2 + C_1 \big(1 + \sqrt{\log(2HL)}\big)$ multiplies the estimation error of the WCP weight function. In particular, the bound need not be sharp even when $H$ is small.

In fact, when $H$ is small (e.g., $H=1$ when individuals are sampled from the entire population), a sharper coverage bound can be achieved by a slight modification of the proof in Appendix~\ref{app:thm2}. In this section, we state a corollary of Theorem~2 that provides such a refined guarantee and present its proof, highlighting only the modified steps.

\begin{apxcor}[Lower bound on the coverage probability for small $H$] \label{cor:app:smallh}
    Define $\Dc := \dttra \cup \dtwo$, where \(\dttra\) is given in Section~4.
    Denote $\nrawone := |\irawone|$, $n_2 := |\itwo|$, and $n_{1h} := |\ioneh|$ for each $h \in [H]$.
    Under Assumptions~\ref{ass:super}--\ref{ass:invariant} and Assumption~\ref{ass:kliep-1},
    the selected prediction set $\widehat C_{\alpha, \hat \ell}$ satisfies the following.
    There exists an event $\mathcal E(\Dc)$ that depends solely on $\Dc$ such that
    $\Pb(\mathcal E(\Dc) \bigm| \tilde n_1, n_2, \{n_{1h}\}_{h=1}^H ) \to 1$
    as $n_{11}, \dots, n_{1H}, n_2 \to \infty$, and on $\mathcal E(\Dc)$,
    \[
    \begin{adjustbox}{max width=\linewidth}
    $\displaystyle
    \begin{aligned}
        &\Pb\left(Y_2 \in \widehat{C}_{\alpha, \hat \ell}(\Xb_2, \Zbt_2) \Bigm| \Dc, \tilde n_1, n_2, \{n_{1h}\}_{h=1}^H \right)  \ge 1-\alpha - C_2 \cdot \max_{\zbt, h} \frac{v_{2\zbt} p_{2\gb}}{v_{1\zbt} q_{1h\gb}} \cdot \frac{2 + C_1\big(H + \sqrt{\log(2L)}\big)}{\sqrt{\min_{1 \le h \le H} n_{1h}}} \\ 
        & \qquad \qquad \quad - \big(2 + C_1(H-1)\big) \cdot \max_{1 \le h \le H} \Eb\left[|\widehat w_h(\Xb_{1h}, \Zbt_{1h}) - w_h(\Xb_{1h}, \Zbt_{1h})| \Bigm| \Dc, \tilde n_1, n_2, \{n_{1h}\}_{h=1}^H  \right]
    \end{aligned}
    $
    \end{adjustbox}
    \]
    for universal constants $C_1, C_2 > 0$, which are the same constants as those appearing in Theorem~2.
    Here, the training covariates $(\Xb_{1h}, \Zbt_{1h}) \sim P_{(\Xb_1, \Zb_1) \mid \Gb_1} \times P_{\Gb_{1h} \mid \rfull_{1h}=1}$ and
    the test point $(\Xb_2, Y_2, \Zbt_2) \sim P_{(\Xb_2, Y_2, \Zb_2) \mid \Gb_2} \times P_{\Gb_2}$ are both independent of $\Dc$.
\end{apxcor}

\begin{proof}[Proof of Corollary~\ref{cor:app:smallh}]
In this proof, we use the same notation as in Appendix~\ref{app:thm2}.
Define
\[
\hast := \argmax_{1 \le h \le H} \frac{\Delta_h(\Dc)}{\sqrt{n_h}},
\]
so that $\Delta(\Dc) = \Delta_{\hast}(\Dc)/\sqrt {n_{\hast}}$.
We then modify the argument leading to \eqref{eqn:thm2-bound-6} as follows:
\begin{align*}
    &\Eb[W(\Dc, \dcal) \mid \Dc] \\
    &=\int_0^\infty \Pb\left(W(\Dc, \dcal) > u \Delta(\Dc) \Bigm| \Dc\right) du \cdot \Delta(\Dc)  \\
    &\le \Bigg(2 + \frac{\sqrt{\log(2L)}}{\sqrt{C_4/C^2} \Delta_{\hast}(\Dc)} + 3L \sum_{h=1}^H \int_{2 + \frac{\sqrt{\log(2L)}}{\sqrt{C_4/C^2} \Delta_{\hast}(\Dc)}}^\infty \exp\left(-\frac{C_4(u-2)^2\Delta_h(\Dc)^2}{C^2}\right) du\Bigg) \cdot \Delta(\Dc) \\
    &\le 2\Delta(\Dc) + \frac{\sqrt{\log(2L)}}{\sqrt{C_4/C^2}\sqrt{n_{\hast}}} + 3L\int_{2 + \frac{\sqrt{\log(2L)}}{\sqrt{C_4/C^2} \Delta_{\hast}(\Dc)}}^\infty \exp\left(-\frac{C_4(u-2)^2\Delta_{\hast}(\Dc)^2}{C^2}\right) du \cdot \Delta(\Dc)  \\
    & \qquad + 3L(H-1)\int_{2 + \frac{\sqrt{\log(2L)}}{4C_3\sqrt{C_4}}}^\infty \exp\left(-C_4(u-2)^2(4C_3)^2\right) du \cdot \Delta(\Dc) \\
    &\le 2\Delta(\Dc) + \frac{\sqrt{\log(2L)}}{\sqrt{C_4/C^2}\sqrt{n_{\hast}}} + 3L\int_{\sqrt{\log(2L)}}^\infty \exp(-s^2) ds \cdot \frac{\Delta(\Dc) }{\sqrt{C_4/C^2} \Delta_{\hast}(\Dc)}  \\
    & \qquad + (H-1) \cdot 3L\int_{\sqrt{\log(2L)}}^\infty \exp(-s^2) ds \cdot \frac{\Delta(\Dc) }{4C_3\sqrt{C_4}}  \\
    &\le 2\Delta(\Dc) + C \cdot\frac{1 + \sqrt{\log(2L)}}{\sqrt{C_4}\sqrt{n_{\hast}}} + (H-1) \cdot \frac{\Delta(\Dc)}{4C_3\sqrt{C_4}} \\
    &\le \big(2 + C_5(H-1)\big) \Delta(\Dc) + \frac{4\sqrt{2}CC_3C_5\Big(1 + \sqrt{\log(2L)}\Big)}{\sqrt{\min_{1 \le h \le H} n_{1h}}} \\
    &\le \big(2 + C_5(H-1)\big) \cdot \max_{1 \le h \le H} \Eb\left[|\widehat w_h(\Xb_{1h}, \Zbt_{1h}) - w_h(\Xb_{1h}, \Zbt_{1h})| \Bigm| \Dc, \tilde n_1, n_2, \{n_{1h}\}_{h=1}^H  \right] \\
    & \qquad + \left(2 + C_5\Big(H + \sqrt{\log(2L)}\Big)\right) \cdot \frac{4\sqrt{2} CC_3}{\sqrt{\min_{1 \le h \le H} n_{1h}}},
\end{align*}
where $C_5 = 1/(4C_3\sqrt{C_4})$.
Here, the second inequality uses the fact that $\Delta_h(\Dc) \ge 4CC_3$ for all $h \neq h^\ast$.
The fourth inequality follows from the bound $3L \int_{\sqrt{\log(2L)}}^\infty \exp(-s^2) ds \le 1$ for $L \ge 1$,
and the final inequality follows from \eqref{eqn:thm2-bound-7}.
\end{proof}

The advantage of Corollary~\ref{cor:app:smallh} is most apparent in the small-$H$ regime. For example, when $H=1$, Corollary~\ref{cor:app:smallh} separates the dependence on $H$ and $L$, and the estimation error is multiplied by a constant factor equal to $2$, yielding a substantially sharper bound than Theorem~2. On the other hand, the corollary bound scales linearly with $H$, and can therefore be looser than the bound in Theorem~2 when $H$ is large. Taken together, these results suggest using Theorem~2 in the large-$H$ regime and Corollary~\ref{cor:app:smallh} in the small-$H$ regime.

\section{Theoretical Results for Group-Conditional Methods} \label{app:theoretical-group}

In this section, we present theoretical results for group-conditional prediction sets constructed as in Appendix~\ref{app:group}, paralleling the results for marginal sets given in Theorem~2 and Corollary~1.
For each subgroup $\gb \in \kg$, let $\widehat C_{\alpha, \gb, 1}, \dots, \widehat C_{\alpha, \gb, L}$ denote $L$ candidate group-conditional prediction sets of the form~\eqref{eqn:group-conditional}, and let $\widehat \ell_\gb \in [L]$ denote the selected index.

The following corollary is obtained by applying Corollary~\ref{cor:app:smallh} to the group-conditional setting.
As will be discussed in Appendix~\ref{app:group}, conditioning on $\Gb_{ti}$ eliminates distributional differences across strata, so this setting corresponds to the special case $H=1$ in Corollary~\ref{cor:app:smallh}.
Therefore, substituting $H=1$ into Corollary~\ref{cor:app:smallh} and applying the corresponding group-conditional modification yields the coverage bound for the selected prediction set $\widehat C_{\alpha, \gb, \hat \ell_\gb}$.

\begin{apxcor}[Lower bound on the coverage probability of the group-conditional set] \label{cor:group-1}
    For each subgroup $\gb \in \kg$, define $\Dc_\gb := \dttrag \cup \dtwog$, where $\dttrag$ and $\dtwog$ are given in Appendix~\ref{app:group}.
    Conditional on $n_{1\gb} = |\ioneg|$, $\tilde n_{1\gb} = |\tilde{\mathcal I}_{1\gb}|$, and $n_{2\gb} = |\itwog|$, and under Assumptions~\ref{ass:super}--\ref{ass:invariant} and Assumption~\ref{ass:kliep-1}, the selected subgroup-specific prediction set $\widehat C_{\alpha, \gb, \hat \ell_\gb}$ satisfies the following. There exists an event $\mathcal E(\Dc_\gb)$ that depends solely on $\Dc_\gb$ such that
    $\Pb(\mathcal E(\Dc_\gb) \mid n_{1\gb}, \tilde n_{1\gb}, n_{2\gb} ) \to 1$
    as $n_{1\gb}, n_{2\gb} \to \infty$, and on $\mathcal E(\Dc_\gb)$,
    \[
    \begin{adjustbox}{max width=\linewidth}
    $\displaystyle
    \begin{aligned}
        &\Pb\left(Y_{2\gb} \in \widehat{C}_{\alpha, \gb, \hat \ell_\gb}(\Xb_{2\gb}, \Zb_{2\gb}) \,\middle|\, \Dc_\gb, n_{1\gb}, \tilde n_{1\gb}, n_{2\gb} \right)  \ge 1-\alpha - C_2 \cdot \max_{\zbt} \frac{v_{2\zbt}}{v_{1\zbt}} \cdot \frac{2 + C_1\big(1+\sqrt{\log(2L)}\big)}{\sqrt{n_{1\gb}}} \\ 
        &\qquad \qquad \qquad \qquad \qquad \qquad \quad  - 2 \Eb\left[\left|\widehat w_{\gb}(\Xb_{1\gb}, \Zb_{1\gb}) - w_{\gb}(\Xb_{1\gb}, \Zb_{1\gb})\right| \bigm| \Dc_\gb, n_{1\gb}, \tilde n_{1\gb}, n_{2\gb}  \right] 
    \end{aligned}
    $
    \end{adjustbox}
    \]
    for universal constants $C_1, C_2 > 0$, which are the same constants as those appearing in Theorem~2.
    Here, $w_{\gb}$ and $\widehat w_{\gb}$ denote the subgroup-specific WCP weight function and its estimator, respectively, as given in Appendix~\ref{app:group}.
    Moreover, the training covariates $(\Xb_{1\gb}, \Zb_{1\gb}) \sim P_{(\Xb_1, \Zb_1) \mid \Gb_1 = \gb}$ and the test point $(\Xb_{2\gb}, Y_{2\gb}, \Zb_{2\gb}) \sim P_{(\Xb_2, Y_2, \Zb_2) \mid \Gb_2 = \gb}$ are both independent of $\Dc_\gb$.
\end{apxcor}

The following corollary then follows directly from Corollary~\ref{cor:group-1} and a group-conditional analogue of Proposition~\ref{prop:weight}; in particular, the argument does not rely on any assumptions about the relative orders of the sample sizes $n_{1\gb}$ and $n_{2\gb}$.

\begin{apxcor}[Coverage probability of the group-conditional prediction set] \label{cor:group-2}
    For each $\gb \in \k$, define $\Dc_\gb = \dttrag \cup \dtwog$, where $\dttrag$ and $\dtwog$ are given in Appendix~\ref{app:group}.
    Conditional on $n_{1\gb}$, $\tilde n_{1\gb}$, and $n_{2\gb}$, and under Assumptions~\ref{ass:super}--\ref{ass:invariant} and Assumption~\ref{ass:kliep-1}, as $n_{1\gb}, n_{2\gb} \rightarrow \infty$, the selected subgroup-specific prediction set $\widehat C_{\alpha, \gb, \hat \ell_\gb}$ satisfies
    \[
    \begin{adjustbox}{max width=\linewidth}
    $\displaystyle
    \begin{aligned}
        &\Pb\left(Y_{2\gb} \in \widehat{C}_{\alpha, \gb, \hat \ell_\gb}(\Xb_{2\gb}, \Zb_{2\gb}) \mid \Dc_\gb, n_{1\gb}, \tilde n_{1\gb}, n_{2\gb} \right) \ge 1 - \alpha  - \begin{cases}
        O_p\big(n_{1\gb}^{-\frac{1}{2}}\big) &   d = 0, \ l = s,\\
        O_p\big(\min(n_{1\gb}, n_{2\gb})^{-\frac{1}{2}}\big) &   d = 0, \ l > s,\\
        O_p\big(\min(n_{1\gb}, n_{2\gb})^{-\frac{1}{2+\gamma}}\big)  &  d = 1, 2, \\
        O_p\!\left(\Bigl(\tfrac{\log\min(n_{1\gb}, n_{2\gb})}{\min(n_{1\gb}, n_{2\gb})}\Bigr)^{\frac{1}{d}}\right) &  d \ge 3,
        \end{cases}
    \end{aligned}
    $
    \end{adjustbox}
    \]
    for any arbitrarily small $\gamma > 0$.
    The $O_p(\cdot)$ terms are taken with respect to the randomness in $\Dc_\gb$, conditional on $n_{1\gb}$, $\tilde n_{1\gb}$, and $n_{2\gb}$.
    Moreover, $(\Xb_{2\gb}, Y_{2\gb}, \Zb_{2\gb}) \sim P_{(\Xb_2, Y_2, \Zb_2) \mid \Gb_2 = \gb}$ is a subgroup-specific test point independent of $\Dc_\gb$.
\end{apxcor}

\section{Useful Lemmas} \label{app:proof-2}

Here we provide several lemmas, along with their proofs, that are used in the proofs in Appendix~\ref{app:kliep} and Appendix~\ref{app:thm2}.

\begin{apxlemma}[Ahlfors--David regularity of convex compact sets with nonempty interior] \label{lem:thickness}
Let $S \subset \Rb^d$ be a nonempty compact convex set with nonempty interior.
Then there exist constants $r_0>0$ and $c_0>0$ such that for all $x\in S$ and all $0 < r \le r_0$,
\[
\lambda( B(x,r) \cap S) \;\ge\; c_0\, r^d,
\]
where $\lambda$ denotes the Lebesgue measure on $\mathbb{R}^d$.
\end{apxlemma}

\begin{proof}[Proof of Lemma~\ref{lem:thickness}]
Since $S$ has nonempty interior, there exist $z \in S$ and $\rho>0$ such that 
$B(z,\rho) \subset S$.  
Because $S$ is compact, its diameter $D := \sup\{\|u - v\|_2 : u,v \in S\}$ is finite.
Fix any $x \in S$ and consider the convex hull
\[
K_x := \operatorname{conv}\left( \{x\} \cup B(z,\rho)\right),
\]
which is contained in $S$ by convexity.  
Let $v := z - x$. If $\|v\|_2 \le \rho$, then $x \in B(z,\rho)$ and thus $K_x = B(z,\rho)$. 
Hence, for any $0 < r \le \rho/2$, one can always find a closed ball of radius $r/2$ contained in $B(x,r) \cap B(z,\rho)$, and therefore
\begin{equation} \label{eqn:thickness-1}
\lambda(B(x, r) \cap K_x) = \lambda(B(x, r) \cap B(z, \rho))\ge \frac{v_d}{2^d} r^d,
\end{equation}
where $v_d = \lambda(B(0,1))$.  
Now suppose $\rho < \|v \| \le D$. Then $x \notin B(z,\rho)$, and we instead write
\begin{align*}
    K_x &= \{x + tu : u \in B(v, \rho), \ 0 \le t \le 1\} \\
    &= \left\{x + tu : u \in B\left(\frac{v}{\|v\|_2}, \frac{\rho}{\|v\|_2}\right), \ 0 \le t \le \|v\|_2\right\} \\
    &\supseteq \left\{x + tu : u \in B\left(\frac{v}{\|v\|_2}, \frac{\rho}{D}\right), \ 0 \le t \le \rho\right\} \\
    &= \left\{x + tu : u \in B\left(\rho\frac{v}{\|v\|_2}, \frac{\rho^2}{D}\right), \ 0 \le t \le 1\right\} =: C_x.
\end{align*}
For $r_* := \rho(1-\rho/D) > 0$, the volume $c_* := \lambda(B(x, r_*) \cap C_x) > 0$ depends only on $\rho$ and $D$, not on $x$. 
Then, for any $0 < r \le r_*$, homogeneity of the Lebesgue measure yields
\begin{equation} \label{eqn:thickness-2}
\lambda(B(x, r) \cap K_x) \ge \lambda(B(x, r) \cap C_x) = c_*\frac{r^d}{r_*^d}.
\end{equation}
Finally, define $c_0 := \min(v_d/2^d,\, c_*/r_*^d)$ and $r_0 := \min(\rho/2,\, r_*)$. 
By combining \eqref{eqn:thickness-1} and \eqref{eqn:thickness-2}, it follows that for every $0 < r \le r_0$,
\[
\lambda(B(x, r) \cap S) \ge \lambda(B(x, r) \cap K_x) \ge c_0 r^d.
\]
\end{proof}

The following lemma is a multivariate extension of Lemma~3.1 in \citet{ghosal01}, and we use the same notation as in the proof of Theorem~\ref{thm:kliep-1}. 
The proof follows the original argument, with only the steps requiring multivariate extension detailed here.

\begin{apxlemma}[cf. \citet{ghosal01}, Lemma 3.1] \label{lem:ghosal}
    Let $0 < \epsilon < \frac{1}{2}$, $\sigma > 0$ be given. For any probability measure $F$ on $[-a, a]^d$, where $a \le B(\log(1/\epsilon))^\gamma$ and $\gamma \ge 1/2$ and $B > 0$ are constants, there exists a discrete probability measure $F'$ on $[-a, a]^d$ with at most $N \lesssim (\log(1/\epsilon))^{2d\gamma}$ support points in $[-a, a]^d$ such that
    \[
    \lVert g_{F, \sigma} - g_{F', \sigma} \rVert_\infty \lesssim \epsilon, \qquad \mathrm{supp}(F') \subseteq \mathrm{supp}(F),
    \]
    where $g_{F, \sigma}$ is defined in \eqref{eqn:kliep-2-0}.
    \begin{proof}[Proof of Lemma~\ref{lem:ghosal}]
        Define $M = \max(2\sqrt d a, \sqrt 8 \sigma (\log\frac{1}{\epsilon})^{1/2})$. Then, for any probability measures $F, F'$ on $[-a, a]^d$, we have
        \begin{align*}
            \sup_{\lVert \xb \rVert \ge M} |g_{F, \sigma}(\xb) - g_{F', \sigma}(\xb)| \lesssim \epsilon.
        \end{align*}
        Now, by Taylor's expansion of $e^y$ and putting $y = -\lVert \xb \rVert^2/2\sigma^2$, we have
        \begin{equation} \label{eqn:lem:3-1-0}
        \left| f_\sigma(\xb) - \sum_{j = 0}^{k-1} \frac{(-1)^j (2\pi)^{-d/2}\sigma^{-(2j+d)}}{2^j \cdot j!} \lVert \xb \rVert^{2j} \right| \le (2\pi \sigma^2)^{-d/2} \frac{(e^{1/2}2^{-1/2}\sigma^{-1}\lVert \xb \rVert)^{2k}}{k^k}, \quad \forall k > 1,
        \end{equation}
        and hence for any probability measures $F, F'$ on $[-a, a]^d$ and $k > 1$,
        \begin{align} \label{eqn:lem:3-1-1}
            \sup_{\lVert \xb \rVert \le M}  |g_{F, \sigma}(\xb) - g_{F', \sigma}(\xb)| \nonumber  &= \sup_{\lVert \xb \rVert \le M} \left|\int f_\sigma(\xb-\boldsymbol{z})d(F-F')(\boldsymbol{z})\right| \nonumber \\
            &\le \sup_{\lVert \xb \rVert \le M} \left|\int \sum_{j=0}^{k-1} (2\pi)^{-d/2}\frac{(-1)^j \sigma^{-(2j+d)}\lVert \xb - \boldsymbol{z} \rVert^{2j}}{2^j \cdot j!} d(F - F')(\boldsymbol{z}) \right| \nonumber \\
            & \quad + 2 \sup_{\substack{\lVert \xb \rVert \le M \\ \lVert \boldsymbol{z} \rVert \le \sqrt{d} a}}
             \left| f_\sigma(\xb-\boldsymbol{z}) - \sum_{j=0}^{k-1} (2\pi)^{-d/2}\frac{(-1)^j \sigma^{-(2j+d)}\lVert \xb - \boldsymbol{z} \rVert^{2j}}{2^j \cdot j!} \right|.
        \end{align}
        For each $j = 0, \ldots, k-1$ and for any $\xb = (x_1, \ldots, x_d) \in \Rb^d$, $\boldsymbol{z} = (z_1, \ldots, z_d) \in \mathbb{R}^d$, the term $\lVert \xb - \boldsymbol{z} \rVert^{2j}$ can be expressed as $\lVert \xb - \boldsymbol{z} \rVert^{2j} = \sum_{\ab \in A_j} c_{\xb, \ab, j}  z_1^{\alpha_1}\cdots z_d^{\alpha_d}$, where $A_j$ is defined as
        \begin{align*}
            A_j = \{\ab = (\alpha_1, \ldots, \alpha_d)^\top \in \mathbb{Z}_{\ge 0}^d \mid \sum_{i=1}^d \alpha_i \le 2j\},
        \end{align*}
        and $c_{\xb, \ab, j}$ depends on $\xb$, $\ab$, and $j$. Therefore, for any $k > 1$,
        \begin{align*}
            \sum_{j=0}^{k-1} (2\pi)^{-d/2}\frac{(-1)^j \sigma^{-(2j+d)}\lVert \xb - \boldsymbol{z} \rVert^{2j}}{2^j \cdot j!} &= \sum_{j=0}^{k-1} b_j \left(\sum_{\ab \in A_j} c_{\xb, \ab, j} z_1^{\alpha_1}\cdots z_d^{\alpha_d}\right)\\ &= \sum_{\ab \in A_{k-1}} \left(\sum_{j \ge \lVert \ab \rVert_1/2} b_j c_{\xb, \ab, j}\right) z_1^{\alpha_1}\cdots z_d^{\alpha_d},
        \end{align*}
        for $b_j = (-1)^j 2^{-j} (2\pi)^{-d/2}\sigma^{-(2j+d)}/j!$. Therefore, the first term on the RHS of \eqref{eqn:lem:3-1-1} vanishes, if
        \begin{equation} \label{eqn:lem-3-1-2}
            \int z_1^{\alpha_1}\cdots z_d^{\alpha_d} dF(\boldsymbol{z}) = \int z_1^{\alpha_1}\cdots z_d^{\alpha_d} dF'(\boldsymbol{z}), \quad \forall \ab = (\alpha_1, \ldots, \alpha_d)^\top \in A_{k-1}.
        \end{equation}
         Moreover, with $a = e^{1/2}2^{-1/2}\sigma^{-1}\max(3B, \sqrt{18}\sigma)$, by \eqref{eqn:lem:3-1-0}, the second term on the RHS of \eqref{eqn:lem:3-1-1} is bounded by a constant multiple of $\left(a(\log(1/\epsilon)^\gamma\right)^{2k}/k^k$. Therefore, if we choose $k$ to be the smallest integer exceeding $(1+a^2)(\log(1/\epsilon))^{2\gamma}$, it follows that
        \[
        \sup_{\lVert \xb \rVert \le M} |g_{F, \sigma}(\xb) - g_{F', \sigma}(\xb)| \lesssim \epsilon.
        \]
        Finally, by Lemma A.1 in \citet{ghosal01} and \eqref{eqn:lem-3-1-2}, the discrete probability measure $F'$ can be chosen to be a discrete distribution on $\mathrm{supp}(F)$ with at most $N = |A_{k-1}|+1$ support points. If $\epsilon$ is small, then $k > d$ and the following holds:
        \[
        |A_{k-1}| = \binom{2(k-1)+d}{d} \le \frac{(2k+d)^d}{d!} \le \frac{3^d}{d!} k^d.
        \]
        Therefore, $F'$ has at most $N \lesssim (\log(1/\epsilon))^{2d\gamma}$ support points.
    \end{proof}
\end{apxlemma}

In the following lemmas, we use the same notation as in the proof of Theorem~\ref{thm:kliep-1}.

\begin{apxlemma}[cf. \citet{Sugiyama2008}, Lemma 4] \label{lem:kliep-1}
For each $\zbt \in \k$, define
\[
\snz
:= \left\{\inf_{g \in \fnz} \frac{1}{\nttrz} \sum_{i \in \ittraz} g(\Xb_{1i}) \ge \frac{\epsilon_0}{2}\right\},
\]
where $\nb_\zbt := (\nttrz, \ntez)^\top$.
Then $\Pb(\snz \mid \nttrz, \ntez) \to 1$ as $\nttrz, \ntez \to \infty$. Moreover, on the event $\snz$, we have $\hatgnz \subseteq \gnzmbar$.
Here, $\bar{M} := 2\xi_0 / \epsilon_0$, where $\xi_0$ and $\epsilon_0$ are the constants appearing in (A2) of Assumption~\ref{ass:kliep-1}.
\end{apxlemma}

\begin{apxlemma} \label{lem:kliep-2}
    For each $\zbt \in \k$, on the event $\snz$ defined in Lemma \ref{lem:kliep-1},
    \[
    |\munk - 1| \le (\bar{M}^{1/2} + \eta_1^{1/2}) \cdot h_{P_{1\zbt}}(r_\zbt, \rnz).
    \]
    \begin{proof}
    Since $\rnz \in \hatgnz \subseteq \gnzmbar$ on $\snz$ by Lemma~\ref{lem:kliep-1}, 
        \begin{align*}
            |\munk - 1| &= \left|\int_\mathcal{X} \rnz(\xb)\ptrz(\xb)d\xb -  \int_\mathcal{X} r_\zbt(\xb)\ptrz(\xb)d\xb \right| \\
            &\le \int_\mathcal{X} |\rnz(\xb)-r_\zbt(\xb)|\ptrz(\xb)d\xb \\
            &\le \lVert \rnz^{1/2} + r_\zbt^{1/2}\rVert_\infty \cdot \int_\mathcal{X} |\rnz^{1/2}(\xb)-r_\zbt^{1/2}(\xb)|\ptrz(\xb)d\xb \\
            &\le (\bar{M}^{1/2} + \eta_1^{1/2}) \cdot h_{P_{1\zbt}}(r_\zbt, \rnz),
        \end{align*}
        where the last inequality follows from Assumption \ref{ass:kliep-1} and the Cauchy-Schwarz inequality.
    \end{proof}
\end{apxlemma}

\section{Details on the WCP-type Threshold} \label{app:threshold}

In this section, we compare our threshold choice with the original WCP threshold, focusing on the special case $H=1$ for ease of comparison.
Recall that, at miscoverage level $\alpha \in (0,1)$, the WCP-type threshold of the proposed method when $H=1$ is given by
\begin{equation*}
\widehat{t}_\alpha := \text{Quantile}_{1-\alpha}\left( \sum_{i \in \mathcal{I}_{\mathrm{cal}}} \frac{\widehat{w}(\Xb_{1i}, \Zbt_{1i})}{\sum_{j \in \mathcal{I}_{\mathrm{cal}}} \widehat{w}(\Xb_{1j}, \Zbt_{1j})} \delta_{s_i} \right),
\end{equation*}
where we drop the stratum index $h$ since $H = 1$.
By contrast, directly applying the threshold of Section~2.1 with the estimated WCP weight function $\widehat w$ yields
\[
\widehat{t}_\alpha^+(\xb, \zbt) = \text{Quantile}_{1-\alpha}\left(  \frac{\sum_{i \in \mathcal{I}_{\mathrm{cal}}}\widehat{w}(\Xb_{1i}, \Zbt_{1i})\delta_{s_i} + \widehat w(\xb, \zbt)\delta_{+\infty}}{\sum_{j \in \mathcal{I}_{\mathrm{cal}}} \widehat{w}(\Xb_{1j}, \Zbt_{1j}) + \widehat{w}(\xb, \zbt)}  \right).
\]

Since $\widehat{t}_\alpha^+(\xb, \zbt)$ depends on the test point $(\xb, \zbt)$, using the fixed threshold $\widehat{t}_\alpha$ is computationally more convenient. 
Because $\widehat{t}_\alpha \le \widehat{t}_\alpha^+(\xb, \zbt)$ for all $(\xb, \zbt)$, the resulting prediction set is slightly less conservative than that of the original WCP method, similar to the GWCP threshold.
This trade-off is clearly illustrated in \citet{gwcp}.
However, since Corollary~1 implies that the proposed method attains coverage close to the nominal level when the sample size is sufficiently large, the slight loss of conservativeness becomes negligible.
Thus, the proposed approach benefits from the computational simplicity of using the fixed threshold $\widehat{t}_\alpha$ without sacrificing asymptotic coverage.

\section{Implementation Details} \label{app:implement}

In this section, we provide implementation details for constructing both marginal and group-conditional prediction sets.
In both cases, we use \textit{sampling weights} when constructing the score functions.
For concreteness, consider the absolute residual score, where we estimate the mean function $\mu$ using a linear regression model. When the homoscedasticity assumption is violated and the sample subgroup proportions differ from those in the population, the weighted least squares (WLS) estimator remains consistent, whereas the ordinary least squares (OLS) estimator may not \citep{wls}.
Because complex surveys typically involve oversampling---which can induce discrepancies between the sample and population demographics---we use WLS with the sampling weights when estimating $\mu$ for the absolute residual score. Similarly, we use the same weights when estimating $q_{\alo}$ and $q_{\ahi}$ for the CQR score.
More specifically, because the score functions are constructed using observations for which $(\Xb_{1i}, Y_{1i}, \Zbt_{1i})$ are fully observed, we use the sampling weight $\wfull_{1i}$.

\subsection{Implementation Details for Marginal Prediction Sets} \label{app:marg}

To fit a linear model, each component of the categorical covariate variable $\Zbt_{ti} \in \k$ is encoded using dummy variables.
Let $K := K_1 + \cdots + K_l$ and define $\gammab = (\gammab_1^\top, \ldots, \gammab_l^\top)^\top \in \Rb^{K-l}$, where $\gammab_j = (\gamma_{j,1}, \ldots, \gamma_{j,(K_j-1)})^\top \in \Rb^{K_j-1}$ denotes the coefficient vector associated with the dummy encoding of the $j$th component, for $j = 1,\ldots,l$.
Formally, for each $\zbt = (k_1, \ldots, k_l)^\top \in \k$, we define
\begin{equation*}
\gammab^\top \zbt := \sum_{j=1}^l \sum_{s=1}^{K_j-1} \gamma_{j,s}\,\mathbf{1}\{k_j=s\}, 
\qquad \forall \gammab \in \Rb^{K-l}.
\end{equation*}

We now describe how to construct marginal prediction sets for continuous $Y_{ti} \in \Rb$, using absolute residual and CQR scores, and for categorical $Y_{ti} \in [M]$, using soft classifiers.
The procedures below are written with sampling weights. When sampling weights are unavailable (e.g., in simulations), we simply omit them and use the unweighted versions of the same procedures.

\paragraph*{1. Absolute residual score. } 
We estimate the conditional mean function $\mu(\xb,\zbt):=\Eb[Y_t \mid \Xb_t=\xb, \Zbt_t=\zbt]$ under the conditional distribution $P_{Y_t \mid (\Xb_t=\xb,\ \Zbt_t=\zbt)}$, which is invariant over time by Assumption~5, using the weighted least squares (WLS) estimator:
\[
\big(\widehat\beta_0, \hatbb^\top, \hatgb^\top\big)^\top 
= \argmin_{\beta_0 \in \mathbb R, \bb \in \mathbb{R}^d, \gammab \in \mathbb{R}^{K - l}} 
\sum_{i \in \itra} W_{1i}^Y 
(Y_{1i} - \beta_0 - \bb^{\top} \Xb_{1i} - \gammab^{\top} \Zbt_{1i})^2.
\]
The estimated regression function is then given by 
$\widehat \mu(\xb, \zbt) = \widehat \beta_0 + \hatbb^\top\xb + \hatgb^\top\zbt$, 
and the score function is defined as 
$\mathcal S(\xb, y, \zbt) = |y - \widehat \mu(\xb, \zbt)|$. 
Then, for the threshold $\widehat t_\alpha$, 
the resulting prediction set takes the form of an interval:
\[
\widehat C_\alpha(\xb, \zbt) 
= \left\{y \in \Rb : \widehat \mu(\xb, \zbt) - \widehat t_\alpha \le y \le \widehat \mu(\xb, \zbt) + \widehat t_\alpha\right\}.
\]

Using the absolute residual score results in prediction intervals of constant length $2\widehat t_\alpha$ across all values of $(\xb, \zbt)$, which may be conceived as undesirable. 
The CQR method addresses this limitation by producing intervals whose lengths adapt adaptively to both $\xb$ and $\zbt$.

\paragraph*{2. CQR score.} 
The quantile function $q_\theta(\xb, \zbt)$ of the conditional distribution $P_{Y_t \mid (\Xb_t=\xb,\ \Zbt_t=\zbt)}$ at level $\theta \in (0, 1)$, which also remains consistent over time by Assumption~5, can be estimated by minimizing the \textit{weighted pinball loss} as follows:
\[
(\widehat\beta_{0, \theta}, \hatbb_\theta^\top, \hatgb_\theta^\top)^\top  
= \argmin_{\beta_0 \in \mathbb{R}, \bb \in \mathbb{R}^d, \gammab \in \mathbb{R}^{K - l}} 
\sum_{i \in \itra} W_{1i}^Y \cdot
\rho_\theta(Y_{1i} - \beta_0 - \bb^{\top} \Xb_{1i} - \gammab^{\top} \Zbt_{1i}),
\]
where $\rho_\theta(\cdot)$ is the pinball loss function defined by 
$\rho_\theta(x) = (\theta \mathbf{1}\{x \ge 0\} + (1-\theta)\mathbf{1}\{x < 0\})\cdot |x|$.  
The estimated quantile function is then given by 
$\widehat q_\theta(\xb, \zbt) = \widehat \beta_{0, \theta} + \hatbb_\theta^\top \xb + \hatgb_\theta^\top \zbt$, 
and the score function is defined as 
$\mathcal S(\xb, y, \zbt) = \max\{\widehat q_{\alo}(\xb, \zbt) - y, \; y - \widehat q_{\ahi}(\xb, \zbt)\}$.  
For the threshold $\widehat t_\alpha$, 
the resulting prediction set takes the form of an interval:
\[
\widehat{C}_\alpha(\xb, \zbt; \alo, \ahi) = \left\{y \in \mathbb{R} : \widehat{q}_{\alo}(\xb, \zbt) - \widehat t_\alpha \le y \le \widehat{q}_{\ahi}(\xb, \zbt) + \widehat t_\alpha\right\}.
\]
The length of the CQR prediction interval becomes 
$\mathrm{Len}(\widehat{C}_\alpha(\xb, \zbt; \alo,\ahi)) = 2\widehat t_\alpha + \widehat{q}_{\ahi}(\xb, \zbt) - \widehat{q}_{\alo}(\xb, \zbt)$, 
which varies with both $\xb$ and $\zbt$; in particular, within each $\zbt$, it changes linearly in $\xb$ with a constant slope.
The quantile levels $\alo$ and $\ahi$ are treated as hyperparameters, which may be chosen to minimize the prediction interval length. 
Specifically, we set the grids
\begin{equation} \label{eqn:grid}
\mathcal A_{\mathrm{lo}}
:= \left\{ \frac{j}{2(1+n_\alpha)} : j = 1, \ldots, n_\alpha \right\}, 
\qquad 
\mathcal A_{\mathrm{hi}}
:= \left\{\frac{1}{2} + \frac{j}{2(1+n_\alpha)} : j = 1, \ldots, n_\alpha \right\},
\end{equation}
for $\alo$ and $\ahi$, respectively, resulting in a total of $L = n_\alpha^2$ combinations of $(\alo, \ahi)$. 
For each combination, we estimate $\hatqlo$ and $\hatqhi$, construct the prediction set $\widehat C_\alpha(\cdot, \cdot \ ; \alo, \ahi)$, and choose the combination $(\alo^*, \ahi^*)$ that yields the minimum average interval length:
\[
(\alo^*, \ahi^*)
:= \argmin_{\alo \in \mathcal A_{\mathrm{lo}},\, \ahi \in \mathcal A_{\mathrm{hi}}} \
\frac{1}{|\itest|} \sum_{i \in \itest} \mathrm{Len}\big(\widehat C_\alpha(\Xb_{2i}, \Zbt_{2i} ; \alo, \ahi)\big),
\]
using the test dataset $\dtest = \{(\Xb_{2i}, \Zbt_{2i})\}_{i \in \itest}$. Indeed, the pair $(\alo^*, \ahi^*)$ may be selected using \emph{any} dataset, including the calibration dataset $\dcal$. As can be seen from the proof in Appendix~\ref{app:thm2}, the coverage guarantee in Theorem~2 does not depend on the selection mechanism of the prediction set.
Finally, the selected prediction set is then given by $\widehat C_\alpha(\cdot, \cdot \ ; \alo^*, \ahi^*)$.

\paragraph*{3. Soft classification score. }
For categorical outcomes $Y_{ti} \in [M]$, we train a soft classifier $\widehat f(\cdot, \cdot  ; \boldsymbol{\theta}) : \Rb^d \times \k \rightarrow \Delta^{M-1}$, which can incorporate sampling weights during training (e.g., XGBoost, random forest, neural network), where $\boldsymbol{\theta}$ denotes a vector of hyperparameters. 
For each $\xb \in \Rb^d$ and $\zbt \in \k$, write $\widehat f(\xb, \zbt;  \boldsymbol{\theta}) = (\widehat f_1(\xb, \zbt;  \boldsymbol{\theta}), \cdots, \widehat f_M(\xb, \zbt;  \boldsymbol{\theta}))^\top$.
Then, using the score function $\mathcal{S}(\xb, y, \zbt) = 1 - \widehat f_y(\xb, \zbt;  \boldsymbol{\theta})$ together with the threshold $\widehat{t}_\alpha$, the resulting prediction set is given by
\begin{equation*}
\widehat{C}_\alpha(\xb, \zbt;  \boldsymbol{\theta}) = \left\{y \in [M] : \widehat{f}_y(\xb, \zbt;  \boldsymbol{\theta}) \ge 1 - \widehat{t}_\alpha\right\}.
\end{equation*}
Now, let $\boldsymbol{\Theta}$ denote a grid over $\boldsymbol{\theta}$. Then, for each $\boldsymbol{\theta} \in \boldsymbol{\Theta}$, we train $\widehat f(\cdot, \cdot  ; \boldsymbol{\theta})$, construct the prediction set $\widehat{C}_\alpha(\cdot, \cdot;  \boldsymbol{\theta})$, and choose $\boldsymbol{\theta}^* \in \boldsymbol{\Theta}$ to minimize the average prediction set size:
\[
\boldsymbol{\theta}^* := \argmin_{\boldsymbol{\theta} \in \boldsymbol{\Theta}} \frac{1}{|\itest|} \sum_{i \in \itest} |\widehat C_\alpha(\Xb_{2i}, \Zbt_{2i}; \boldsymbol{\theta})|,
\]
using the test dataset $\dtest = \{(\Xb_{2i}, \Zbt_{2i})\}_{i \in \itest}$. Again, the hyperparameter $\boldsymbol{\theta}^*$ can be selected using \emph{any} dataset. Finally, the selected prediction set is given by $\widehat C_\alpha(\cdot, \cdot; \boldsymbol{\theta}^*)$.

\subsection{Construction and Implementation Details of Group-Conditional Prediction Sets} \label{app:group}

\subsubsection{Construction of Group-Conditional Prediction Sets}

To obtain a prediction set $\widehat C_\alpha(\cdot, \cdot)$ that satisfies group-conditional validity, we construct, for each subgroup $\gb \in \kg$, a subgroup-specific prediction set $\widehat C_{\alpha,\gb}(\cdot, \cdot)$ using subgroup-specific datasets $\donegt$ and $\dtwog$, which are defined as follows:
\begin{align*}
\donegt &:= \{(\Xb_{1i}, Y_{1i}, \Zb_{1i}, \rfull_{1i}) \mid i \in \ionet, \Gb_{1i} = \gb \} = \{(\Xb_{1i}, Y_{1i}, \Zb_{1i}, \rfull_{1i})\}_{i \in \ionegt}, \\ 
\dtwog &:= \{(\Xb_{2i}, \Zb_{2i}) \mid i \in \itwo, \Gb_{2i} = \gb \} = \{(\Xb_{2i},  \Zb_{2i})\}_{i \in \itwog}.
\end{align*}
Moreover, define $\doneg$, which serves as the training dataset, as
\[
\doneg := \{(\Xb_{1i}, Y_{1i}, \Zb_{1i}) \mid i \in \ionegt, \rfull_{1i} = 1 \} = \{(\Xb_{1i}, Y_{1i}, \Zb_{1i})\}_{\ioneg}.
\]
Under the WCP framework, $\widehat C_{\alpha,\gb}$ is designed to achieve the group-conditional coverage target
\[
\Pb\big(Y_{2\gb} \in \widehat C_{\alpha,\gb}(\Xb_{2\gb}, \Zb_{2\gb})\big) \ge 1 - \alpha,
\]
for the corresponding subgroup.
In practice, since the WCP weight functions are estimated, the resulting coverage is lower-bounded by $1-\alpha$ up to an additional estimation error term, as characterized in Corollary~\ref{cor:group-1} and Corollary~\ref{cor:group-2}.
Here, the probability is taken over $\donegt$, $\dtwog$, and an independent test point $(\Xb_{2\gb}, Y_{2\gb}, \Zb_{2\gb}) \sim P_{(\Xb_2, Y_2, \Zb_2) \mid \Gb_2 = \gb}$, where the conditional distribution is specified in Assumption~\ref{ass:super}.
We then define $\widehat C_\alpha(\xb, \zbt) := \widehat C_{\alpha,\gb}(\xb, \zb)$ for each $\xb \in \Rb^d$ and $\zbt = (\gb^\top, \zb^\top)^\top \in \k$.

For each subgroup $\gb$, we first randomly split $\doneg$ into a training set $\dtrag$ and a calibration set $\dcalg$, with index sets $\itrag$ and $\icalg$, typically with $|\itrag| = |\icalg| = |\ioneg|/2$.
We then define $\dttrag := \{(\Xb_{1i}, Y_{1i}, \Zb_{1i}, \rfull_{1i})\}_{i \in \ittrag}$, where $\ittrag := \ionegt \setminus \icalg$.
Using the subgroup-specific training dataset $\dtrag$, we train the subgroup-specific score function $\mathcal{S}_\gb(\xb, y, \zb)$, and compute the nonconformity scores $s_i = \mathcal{S}_\gb(\Xb_{1i}, Y_{1i}, \Zb_{1i})$ for each $i \in \icalg$.
By Assumption~\ref{ass:super}, the subgroup-specific calibration dataset $\dcalg$ satisfies
\[
(\Xb_{1i}, Y_{1i}, \Zb_{1i}) \iid P_{(\Xb_1, Y_1, \Zb_1) \mid \Gb_1 = \gb}, \qquad i \in \icalg,
\]
while the subgroup-specific target distribution is given by $P_{(\Xb_2, Y_2, \Zb_2) \mid \Gb_2 = \gb}$.
Accordingly, by Assumption~5, the subgroup-specific WCP weight function is defined as
\[
w_\gb(\xb, \zb) := \Bigg(\frac{dP_{\Xb_2 \mid \Zbt_2 = \zbt}}{dP_{\Xb_1 \mid \Zbt_1 = \zbt}}\Bigg)(\xb) \cdot \frac{P_{\Zb_2 \mid \Gb_2 = \gb}(\zb)}{P_{\Zb_1 \mid \Gb_1 = \gb}(\zb)} = r_\zbt(\xb) \cdot \frac{v_{2\zbt}}{v_{1\zbt}},
\]
for all $\xb \in \Rb^d$ and $\zbt = (\gb^\top, \zb^\top)^\top \in \k$.
Note that, since conditioning on the subgroup membership variable $\Gb_{ti}$ eliminates distributional differences across strata, the weight function no longer depends on the stratum.
We estimate this weight function via $\widehat r_\zbt(\cdot)$ using the KLIEP method described in Appendix~\ref{app:kliep-1}, based on $\{\Xb_{1i} \mid i \in \ittrag, \Zb_{1i} = \zb\}$ and $\{\Xb_{2i} \mid i \in \itwog, \Zb_{2i} = \zb\}$.
Then, we define and use the following subgroup-specific threshold:
\[
\widehat t_{\alpha, \gb} := \q_{1-\alpha}\left(\sum_{i \in \icalg} \frac{\widehat w_\gb(\Xb_{1i}, \Zb_{1i})}{\sum_{j \in \icalg} \widehat w_\gb(\Xb_{1j}, \Zb_{1j})} 
\cdot \delta_{s_i}\right).
\]
Using this threshold, the resulting group-conditional prediction set is given by
\begin{equation} \label{eqn:group-conditional}
    \widehat C_{\alpha}(\xb, \zbt) = \widehat C_{\alpha, \gb}(\xb, \zb) = \left\{y \in \mathcal Y : \mathcal S_\gb(\xb, y, \zb) \le \widehat t_{\alpha, \gb}\right\}, \qquad \forall \xb \in \mathbb R^d, \ \zbt = (\gb^\top, \zb^\top)^\top. 
\end{equation}

\subsubsection{Implementation Details for Group-Conditional Prediction Sets}
To fit a linear model, each component of the non-demographic categorical covariates $\Zb_{ti} \in \kz$ is encoded using dummy variables.
Let $K := K_{s+1} + \cdots + K_l$ and define $\gammab = (\gammab_{s+1}^\top, \ldots, \gammab_l^\top)^\top \in \Rb^{K-(l-s)}$, where $\gammab_j = (\gamma_{j,1}, \ldots, \gamma_{j,(K_j-1)})^\top \in \Rb^{K_j-1}$ denotes the coefficient vector associated with the dummy encoding of the $j$th component, for $j = s+1,\ldots,l$.
Formally, for each $\zb = (k_{s+1}, \ldots, k_l)^\top \in \kz$, we define
\begin{equation*}
\gammab^\top \zb := \sum_{j=s+1}^l \sum_{s=1}^{K_j-1} \gamma_{j,s}\,\mathbf{1}\{k_j=s\}, 
\qquad \forall \gammab \in \Rb^{K-(l-s)}.
\end{equation*}
We now describe the procedure for constructing group-conditional prediction sets for continuous outcomes $Y_{ti} \in \Rb$, using absolute residual and CQR scores, and for categorical outcomes $Y_{ti} \in [M]$, using soft classifiers.
The procedures below are written with sampling weights. When sampling weights are unavailable (e.g., simulations), we simply omit them and use the unweighted versions of the same procedures.

\paragraph{1. Absolute residual score. }
We estimate the subgroup-specific conditional mean function $\mu_\gb(\xb, \zb) := \Eb[Y_t \mid \Xb_t = \xb, \Zb_t = \zb, \Gb_t = \gb]$ separately for each $\gb \in \kg$:
\[
(\widehat \beta_{0, \gb}, \widehat{\boldsymbol \beta}_{\gb}^\top, \widehat \gammab_\gb^\top)^\top  = \argmin_{\beta_{0} \in \mathbb R, \boldsymbol\beta \in \mathbb{R}^d, \gammab \in \Rb^{K-(l-s)}} \sum_{i \in \itrag} W_{1i}^Y(Y_{1i} - \beta_0 - \boldsymbol\beta^{\top} \Xb_{1i} - \gammab^\top \Zb_{1i})^2.
\]
Then, the score function for subgroup $\gb$ is defined as $\mathcal{S}_\gb(\xb, y, \zb) = |y - \widehat \mu_\gb(\xb, \zb)| = |y - \widehat{\beta}_{0, \gb} - \widehat{\boldsymbol\beta}_\gb^\top \xb - \hatgb_{\gb}^\top \zb|$.
For the subgroup-specific threshold $\widehat t_{\alpha, \gb}$, the resulting group-conditional prediction set forms an interval:
\begin{align*}
\widehat C_{\alpha}(\xb, \zbt) = \widehat C_{\alpha, \gb}(\xb, \zb) = \left\{y \in \mathbb{R} : \widehat \mu_\gb(\xb, \zb) - \widehat t_{\alpha, \gb} \le y \le \widehat \mu_\gb(\xb, \zb) + \widehat t_{\alpha, \gb}\right\}.
\end{align*}
This construction yields prediction intervals with constant length $2\widehat t_{\alpha, \gb}$ within each subgroup, while allowing the lengths to vary across subgroups.

\paragraph{2. CQR score. } The subgroup-specific quantile function $q_{\theta, \gb}(\xb, \zb)$ of $Y_t \mid (\Xb_t = \xb, \Zb_t = \zb, \Gb_t = \gb)$ at level $\theta \in (0, 1)$ can be estimated by minimizing the \textit{weighted pinball loss}:
\[
(\widehat\beta_{0, \theta, \gb}, \hatbb_{\theta, \gb}^\top, \hatgb_{\theta, \gb}^\top)^\top  
= \argmin_{\beta_0 \in \mathbb{R}, \bb \in \mathbb{R}^d, \gammab \in \Rb^{K-(l-s)}} 
\sum_{i \in \itrag} W_{1i}^Y 
\rho_\theta(Y_{1i} - \beta_0 - \bb^{\top} \Xb_{1i} - \gammab^\top \Zb_{1i}),
\]
where $\rho_\theta$ is the pinball loss function defined by 
$\rho_\theta(x) = (\theta \mathbf{1}\{x \ge 0\} + (1-\theta)\mathbf{1}\{x < 0\})\cdot |x|$.  
The estimated quantile function is then given by 
$\widehat q_{\theta, \gb}(\xb, \zb) = \widehat \beta_{0, \theta, \gb} + \widehat \bb_{\theta, \gb}^\top \xb + \hatgb_{\theta, \gb}^\top \zb$, 
and the score function is defined as 
$\mathcal S_\gb(\xb, y, \zb) = \max\{\widehat q_{\alpha_{\mathrm{lo}, \gb}, \gb}(\xb, \zb) - y,\; y - \widehat q_{\alpha_{\mathrm{hi}, \gb}, \gb}(\xb, \zb)\}$.  
For the threshold $\widehat t_{\alpha, \gb}$, 
the resulting prediction set takes the form of an interval:
\[
\widehat C_\alpha(\xb, \zbt; \alpha_{\mathrm{lo}, \gb}, \alpha_{\mathrm{hi}, \gb}) 
= \left\{y \in \mathbb{R} : \widehat q_{\alpha_{\mathrm{lo}, \gb}, \gb}(\xb, \zb) - \widehat t_{\alpha, \gb} \le y \le \widehat q_{\alpha_{\mathrm{hi}, \gb}, \gb}(\xb, \zb) + \widehat t_{\alpha, \gb}\right\}.
\]
In this case, the length of the prediction interval becomes 
$\mathrm{Len}(\widehat C_{\alpha, \gb}(\xb, \zb; \alpha_{\mathrm{lo}, \gb}, \alpha_{\mathrm{hi}, \gb})) = 2\widehat t_{\alpha, \gb} + \widehat q_{\alpha_{\mathrm{hi}, \gb}, \gb}(\xb, \zb) - \widehat q_{\alpha_{\mathrm{lo}, \gb}, \gb}(\xb, \zb)$, 
which depends on both $\xb$, $\zb$, and $\gb$; specifically, within each subgroup $\gb$ it is linear in $(\xb, \zb)$ with a subgroup-specific slope.
As in the marginal method, the quantile levels $\alpha_{\mathrm{lo}, \gb}$ and $\alpha_{\mathrm{hi}, \gb}$ are treated as hyperparameters and may be chosen to minimize the prediction interval length.
Given the grids $\mathcal A_{\mathrm{lo}}$ and $\mathcal A_{\mathrm{hi}}$ defined in \eqref{eqn:grid}, for each subgroup $\gb$ we select the pair $(\alpha_{\mathrm{lo}, \gb}^\ast, \alpha_{\mathrm{hi}, \gb}^\ast)$ that minimizes the average interval length:
\[
(\alpha_{\mathrm{lo}, \gb}^\ast, \alpha_{\mathrm{hi}, \gb}^\ast)
:= \argmin_{\alpha_{\mathrm{lo}, \gb} \in \mathcal A_{\mathrm{lo}},\, \alpha_{\mathrm{hi}, \gb} \in \mathcal A_{\mathrm{hi}}} 
\frac{1}{|\itestg|} \sum_{i \in \itestg} \mathrm{Len}\bigl(\widehat C_{\alpha, \gb}(\Xb_{2i}, \Zb_{2i} ; \alpha_{\mathrm{lo}, \gb}, \alpha_{\mathrm{hi}, \gb})\bigr),
\]
using the subgroup-specific test dataset
$
\dtestg = \{(\Xb_{2i}, \Zb_{2i})\}_{i \in \itestg}.
$
As in the marginal case, the pair $(\alpha_{\mathrm{lo}, \gb}^\ast, \alpha_{\mathrm{hi}, \gb}^\ast)$ may be selected using \emph{any} dataset, including the subgroup-specific calibration dataset $\dcalg$.
Finally, the selected prediction set is given by
$\widehat C_\alpha(\cdot, \gb) = \widehat C_\alpha(\cdot, \gb; \alpha_{\mathrm{lo}, \gb}^\ast, \alpha_{\mathrm{hi}, \gb}^\ast)$.

\paragraph{3. Soft classification score. }
For the categorical outcome $Y_{ti} \in [M]$, for each subgroup $\gb$, we train a soft classifier
$\widehat f_\gb(\xb, \zb; \boldsymbol{\theta}_{\gb}) = (\widehat f_{\gb, 1}(\xb, \zb; \boldsymbol{\theta}_{\gb}), \ldots, \widehat f_{\gb, M}(\xb, \zb; \boldsymbol{\theta}_{\gb}))^\top$
using $\dtrag$, while incorporating sampling weights during training, where $\boldsymbol{\theta}_{\gb}$ denotes a vector of hyperparameters for subgroup $\gb$.
With the subgroup-specific score function $\mathcal{S}_\gb(\xb, y, \zb) := 1 - \widehat f_{\gb, y}(\xb, \zb; \boldsymbol{\theta}_{\gb})$ and the threshold $\widehat t_{\alpha, \gb}$, the resulting prediction set is given by
\[
\widehat C_{\alpha}(\xb, \zbt; \boldsymbol{\theta}_{\gb}) = \widehat{C}_{\alpha, \gb}(\xb, \zb; \boldsymbol{\theta}_{\gb}) = \left\{y \in [M] : \widehat{f}_{\gb, y}(\xb, \zb; \boldsymbol{\theta}_{\gb}) \ge 1 - \widehat{t}_{\alpha, \gb}\right\}.
\]
For a grid $\boldsymbol{\Theta}$, for each subgroup $\gb$ we choose $\boldsymbol{\theta}_{\gb}^*$ to minimize the average prediction set size:
\[
\boldsymbol{\theta}_{\gb}^* := \argmin_{\boldsymbol{\theta}_\gb \in \boldsymbol{\Theta}} \frac{1}{|\itestg|} \sum_{i \in \itestg} |\widehat C_{\alpha, \gb}(\Xb_{2i}, \Zb_{2i}; \boldsymbol{\theta}_\gb)|,
\]
using 
$
\dtestg = \{(\Xb_{2i}, \Zb_{2i})\}_{i \in \itestg}.
$
Again, the hyperparameter $\boldsymbol{\theta}_{\gb}^*$ can be selected using \emph{any} dataset. Finally, the selected prediction set is given by $\widehat C_\alpha(\cdot, \gb) = \widehat C_\alpha(\cdot, \gb; \boldsymbol{\theta}_{\gb}^*)$.

\subsection{Implications of Score Choice for Interval Length}
For the continuous outcome, Table~\ref{table:comparison_2} summarizes the resulting interval lengths under the two marginal and two group-conditional methods, where ``AR'' denotes the use of absolute residual scores.
With the absolute residual score, marginal prediction sets yield intervals of constant length, whereas group-conditional sets yield lengths that vary across subgroups but remain constant in $(\xb, \zb)$ within each subgroup.
With the CQR score, marginal sets yield lengths that vary linearly in $(\xb, \zb)$ with a common slope across subgroups $\gb$, while group-conditional sets allow subgroup-specific slopes.
In practice, users may choose the specification that matches their desired interval-length behavior; importantly, our methodology asymptotically guarantees the target coverage regardless of this choice.

\begin{table}[!htbp] 
\centering
\renewcommand{\arraystretch}{1.6}
\caption{Comparison of interval lengths across four choices for continuous outcome}
\label{table:comparison_2}
\vspace{-3mm}
\begin{adjustbox}{max width=\linewidth}
\begin{tabular}{ll}
\toprule
\textbf{Modeling choice} & \textbf{Interval Length} ($\mathrm{Len}(\widehat{C}_\alpha(\xb, \zbt)) = \mathrm{Len}(\widehat{C}_\alpha(\xb, \zb, \gb))$) \\
\midrule
AR (Marginal) & $2\widehat t_\alpha$ \\
AR (Group-conditional) & $2\widehat t_{\alpha, \gb}$ \\
CQR (Marginal) & $2\widehat t_\alpha + (\widehat \beta_{0, \ahi^*} - \widehat \beta_{0, \alo^*}) + (\widehat \bb_{\ahi^*} - \widehat \bb_{\alo^*})^\top\xb + (\hatgb_{\ahi^*} - \hatgb_{\alo^*})^\top \zbt$ \\
CQR (Group-conditional) & $2\widehat t_{\alpha, \gb} + \big(\widehat \beta_{0, \alpha_{\mathrm{hi}, \gb}^\ast, \gb} - \widehat \beta_{0, \alpha_{\mathrm{lo}, \gb}^\ast, \gb} \big) + \big(\widehat \bb_{\alpha_{\mathrm{hi}, \gb}^\ast, \gb} - \widehat \bb_{\alpha_{\mathrm{lo}, \gb}^\ast, \gb} \big)^\top\xb + \big(\hatgb_{\alpha_{\mathrm{hi}, \gb}^\ast, \gb} - \hatgb_{\alpha_{\mathrm{lo}, \gb}^\ast, \gb} \big)^\top\zb$\\
\bottomrule
\end{tabular}
\end{adjustbox}

\end{table}

\section{Reduction of Simplified Designs} \label{app:reduction}

In this section, we show how each sampling design in Table~\ref{tab:simplified-designs} reduces to its corresponding equivalent representation under the assumptions in Section~3.2.
As shown in Table~1, the three designs with Poisson sampling reduce to PPSWOR within a single stratum, and the PPSWOR within selected PSUs reduces to PPSWOR within strata.

\subsection{Poisson Sampling From the Entire Population} \label{app:poisson-pop}

We begin by formally specifying Poisson sampling from the entire population, which serves as a baseline for Poisson sampling designs, in the following assumption.

\begin{apxass}[Poisson sampling from the entire population] \label{ass:poisson-pop}
For each $t \in \{1,2\}$, let $\pi_{t\gb} \in (0,1]$ denote a pre-specified inclusion probability for units in subgroup $\gb$. Each unit with $\Gb_{ti} = \gb$ is then independently sampled with probability $\pi_{t\gb}$. The sampling mechanism is assumed to be independent of $(\rcov_{ti}, \rfull_{ti})$ and $(\Xb_{ti}, Y_{ti}, \Zb_{ti})$.
\end{apxass}

Then, Assumptions~\ref{ass:super} and~\ref{ass:nonresponse} and Assumption~\ref{ass:poisson-pop} together characterize the distribution of subgroup-wise sample sizes of Poisson sampling from the entire population. Throughout, $\mathcal{P}(\lambda)$ denotes the Poisson distribution with rate $\lambda>0$.

\begin{apxthm}[Distribution of subgroup-wise sample sizes]  \label{thm:poisson-pop}
    For each $t \in \{1,2\}$ and $\gb \in \kg$, let $N_{t\gb}$ and $n_{t\gb}$ denote the population size and sample size of subgroup $\gb$ at time $t$, respectively, and assume that $\lim_{N_{t\gb} \rightarrow \infty} N_{t\gb}\pi_{t\gb} = m_{t\gb}$, where $m_{t\gb} > 0$.
    Then, under Assumptions~\ref{ass:super} and~\ref{ass:nonresponse} and Assumption~\ref{ass:poisson-pop}, as $N_{1\gb}, N_{2\gb} \rightarrow \infty$ for all $\gb$, we have
    $n_1 \convd \mathcal{P}(m_1^y)$ and $n_2 \convd \mathcal{P}(m_2^{\xb\zb})$, where $m_1^y := \sum_{\gb} m_{1\gb} a_{1\gb}^y$ and $m_2^{xz} := \sum_{\gb} m_{2\gb} a_{2\gb}^{xz}$, and
    \[
    (n_{t\gb})_{\gb} \mid n_t \convd \mathrm{Mult}(n_t, (q_{t\gb})_{\gb}), \qquad t \in \{1, 2\},
    \]
    where $q_{1\gb} := m_{1\gb}a_{1\gb}^y/m_1^y$ and $q_{2\gb} := m_{2\gb}a_{2\gb}^{xz}/m_2^{xz}$ for each $\gb$.
\end{apxthm}

\begin{proof}[Proof of Theorem~\ref{thm:poisson-pop}]
For each $t \in \{1,2\}$ and $\gb \in \kg$, let $n_{t\gb}^0$ denote the number of sampled units in subgroup $\gb$ before nonresponse. Since units in subgroup $\gb$ are independently included with probability $\pi_{t\gb}$, we have $n_{t\gb}^0 \sim B(N_{t\gb}, \pi_{t\gb})$. Moreover, since $n_{1\gb} \mid n_{1\gb}^0 \sim B(n_{1\gb}^0, a_{1\gb}^y)$ and $n_{2\gb} \mid n_{2\gb}^0 \sim B(n_{2\gb}^0, a_{2\gb}^{xz})$ hold independently across subgroups, it follows that, for each $\gb$,
\[
n_{1\gb} \sim B(N_{1\gb}, \pi_{1\gb}a_{1\gb}^y), \qquad n_{2\gb} \sim B(N_{2\gb}, \pi_{2\gb}a_{2\gb}^{xz}).
\]
Then, by the Poisson limit theorem and the assumption $\lim_{N_{t\gb} \rightarrow \infty} N_{t\gb}\pi_{t\gb} = m_{t\gb}$, we have $n_{1\gb} \convd \mathcal P(m_{1\gb}a_{1\gb}^y)$ and $n_{2\gb} \convd \mathcal P(m_{2\gb}a_{2\gb}^{xz})$ for each $\gb$.
Since $\{n_{t\gb}\}_{\gb}$ are independent, we obtain $n_1 \convd \mathcal{P}(m_1^y)$ and $n_2 \convd \mathcal{P}(m_2^{\xb\zb})$, where $m_1^y := \sum_{\gb} m_{1\gb} a_{1\gb}^y$ and $m_2^{xz} := \sum_{\gb} m_{2\gb} a_{2\gb}^{xz}$, and
\[
(n_{t\gb})_{\gb} \mid n_t \convd \mathrm{Mult}(n_t, (q_{t\gb})_{\gb}), \qquad t \in \{1, 2\},
\]
where $q_{1\gb} := m_{1\gb}a_{1\gb}^y/m_1^y$ and $q_{2\gb} := m_{2\gb}a_{2\gb}^{xz}/m_2^{xz}$ for each $\gb$.
\end{proof}

Theorem~\ref{thm:poisson-pop} implies that, when subgroup-wise population sizes are sufficiently large, for each $t \in \{1,2\}$ and conditional on the realized sample size $n_t$, the sampled subgroup labels may be modeled as i.i.d.\ draws from a categorical distribution.
Specifically, under Assumptions~\ref{ass:super} and~\ref{ass:nonresponse} and Assumption~\ref{ass:poisson-pop}, and conditional on $n_1$ and $n_2$, the survey datasets $\done$ and $\dtwo$ may be regarded as i.i.d.\ samples:
\begin{align*} 
(\Xb_{1i}, Y_{1i}, \Zb_{1i}, \Gb_{1i}) &\iid P_{(\Xb_1, Y_1, \Zb_1) \mid \Gb_1} \times P_{\Gb_{1} \mid \rfull_{1} = 1}, \quad i \in \ione,  \\
(\Xb_{2i}, \Zb_{2i}, \Gb_{2i}) &\iid P_{(\Xb_2, \Zb_2) \mid \Gb_2} \times P_{\Gb_{2} \mid \rcov_{2} = 1}, \quad i \in \itwo,
\end{align*}
where $P_{\Gb_{1} \mid \rfull_{1} = 1} := \mathrm{Cat}((q_{1\gb})_{\gb})$ and $P_{\Gb_{2} \mid \rcov_{2} = 1} := \mathrm{Cat}((q_{2\gb})_{\gb})$, with $q_{1\gb}$ and $q_{2\gb}$ specified in Theorem~\ref{thm:poisson-pop}.
Therefore, Poisson sampling from the entire population reduces to the PPSWOR-within-a-single-stratum design, corresponding to Section~\ref{sec-survey-prob} with $H=1$.

\subsection{Poisson Sampling Within Strata} \label{app:poisson-strata}

Under Poisson sampling, units are sampled independently with inclusion probabilities that may depend on the demographic subgroup but do not depend on stratum membership. Consequently, stratification plays no role: whether Poisson sampling is applied within each stratum or to the entire population yields the same sampling distribution. Therefore, Poisson sampling within strata is equivalent to Poisson sampling from the entire population.

\subsection{Poisson Sampling Within Selected PSUs} \label{app:poisson-psu}

Here, the finite population is partitioned into $J$ primary sampling units (PSUs), indexed by $j = 1, \dots, J$.
We begin by formally specifying the sampling design of Poisson sampling within selected PSUs in the following assumption.

\begin{apxass}[Poisson sampling within selected PSUs] \label{ass:poisson-psu}
For each $t \in \{1,2\}$, let $\pi_{t\gb} \in (0,1]$ denote a pre-specified inclusion probability for units in subgroup $\gb$.
Sampling proceeds in two stages: we first sample PSUs and then sample individuals within the selected PSUs.

\begin{enumerate}
    \item Let $M_{tj}$ denote the measure of size (MOS) of PSU $j$ at time $t$, with $\sum_{j=1}^J M_{tj} = 1$. We then select $J_0$ PSUs via PPSWOR using $\{M_{tj}\}_{j=1}^J$, so that the inclusion probability of PSU $j$ is $J_0 M_{tj}$. Let $\mathcal J_t \subset [J]$ denote the set of PSUs selected at time $t$.

    \item To achieve the pre-specified inclusion probability, for each selected PSU $j \in \mathcal J_t$ we independently sample units in subgroup $\gb$ with conditional inclusion probability $\pi_{t\gb}/(J_0 M_{tj})$. The sampling mechanism is assumed to be independent of $(\rcov_{ti}, \rfull_{ti})$ and $(\Xb_{ti}, Y_{ti}, \Zb_{ti})$. Then, the overall inclusion probability for a unit in subgroup $\gb$ in PSU $j$ becomes
    \begin{align*}
    &\Pb(\text{PSU } j \text{ is selected}) \cdot \Pb(\text{a unit in subgroup } \gb \text{ in PSU } j \text{ is selected} \mid \text{PSU } j \text{ is selected}) \\
    &= J_0 M_{tj} \cdot \frac{\pi_{t\gb}}{J_0 M_{tj}} = \pi_{t\gb}.
    \end{align*}
\end{enumerate}
\end{apxass}

Then, under Assumptions~\ref{ass:super} and~\ref{ass:nonresponse} and Assumption~\ref{ass:poisson-psu}, conditional on the PSU selection, we can characterize the distribution of subgroup-wise sample sizes. 

\begin{apxthm}[Distribution of subgroup-wise sample sizes] \label{thm:poisson-psu}
For each $t \in \{1,2\}$, subgroup $\gb \in \kg$, and selected PSU $j \in \mathcal J_t$, let $N_{tj\gb}$ and $n_{tj\gb}$ denote the population size and sample size of subgroup $\gb$ in PSU $j$ at time $t$, respectively.
Moreover, define $n_{t\gb} := \sum_{j \in \mathcal J_t} n_{tj\gb}$ and $n_t := \sum_{\gb} n_{t\gb}$.
Assume that $\lim_{N_{tj\gb} \to \infty} N_{tj\gb}\pi_{t\gb} = m_{tj\gb}$ for each $t$, $j$, and $\gb$.
Then, under Assumptions~\ref{ass:super} and~\ref{ass:nonresponse} and Assumption~\ref{ass:poisson-psu}, as $N_{tj\gb} \to \infty$ for all $t$, $j$, and $\gb$ and conditional on the selected PSUs $\mathcal J_t$, we have $n_1 \convd \mathcal P(m_1^y)$ and $n_2 \convd \mathcal P(m_2^{xz})$, where $m_1^y := \sum_{\gb}\sum_{j \in \mathcal J_1} m_{1j\gb}a_{1\gb}^y$ and $m_2^{xz} := \sum_{\gb}\sum_{j \in \mathcal J_2} m_{2j\gb}a_{2\gb}^{xz}$, and
\[
(n_{t\gb})_{\gb} \mid n_t \convd \mathrm{Mult}(n_t, (q_{t\gb})_\gb), \qquad t \in \{1, 2\},
\]
where $q_{1\gb} := \sum_{j \in \mathcal J_1} m_{1j\gb}a_{1\gb}^y/m_1^y$ and $q_{2\gb} := \sum_{j \in \mathcal J_2} m_{2j\gb}a_{2\gb}^{xz}/m_2^{xz}$ for each $\gb$.
\end{apxthm}

\begin{proof}
The proof is essentially the same as that of Theorem~\ref{thm:poisson-pop}. For each $\gb \in \kg$ and $j \in \mathcal J_t$, we have $n_{1j\gb} \sim B(N_{1j\gb}, \pi_{1\gb}a_{1\gb}^y)$ and $n_{2j\gb} \sim B(N_{2j\gb}, \pi_{2\gb}a_{2\gb}^{xz})$. Then, by the Poisson limit theorem and the assumption $\lim_{N_{tj\gb} \to \infty} N_{tj\gb}\pi_{t\gb} = m_{tj\gb}$, we have $n_{1j\gb} \convd \mathcal P(m_{1j\gb}a_{1\gb}^y)$ and $n_{2j\gb} \convd \mathcal P(m_{2j\gb}a_{2\gb}^{xz})$. Moreover, since $\{n_{tj\gb}\}_{j,\gb}$ are independent, we obtain
\[
n_{1\gb} \convd \mathcal P\big(\sum_{j \in \mathcal J_1} m_{1j\gb}a_{1\gb}^y\big), \quad n_{2\gb} \convd \mathcal P\big(\sum_{j \in \mathcal J_2} m_{2j\gb}a_{2\gb}^{xz}\big), \qquad \forall \gb \in \kg,
\]
and $n_1 \convd \mathcal P(m_1^y)$ and $n_2 \convd \mathcal P(m_2^{xz})$, where
\[
m_1^y := \sum_{\gb}\sum_{j \in \mathcal J_1} m_{1j\gb}a_{1\gb}^y, \qquad m_2^{xz} := \sum_{\gb}\sum_{j \in \mathcal J_2} m_{2j\gb}a_{2\gb}^{xz}.
\]
Therefore, for each $t \in \{1, 2\}$ we obtain $(n_{t\gb})_{\gb} \mid n_t \convd \mathrm{Mult}(n_t, (q_{t\gb})_\gb)$, where 
\[
q_{1\gb} := \sum_{j \in \mathcal J_1} m_{1j\gb}a_{1\gb}^y/m_1^y \quad \text{and}  \quad q_{2\gb} := \sum_{j \in \mathcal J_2} m_{2j\gb}a_{2\gb}^{xz}/m_2^{xz}
\]
for each $\gb \in \kg$.
\end{proof}

Theorem~\ref{thm:poisson-psu} implies that, when population sizes are sufficiently large in each subgroup--PSU cell, conditional on the PSU selection and the realized sample sizes $n_1$ and $n_2$, the sampled subgroup labels may be modeled as i.i.d.\ draws. Consequently, as in Appendix~\ref{app:poisson-pop}, the survey datasets $\done$ and $\dtwo$ may be regarded as i.i.d.\ samples, and thus Poisson sampling within selected PSUs again reduces to the PPSWOR-within-a-single-stratum design.

\subsection{PPSWOR Within Selected PSUs} \label{app:ppswor-psu}

As in Appendix~\ref{app:poisson-psu}, the finite population is partitioned into $J$ primary sampling units (PSUs), indexed by $j = 1, \dots, J$. We begin by formally specifying the PPSWOR-within-selected-PSUs sampling design in the following assumption.

\begin{apxass}[PPSWOR within selected PSUs] \label{ass:ppswor-psu}
For each $t \in \{1,2\}$, the survey is designed to produce a self-weighting sample within each subgroup by targeting a pre-specified inclusion probability $\pi_{t\gb} \in (0,1]$ for all units in subgroup $\gb$. Sampling proceeds in two stages: we first select PSUs and then sample individuals within the selected PSUs.

\begin{enumerate}
    \item Let $M_{tj}$ denote the measure of size (MOS) of PSU $j$ at time $t$, with $\sum_{j=1}^J M_{tj} = 1$. We then select $J_0$ PSUs via PPSWOR using $\{M_{tj}\}_{j=1}^J$, so that the inclusion probability of PSU $j$ is $J_0 M_{tj}$. Let $\mathcal J_t \subset [J]$ denote the set of PSUs selected at time $t$.

    \item Let $N_{tj\gb}$ denote the population size of subgroup $\gb$ in PSU $j$. Within each selected PSU $j \in \mathcal J_t$, we sample $m_{tj}$ units via PPSWOR, using inclusion probabilities as size measures. Sampling is assumed to be independent of $(\rcov_{ti}, \rfull_{ti})$ and $(\Xb_{ti}, Y_{ti}, \Zb_{ti})$. The overall selection probability for an individual in subgroup $\gb$ in PSU $j$ is then
    \begin{align*}
    &\Pb(\text{PSU } j \text{ is selected}) \cdot \Pb(\text{a unit in subgroup } \gb \text{ in PSU } j \text{ is selected} \mid \text{PSU } j \text{ is selected}) \\
    &= J_0 M_{tj} \cdot m_{tj} \frac{\pi_{t\gb}}{\sum_{\gb'} N_{tj\gb'}\pi_{t\gb'}}.
    \end{align*}
    The within-PSU sample size $m_{tj}$ is chosen to satisfy the self-weighting condition, namely $m_{tj} = \sum_{\gb} N_{tj\gb} \pi_{t\gb}/(J_0M_{tj})$ for each selected PSU $j \in \mathcal J_t$.
\end{enumerate}
\end{apxass}

As in Section~\ref{sec-survey-ass}, the subgroup composition can be approximated by a with-replacement model. Specifically, for each $j \in \mathcal J_t$, the subgroup probabilities are given by $\tilde q_{tj\gb} := N_{tj\gb}\pi_{t\gb}/\sum_{\gb'} N_{tj\gb'}\pi_{t\gb'}$, which motivates the multinomial approximation in Assumption~\ref{ass:multinomial-ppswor-psu}.

\begin{apxass}[Multinomial approximation for subgroup-wise sample counts] \label{ass:multinomial-ppswor-psu}
For each $t \in \{1,2\}$ and $j \in \mathcal J_t$, letting $m_{tj\gb}$ denote the number of sampled units in subgroup $\gb$ within the selected PSU $j \in \mathcal J_t$ at time $t$ (before nonresponse), we assume $(m_{tj\gb})_{\gb} \sim \mathrm{Mult}(m_{tj}, (\tilde q_{tj\gb})_{\gb})$.
\end{apxass}

The following theorem characterizes the distribution of within-PSU subgroup-wise sample sizes under Assumptions~\ref{ass:super} and~\ref{ass:nonresponse} and Assumptions~\ref{ass:ppswor-psu} and~\ref{ass:multinomial-ppswor-psu}.

\begin{apxthm}[The distribution of within-PSU subgroup-wise sample sizes] \label{thm:ppswor-psu}
    For each $t \in \{1,2\}$ and $j \in \mathcal J_t$, let $\mathcal D_{tj}$ denote the subset of $\mathcal D_t$ in PSU $j$, and let $n_{tj} = |\mathcal D_{tj}|$. For each $\gb \in \kg$, let $n_{tj\gb}$ denote the sample size of subgroup $\gb$ in the selected PSU $j$ at time $t$. Then, under Assumptions~\ref{ass:super} and~\ref{ass:nonresponse} and Assumptions~\ref{ass:ppswor-psu} and~\ref{ass:multinomial-ppswor-psu}, we have
    \[
    (n_{tj\gb})_{\gb} \mid n_{tj} \sim \mathrm{Mult}(n_{tj}, (q_{tj\gb})_{\gb}), \qquad t \in \{1, 2\}, \ j \in \mathcal J_t,
    \]
    where $q_{1j\gb} = N_{1j\gb}\pi_{1\gb}a_{1\gb}^y/\sum_{\gb'} N_{1j\gb'}\pi_{1\gb'}a_{1\gb'}^y$ and $q_{2j\gb} = N_{2j\gb}\pi_{2\gb}a_{2\gb}^{xz}/\sum_{\gb'} N_{2j\gb'}\pi_{2\gb'}a_{2\gb'}^{xz}$.
\end{apxthm}

Theorem~\ref{thm:ppswor-psu} implies that, conditional on the PSU selection and the within-PSU realized sample sizes, the sampled subgroup labels may be modeled as i.i.d.\ draws from a categorical distribution within each selected PSU.
For each selected PSU $j \in \mathcal J_t$, let $\mathcal I_{1j} \subset [N_1]$ and $\mathcal I_{2j} \subset [N_2]$ denote the index sets corresponding to the datasets $\mathcal D_{1j}$ and $\mathcal D_{2j}$, respectively.
Then, under Assumptions~\ref{ass:super} and~\ref{ass:nonresponse} and Assumptions~\ref{ass:ppswor-psu} and~\ref{ass:multinomial-ppswor-psu}, conditional on $\mathcal J_1$ and $\mathcal J_2$ and on $n_{1j}$ and $n_{2j}$, the survey datasets $\mathcal D_{1j}$ and $\mathcal D_{2j}$ may be regarded as i.i.d.\ samples:
\begin{align*} 
(\Xb_{1i}, Y_{1i}, \Zb_{1i}, \Gb_{1i}) &\iid P_{(\Xb_1, Y_1, \Zb_1) \mid \Gb_1} \times P_{\Gb_{1j} \mid \rfull_{1j} = 1}, \quad i \in \mathcal I_{1j}, \ \ j \in \mathcal J_1,  \\
(\Xb_{2i}, \Zb_{2i}, \Gb_{2i}) &\iid P_{(\Xb_2, \Zb_2) \mid \Gb_2} \times P_{\Gb_{2j} \mid \rcov_{2j} = 1}, \quad i \in \mathcal I_{2j}, \ \ j \in \mathcal J_2, 
\end{align*}
where $P_{\Gb_{1j} \mid \rfull_{1j} = 1} := \mathrm{Cat}((q_{1j\gb})_{\gb})$ and $P_{\Gb_{2j} \mid \rcov_{2j} = 1} := \mathrm{Cat}((q_{2j\gb})_{\gb})$, with $q_{1j\gb}$ and $q_{2j\gb}$ specified in Theorem~\ref{thm:ppswor-psu}. Therefore, conditional on the PSU selection, the PPSWOR-within-selected-PSUs design reduces to the PPSWOR-within-strata design, corresponding to the setting of Section~\ref{sec-survey-prob} with $H = |\mathcal J_1|$.

\section{Extension to Geographic Sampling Domains} \label{app:geo}

In Section~\ref{sec-survey-1-1}, we mentioned that our analysis can also be applied to complex surveys that adopt geographic sampling domains. Here, we describe how to apply our method to surveys designed to yield self-weighting samples within each geographic stratum, such as the U.S.\ Current Population Survey (CPS) \citep{cps}.

In this case, there is no need to decompose the categorical covariates $\Zbt_{ti}$ into non-demographic and demographic variables. Instead, letting $H$ denote the number of geographic strata, we define $S_{ti} \in [H]$ as the stratum membership of the $i$th unit at time $t$, and treat $S_{ti}$ as fixed by design.
Since sampling domains are selected to capture key population heterogeneity, the geographic stratum indicator $S_{ti}$ is a particularly informative covariate for prediction. Therefore, we use $(\Xb_{ti}, \Zbt_{ti}, S_{ti})$ as covariates, and the datasets $\done$ and $\dtwo$ have the following forms:
\begin{align*}
    \done = \{(\Xb_{1i}, Y_{1i}, \Zbt_{1i}, S_{1i}, \wfull_{1i})\}_{i \in \ione}, \qquad
    \dtwo = \{(\Xb_{2i}, \Zbt_{2i}, S_{2i}, \wcov_{2i})\}_{i \in \itwo}.
\end{align*}

With these definitions, we impose several working assumptions and apply WCP in essentially the same way as in the demographic sampling domain case. In particular, the categorical covariates $\Zbt_{ti}$ play the role of the non-demographic categorical covariates $\Zb_{ti}$, and the stratum membership $S_{ti}$ plays the role of the demographic subgroup variable $\Gb_{ti}$.

\begin{apxass}[Superpopulation framework within geographic strata] \label{ass:app:geo-super}
    For each \(t \in \{1, 2\}\) and \(h \in [H]\), for all units $i \in [N_t]$ with $S_{ti}= h$, the variables $(\Xb_{ti}, Y_{ti}, \Zbt_{ti})$ are assumed to be independently drawn from an underlying distribution
    \begin{equation*}
        P_{(\Xb_t, Y_t, \Zbt_t) \mid S_t = h} 
        = P_{Y_t \mid \Xb_t, \Zbt_t, S_t = h} \times P_{\Xb_t \mid \Zbt_t, S_t = h} \times P_{\Zbt_t \mid S_t = h},
    \end{equation*}
    where the right-hand side is a factorization of the distribution on the left-hand side.
\end{apxass}


\begin{apxass}[Nonresponse mechanism] \label{ass:app:geo-nonresponse}
    The response indicators \((\rcov_{ti}, \rfull_{ti})\) are independent across units and depend only on \(S_{ti}\).
    Specifically, for each $t \in \{1, 2\}$ and $h \in [H]$, for any unit \(i \in [N_t]\) with \(S_{ti} = h\), we have
    $\Pb(\rcov_{ti} = 1) = a_{th}^{xz}$, $\Pb(\rfull_{ti} = 1) = a_{th}^y$, and
    \((\rcov_{ti}, \rfull_{ti}) \indep (\Xb_{ti}, Y_{ti}, \Zbt_{ti})\),
    where $a_{th}^{xz}, a_{th}^y \in (0, 1]$ are unknown parameters.
\end{apxass}

As before, we consider the six simplified designs listed in Table~\ref{tab:simplified-designs}. Recall that, in the previous setting, all designs reduce to PPSWOR within strata, as shown in Appendix~\ref{app:reduction}. In the present setting, a similar reduction argument applies, but the six designs now reduce to two distinct designs. Specifically, the three Poisson sampling designs reduce to PPSWOR from the entire population, whereas PPSWOR within PSUs reduces to PPSWOR within strata; in this case, these two designs are distinct. Briefly, this is because the stratum membership variable $S_{ti}$ not only defines the sampling domain but is also used in the sampling design.

\begin{apxass}[Sampling designs] \label{ass:app:geo-design}
    We consider the following two designs. Under both designs, for each $t \in \{1, 2\}$, the survey is designed to yield a self-weighting sample within each geographic stratum by targeting a pre-specified inclusion probability $\pi_{th} \in (0, 1]$ for units in stratum $h$. Moreover, the sampling is independent of $(\rcov_{ti}, \rfull_{ti})$ and $(\Xb_{ti}, Y_{ti}, \Zbt_{ti})$.
    Let $N_{th}$ denote the population size of stratum $h$, which is treated as known.
    \begin{itemize}
        \item[(a)] {\upshape (PPSWOR from the entire population).} For each $t \in \{1, 2\}$, we sample $m_t$ units using PPSWOR, where the overall inclusion probability of an individual in stratum $h$ is given by $m_t \cdot (\pi_{th} / \sum_{h'=1}^H N_{th'}\pi_{th'})$. The sample size $m_t$ is chosen to satisfy the self-weighting assumption, namely $m_t = \sum_{h=1}^H N_{th} \pi_{th}$.
    
        \item[(b)] {\upshape (PPSWOR within strata).} For each $t \in \{1, 2\}$ and each stratum $h \in [H]$, we sample $m_{th}$ units using PPSWOR, where the overall inclusion probability of an individual is given by $m_{th} \cdot (\pi_{th} /  (N_{th}\pi_{th})) = m_{th}/N_{th}$ (because all units within the same stratum share the same inclusion probability, this is in fact SRSWOR within strata). The sample size $m_{th}$ is chosen to satisfy the self-weighting assumption, namely $m_{th} = N_{th} \pi_{th}$ for each $h$.
    \end{itemize}
\end{apxass}

For design (a) in Assumption~\ref{ass:app:geo-design}, because the population is much larger than the sample size, the stratum composition can be well approximated by a with-replacement model with stratum probabilities $\tilde q_{th} := N_{th}\pi_{th}/\sum_{h'=1}^H N_{th'}\pi_{th'}$ for each $h \in [H]$. This leads to the multinomial approximation in Assumption~\ref{ass:app:geo-multinomial}. 

\begin{apxass}[Multinomial approximation for stratum-wise sample counts] \label{ass:app:geo-multinomial}
    For design (a) in Assumption~\ref{ass:app:geo-design}, let $n_{th}^0$ denote the number of sampled units in stratum $h$ at time $t$ (before nonresponse). We assume $(n_{th}^0)_{h} \sim \mathrm{Mult}(m_t, (\tilde q_{th})_{h})$.
\end{apxass}

Due to nonresponse, the realized sample size within each stratum is random.
The following theorem characterizes the distribution of stratum-wise sample sizes under the assumptions introduced so far, for both designs in Assumption~\ref{ass:app:geo-design}.

\begin{apxthm}[The distribution of stratum-wise sample sizes] \label{thm:app:geo-multinomial}
    For each $t \in \{1, 2\}$ and $h \in [H]$, let $\mathcal D_{th}$ denote the subset of $\mathcal D_t$ in stratum $h$, and let $n_{th} := |\mathcal D_{th}|$. 
    \begin{itemize}
        \item[(a)] {\upshape (PPSWOR from the entire population).} Under Assumptions~\ref{ass:app:geo-super}, \ref{ass:app:geo-nonresponse}, \ref{ass:app:geo-design}(a), and \ref{ass:app:geo-multinomial}, we have $(n_{th})_h \mid n_t \sim \mathrm{Mult}(n_t, (q_{th})_h)$ for each $t \in \{1, 2\}$, where $q_{1h} = N_{1h}\pi_{1h}a_{1h}^y/\sum_{h=1}^H N_{1h}\pi_{1h}a_{1h}^y$ and $q_{2h} = N_{2h}\pi_{2h}a_{2h}^{xz}/\sum_{h=1}^H N_{2h}\pi_{2h}a_{2h}^{xz}$ for each $h \in [H]$.

        \item[(b)] {\upshape (PPSWOR within strata).} Under Assumptions~\ref{ass:app:geo-super}, \ref{ass:app:geo-nonresponse}, and \ref{ass:app:geo-design}(b), for each $h \in [H]$, we have $n_{1h} \sim \mathrm{Binomial}(N_{1h}\pi_{1h}, a_{1h}^y)$ and $n_{2h} \sim \mathrm{Binomial}(N_{2h}\pi_{2h}, a_{2h}^{xz})$.
    \end{itemize}
\end{apxthm}

We finally assume that the conditional distribution of the outcome given the covariates, including stratum membership, is invariant over time.


\begin{apxass}[Temporal invariance of the conditional distribution] \label{ass:app:geo-invariant} 
    For $P_{Y_1 \mid (\Xb_1, \Zbt_1, S_1)}$ and $P_{Y_2 \mid (\Xb_2, \Zbt_2, S_2)}$ defined in Assumption~\ref{ass:app:geo-super}, we assume that
    \[
    P_{Y_1 \mid (\Xb_1 = \xb, \Zbt_1 = \zbt, S_1 = h)} = P_{Y_2 \mid (\Xb_2 = \xb, \Zbt_2 = \zbt, S_2 = h)}, 
    \quad \forall \xb \in \mathbb R^d,\ \zbt \in \k,\ h \in [H].
    \]
\end{apxass}

\noindent
We now specify the problem setup under Assumptions~\ref{ass:app:geo-super}--\ref{ass:app:geo-invariant} for both designs.

\subsection{PPSWOR From the Entire Population} \label{app:geo-pop}

Under Assumptions~\ref{ass:app:geo-super}, \ref{ass:app:geo-nonresponse}, \ref{ass:app:geo-design}(a), and \ref{ass:app:geo-multinomial}, the training and target distributions for design (a) in Assumption~\ref{ass:app:geo-design} can be characterized explicitly.  
Although the stratum membership $S_{ti}$ is treated as fixed at the population level, Theorem~\ref{thm:app:geo-multinomial}(a) implies that, conditional on the realized sample size $n_t$, the stratum-wise sample counts $(n_{th})_h$ may be modeled as arising from i.i.d.\ draws from a categorical distribution for each $t \in \{1, 2\}$. 
Then, under Assumptions~\ref{ass:app:geo-super}, \ref{ass:app:geo-nonresponse}, \ref{ass:app:geo-design}(a), and \ref{ass:app:geo-multinomial}, and conditional on $n_1$ and $n_2$, the survey datasets $\done$ and $\dtwo$ may be regarded as i.i.d.\ samples:
\begin{align} 
(\Xb_{1i}, Y_{1i}, \Zbt_{1i}, S_{1i}) &\iid P_{(\Xb_1, Y_1, \Zbt_1) \mid S_1} \times P_{S_1 \mid \rfull_1 = 1}, \quad i \in \ione, \label{eqn:app:geo-pop-done} \\
(\Xb_{2i}, \Zbt_{2i}, S_{2i}) &\iid P_{(\Xb_2, \Zbt_2) \mid S_2} \times P_{S_{2} \mid \rcov_{2} = 1}, \quad \ i \in \itwo, \label{eqn:app:geo-pop-dtwo}
\end{align}
where $P_{(\Xb_1, Y_1, \Zbt_1) \mid S_1}$ and $P_{(\Xb_2, \Zbt_2) \mid S_2}$ are defined in Assumption~\ref{ass:app:geo-super}. Moreover, $P_{S_1 \mid \rfull_1 = 1} := \mathrm{Cat}((q_{1h})_h)$ and $P_{S_2 \mid \rcov_2 = 1} := \mathrm{Cat}((q_{2h})_h)$, where $\mathrm{Cat}(\cdot)$ denotes the categorical distribution and $q_{1h}$ and $q_{2h}$ are specified in Theorem~\ref{thm:app:geo-multinomial}(a).
The target distribution is then given by
\begin{equation} \label{eqn:app:geo-pop-target}
(\Xb_2, Y_2, \Zbt_2, S_2) \sim P_{(\Xb_2, Y_2, \Zbt_2) \mid S_2} \times P_{S_2},
\end{equation}
where $P_{(\Xb_2, Y_2, \Zbt_2) \mid S_2}$ is defined in Assumption~\ref{ass:app:geo-super}, $P_{S_2} := \mathrm{Cat}((p_{2h})_{h})$, and $p_{2h} := N_{2h}/N_2$.

In this case, the training distribution is \eqref{eqn:app:geo-pop-done} and the target distribution is \eqref{eqn:app:geo-pop-target}. Therefore, the construction of the DA-WCP prediction set follows the procedure in Section~\ref{sec-method}, after the substitutions $\Zb_{ti}\mapsto \Zbt_{ti}$ and $\Gb_{ti}\mapsto S_{ti}$.
Specifically, for any $\xb \in \Rb^d$, $\zbt \in \k$, and $h \in [H]$, the WCP weight function at covariate $(\xb, \zbt, h)$ is given by
\begin{equation*} \label{eqn:app:geo-pop-weight}
    w_h(\xb, \zbt, h) := \left(\frac{dP_{\Xb_2 \mid \Zbt_2 = \zbt, S_2 = h}}{dP_{\Xb_1 \mid \Zbt_1 = \zbt, S_1 = h}}\right)(\xb) \cdot \frac{P_{\Zbt_2 \mid S_2 = h}(\zbt)}{P_{\Zbt_1 \mid S_1 = h}(\zbt)} \cdot \frac{P_{S_2}(h)}{P_{S_1 \mid \rfull_1 = 1}(h)} 
    = r(\xb; \zbt, h) \cdot \frac{v_{2h}(\zbt)}{v_{1h}(\zbt)} \cdot \frac{p_{2h}}{q_{1h}},
\end{equation*}
where $r(\cdot \ ; \zbt, h) := \Big(\tfrac{dP_{\Xb_2 \mid \Zbt_2 = \zbt, S_2 = h}}{dP_{\Xb_1 \mid \Zbt_1 = \zbt, S_1 = h}}\Big)(\cdot)$ denotes the density ratio function for the subgroup--stratum cell $(\zbt, h)$, $v_{1h}(\zbt) := P_{\Zbt_1 \mid S_1 = h}(\zbt)$ and $v_{2h}(\zbt) := P_{\Zbt_2 \mid S_2 = h}(\zbt)$ denote the proportions of subgroup $\zbt$ within stratum $h$, and $q_{1h}$ is the proportion of stratum $h$ under $P_{S_1 \mid \rfull_1 = 1}$ in~\eqref{eqn:app:geo-pop-done}.

We estimate these quantities using our survey datasets.
For each $(\zbt, h)$, the density ratio function $r(\cdot \ ; \zbt, h)$ is estimated via the KLIEP method (see Appendix~\ref{app:kliep-1}), applied to $\{\Xb_{1i} \mid i \in \ione, \Zbt_{1i} = \zbt, S_{1i} = h\}$ and $\{\Xb_{2i} \mid i \in \itwo, \Zbt_{2i} = \zbt, S_{2i} = h\}$.
Moreover, we estimate 
\[
\widehat v_{1h}(\zbt) = \frac{\sum_{i \in \ittra} \mathbf 1\{\Zbt_{1i} = \zbt, S_{1i} = h\}}{\sum_{i \in \ittra} \mathbf 1\{S_{1i} = h \}}, \qquad \widehat v_{2h}(\zbt) = \frac{\sum_{i \in \itwo} \mathbf 1\{\Zbt_{2i} = \zbt, S_{2i} = h\}}{\sum_{i \in \itwo} \mathbf 1\{S_{2i} = h\}}.
\]
Finally, for each $h \in [H]$, we estimate $q_{1h}$ by $\widehat q_{1h} = \sum_{i \in \itra} \mathbf 1\{S_{1i} = h\} / \ntra$.

Given an estimated WCP weight function $\widehat w$ and nonconformity scores $\{s_i\}_{i \in \mathcal{I}_{\mathrm{cal}}}$, for a miscoverage level $\alpha \in (0,1)$, we compute the WCP-type threshold as
\begin{equation*}
    \widehat{t}_\alpha := \text{Quantile}_{1-\alpha}\Bigg(\sum_{i \in \ical} \frac{\widehat{w}(\Xb_{1i}, \Zbt_{1i}, S_{1i})}{\sum_{j \in \ical} \widehat{w}(\Xb_{1j}, \Zbt_{1j}, S_{1j})} \cdot\delta_{s_i} \Bigg).
\end{equation*}
We then construct the marginal prediction set for $Y_2$ at covariates $(\Xb_2, \Zbt_2, S_2)=(\xb,\zbt, h)$ with miscoverage level $\alpha$ as
$\widehat{C}_{\alpha}(\xb, \zbt, h) = \{y \in \mathcal Y : \mathcal{S}(\xb, y, \zbt, h) \le \widehat{t}_\alpha \}$, where $\mathcal{S}(\xb, y, \zbt, h)$ denotes a score function.
The analogous theoretical coverage guarantees given in Section~\ref{sec-theoretical} then hold for the selected prediction set $\widehat C_{\alpha, \hat \ell}$ with $\widehat \ell \in [L]$, under Assumptions~\ref{ass:app:geo-super}--\ref{ass:app:geo-invariant} and appropriate regularity conditions tailored to this setting.

\subsection{PPSWOR Within Strata} \label{app:geo-strata}

For design (b) in Assumption~\ref{ass:app:geo-design}, Theorem~\ref{thm:app:geo-multinomial}(b) implies that the stratum-wise sample counts can be modeled as independent binomial draws across strata. However, the Poisson limit theorem cannot be invoked to approximate them by Poisson draws, since $a_{1h}^y$ and $a_{2h}^{xz}$ are positive constants. Consequently, the conditional distribution of $(n_{th})_h \mid n_t$ cannot be approximated by a multinomial distribution. As a result, even conditional on the realized sample size, the sample stratum labels cannot be modeled as i.i.d.\ draws from a categorical distribution. Therefore, the survey datasets $\done$ and $\dtwo$ cannot be treated as i.i.d.\ samples, and marginal validity cannot be established directly using our WCP-based analysis.

Instead, we exploit the fact that the survey data can be viewed as i.i.d.\ samples \emph{within} each stratum, and we construct prediction sets with stratum-conditional coverage guarantee.
Specifically, for each stratum $h \in [H]$, define $\ioneh = \{i \in \ione \mid S_{1i} = h\}$ and $\itwoh = \{i \in \itwo \mid S_{2i} = h\}$.
Then, under Assumptions~\ref{ass:app:geo-super}, \ref{ass:app:geo-nonresponse}, and \ref{ass:app:geo-design}(b), and conditional on $n_{1h} = |\ioneh|$ and $n_{2h} = |\itwoh|$, the datasets $\doneh$ and $\dtwoh$ may be regarded as i.i.d.\ samples:
\begin{align*} 
(\Xb_{1i}, Y_{1i}, \Zbt_{1i}) &\iid P_{(\Xb_1, Y_1, \Zbt_1) \mid S_1 = h}, \qquad i \in \ioneh,\ h \in [H], \\
(\Xb_{2i}, \Zbt_{2i}) &\iid P_{(\Xb_2, \Zbt_2) \mid S_2 = h}, \qquad \ \ \ i \in \itwoh,\ h \in [H]. 
\end{align*}
The target distribution for stratum $h$ is $P_{(\Xb_2, Y_2, \Zbt_2)\mid S_2 = h}$, the population-level conditional distribution of $(\Xb_2, Y_2, \Zbt_2)$ given $S_2 = h$, as defined in Assumption~\ref{ass:app:geo-super}.

For each stratum $h \in [H]$, let $w_h$ denote the stratum-specific WCP weight function with $P_{(\Xb_1, Y_1, \Zbt_1) \mid S_1 = h}$ as a training distribution and $P_{(\Xb_2, Y_2, \Zbt_2)\mid S_2 = h}$ as a target distribution. Then, for any $\xb \in \Rb^d$ and $\zbt \in \k$, under Assumptions~\ref{ass:app:geo-super}, \ref{ass:app:geo-nonresponse}, \ref{ass:app:geo-design}(b), and \ref{ass:app:geo-invariant}, the WCP weight at covariate $(\xb,\zbt)$ in stratum $h$ is given by
\begin{align*}
w_h(\xb, \zbt) := \left(\frac{dP_{\Xb_2 \mid \Zbt_2 = \zbt, S_2 = h}}{dP_{\Xb_1 \mid \Zbt_1 = \zbt, S_1 = h}}\right)(\xb) \cdot \frac{P_{\Zbt_2 \mid S_2 = h}(\zbt)}{P_{\Zbt_1 \mid S_1 = h}(\zbt)} = r(\xb; \zbt, h) \cdot \frac{v_{2h}(\zbt)}{v_{1h}(\zbt)},
\end{align*}
where $r(\,\cdot\,;\zbt,h) := \Big(\tfrac{dP_{\Xb_2 \mid \Zbt_2 = \zbt, S_2 = h}}{dP_{\Xb_1 \mid \Zbt_1 = \zbt, S_1 = h}}\Big)(\,\cdot\,)$, $v_{1h}(\zbt) := P_{\Zbt_1 \mid S_1 = h}(\zbt)$, and $v_{2h}(\zbt) := P_{\Zbt_2 \mid S_2 = h}(\zbt)$. These quantities are estimated using our survey datasets as described in Appendix~\ref{app:geo-pop}.

For each $h \in [H]$, given an estimated stratum-specific WCP weight function $\widehat w_h$ and nonconformity scores $\{s_i\}_{i \in \icalh}$, we compute the stratum-specific WCP-type threshold as
\begin{equation*}
    \widehat{t}_{\alpha, h} := \text{Quantile}_{1-\alpha}\Bigg(\sum_{i \in \icalh} \frac{\widehat{w}_h(\Xb_{1i}, \Zbt_{1i})}{\sum_{j \in \icalh} \widehat{w}_h(\Xb_{1j}, \Zbt_{1j})} \cdot\delta_{s_i} \Bigg).
\end{equation*}
We then construct the stratum-conditional prediction set for $Y_2$ at covariates $(\Xb_2, \Zbt_2)=(\xb,\zbt)$ given $S_2 = h$, as
$\widehat{C}_{\alpha, h}(\xb, \zbt) = \{y \in \mathcal Y : \mathcal{S}_h(\xb, y, \zbt) \le \widehat{t}_{\alpha, h} \}$, where $\mathcal{S}_h(\xb, y, \zbt)$ is a stratum-specific score function trained by $\dtrah$.
The analogous theoretical results in Appendix~\ref{app:theoretical-group} then hold for the selected prediction set $\widehat C_{\alpha, h, \hat \ell_h}$ with $\widehat \ell_h \in [L]$, under Assumptions~\ref{ass:app:geo-super}, \ref{ass:app:geo-nonresponse}, \ref{ass:app:geo-design}(b), and \ref{ass:app:geo-invariant}, along with appropriate regularity conditions for this setting.

\section{When Outcomes are Partially Observed at $t = 2$} \label{app:partial}

\subsection{Problem Setup}
In this section, we describe how to construct prediction sets when outcomes are partially observed at $t=2$. Let $\dtwot$ denote the dataset with fully observed covariates at $t=2$:
\[
\dtwot = \{(\Xb_{2i}, Y_{2i}, \Zbt_{2i}, \rfull_{2i}, \wcov_{2i}, \wfull_{2i})\}_{i \in \itwot},
\]
and let $\dtwo$ denote the dataset with fully observed covariates and outcomes at $t=2$:
\[
\dtwo = \{(\Xb_{2i}, Y_{2i}, \Zbt_{2i}, \wfull_{2i}) : i \in \itwot, \rfull_{2i} = 1\} = \{(\Xb_{2i}, Y_{2i}, \Zbt_{2i}, \wfull_{2i})\}_{i \in \itwo}.
\]
Then, analogous to Theorem~1, under Assumptions~1--4 we obtain the following corollary.
\begin{apxcor}[The distribution of subgroup-wise sample sizes within stratum] \label{cor:dtwo}
For each $h \in [H]$, let $\mathcal D_{2h}$ denote the subset of the previously defined $\mathcal D_2$ corresponding to stratum $h$, and let $n_{2h} = |\mathcal D_{2h}|$. For each $\gb \in \kg$, let $n_{2h\gb}$ denote the sample size of subgroup $\gb$ in stratum $h$ at $t=2$. Then, under Assumptions~\ref{ass:super}--\ref{ass:multinomial}, we have
\[
(n_{2h\gb})_{\gb} \mid n_{2h} \sim \mathrm{Mult}(n_{2h}, (q_{2h\gb})_{\gb}),
\]
where $q_{2h\gb} = N_{2h\gb}\pi_{2\gb}a_{2\gb}^{y}/\sum_{\gb'} N_{2h\gb'}\pi_{2\gb'}a_{2\gb'}^{y}$.
\end{apxcor}
Note that the definition of $q_{2h\gb}$ here differs from that in Section~\ref{sec-survey-prob}, since $q_{2h\gb}$ now denotes the subgroup probabilities in the $t=2$ sample with fully observed covariates and outcomes. Analogous to Section~3.3, for each stratum $h \in [H]$, under Assumptions~\ref{ass:super}--\ref{ass:multinomial} and conditional on $n_{1h}$ and $n_{2h}$, the survey datasets $\doneh$ and $\dtwoh$, which serve as training data, can be regarded as i.i.d.\ samples:
\begin{align} 
    (\Xb_{1i}, Y_{1i}, \Zb_{1i}, \Gb_{1i}) &\iid P_{(\Xb_1, Y_1, \Zb_1) \mid \Gb_1} \times P_{\Gb_{1h} \mid \rfull_{1h} = 1}, \quad i \in \ioneh, \ \ h \in [H], \label{eqn:doneh} \\
    (\Xb_{2i}, Y_{2i}, \Zb_{2i}, \Gb_{2i}) &\iid P_{(\Xb_2, Y_2, \Zb_2) \mid \Gb_2} \times P_{\Gb_{2h} \mid \rfull_{2h} = 1}, \quad i \in \itwoh, \ \ h \in [H], \label{eqn:dtwoh}
\end{align}
where $P_{(\Xb_1, Y_1, \Zb_1) \mid \Gb_1}$ and $P_{(\Xb_2, Y_2, \Zb_2) \mid \Gb_2}$ are defined in Assumption~1. Moreover, $P_{\Gb_{1h} \mid \rfull_{1h} = 1} := \mathrm{Cat}((q_{1h\gb})_{\gb})$ and $P_{\Gb_{2h} \mid \rfull_{2h} = 1} := \mathrm{Cat}((q_{2h\gb})_{\gb})$, where $q_{1h\gb}$ and $q_{2h\gb}$ are specified in Theorem~1 and Corollary~\ref{cor:dtwo}, respectively. The target distribution, i.e., the $t=2$ population, is given by
\begin{equation} \label{eqn:target}
    (\Xb_2, Y_2, \Zb_2, \Gb_2) \sim P_{(\Xb_2, Y_2, \Zb_2) \mid \Gb_2} \times P_{\Gb_2},
\end{equation}
where $P_{\Gb_2} := \mathrm{Cat}((p_{2\gb})_{\gb})$ and $p_{2\gb} := N_{2\gb}/N_2$ for each $\gb \in \kg$.

\subsection{Construction of the Prediction Set}
\paragraph*{(a) Data Split.}
For each $t \in \{1, 2\}$ and $h \in [H]$, we first randomly split the stratum-specific dataset $\mathcal D_{th}$ into a training set $\mathcal D_{\mathrm{tra}, t, h}$ and a calibration set $\mathcal D_{\mathrm{cal}, t, h}$, with index sets $\mathcal I_{\mathrm{tra}, t, h}$ and $\mathcal I_{\mathrm{cal}, t, h}$, where $n_{\mathrm{tra}, t, h} = |\mathcal I_{\mathrm{tra}, t, h}|$ and $n_{\mathrm{cal}, t, h} = |\mathcal I_{\mathrm{cal}, t, h}|$ are typically chosen as $n_{\mathrm{tra}, t, h} = n_{\mathrm{cal}, t, h} = n_{th}/2$. We then define $\mathcal D_{\mathrm{tra}, t} := \bigcup_{h=1}^H \mathcal D_{\mathrm{tra}, t, h}$ and $\mathcal D_{\mathrm{cal}, t} := \bigcup_{h=1}^H \mathcal D_{\mathrm{cal}, t, h}$, with corresponding index sets $\mathcal I_{\mathrm{tra}, t} := \bigcup_{h=1}^H \mathcal I_{\mathrm{tra}, t, h}$ and $\mathcal I_{\mathrm{cal}, t} := \bigcup_{h=1}^H \mathcal I_{\mathrm{cal}, t, h}$. Finally, we define $\dtra := \mathcal D_{\mathrm{tra}, 1} \cup \mathcal D_{\mathrm{tra}, 2}$ and $\dcal := \mathcal D_{\mathrm{cal}, 1} \cup \mathcal D_{\mathrm{cal}, 2}$.

\paragraph*{(b) Score functions and nonconformity scores.}
Using $\dtra$ defined above, we construct the score function $\mathcal S(\xb, y, \zbt)$, following step~(a) of Section~\ref{sec-method}. For each $t \in \{1, 2\}$, we then compute the nonconformity scores as $s_{ti} = \mathcal S(\Xb_{ti}, Y_{ti}, \Zbt_{ti})$ for $i \in \mathcal I_{\mathrm{cal}, t}$.

\paragraph*{(c) Estimation of the WCP weight function.}
For each $t \in \{1, 2\}$ and $h \in [H]$, let $w_{th}$ denote the WCP weight function with $\mathcal D_{th}$ as training data and \eqref{eqn:target} as the target distribution. Then, for any $\xb \in \Rb^d$ and $\zbt = (\gb^\top, \zb^\top)^\top \in \k$, $w_{th}(\xb, \zbt)$ is given by
\begin{equation*} 
    w_{1h}(\xb, \zbt) = r_{\zbt}(\xb) \cdot \frac{v_{2\zbt}}{v_{1\zbt}} \cdot \frac{p_{2\gb}}{q_{1h\gb}}, \qquad
    w_{2h}(\xb, \zbt) = \frac{p_{2\gb}}{q_{2h\gb}},
\end{equation*}
where $r_\zbt(\cdot) := \Big(\tfrac{dP_{\Xb_2 \mid \Zbt_2 = \zbt}}{dP_{\Xb_1 \mid \Zbt_1 = \zbt}}\Big)(\cdot)$, $v_{1\zbt} := P_{\Zb_1 \mid \Gb_1 = \gb}(\zb)$, and $v_{2\zbt} := P_{\Zb_2 \mid \Gb_2 = \gb}(\zb)$. 

Following step~(b) of Section~\ref{sec-method}, we use $\dcal$ only for computing scores, and we use the remaining data $\dttra := \{(\Xb_{1i}, Y_{1i}, \Zbt_{1i}, \rfull_{1i})\}_{i \in \tilde{\mathcal I}_{\mathrm{tra}, 1}} \cup \{(\Xb_{2i}, Y_{2i}, \Zbt_{2i}, \rfull_{2i})\}_{i \in \tilde{\mathcal I}_{\mathrm{tra}, 2}}$ for estimation, where $\tilde{\mathcal I}_{\mathrm{tra}, t} := \tilde{\mathcal I}_t \setminus \mathcal I_{\mathrm{cal}, t}$ for each $t \in \{1, 2\}$.
In particular, for each $\zbt = (\gb^\top, \zb^\top)^\top \in \k$, we estimate $r_\zbt$ using the KLIEP method applied to $\{\Xb_{1i} \mid i \in \tilde{\mathcal I}_{\mathrm{tra}, 1}, \Zbt_{1i} = \zbt\}$ and $\{\Xb_{2i} \mid i \in \tilde{\mathcal I}_{\mathrm{tra}, 2}, \Zbt_{2i} = \zbt\}$. Moreover, we estimate
\[
\widehat v_{1\zbt} = \frac{\sum_{i \in \tilde{\mathcal I}_{\mathrm{tra}, 1}} \mathbf 1\{\Zbt_{1i} = \zbt\}}{\sum_{i \in \tilde{\mathcal I}_{\mathrm{tra}, 1}} \mathbf 1\{\Gb_{1i} = \gb\}}, \qquad \widehat v_{2\zbt} = \frac{\sum_{i \in \tilde{\mathcal I}_{\mathrm{tra}, 2}} \mathbf 1\{\Zbt_{2i} = \zbt\}}{\sum_{i \in \tilde{\mathcal I}_{\mathrm{tra}, 2}} \mathbf 1\{\Gb_{2i} = \gb\}}.
\]

For each $t \in \{1, 2\}$ and each stratum $h \in [H]$, although $\{\Gb_{ti}\}_{i \in \mathcal I_{\mathrm{tra}, t, h}}$ is not an i.i.d.\ sample from the population at time $t$ within stratum $h$, conditional on $n_{th}$ it may be modeled as an i.i.d.\ sample from $P_{\Gb_{th} \mid \rfull_{th} = 1}$ defined in~\eqref{eqn:doneh} and~\eqref{eqn:dtwoh}. Therefore, we estimate
\[
\widehat q_{1h\gb} = \frac{\sum_{i \in \mathcal I_{\mathrm{tra}, 1, h}} \mathbf 1\{\Gb_{1i} = \gb\}}{n_{\mathrm{tra}, 1, h}}, \qquad
\widehat q_{2h\gb} = \frac{\sum_{i \in \mathcal I_{\mathrm{tra}, 2, h}} \mathbf 1\{\Gb_{2i} = \gb\}}{n_{\mathrm{tra}, 2, h}}.
\]
Finally, for any $\xb \in \mathbb{R}^d$ and $\zbt = (\gb^\top, \zb^\top)^\top \in \k$, we estimate the WCP weight function at time $t$ in stratum $h$ as
\begin{equation*} 
    \widehat w_{1h}(\xb, \zbt) = \widehat r_{\zbt}(\xb) \cdot \frac{\widehat v_{2\zbt}}{\widehat v_{1\zbt}} \cdot \frac{p_{2\gb}}{\widehat q_{1h\gb}}, \qquad
    \widehat w_{2h}(\xb, \zbt) = \frac{p_{2\gb}}{\widehat q_{2h\gb}}.
\end{equation*}

\paragraph*{(d) Construction of the prediction set.}

Given the nonconformity scores $\{s_{ti}\}_{i \in \mathcal{I}_{\mathrm{cal}, t}}$ for $t=1,2$ and the estimated WCP weight functions $\{\widehat{w}_{th}\}_{t,h}$, we proceed as in step~(c) of Section~\ref{sec-method}. In particular, whereas Section~\ref{sec-method} uses $H$ i.i.d.\ calibration datasets from $t=1$, here we have $H$ i.i.d.\ calibration datasets from each of $t=1$ and $t=2$. Extending the construction in Section~4.1(c) in a straightforward manner therefore yields the following WCP-type threshold for a miscoverage level $\alpha \in (0,1)$:
\begin{equation} \label{eqn:d2-threshold}
\widehat{t}_\alpha = \text{Quantile}_{1-\alpha}\Bigg(\sum_{t=1}^2 \sum_{h=1}^H \frac{1}{2H} \sum_{i \in \mathcal I_{\mathrm{cal}, t, h}} \frac{\widehat{w}_{th}(\Xb_{ti}, \Zbt_{ti})}{\sum_{j \in \mathcal I_{\mathrm{cal}, t, h}} \widehat{w}_{th}(\Xb_{tj}, \Zbt_{tj})} \cdot\delta_{s_{ti}} \Bigg).
\end{equation}
The marginal prediction set $\widehat{C}_{\alpha}(\xb,\zbt)$ at miscoverage level $\alpha$ is then constructed as
\begin{equation} \label{eqn:d2-set}
\widehat{C}_{\alpha}(\xb, \zbt) = \left\{y \in \mathbb{R} : \mathcal{S}(\xb, y, \zbt) \le \widehat{t}_\alpha \right\}, \quad \forall \xb \in \Rb^d,\ \zbt \in \k.
\end{equation}

\paragraph*{(e) Selection of the prediction set.}

Following step~(d) of Section~\ref{sec-method}, we construct $L$ prediction sets utilizing different score functions $\mathcal S_1, \dots, \mathcal S_L$. For each $\ell \in [L]$, we apply steps (b)--(d) with $\mathcal S_\ell$ and denote the resulting prediction set in~\eqref{eqn:d2-set} by $\widehat C_{\alpha,\ell}$. We then select the candidate with the smallest average interval length as in Appendix~\ref{app:implement}. Let $\widehat \ell \in [L]$ denote the selected index; the selected prediction set is $\widehat C_{\alpha,\hat \ell}$.

\subsection{Theoretical Guarantee}

Here, we derive a theoretical coverage guarantee only for the proposed marginal prediction sets. Let $\widehat C_{\alpha,1}, \dots, \widehat C_{\alpha,L}$ denote $L$ candidate prediction sets of the form~\eqref{eqn:d2-set}, and let $\widehat \ell \in [L]$ be the selected index. The following corollary, which is a natural extension of Theorem~\ref{thm:coverage}, provides a lower bound on the coverage of the selected set $\widehat C_{\alpha, \hat \ell}$.

\begin{apxcor}[Lower bound on the coverage probability of $\widehat C_{\alpha, \hat \ell}$] \label{cor:d2-coverage}
    For each $t \in \{1, 2\}$, denote $\tilde n_{t} := |\tilde{\mathcal I}_t|$ and $n_{th} := |\mathcal I_{th}|$ for each $h \in [H]$.
    Under Assumptions~\ref{ass:super}--\ref{ass:invariant} and Assumption~\ref{ass:kliep-1}, the selected prediction set $\widehat C_{\alpha, \hat \ell}$ satisfies the following.
    There exists an event $\mathcal E(\dttra)$, depending solely on $\dttra$, such that $\Pb(\mathcal E(\dttra) \;|\; \tilde n_1, \tilde n_2, \{n_{th}\}_{t, h} ) \to 1$ as $n_{11}, \dots, n_{1H}, n_{21}, \dots, n_{2H} \to \infty$, and on $\mathcal E(\dttra)$,
    \[
    \begin{adjustbox}{max width=\linewidth}
    $\displaystyle
    \begin{aligned}
        &\Pb\big(Y_2 \in \widehat{C}_{\alpha, \hat \ell}(\Xb_2, \Zbt_2) \Bigm| \dttra, \tilde n_1, \tilde n_2, \{n_{th}\}_{t, h} \big)  \ge 1-\alpha - \left(2 + C_1\Big(1 +\sqrt{\log(4HL)}\Big)\right) \\ 
        &\times \Bigg( \underbrace{\max_{\substack{
        t = 1, 2 \\
        1 \le h \le H
        }} \Eb\left[\big|\widehat w_{th}(\Xb_{th}, \Zbt_{th}) - w_{th}(\Xb_{th}, \Zbt_{th})\big| \Bigm| \dttra, \tilde n_1, \tilde n_2, \{n_{th}\}_{t, h}  \right]}_{(\mathrm{I})} + 
        \underbrace{\max_{t, \zbt, h} \frac{v_{2\zbt} p_{2\gb}}{v_{1\zbt} q_{th\gb}} \cdot \frac{C_2}{\sqrt{\min_{t, h} n_{th}}}}_{(\mathrm{II})} \Bigg)
    \end{aligned}
    $
    \end{adjustbox}
    \]
    for some universal constants $C_1, C_2> 0$.
    Here, $(\Xb_{th}, \Zbt_{th}) \sim P_{(\Xb_t, \Zb_t) \mid \Gb_t} \times P_{\Gb_{th} \mid \rfull_{th}=1}$ for each $t \in \{1, 2\}$ and $h \in [H]$, and
    the test point $(\Xb_2, Y_2, \Zbt_2) \sim P_{(\Xb_2, Y_2, \Zb_2) \mid \Gb_2} \times P_{\Gb_2}$ (see \eqref{eqn:doneh}, \eqref{eqn:dtwoh}, and \eqref{eqn:target}) are all independent of $\dttra$.
\end{apxcor}

Note that the only difference from Theorem~2 is that the maximum (or minimum) now also ranges over time $t$, and the prefactor involves $\log(4HL)$, since there are $2H$ i.i.d.\ calibration datasets here. Combining Corollary~\ref{cor:d2-coverage} and Theorem~\ref{thm:kliep-1} then yields the following theoretical guarantee for the selected prediction set $\widehat C_{\alpha,\hat \ell}$.

\begin{apxcor}[Theoretical coverage guarantee for $\widehat C_{\alpha, \hat \ell}$] \label{cor:d2-theo}
    Conditional on $\nrawone$, $\tilde n_2$, and $\{n_{th}\}_{t, h}$, and under Assumptions~\ref{ass:super}--\ref{ass:invariant} and Assumption~\ref{ass:kliep-1}, as $n_{11}, \dots, n_{1H}, n_{21}, \dots, n_{2H} \to \infty$, if $\min_h n_{1h} \gtrsim \tilde n_1/H$ and $\min_h n_{2h} \gtrsim \tilde n_2/H$, then the selected prediction set $\widehat C_{\alpha, \hat \ell}$ satisfies
    \[
    \begin{adjustbox}{max width=\linewidth}
    $\displaystyle
    \begin{aligned}
        \Pb\big(Y_2 \in \widehat{C}_{\alpha,\hat\ell}(\Xb_2, \Zbt_2) \Bigm| \dttra, \nrawone, \tilde n_2, \{n_{th}\}_{t, h} \big) 
        \ge 1 - \alpha  - 
        \begin{cases}
            O_p\big((\min_{t, h} n_{th})^{-\frac{1}{2}}\big) & d = 0,\\
            O_p\big(\min(\tilde n_1, \tilde n_2)^{-\frac{1}{2+\gamma}}\big) & d = 1, 2, \\
            O_p\Big(\Bigl(\tfrac{\log\min(\tilde n_1, \tilde n_2)}{\min(\tilde n_1, \tilde n_2)}\Bigr)^{\frac{1}{d}}\Big) & d \ge 3,
        \end{cases}
    \end{aligned}
    $
    \end{adjustbox}
    \]
    for any arbitrarily small $\gamma > 0$. The $O_p(\cdot)$ terms are taken with respect to the randomness in $\dttra$, conditional on $\nrawone$, $\tilde n_2$, and $\{n_{th}\}_{t, h}$.
    Moreover, $(\Xb_2, Y_2, \Zbt_2) \sim P_{(\Xb_2, Y_2, \Zb_2) \mid \Gb_2} \times P_{\Gb_2}$ is an independent test point drawn from the $t=2$ population.
\end{apxcor}

\section{Computing Empirical Coverage in Real Data} \label{app:eval}
Here, we describe how to compute the empirical coverage of both marginal and group-conditional prediction sets using the available datasets in the real data application. 
Although our method constructs prediction sets for $Y_2$ at time $t=2$ based only on the covariates $(\Xb_{2i}, \Zbt_{2i})$, the empirical coverage can be evaluated only for units with observed outcomes $Y_{2i}$ at $t=2$. 
Since we do not have i.i.d.\ samples from the target distribution (the $t=2$ population), we instead use the subset of $\mathcal{D}_2$ with observed outcomes as the test dataset. 
Based on this test dataset, we then estimate the coverage probability, i.e., the empirical coverage under the target distribution.

\subsection{Marginal Empirical Coverage} \label{app:eval-1}
We first describe how to compute marginal empirical coverage using the observed outcomes.
When outcomes are partially observed at $t = 2$, the dataset at time $t = 2$ with fully observed covariates can be written as
\[
\dtwo = \{(\Xb_{2i}, Y_{2i}, \Zbt_{2i}, \rfull_{2i}, \wcov_{2i}, \wfull_{2i})\}_{i \in \itwo}.
\]
Let $\dtwo^Y$ denote the subset of $\dtwo$ with observed outcomes:
\[
\dtwo^Y := \{(\Xb_{2i}, Y_{2i}, \Zbt_{2i}, \wfull_{2i}) \mid i \in \itwo, \rfull_{2i} = 1 \} = \{(\Xb_{2i}, Y_{2i}, \Zbt_{2i}, \wfull_{2i})\}_{i \in \itwo^Y}.
\]
We begin by randomly selecting a subset of proportion $\beta \in (0,1)$ from $\dtwo^Y$ to form the test dataset $\dtest$, with index set $\itest \subset \itwo^Y$.
Then, the remaining covariates are
\[
\dttwo := \{(\Xb_{2i}, \Zbt_{2i}) \mid i \in \itwo \setminus \itest \} = \{(\Xb_{2i}, \Zbt_{2i})\}_{i \in \ittwo}.
\]
We use the covariates in $\dttwo$ to estimate the WCP weight function, and use the test dataset $\dtest$ to compute empirical coverage.
The proportion $\beta$ can be chosen to ensure sufficient sample sizes for both $\dttwo$ and $\dtest$; in our application to the NHANES dataset, we set $\beta = 0.6$.

Recall that the prediction set \(\widehat{C}_\alpha(\cdot, \cdot) \equiv \widehat{C}_\alpha(\cdot, \cdot; \drawone, \dttwo)\) constructed with $\drawone$ and $\dttwo$ is said to satisfy marginal validity at miscoverage level \(\alpha \in (0,1)\) if
\begin{equation} \label{eqn:eval-1-2}
    \Pb\big(Y_2 \in \widehat C_\alpha(\Xb_2, \Zbt_2)\big) \ge 1 - \alpha,
\end{equation}
where the probability is taken over \(\drawone\), \(\dttwo\), and an independent test point \((\Xb_2, Y_2, \Zbt_2)\) drawn from the target distribution $P_{(\Xb_2, Y_2, \Zb_2) \mid \Gb_2} \times P_{\Gb_2}$.
However, our test dataset $\dtest$ does not follow the target distribution, since it consists only of individuals in $\dtwo$ with $\rfull_{2i}=1$.
For each stratum $h \in [H]$, define the stratum-specific test dataset as $\dtesth = \{(\Xb_{2i}, Y_{2i}, \Zbt_{2i}, \wfull_{2i})\}_{i \in \itesth}$.
Then, analogously to Section~\ref{sec-survey}, conditioning on \(\ntesth = |\itesth|\) implies
\begin{equation} \label{eqn:eval-1-4}
(\Xb_{2i}, Y_{2i}, \Zb_{2i}, \Gb_{2i}) \iid P_{(\Xb_2, Y_2, \Zb_2) \mid \Gb_2} \times P_{\Gb_{2h} \mid \rfull_{2h} = 1}, \qquad i \in \itesth, \ \ h \in [H],
\end{equation}
where for each $h \in [H]$, the distribution $P_{\Gb_{2h} \mid \rfull_{2h} = 1}$ is defined as
\[
P_{\Gb_{2h} \mid \rfull_{2h} = 1} := \mathrm{Mult}((\tilde q_{2h\gb})_\gb), \qquad \tilde q_{2h\gb} := \frac{N_{2h\gb}\pi_{2\gb}a_{2\gb}^{y}}{\sum_{\gb'} N_{2h\gb'}\pi_{2\gb'}a_{2\gb'}^{y}}.
\]

Now, we first aim to estimate \(\Pb(Y_2 \in \widehat{C}_\alpha(\Xb_2, \Zbt_2) \mid \drawone, \dttwo )\) for fixed \(\drawone\) and \(\dttwo\), using \(\dtest\). To this end, since \((\Xb_2, Y_2, \Zbt_2)\) is independent of \(\drawone\) and \(\dttwo\), we apply the following simple trick:
\begin{align*}
    &\mathbb{P}_{(\Xb_2, Y_2, \Zbt_2) \sim P_{(\Xb_2, Y_2, \Zb_2) \mid \Gb_2} \times P_{\Gb_2}}\left(Y_2 \in \widehat{C}_\alpha(\Xb_2, \Zbt_2) \ \middle|\  \drawone, \dttwo \right) \\
    &= \mathbb{E}_{(\Xb_2, Y_2, \Zbt_2) \sim P_{(\Xb_2, Y_2, \Zb_2) \mid \Gb_2} \times P_{\Gb_2}}\left[\mathbf 1\{Y_2 \in \widehat{C}_\alpha(\Xb_2, \Zbt_2)\} \ \middle|\  \drawone, \dttwo \right]  \\
    &= \mathbb{E}_{(\Xb_2, Y_2, \Zbt_2) \sim P_{(\Xb_2, Y_2, \Zb_2) \mid \Gb_2} \times P_{\Gb_{2h} \mid \rfull_{2h} = 1}}\left[\mathbf 1\{Y_2 \in \widehat{C}_\alpha(\Xb_2, \Zbt_2)\} \cdot \frac{P_{\Gb_2}(\Gb_2)}{P_{\Gb_{2h} \mid \rfull_{2h} = 1}(\Gb_2)} \ \middle| \ \drawone, \dttwo\right],
\end{align*}
for each $h \in [H]$. Then, conditional on $\drawone$ and $\dttwo$, the following quantity is an unbiased estimator of \(\Pb(Y_2 \in \widehat{C}_\alpha(\Xb_2, \Zb_2) \mid \drawone, \dttwo)\):
\begin{equation*}
     \frac{1}{|\itest|} \sum_{h =1}^H \sum_{i \in \itesth}  \frac{P_{\Gb_2}(\Gb_{2i})}{ P_{\Gb_{2h} \mid \rfull_{2h} = 1}(\Gb_{2i})}  \cdot \mathbf 1\big\{Y_{2i} \in \widehat{C}_\alpha(\Xb_{2i}, \Zbt_{2i}; \drawone, \dttwo)\big\}.
\end{equation*}
For each subgroup $\gb \in \kg$ and stratum $h \in [H]$, we estimate $P_{\Gb_{2h} \mid \rfull_{2h} = 1}(\gb)$ by
\begin{equation*}
\widehat P_{\Gb_{2h} \mid \rfull_{2h} = 1}(\gb) := \frac{\sum_{i \in \itwo^Y \setminus \itest} \mathbf 1\{S_{2i} = h, \Gb_{2i} = \gb\}}{\sum_{i \in \itwo^Y \setminus \itest} \mathbf 1\{S_{2i} = h\}},
\end{equation*}
where $S_{2i} \in [H]$ denotes the stratum membership of the $i$th unit at $t = 2$.
We then randomly split $\done$ into training and calibration sets $B$ times to reduce the variance of the estimator.
Let $\widehat C_\alpha(\cdot, \cdot \ ; \drawone^{(b)}, \dttwo)$ denote the marginal prediction set obtained from the $b$th split.
We estimate the left-hand side of~\eqref{eqn:eval-1-2}, the marginal coverage probability, by
\begin{equation*}
\begin{adjustbox}{max width=\linewidth}
$\displaystyle
\mathbb{P}\big(Y_2 \in \widehat{C}_\alpha(\Xb_2, \Zbt_2)\big) \approx \frac{1}{B} \sum_{b=1}^B \left\{ \frac{1}{|\itest|} \sum_{h =1}^H \sum_{i \in \itesth}  \frac{P_{\Gb_2}(\Gb_{2i})}{ \widehat P_{\Gb_{2h} \mid \rfull_{2h} = 1}(\Gb_{2i})}  \cdot \mathbf 1\big\{Y_{2i} \in \widehat{C}_\alpha(\Xb_{2i}, \Zbt_{2i}; \drawone^{(b)}, \dttwo)\big\} \right\}.
$
\end{adjustbox}
\end{equation*}

\subsection{Subgroup-wise Empirical Coverage} \label{app:eval-2}
For each \(\gb \in \kg\), let $\dtwog^Y$ denote the subset of $\dtwog$ with observed outcomes:
\[
\dtwog^Y = \{(\Xb_{2i}, Y_{2i}, \Zb_{2i}, W_{2i}^Y) \mid i \in \itwog, \rfull_{2i} = 1 \} = \{(\Xb_{2i}, Y_{2i}, \Zb_{2i}, W_{2i}^Y)\}_{i \in \itwog^Y}.
\]
We first randomly select a subset of proportion $\beta \in (0,1)$ from $\dtwog^Y$ to form the subgroup-specific test dataset $\dtestg$, with index set $\itestg \subset \itwog^Y$.
Then, the remaining covariates are
\[
\dttwog := \left\{(\Xb_{2i}, \Zb_{2i}) \ \middle| \ i \in \itwog \setminus \itestg \right\} = \left\{(\Xb_{2i}, \Zb_{2i})\right\}_{i \in \ittwog}.
\]
We use \(\dttwog\) to construct group-conditional prediction sets, and \(\dtestg\) to compute empirical coverage.
The proportion $\beta$ can be chosen to ensure sufficient sample sizes for both $\dttwog$ and $\dtestg$; in our application to the NHANES dataset, we set $\beta = 0.6$.

Recall that the prediction set \(\widehat{C}_{\alpha, \gb}(\cdot, \cdot) \equiv \widehat{C}_{\alpha, \gb}(\cdot, \cdot \ ; \donegt, \dttwog)\) constructed with $\donegt$ and $\dttwog$ is said to satisfy group-conditional validity at miscoverage level \(\alpha \in (0, 1)\) if
\begin{equation} \label{eqn:eval-2-2}
    \Pb\big(Y_2 \in \widehat C_{\alpha, \gb}(\Xb_2, \Zb_2) \bigm| \Gb_2 = \gb \big) \ge 1-\alpha, \qquad \forall \gb \in \kg,
\end{equation}
where the probability is taken over \(\donegt\), \(\dttwog\), and an independent test point \((\Xb_2, Y_2, \Zb_2)\) drawn from subgroup $\gb$ of the $t = 2$ population.
Since $(\Xb_{2i}, Y_{2i}, \Zb_{2i}) \iid P_{(\Xb_2, Y_2, \Zb_2) \mid \Gb_2 = \gb}$ for $i \in \itestg$, the subgroup-specific test dataset $\dtestg$ can be regarded as an i.i.d.\ sample from the target distribution within subgroup $\gb$.
Accordingly, conditional on $\donegt$ and $\dttwog$, the following estimator is an unbiased estimator of $\Pb(Y_2 \in \widehat C_{\alpha, \gb}(\Xb_2, \Zb_2) \mid \Gb_2 = \gb, \donegt, \dttwog)$:
\begin{equation*} \label{eqn:group-2}
    \frac{1}{|\itestg|} \sum_{i \in \itestg} \mathbf 1\big\{Y_{2i} \in \widehat{C}_{\alpha, \gb}(\Xb_{2i}, \Zb_{2i} ; \donegt, \dttwog)\big\}.
\end{equation*}
We then randomly split $\doneg$ into subgroup-specific training and calibration sets $B$ times to reduce the variance of the resulting estimator.
Let $\widehat C_{\alpha, \gb}(\cdot, \cdot  \ ; \donegt^{(b)}, \dttwog)$ denote the resulting group-conditional prediction set for the $b$th split.
For each $\gb$, the left-hand side of \eqref{eqn:eval-2-2}, that is, the subgroup-wise coverage probability, is then estimated as
\begin{align*}
\Pb\big(Y_2 \in \widehat{C}_{\alpha, \gb}(\Xb_2, \Zb_2) \bigm| \Gb_2 = \gb\big) \approx \frac{1}{B}\sum_{b=1}^B \left\{ \frac{1}{|\itestg|} \sum_{i \in \itestg} \mathbf 1\big\{Y_{2i} \in \widehat{C}_{\alpha, \gb}(\Xb_{2i},\Zb_{2i} ; \donegt^{(b)}, \dttwog)\big\} \right\}.
\end{align*}

\section{Coverage Properties of GWCP Relative to DA-WCP} \label{app:coverage}

In this section, we provide a theoretical analysis of the conditions under which GWCP undercovers or overcovers, thereby identifying regimes in which DA-WCP, our proposed method, can outperform GWCP.
For ease of comparison, we assume that the categorical covariate consist only of demographic subgroup variables (i.e., $\Zbt_{ti} = \Gb_{ti}$ and $\zbt = \gb$), consider a single stratum (i.e., $H = 1$) so that we suppress the stratum index $h$, and use fixed quantile levels $(\alo, \ahi)$ for the CQR score.
Under this setting, the GWCP and DA-WCP thresholds both take the form
\begin{align*}
\widehat t_\alpha
= 
\q_{1-\alpha}\!\left(
  \sum_{i \in \mathcal{I}_{\mathrm{cal}}}
  \frac{\widehat{w}(\Xb_{1i}, \Gb_{1i})}
       {\sum_{j \in \mathcal{I}_{\mathrm{cal}}} \widehat{w}(\Xb_{1j}, \Gb_{1j})}
  \cdot \delta_{s_i}
\right), 
\end{align*}
where $\widehat w(\xb, \gb) = \rng(\xb) \cdot p_{2\gb}/\widehat q_{1\gb}$ for DA-WCP and $\widehat w(\xb, \gb) = p_{2\gb}/\widehat q_{1\gb}$ for GWCP.

Here, $s_i = |Y_{1i} - \widehat \mu(\Xb_{1i}, \Gb_{1i})|$ for the absolute residual score, and 
$s_i = \max\{\widehat q_{\alo}(\Xb_{1i}, \Gb_{1i}) - Y_{1i},\, Y_{1i} - \widehat q_{\ahi}(\Xb_{1i}, \Gb_{1i})\}$ for the CQR score.  
We denote the DA-WCP threshold by $\tdawcp$ and the GWCP threshold by $\tgwcp$.  
While DA-WCP achieves coverage close to the nominal level across all scenarios, as guaranteed by Corollary~1,  
GWCP will \emph{undercover} if $\tgwcp < \tdawcp$ and \emph{overcover} if $\tgwcp > \tdawcp$.
We first prove the following theorem and then identify the settings in which GWCP undercovers or overcovers.

\begin{apxthm}[Limit representation of the DA-WCP threshold] \label{thm:threshold}
Define $\Dc := \dttra \cup \dtwo$.
Conditional on $\nrawone = |\irawone|$, $n_1 = |\ione|$, and $n_2 = |\itwo|$, for any $\alpha \in (0, 1)$, as $n_1,n_2 \rightarrow \infty$, the DA-WCP threshold $\tdawcp$ satisfies
\[
\begin{adjustbox}{max width=\linewidth}
$\displaystyle
\tdawcp = \inf \left\{y \in \mathbb R : 
\sum_{\gb \in \kg} p_{2\gb} \cdot 
\Eb\left[r_\gb(\Xb_1) \cdot \mathbf 1\{S_{1\gb} \le y\} 
\bigm| \Dc,\, \Gb_1 = \gb \right] 
\ge 1-\alpha \right\} + o_p(1),
$
\end{adjustbox}
\]
where $S_{1\gb} := \mathcal S(\Xb_1, Y_1, \gb)$ has a non-atomic distribution and its support is an interval in $\Rb$ for each $\gb$, and the expectation is taken with respect to $(\Xb_1, Y_1) \sim P_{(\Xb_1, Y_1) \mid \Gb_1 = \gb}$.
\end{apxthm}

\begin{proof}[Proof of Theorem~\ref{thm:threshold}]
For the DA-WCP weight function $\widehat w(\xb, \gb) = \rng(\xb) \cdot p_{2\gb}/\widehat q_{1\gb}$, define
\[
\widehat G(y) := 
\frac{\sum_{i \in \ical} \widehat w(\Xb_{1i}, \Gb_{1i}) \cdot \mathbf 1\{s_i \le y\}}
     {\sum_{i \in \ical} \widehat w(\Xb_{1i}, \Gb_{1i})}, 
\qquad \forall y \in \mathbb R. 
\]
By Lemma~\ref{lem:kliep-1}, the event $\edc$ defined by
\[
\edc := \bigcap_{\gb \in \kg} \left\{\widehat q_{1\gb} > \frac{q_{1\gb}}{2}, \ S_\gb(n_{\gb}) \right\},
\]
where $S_\gb(n_{\gb}) = \snz$ in this case, satisfies $\Pb(\edc \mid \tilde n_1, n_1, n_2) \rightarrow 1$ as $n_1, n_2 \rightarrow \infty$.
On $\edc$, since $\lVert \widehat r_\gb \rVert_\infty \le \bar M$ for all $\gb \in \kg$, we have
$\widehat w(\xb, \gb) = \rng(\xb) \cdot p_{2\gb}/\widehat q_{1\gb} \le 2 \bar M/(\min_\gb q_{1\gb})$.
Therefore, on $\edc$, by the strong law of large numbers, for each $y \in \mathbb R$,
\begin{equation} \label{eqn:coverage-1}
\widehat G(y) 
= 
\frac{\sum_{i \in \ical} \widehat w(\Xb_{1i}, \Gb_{1i}) \cdot \mathbf 1\{s_i \le y\}}
     {\sum_{i \in \ical} \widehat w(\Xb_{1i}, \Gb_{1i})} 
\xrightarrow[]{a.s.} 
\frac{\Eb[\widehat w(\Xb_1, \Gb_1) \cdot \mathbf 1\{S_{1} \le y\} \mid \Dc]}
     {\Eb[\widehat w(\Xb_1, \Gb_1) \mid \Dc]} =: \tilde G(y),
\end{equation}
where $S_{1\gb} =  \mathcal S(\Xb_1, Y_1, \Gb_1)$ and the expectation is taken with respect to
$(\Xb_1, Y_1,  \Gb_1) \sim P_{(\Xb_1, Y_1) \mid \Gb_1} \times P_{\Gb_1 \mid \rfull_1 = 1}$.
For the numerator of \eqref{eqn:coverage-1}, for each $y \in \Rb$ we have
\begin{align*}
\Eb\left[\widehat w(\Xb_1, \Gb_1) \cdot \mathbf 1\{S_{1} \le y\} \bigm| \Dc \right] 
&= \sum_\gb q_{1\gb} \cdot \Eb\left[\frac{p_{2\gb}}{\widehat q_{1\gb}} \cdot \rng(\Xb_1) \mathbf 1\{S_{1\gb} \le y\} \Bigm| \Dc, \Gb_1 = \gb \right] \\
&= \sum_\gb (p_{2\gb} + o_p(1)) \cdot \Eb\left[\rng(\Xb_1) \mathbf 1\{S_{1\gb} \le y\} \bigm| \Dc, \Gb_1 = \gb \right] \\
&= \sum_\gb (p_{2\gb} + o_p(1)) \cdot \left\{\Eb\left[r_\gb(\Xb_1) \mathbf 1\{S_{1\gb} \le y\} \bigm| \Dc, \Gb_1 = \gb \right] + o_p(1)\right\} \\
&= \sum_\gb p_{2\gb} \cdot \Eb[r_\gb(\Xb_1) \mathbf 1\{S_{1\gb} \le y\} \mid \Dc, \Gb_1 = \gb] + o_p(1),
\end{align*}
where the second equality uses $\widehat q_{1\gb} = q_{1\gb} + o_p(1)$ and the third equality follows from \eqref{eqn:bound-3-3}.
For the denominator, $\Eb[\widehat w(\Xb_1, \Gb_1) \mid \Dc] = 1 + o_p(1)$ by \eqref{eqn:thm2-term3-2} and Proposition~\ref{prop:weight}.
Thus,
\begin{equation} \label{eqn:coverage-2}
\tilde G(y) 
= \sum_\gb p_{2\gb} \cdot \Eb[r_\gb(\Xb_1) \cdot \mathbf 1\{S_1 \le y\} \mid \Dc, \Gb_1 = \gb] + o_p(1) 
= G(y) + o_p(1), \quad \forall y \in \Rb.
\end{equation}
Fix an arbitrary $\varepsilon > 0$ and $y \in \mathbb R$. Then,
\begin{align*}
&\Pb\!\left(|\widehat G(y) - G(y)| > \varepsilon \right) \\
&\le 
\Eb\!\left[\,\Pb\!\left(|\widehat G(y) - \tilde G(y)| > \tfrac{\varepsilon}{2} \,\big|\, \Dc \right)\!\right]
+ \Pb\!\left(|\tilde G(y) - G(y)| > \tfrac{\varepsilon}{2}\right) \\
&\le 
\Eb\!\left[\Pb\!\left(|\widehat G(y) - \tilde G(y)| > \tfrac{\varepsilon}{2} \,\big|\, \Dc \right) 
\mathbf 1\{\edc\}\right]
+ \Pb(\edc^c)
+ \Pb\!\left(|\tilde G(y) - G(y)| > \tfrac{\varepsilon}{2}\right).
\end{align*}
As $n_1, n_2 \rightarrow \infty$, the second term tends to $0$ since $\Pb(\edc) \to 1$, and the third term vanishes by \eqref{eqn:coverage-2}.
For the first term, by \eqref{eqn:coverage-1}, the integrand inside the expectation converges to $0$ almost surely and is bounded; hence, the entire term also vanishes by the dominated convergence theorem.
Therefore, $\widehat G(y) = G(y) + o_p(1)$ for all $y \in \mathbb R$.
Moreover, since $G$ is monotone increasing and satisfies $\lim_{y \to -\infty} G(y) = 0$ and $\lim_{y \to +\infty} G(y) = 1$, it follows that $G$ is a CDF, and hence $\lVert \widehat G - G \rVert_\infty = o_p(1)$.
Finally, because $S_{1\gb}$ has a non-atomic distribution and its support is an interval in $\Rb$ for each $\gb$, $G$ is strictly increasing, and thus by continuity of the quantile functional, we conclude that
\[
\tdawcp = \q_{1-\alpha}(\widehat G) = \q_{1-\alpha}(G) + o_p(1).
\]
\end{proof}

Similarly, we can derive a limit representation of the GWCP threshold $\tgwcp$. 
The only difference is that the GWCP weight function does not include the density-ratio term, 
and hence $r_\gb(\Xb_1)$ is replaced by $1$ in Theorem~\ref{thm:threshold}.

\begin{apxcor}[Limit representation of the GWCP threshold] \label{cor:threshold}
Conditional on $\nrawone = |\drawone|$, $n_1 = |\done|$, and $n_2 = |\dtwo|$,  
for $\alpha \in (0, 1)$, as $n_1, n_2 \rightarrow \infty$,  
the GWCP threshold $\tgwcp$ satisfies
\[
\tgwcp = \inf \left\{y \in \mathbb R : 
\sum_{\gb \in \k} p_{2\gb} \cdot 
\Pb(S_{1\gb} \le y 
\mid \Dc,\, \Gb_1 = \gb)
\ge 1-\alpha \right\} + o_p(1),
\]
where $S_{1\gb} := \mathcal S(\Xb_1, Y_1, \gb)$ has a non-atomic distribution and its support is an interval in $\Rb$ for each $\gb$, and the expectation is taken with respect to  
$(\Xb_1, Y_1) \sim P_{(\Xb_1, Y_1) \mid \Gb_1 = \gb}$.
\end{apxcor}

For each $\gb \in \k$ and $y \in \mathbb R$, the following holds:
\begin{equation} \label{eqn:coverage-3}
\Eb\left[r_\gb(\Xb_1) \cdot \mathbf 1\{S_{1\gb} \le y\} 
\,\big|\, \Dc,\, \Gb_1 = \gb\right] 
= 
\Eb\left[r_\gb(\Xb_1) \cdot \Pb(S_{1\gb} \le y \mid \Dc, \Xb_1, \Gb_1 = \gb) 
\,\big|\, \Dc,\, \Gb_1 = \gb \right].
\end{equation}
If $r_\gb(\Xb_1)$ exhibits little variation for each $\gb$, then $r_\gb(\Xb_1) \approx 1$ 
since $\Eb[r_\gb(\Xb_1) \mid \Gb_1 = \gb] = 1$, 
and consequently, by Theorem~\ref{thm:threshold} and Corollary~\ref{cor:threshold}, we have $\tdawcp \approx \tgwcp$. 
On the other hand, if $\Pb(S_{1\gb} \le y \mid \Dc, \Xb_1, \Gb_1 = \gb) \appropto r_\gb(\Xb_1)$ (approximately proportional) for each $\gb$ and $y \in \mathbb R$, we can infer that $\tdawcp < \tgwcp$.
Similarly, $\Pb(S_{1\gb} \le y \mid \Dc, \Xb_1, \Gb_1 = \gb) \appropto 1/r_\gb(\Xb_1)$ implies $\tdawcp > \tgwcp$.
We now examine in detail when each of these situations may arise for both the absolute residual score and the CQR score.

\subsection{Absolute Residual Score} \label{app:coverage-ar}

For the absolute residual score, suppose that the outcome $Y_1$ is generated from the following heteroscedastic model with conditional standard deviation $\sigma(\Xb_1, \Gb_1)$:
\[
Y_1 = \mu(\Xb_1, \Gb_1) + \sigma(\Xb_1, \Gb_1) \cdot \varepsilon, 
\qquad 
\varepsilon \indep (\Xb_1, \Gb_1),
\]
where $\Eb[\varepsilon] = 0$ and $\Eb[\varepsilon^2] = 1$.
We consider the following two situations.

\paragraph{(a) The estimated mean function $\widehat \mu$ is consistent.}

If $\widehat \mu$ is consistent (i.e., $\widehat \mu \approx \mu$), then for the absolute residual score $\mathcal S(\xb, y, \gb) = |y - \widehat \mu(\xb, \gb)|$, for each $\gb$ and $y > 0$,
\begin{align*}
\Pb\big(S_{1\gb} \le y \bigm| \Dc, \Xb_1, \Gb_1 = \gb \big) &= \Pb\big(\widehat \mu(\Xb_1, \gb) - y \le Y_1 \le \widehat \mu(\Xb_1, \gb) + y \bigm| \Dc, \Xb_1, \Gb_1 = \gb \big) \\
&\approx \Pb\big( \mu(\Xb_1, \gb) - y \le Y_1 \le  \mu(\Xb_1, \gb) + y \bigm| \Dc, \Xb_1, \Gb_1 = \gb \big) \\
&= \Pb\left( -\frac{y}{\sigma(\Xb_1, \gb)} \le \varepsilon \le \frac{y}{\sigma(\Xb_1, \gb)} \Bigm| \Dc, \Xb_1, \Gb_1 = \gb \right).
\end{align*}
Therefore, by the previous argument, for each $\gb$ and $y > 0$, if $\sigma(\xb, \gb) \appropto r_\gb(\xb)$ as a function of $\xb$, then $\Pb(S_1 \le y \mid \Dc, \Xb_1 = \xb, \Gb_1 = \gb)$ tends to decrease as $r_\gb(\xb)$ increases, implying that $\tdawcp > \tgwcp$ and GWCP undercovers.
Similarly, if $\sigma(\xb, \gb) \appropto 1/r_\gb(\xb)$ as a function of $\xb$ for each $\gb$, then $\tdawcp < \tgwcp$ and GWCP overcovers.

\paragraph{(b) The estimated mean function $\widehat \mu$ is inconsistent.}
However, $\widehat \mu$ can be inconsistent when the true mean function $\mu$ is nonlinear but $\widehat \mu$ is estimated using a linear model.
In this case, $\widehat \mu \approx \mu$ need not hold even with sufficiently large sample sizes.
Unlike the previous setting, it is generally not possible to explicitly characterize when GWCP overcovers or undercovers; instead, the relationship can be described as follows.
For each $\gb$ and $y > 0$, since
\[
\Pb\big(S_1 \le y \bigm| \Dc, \Xb_1, \Gb_1 = \gb \big) = \Pb\left(\widehat \mu(\Xb_1, \gb) - y \le Y_1 \le \widehat \mu(\Xb_1, \gb) + y \bigm| \Dc, \Xb_1, \Gb_1 = \gb\right),
\]
the relation $\Pb\left(\widehat \mu(\Xb_1, \gb) - y \le Y_1 \le \widehat \mu(\Xb_1, \gb) + y \mid \Dc, \Xb_1, \Gb_1 = \gb\right) \appropto r_\gb(\Xb_1)$ as a function of $\Xb_1$ implies $\tdawcp < \tgwcp$, and hence GWCP overcovers,
whereas the opposite relation implies $\tdawcp > \tgwcp$, and hence GWCP undercovers.

\subsection{CQR Score} \label{app:coverage-cqr}

\paragraph{(a) Homoscedastic error \& $\hatqlo$ and $\hatqhi$ are consistent.}
For the CQR score, we first consider the special case in which the estimated quantile functions $\hatqlo$ and $\hatqhi$ are consistent (i.e., $\hatqlo \approx \qlo$ and $\hatqhi \approx \qhi$), and the error is homoscedastic.
Then we can write $Y_1 = \mu(\Xb_1, \Gb_1) + \epsilon$ with $\epsilon \indep (\Xb_1, \Gb_1)$, and for each $\gb \in \kg$ and $y \in \Rb$ we have
\begin{align*}
\Pb\big(S_{1\gb} \le y \bigm| \Dc, \Xb_1, \Gb_1 = \gb \big) 
&= \Pb\left( \hatqlo(\Xb_1, \gb) - y \le Y_1 \le \hatqhi(\Xb_1, \gb) + y \bigm| \Dc, \Xb_1, \Gb_1 = \gb \right) \\
&\approx \Pb\left( \qlo(\Xb_1, \gb) - y \le Y_1 \le \qhi(\Xb_1, \gb) + y \bigm| \Dc, \Xb_1, \Gb_1 = \gb \right) \\
&= \Pb\left( c_{\alo} - y \le \epsilon \le c_{\ahi} + y \bigm| \Dc, \Xb_1, \Gb_1 = \gb \right),
\end{align*}
where $c_{\alo}$ and $c_{\ahi}$ are the $\alo$ and $\ahi$ quantiles of $\epsilon$, respectively.
Thus, $\Pb(S_{1\gb} \le y \mid \Dc, \Xb_1, \Gb_1 = \gb)$ does not depend on $\Xb_1$, and hence, by the previous argument, we have $\tgwcp \approx \tdawcp$, so both DA-WCP and GWCP achieve coverage close to the nominal level.

\paragraph{(b) General case.}
However, we generally cannot explicitly characterize conditions under which GWCP undercovers or overcovers.
This is because the situation may depend on the choice of $(\alo, \ahi)$ and is therefore more complicated than for the absolute residual score.
Instead, the relationship in each case can only be expressed in the following form.
The relation \[\Pb\left( \hatqlo(\Xb_1, \gb) - y \le Y_1 \le \hatqhi(\Xb_1, \gb) + y \mid \Dc, \Xb_1, \Gb_1 = \gb \right) \appropto r_\gb(\Xb_1)\] as a function of $\Xb_1$, for each $\gb \in \kg$ and $y \in \Rb$, implies $\tdawcp < \tgwcp$, and hence that GWCP overcovers, whereas the opposite relation implies that GWCP undercovers.

\section{Simulation Setup} \label{app:simul-setup}

Here, we present the detailed simulation setup in Section~\ref{sec-simul}.
For simplicity, we assume identical missingness patterns for covariates and outcomes at $t=1$ (i.e., $\rcov_{1i}=1 \Leftrightarrow \rfull_{1i}=1$), so $\drawone=\done$ and $\dttra=\dtra$.
Specifically, we set the realized sample sizes to $n_1 = n_2 = 5{,}000$ and $P_{G_2} = \mathrm{Mult}(0.05, 0.1, 0.15, 0.3, 0.4)$, and generate the datasets $\done=\{(\Xb_{1i}, Y_{1i}, \Gb_{1i})\}_{i=1}^{n_1}$ and $\dtwo=\{(\Xb_{2i}, \Gb_{2i})\}_{i=1}^{n_2}$ as follows:

\begin{enumerate}
    \item For each $i \in [n_1]$, generate $\Gb_{1i} \iid P_{\Gb_1 \mid \rfull_1=1}$, where
    \[
    P_{\Gb_1 \mid \rfull_1=1} =
    \begin{cases}
    P_{\Gb_2}, & \text{if (S2)}, \\
    \mathrm{Mult}(0.2, 0.2, 0.2, 0.2,0.2), & \text{otherwise}.
    \end{cases}
    \]
    
    \item For each $i \in [n_2]$, generate $\Gb_{2i} \iid P_{\Gb_2 \mid \rcov_2=1}$, where
    \[
    P_{\Gb_2 \mid \rcov_2=1} =
    \begin{cases}
    P_{\Gb_2}, & \text{if (S3)}, \\
    \mathrm{Mult}(0.1, 0.15, 0.2, 0.25, 0.3), & \text{otherwise}.
    \end{cases}
    \]
    
    \item For each \(t \in \{1,2\}\) and \(g \in \{1,\dots,5\}\), if \(\Gb_{ti}=g\), generate
    \[
    \Xb_{ti} \iid P_{\Xb_t \mid \Gb_t=g}, \quad \text{where} \quad P_{\Xb_t \mid \Gb_t=g}=\mathcal{N}(\boldsymbol\mu_{tg}, \boldsymbol\Sigma_{tg}).
    \]
    For each \(g \in \{1,\dots,5\}\), set
    \begin{align*}
    \boldsymbol\mu_{1g} &= \big(g-3,\, \sin(\pi g/2),\, 3-g,\, \cos(\pi g/2)\big)^\top, \\
    (\Sigmab_{1g})_{ij} &= 0.2^{\,|i-j|} + 0.1(g-1)\,\mathbf{1}\{i=j\}, \qquad 1 \le i,j \le 4,
    \end{align*}
    and define \(\mub_{2g} = \mub_{1g} + \boldsymbol\Delta_{\mub}\) and \(\Sigmab_{2g} = \Sigmab_{1g} + \boldsymbol\Delta_{\Sigmab}\), where
    \begin{align*}
    \boldsymbol\Delta_{\mub} &=
    \begin{cases}
    (0,0,0,0)^\top, & \text{if (S4)}, \\
    (0.3,\ -0.2,\ 0.4,\ 0.1)^\top, & \text{otherwise},
    \end{cases}
    \\
    \boldsymbol\Delta_{\Sigmab} &=
    \begin{cases}
    \mathrm{diag}(0, 0, 0, 0), & \text{if (S4)}, \\
    \mathrm{diag}(0.05,\ -0.02,\ 0.03,\ 0.01), & \text{otherwise}.
    \end{cases}
    \end{align*}

    \item Finally, at $t=1$, if $(\Xb_{1i}, \Gb_{1i}) = (\xb, g)$, independently generate
    \[
    Y_{1i} \sim P_{Y_1 \mid (\Xb_1=\xb, \Gb_1=g)}, \quad \text{where} \quad
    P_{Y_t \mid (\Xb_t=\xb, \Gb_t=g)} = \mathcal{N}\big(\bb^\top \xb + \gammab_g,\ \sigma^2(\xb, g)\big), \ \ t \in \{1,2\}.
    \]
    We set \( \bb = (5,-3,2,1)^\top \), \( (\gamma_1,\ldots,\gamma_5) = (-10,-5,0,5,10) \), and
    \[
    \sigma^2(\xb, g) =
    \begin{cases}
    5, & \text{if (S5)}, \\
    \exp(\boldsymbol\theta^\top \xb + \eta_g), & \text{otherwise},
    \end{cases}
    \qquad \forall \xb \in \Rb^d,\ g \in \{1,\dots,5\},
    \]
    where $\boldsymbol\theta = (0.5,-0.4,0.3,0.2)^\top$ and $(\eta_1,\ldots,\eta_5) = (1.2,1.6,2.0,2.4,2.8)$.

\end{enumerate}

\section{Additional Simulation Results} \label{app:simul}

Here we present additional simulation results demonstrating that GWCP may either undercover or overcover, whereas the proposed method, DA-WCP, consistently achieves the desired coverage across all settings. Motivated by the theoretical results in Appendix~\ref{app:coverage}, we construct two settings designed to illustrate the following scenarios: (1) \texttt{GWCP\_AR\_m} fails to attain nominal coverage; and (2) both \texttt{GWCP\_AR\_m} and \texttt{GWCP\_CQR\_m} fail to attain nominal coverage.

For the simulation, we adopt an analogous setup with five scenarios (S1)--(S5), as described in Section~\ref{sec-simul} and Appendix~\ref{app:simul-setup}, but use a simplified design to facilitate interpretation. To recall, (S1) represents the most general setting, whereas (S2)--(S5) are constructed by sequentially removing sources of distributional shift or modifying the error structure. In our simulation, we set $X_{ti} \in [0, 1]$, $Y_{ti} \in \mathbb R$, and $Z_{ti} = G_{ti} \in \{1, \ldots, 5\}$. Moreover, we assume $H = 1$, dropping the stratum index, and assume $\rcov_{1i} = 1$ if and only if $\rfull_{1i} = 1$. With realized sample sizes $n_1 = n_2 = 5{,}000$ and $P_{Z_2} = \mathrm{Mult}(0.05, 0.1, 0.15, 0.3, 0.4)$, the datasets are generated by the following distributions:

\begin{enumerate}
    \item $P_{Z_1 \mid \rfull_1 = 1} = \mathrm{Mult}(0.2, \dots, 0.2)$, 
    except in (S2), where $P_{Z_1 \mid \rfull_1 = 1} = P_{Z_2}$.
    
    \item $P_{Z_2 \mid \rcov_2 = 1} = \mathrm{Mult}(0.1, 0.15, 0.2, 0.25, 0.3)$, 
    except in (S3), where $P_{Z_2 \mid \rcov_2 = 1} = P_{Z_2}$.
    
    \item For all scenarios except (S4), for each $k = 1, \ldots, 5$, we set 
    \[
    P_{X_1 \mid Z_1 = k} = \mathrm{Beta}(\alpha_{1k}, \beta_{1k}), \qquad P_{X_2 \mid Z_2 = k} = \mathrm{Beta}(\alpha_{2k}, \beta_{2k}).
    \]
    For (S4), we instead set 
    $P_{X_1 \mid Z_1 = k} = P_{X_2 \mid Z_2 = k} = \mathrm{Beta}(\alpha_{1k}, \beta_{1k})$ 
    for each $k = 1, \ldots, 5$.
    
    \item Finally, for all scenarios except (S5), for each $k = 1, \ldots, 5$ and $x \in \mathbb R$, we set
    \begin{equation*} 
    P_{Y_t \mid (X_t = x, Z_t = k)} = \mathcal{N}(0, \sigma^2(x, k) ), \quad t= 1, 2,
    \end{equation*}
    For (S5), where the error is homoscedastic, we set $\sigma(x, k) \equiv 2$. 
\end{enumerate}

\subsection{Setting 1: GWCP methods with the absolute residual score fail to attain nominal coverage}

For the simulation design, in this setting we set 
$(\alpha_{11}, \ldots, \alpha_{15}) = (2.0, 2.2, 2.4, 2.6, 2.8)$, 
and define $\beta_{1k} = \alpha_{2k} = 6 - \alpha_{1k}$ and $\beta_{2k} = \alpha_{1k}$ for each $k = 1, \ldots, 5$. 
Moreover, for scenarios (S1)--(S4), we specify the conditional standard deviation as 
$\sigma(x, k) = \theta_k x + \eta_k$ for each $k$.
As demonstrated in Appendix~\ref{app:coverage-ar}, \texttt{GWCP\_AR\_m} undercovers when $\sigma(x, k) \appropto r_k(x)$ 
and overcovers when $\sigma(x, k) \appropto 1/r_k(x)$ as a function of $x$ for each $k$. 
In our design, since $r_k(x) = (x / (1 - x))^{6 - 2\alpha_{1k}}$ is an increasing function of $x$ for each $k$, 
we therefore expect \texttt{GWCP\_AR\_m} to undercover when $\theta_k > 0$ and to overcover when $\theta_k < 0$, 
regardless of the choice of $\eta_k$. 
Accordingly, we consider a total of four heteroscedastic error cases: 

\begin{table}[!htbp]
\centering
\begin{threeparttable}
\captionsetup{skip=3pt}
\caption{Four heteroscedastic error cases for Setting 1}
\label{tab:setting-1}
\begin{tabular}{ccc}
\toprule
\textbf{Case} & $(\theta_1,\ldots,\theta_5)$ & $(\eta_1,\ldots,\eta_5)$ \\
\midrule
(a) & $(1.0,1.2,1.4,1.6,1.8)$ & $(1.0,1.2,1.4,1.6,1.8)$ \\
(b) & $(1.0,1.2,1.4,1.6,1.8)$ & $(1.8,1.6,1.4,1.2,1.0)$ \\
(c) & $(-1.0,-1.2,-1.4,-1.6,-1.8)$ & $(3.0,3.2,3.4,3.6,3.8)$ \\
(d) & $(-1.0,-1.2,-1.4,-1.6,-1.8)$ & $(3.8,3.6,3.4,3.2,3.0)$ \\
\bottomrule
\end{tabular}
\end{threeparttable}
\end{table}

We first compare the four marginal methods listed in Table~\ref{table:methods} across scenarios (S1)--(S5).  
Figure~\ref{fig:simul-marg-supp-1} and \ref{fig:simul-length-supp-1} show their empirical coverage and average interval length, respectively, where the panels from top to bottom correspond to cases (a)--(d) listed in Table~\ref{tab:setting-1}.  
Results are averaged over 100 independently generated datasets.

\begin{figure}[!htbp]
    \centering
    \includegraphics[width=\textwidth]{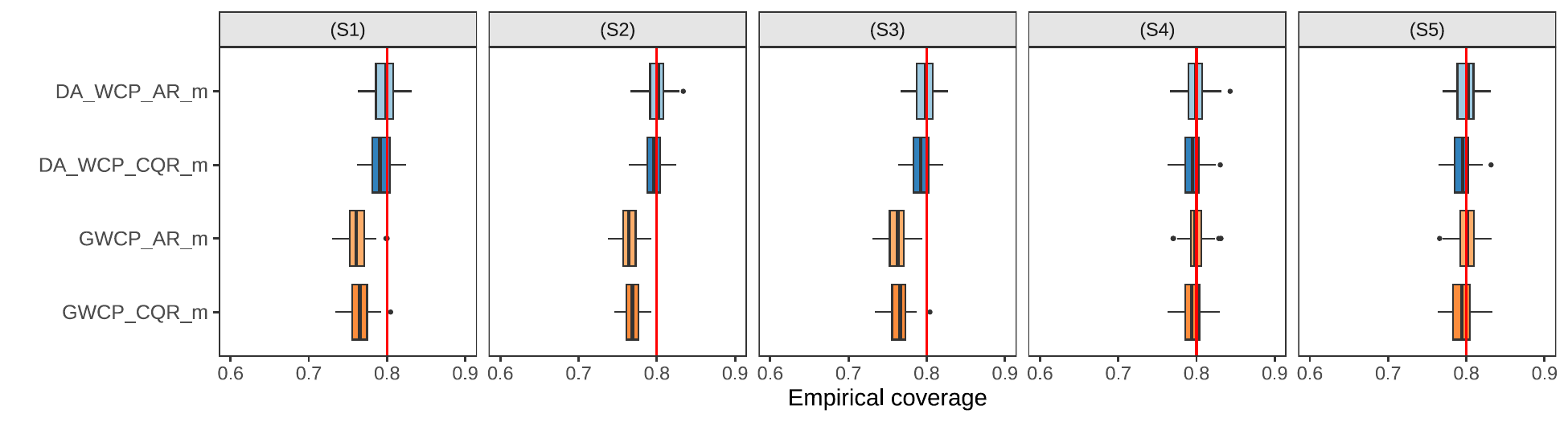}
    \vspace{0.5em} 
    \includegraphics[width=\textwidth]{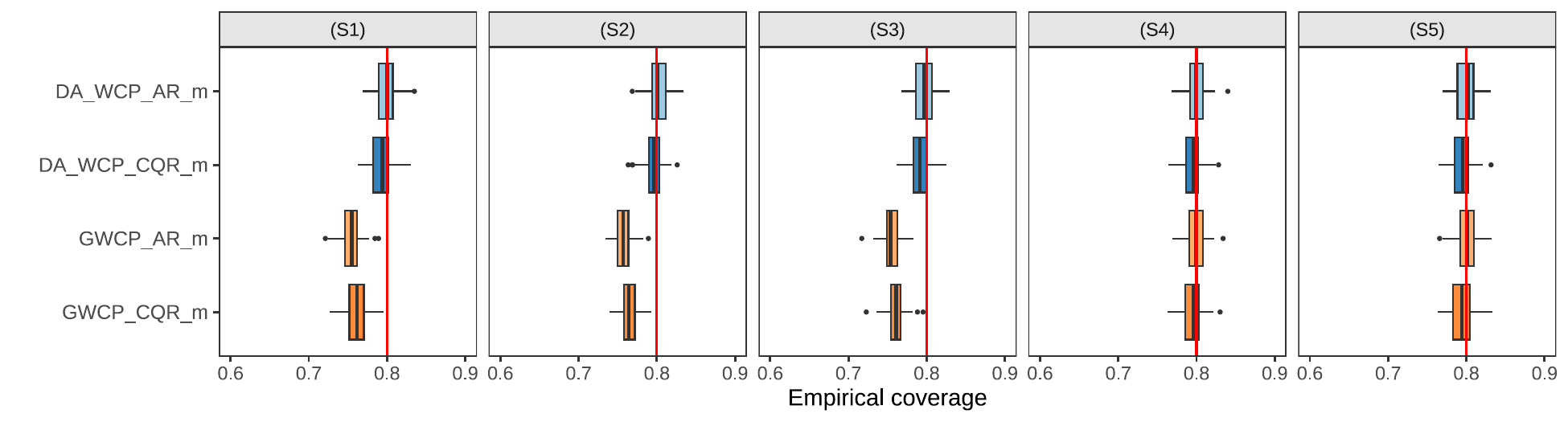}
    \vspace{0.5em} 
    \includegraphics[width=\textwidth]{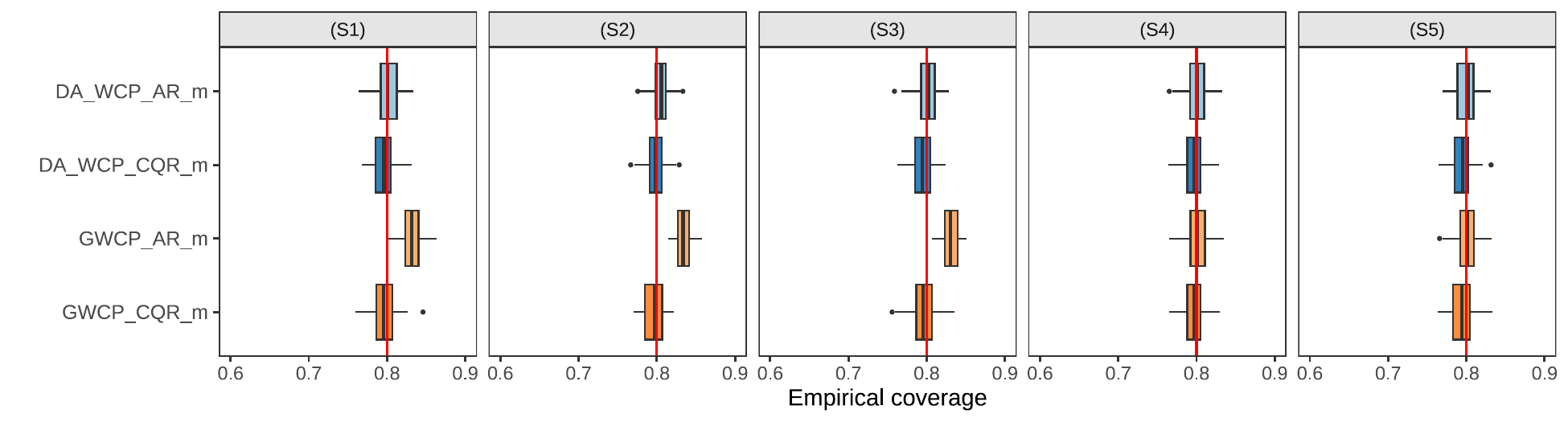}
    \vspace{0.5em} 
    \includegraphics[width=\textwidth]{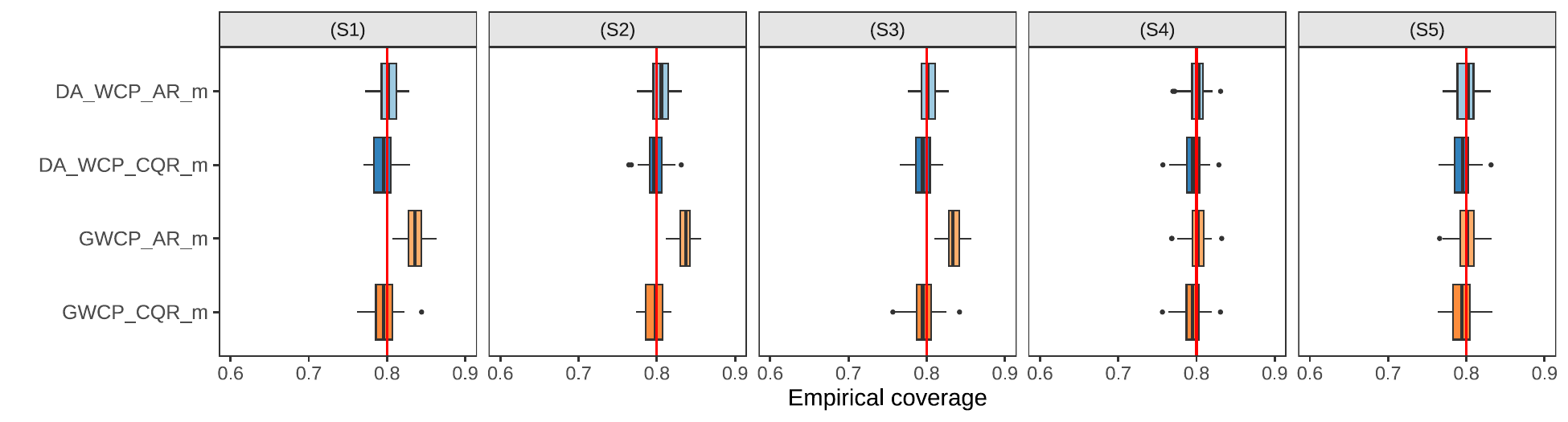}
    \caption{Empirical coverage of four marginal methods under the target level \(1-\alpha=0.8\), evaluated across scenarios (S1)--(S5) using 100 independently generated datasets.  
    From top to bottom, the panels correspond to heteroscedastic error cases (a)--(d) of Setting 1.}
    \label{fig:simul-marg-supp-1}
\end{figure}

\begin{figure}[!htbp]
    \centering
    \includegraphics[width=\textwidth]{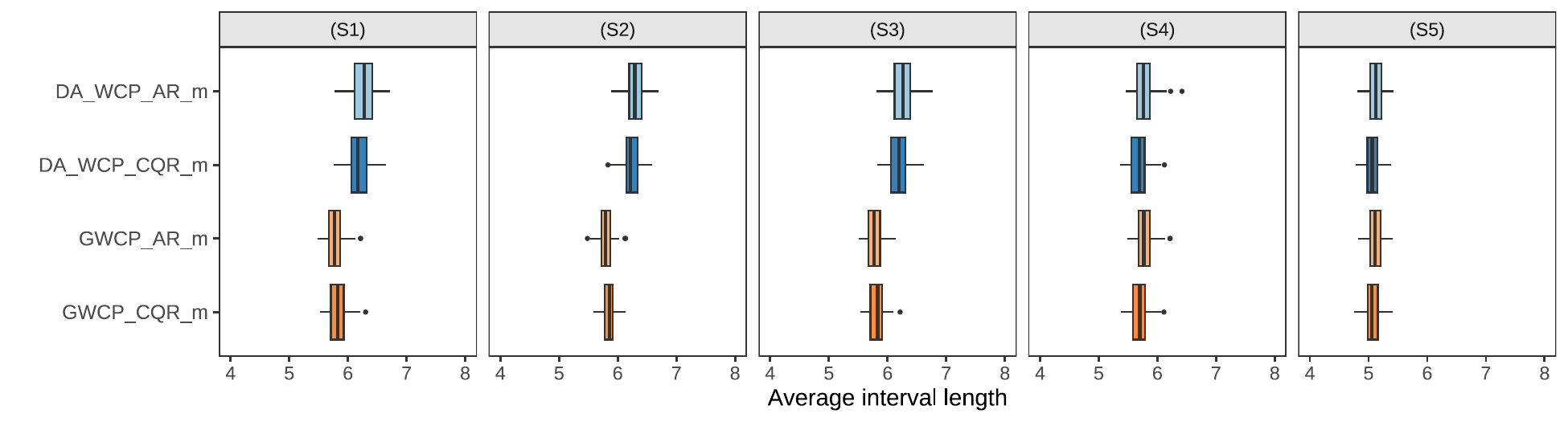}
    \vspace{0.5em} 
    \includegraphics[width=\textwidth]{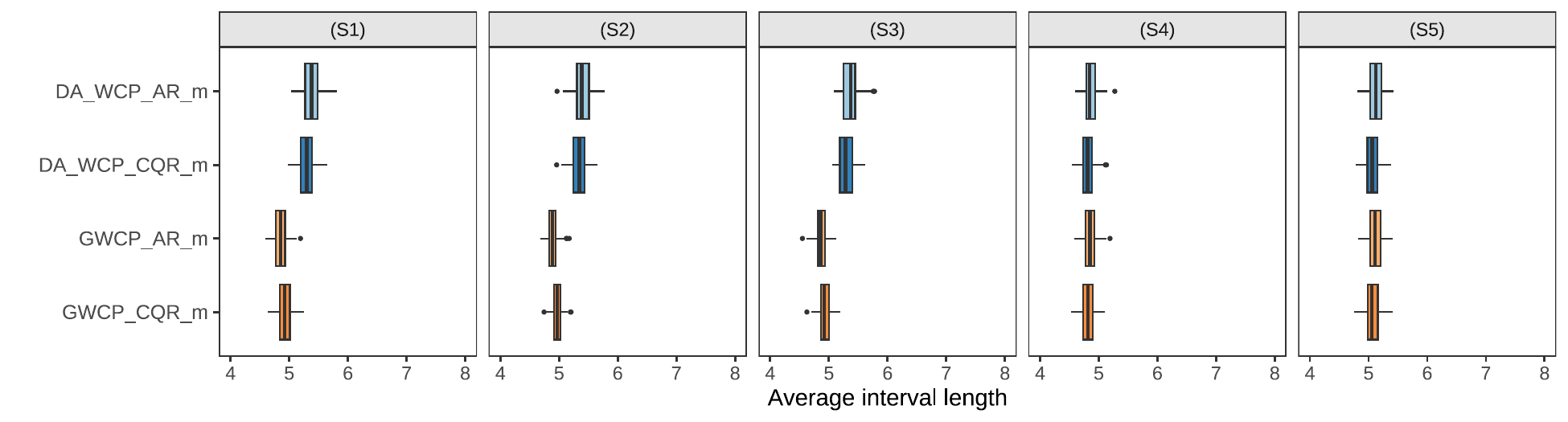}
    \vspace{0.5em} 
    \includegraphics[width=\textwidth]{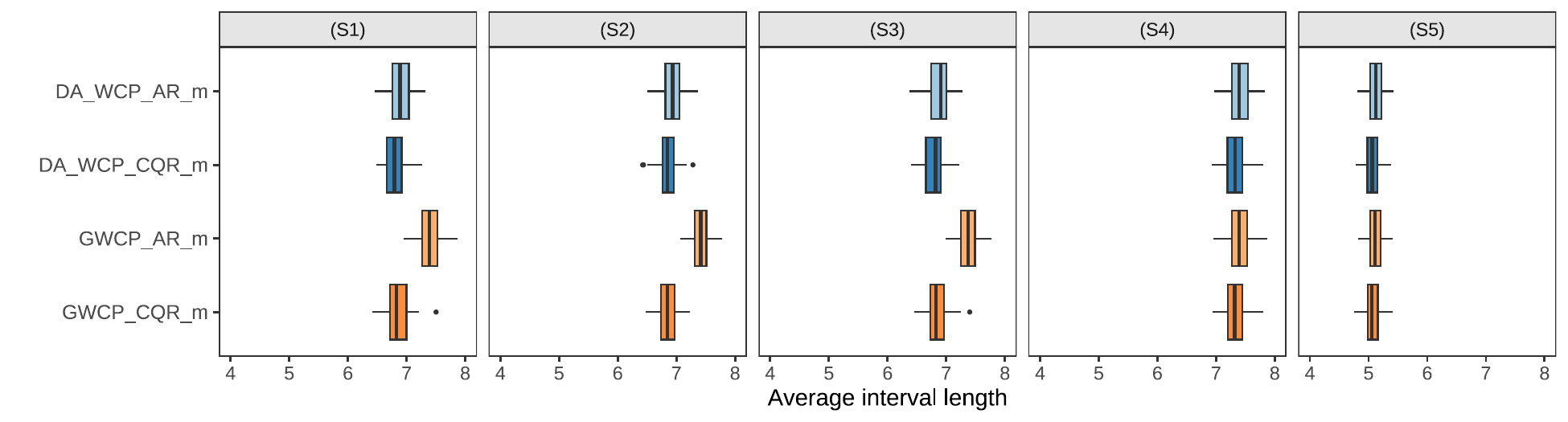}
    \vspace{0.5em} 
    \includegraphics[width=\textwidth]{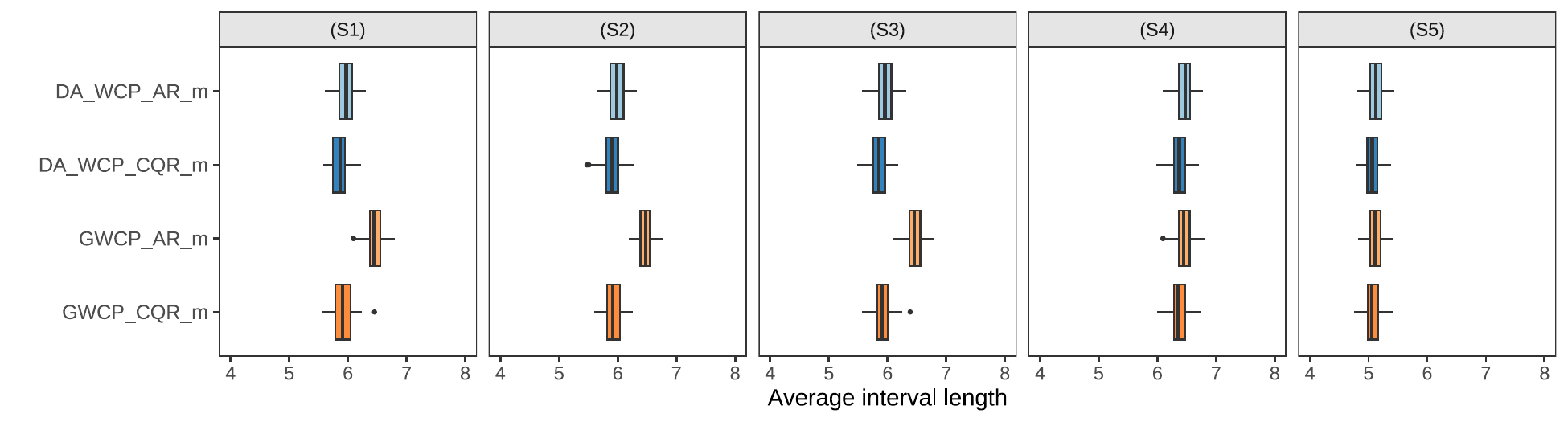}
    \caption{Average interval length of four marginal methods under the target level \(1-\alpha=0.8\), evaluated across scenarios (S1)--(S5) using 100 independently generated datasets.  
    From top to bottom, the panels correspond to heteroscedastic error cases (a)--(d) of Setting 1.}
    \label{fig:simul-length-supp-1}
\end{figure}

Overall, the two DA-WCP methods attain the desired coverage across all scenarios and error cases. 
As expected, cases (a) and (b), as well as cases (c) and (d), exhibit similar coverage patterns across methods, indicating that the choice of $\eta_k$ does not affect coverage.  
Moreover, \texttt{GWCP\_AR\_m} suffers from undercoverage in (a) and (b) and from overcoverage in (c) and (d), as noted above. 
\texttt{GWCP\_CQR\_m} also exhibits undercoverage in (a) and (b), but attains the desired coverage level in (c) and (d), likely due to the set selection step of the CQR methods.

Next, we compare the four group-conditional methods listed in Table~\ref{table:methods} across scenarios (S1)--(S5).  
Figures~\ref{fig:simul-group-supp-1-ab} and~\ref{fig:simul-group-supp-1-cd} present their group-wise empirical coverage, where the panels from top to bottom correspond to cases (a)--(d) listed in Table~\ref{tab:setting-1}.  
As in the marginal setting, cases (a) and (b), as well as cases (c) and (d), exhibit similar group-wise coverage patterns across methods, and the two DA-WCP methods attain the desired coverage across all cases.
Among the two GWCP methods, \texttt{GWCP\_AR\_g} undercovers in (a) and (b) and overcovers in (c) and (d), while \texttt{GWCP\_CQR\_g} also exhibits undercoverage in (a) and (b) but attains the desired coverage level in (c) and (d).
These results indicate that the coverage of GWCP with the absolute residual score can vary considerably under different forms of heteroscedasticity even when $\widehat \mu \approx \mu$.

\begin{figure}[!htbp]
    \centering
    \includegraphics[width=\textwidth]{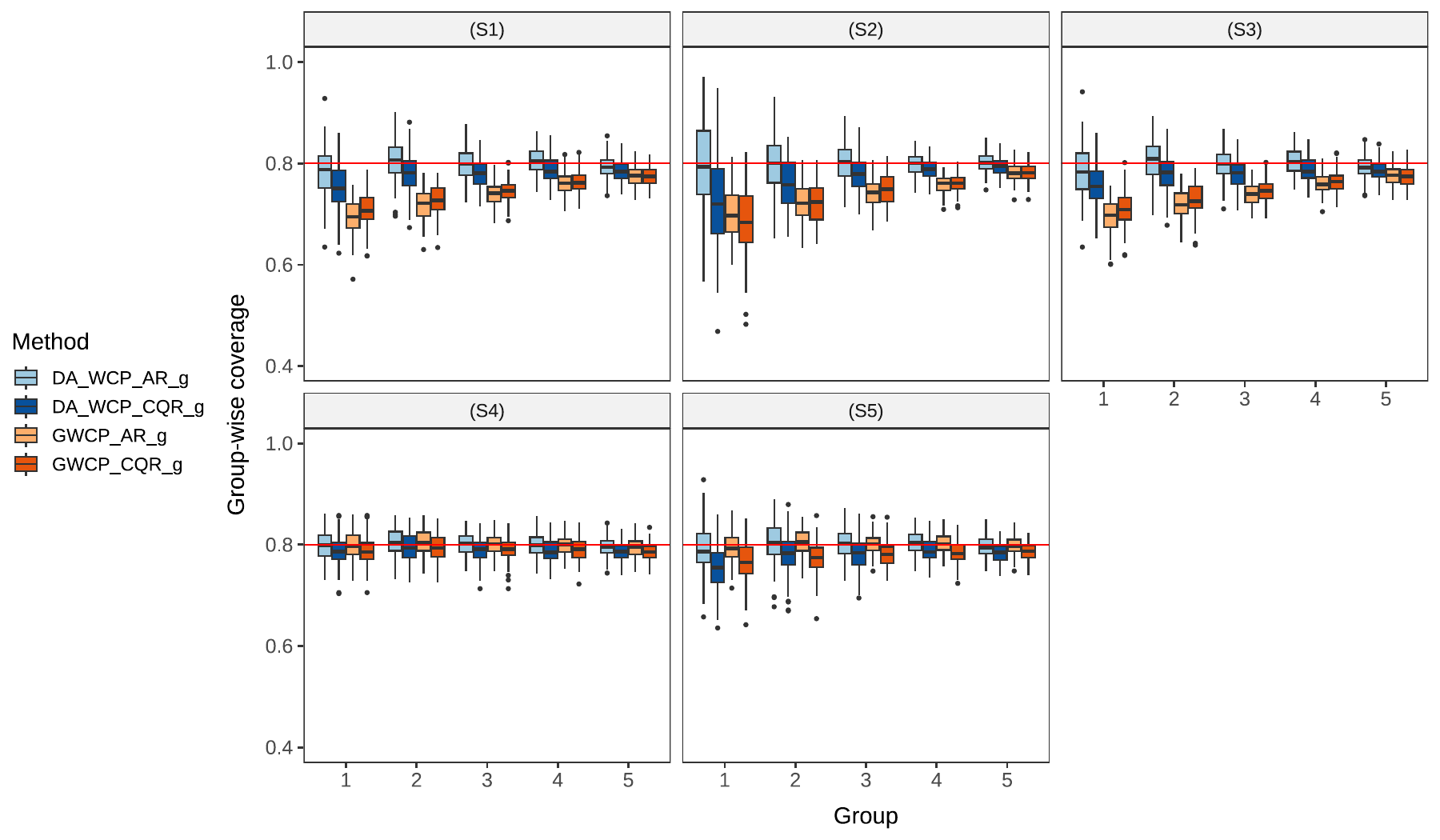}
    \vspace{5mm} 
    \includegraphics[width=\textwidth]{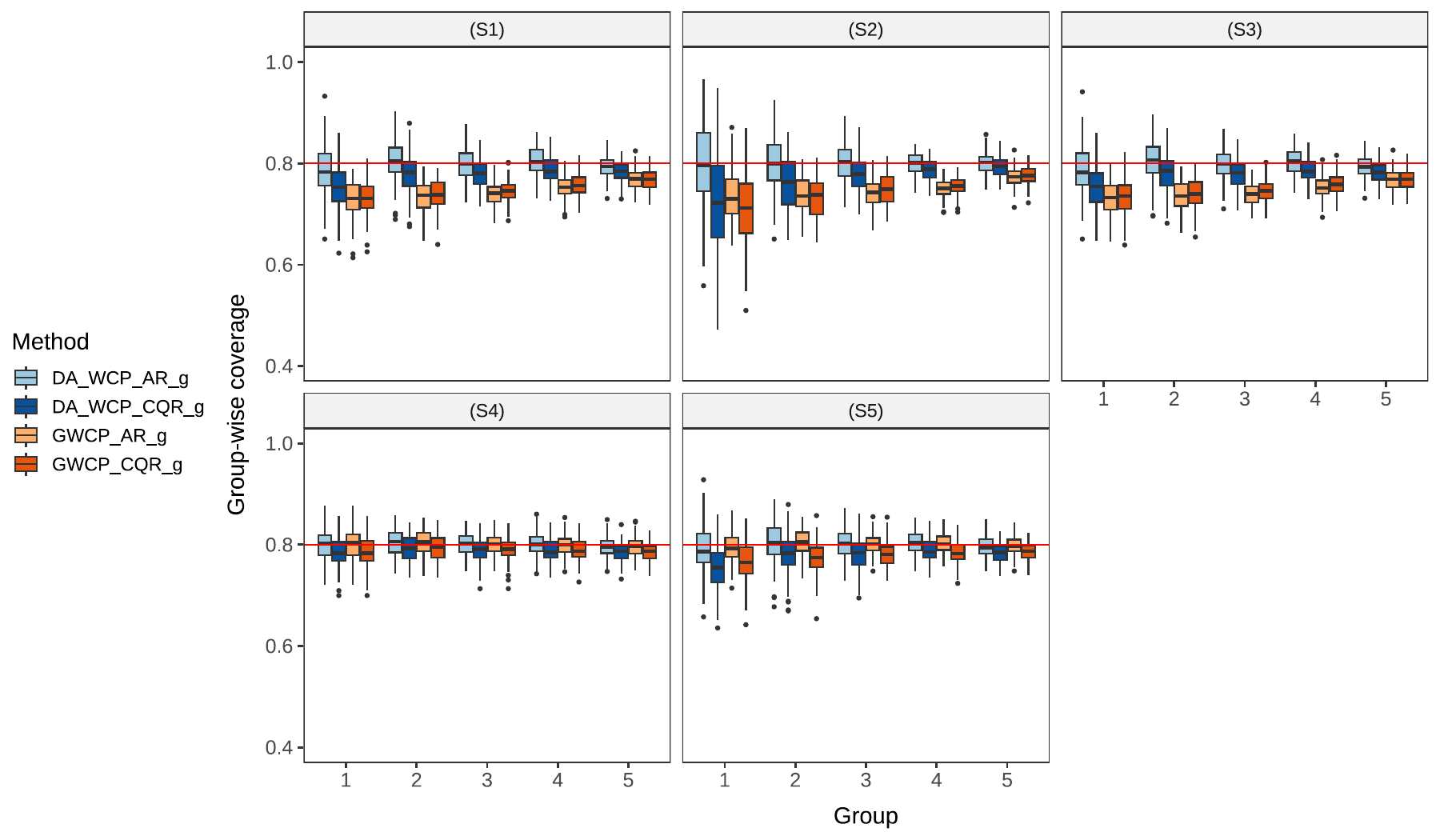}
    \caption{Group-wise empirical coverage of four group-conditional methods under the target level \(1-\alpha=0.8\), evaluated across scenarios (S1)--(S5) using 100 independently generated datasets.  
    The upper and lower panels correspond to cases (a) and (b) of Setting 1, respectively.}
    \label{fig:simul-group-supp-1-ab}
\end{figure}

\begin{figure}[!htbp]
    \centering
    \includegraphics[width=\textwidth]{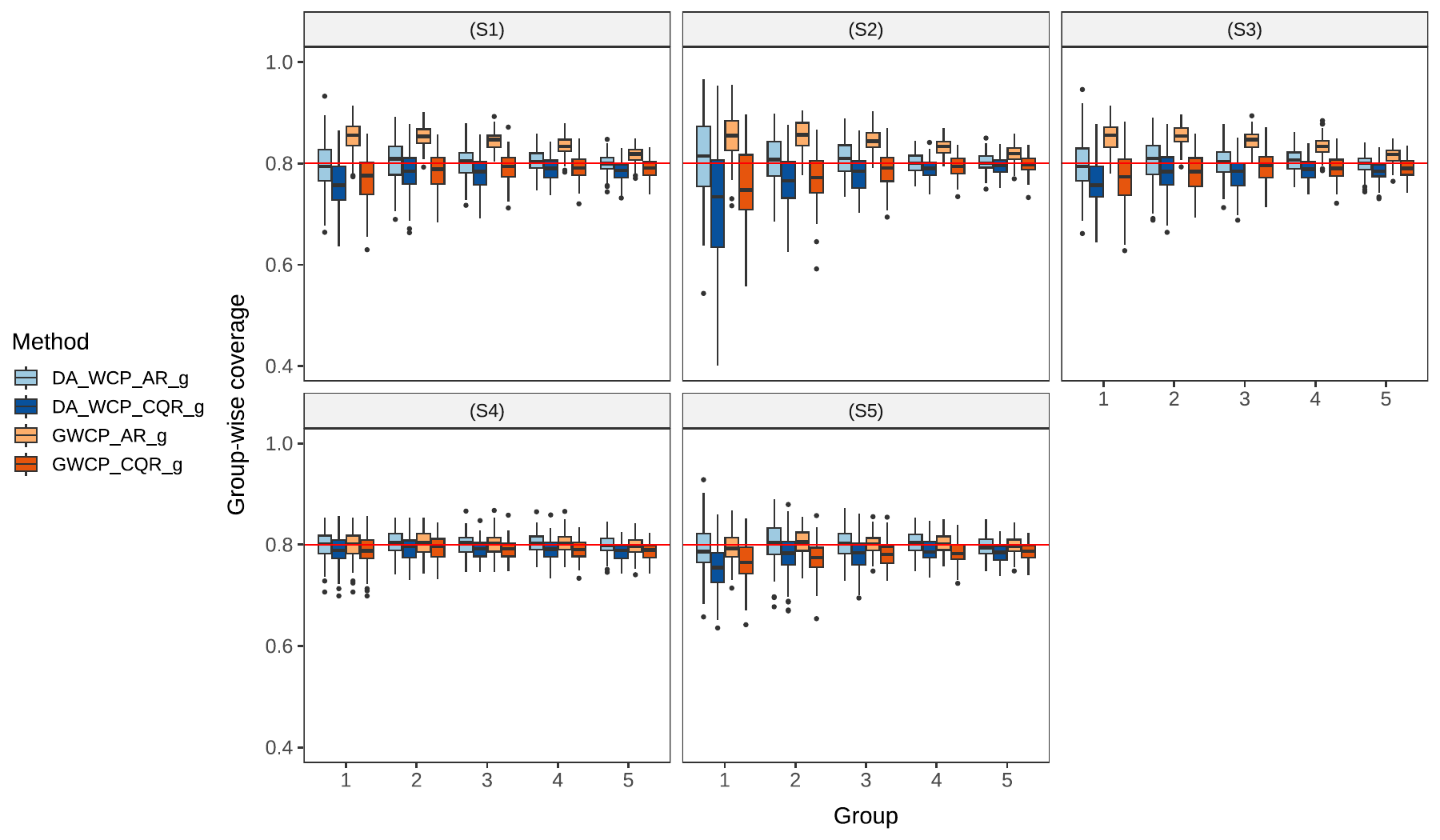}
    \vspace{5mm} 
    \includegraphics[width=\textwidth]{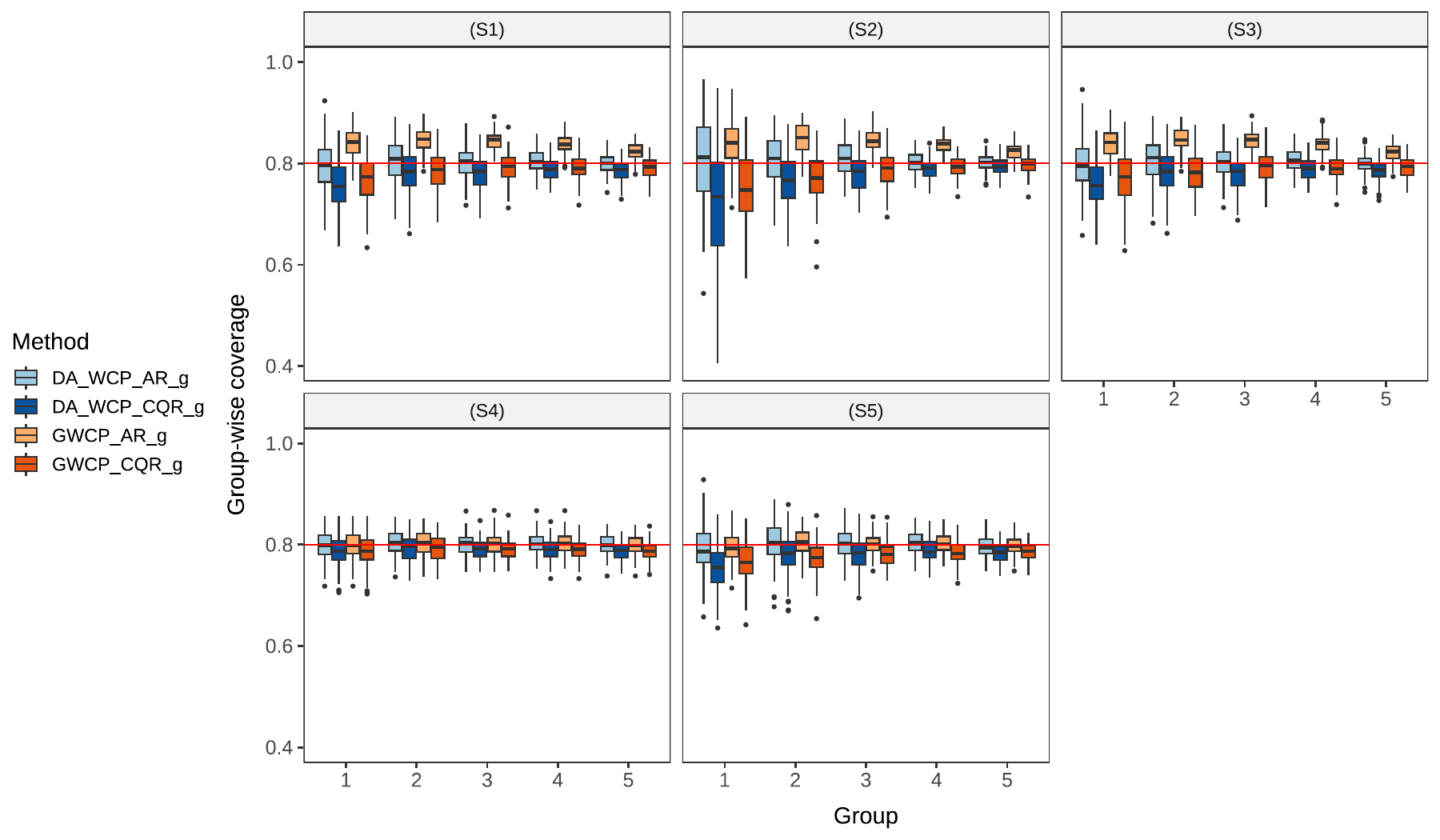}
    \caption{Group-wise empirical coverage of four group-conditional methods under the target level \(1-\alpha=0.8\), evaluated across scenarios (S1)--(S5) using 100 independently generated datasets.  
    The upper and lower panels correspond to cases (c) and (d) of Setting 1, respectively.}
    \label{fig:simul-group-supp-1-cd}
\end{figure}

\subsection{Setting 2 : All GWCP methods fail to attain nominal coverage}

In this setting, we set 
$(\alpha_{11}, \ldots, \alpha_{15}) = (0.95, 0.90, 0.85, 0.80, 0.75)$, 
and define $\beta_{1k} = \alpha_{1k}$ and $\alpha_{2k} = \beta_{2k} = 3 - \alpha_{1k}$ for each $k = 1, \ldots, 5$. 
Moreover, for scenarios (S1)--(S4), we specify the conditional standard deviation as 
$\sigma(x, k) = \sigma_k  |2x - 1| + 1$ for case (a), 
and $\sigma(x, k) = 4 - \sigma_k  |2x - 1|$ for case (b), 
where $(\sigma_1, \ldots, \sigma_5) = (2.2, 2.4, 2.6, 2.8, 3.0)$.

Then, the group-wise density ratio function becomes $r_k(x) \propto (x(1-x))^{3-2\alpha_{1k}}$ for each $k$, 
taking higher values near the center of $[0, 1]$ and lower values toward its boundaries. 
Since we employ a linear model to estimate the mean and quantile functions, $\widehat \mu$ is consistent, whereas $\hatqlo$ and $\hatqhi$ are not. 
As shown in the left panel of Figure~\ref{fig:quantile}, for case (a), where $\sigma(x, k)$ is smaller near the center, 
we observe that $\hatqlo(x, k) < \qlo(x, k)$ near the boundaries and $\hatqlo(x, k) > \qlo(x, k)$ near the center. 
Consequently, for each fixed $k \in \{1, \dots, 5\}$ and $y \in \Rb$, the quantity $\Pb( \hatqlo(x, k) - y \le Y_1 \le \hatqhi(x, k) + y \mid \Dc, X_1 = x, Z_1 = k )$, viewed as a function of $x$, is lower near the boundaries and higher near the center, which is expected to result in overcoverage of GWCP for both absolute residual and CQR scores, as demonstrated in Appendix~\ref{app:coverage}.
Conversely, for case (b), where $\sigma(x, k)$ is larger near the center, 
$\hatqlo(x, k) > \qlo(x, k)$ near the boundaries and $\hatqlo(x, k) < \qlo(x, k)$ near the center, as shown in the right panel of Figure~\ref{fig:quantile}, 
which is expected to lead to undercoverage of GWCP for both scores.

\begin{figure}[!htbp]
    \centering
    \includegraphics[width=0.45\linewidth]{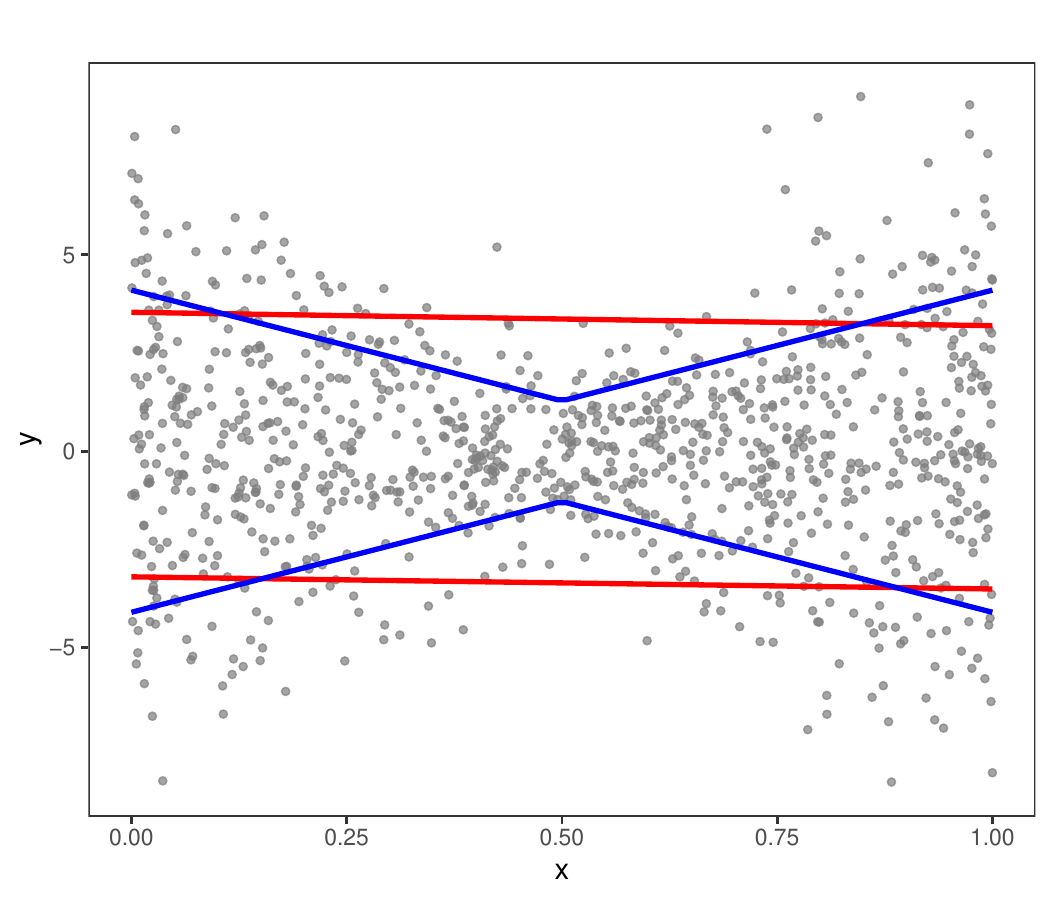}
    \hspace{0.5em}
    \includegraphics[width=0.45\linewidth]{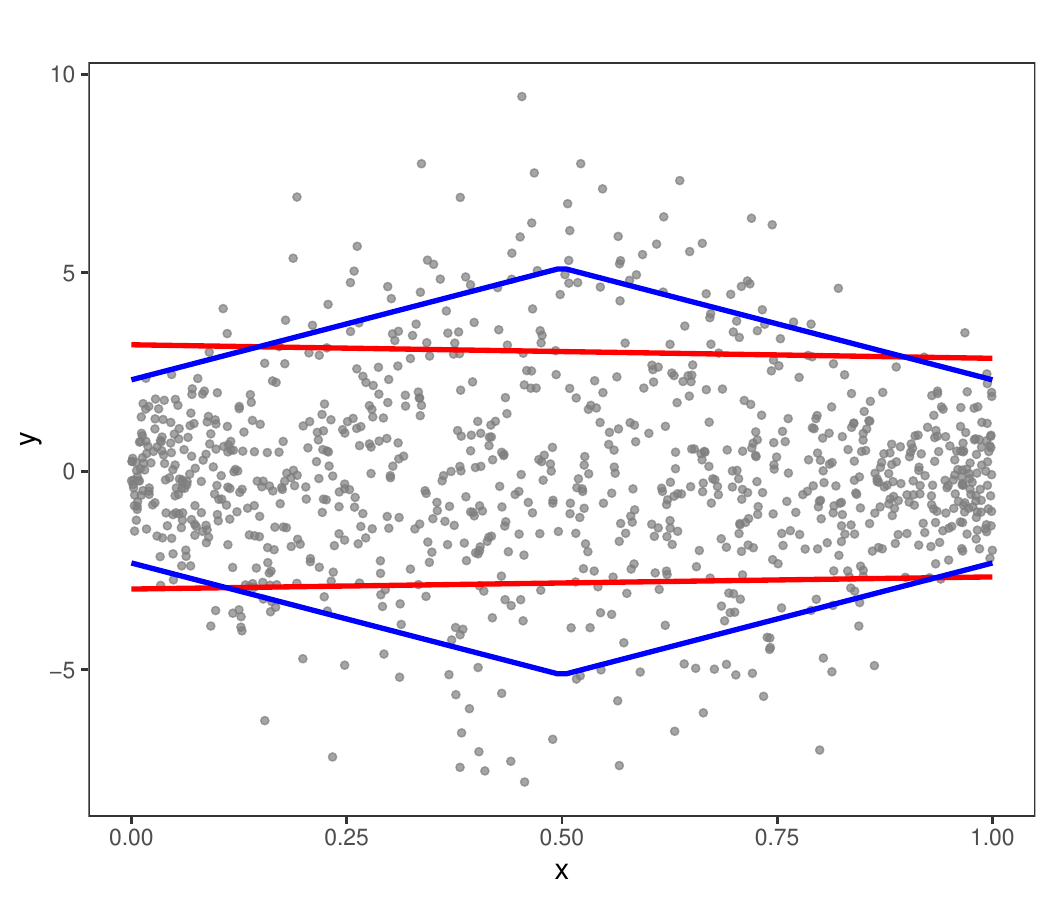}
    \caption{Comparison of the true quantile functions $\qlo$ and $\qhi$ (blue) with the estimated quantile functions $\hatqlo$ and $\hatqhi$ (red) for case (a) (left) and case (b) (right), at $\alpha = 0.2$, based on the dataset $\done$ generated with $n_1 = 5000$ for group $k = 5$.}
    \label{fig:quantile}
\end{figure}

We first compare the four marginal methods listed in Table~\ref{table:methods} across scenarios (S1)--(S5).  
Figures~\ref{fig:simul-marg-supp-2} and \ref{fig:simul-length-supp-2} show their empirical coverage and average interval length, respectively.  
Results are averaged over 100 independently generated datasets.  
Overall, the two DA-WCP methods attain the desired coverage across all scenarios and error cases, with both scores producing intervals of similar length.  
In contrast, in scenarios (S1)--(S3), GWCP exhibits overcoverage in (a) and undercoverage in (b) for both scores, as expected.

Next, we compare the four group-conditional methods listed in Table~\ref{table:methods} across scenarios (S1)--(S5).  
Figure~\ref{fig:simul-group-supp-2} presents their group-wise empirical coverage.  
As in the marginal case, in scenarios (S1)--(S3), GWCP overcovers in (a) and undercovers in (b) for both scores, while DA-WCP consistently achieves the desired coverage.  
These results indicate that, for the CQR score, GWCP may fail to achieve the desired coverage when the estimated quantile functions are inconsistent, which can occur under model misspecification, whereas DA-WCP remains robust to such misspecification.

\begin{figure}[!htbp]
    \centering
    \includegraphics[width=0.95\textwidth]{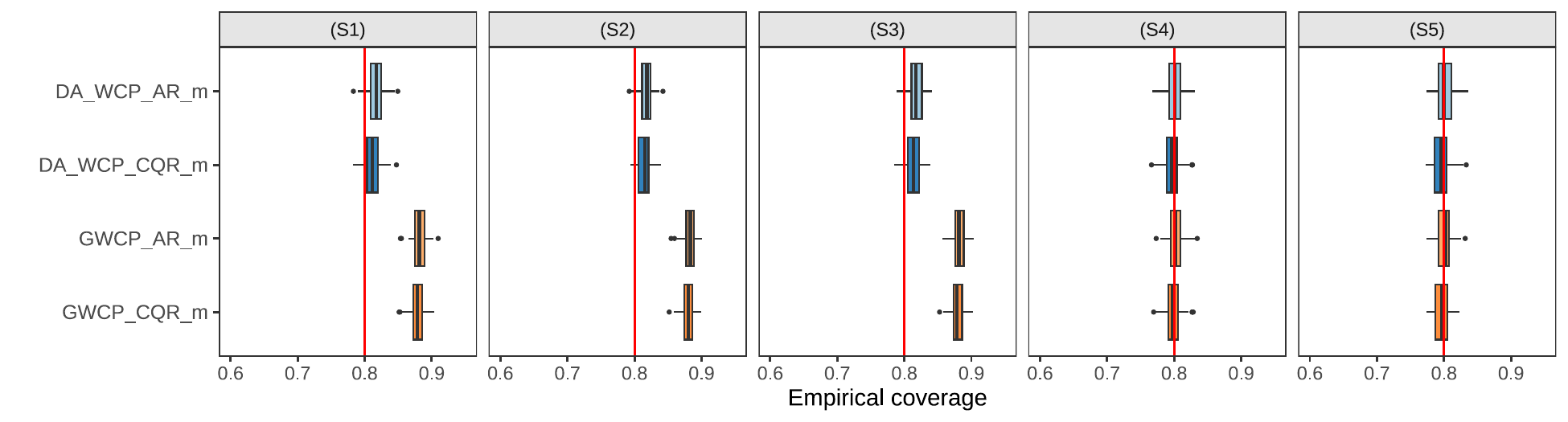}
    \vspace{0.5em} 
    \includegraphics[width=0.95\textwidth]{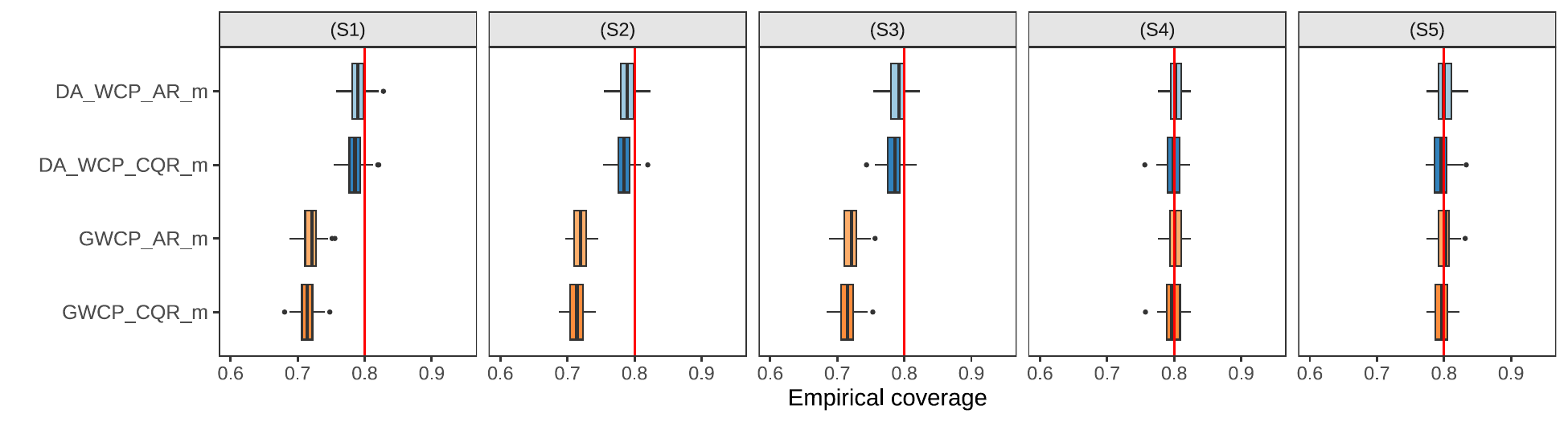}
    \caption{Empirical coverage of four marginal methods under the target level \(1 - \alpha = 0.8\), 
    evaluated across scenarios (S1)--(S5) using 100 independently generated datasets. 
    The upper and lower panels correspond to cases (a) and (b), respectively.}
    \label{fig:simul-marg-supp-2}
\end{figure}

\begin{figure}[!htbp]
    \centering
    \includegraphics[width=0.95\textwidth]{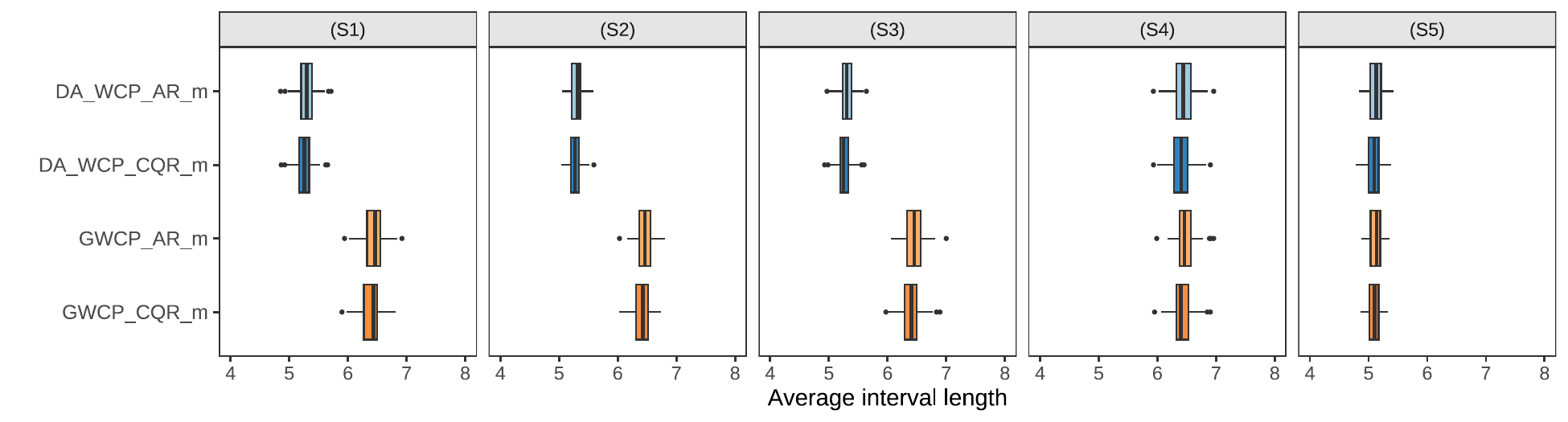}
    \vspace{0.5em} 
    \includegraphics[width=0.95\textwidth]{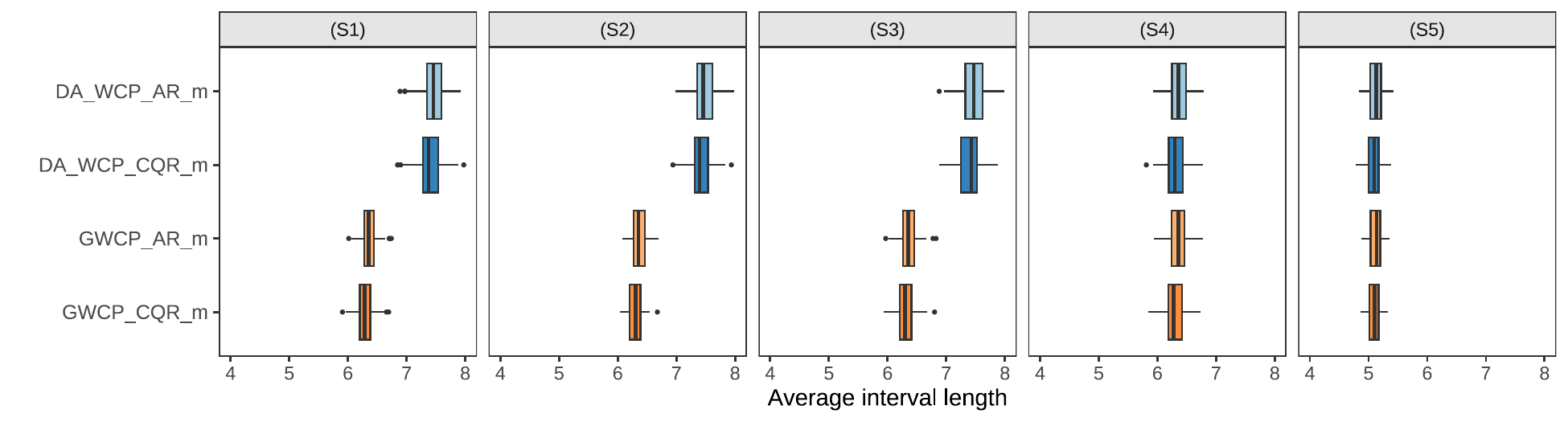}
    \caption{Average interval length of four marginal methods under the target level \(1 - \alpha = 0.8\), 
    evaluated across scenarios (S1)--(S5) using 100 independently generated datasets. 
    The upper and lower panels correspond to cases (a) and (b), respectively.}
    \label{fig:simul-length-supp-2}
\end{figure}

\begin{figure}[!htbp]
    \centering
    \includegraphics[width=\textwidth]{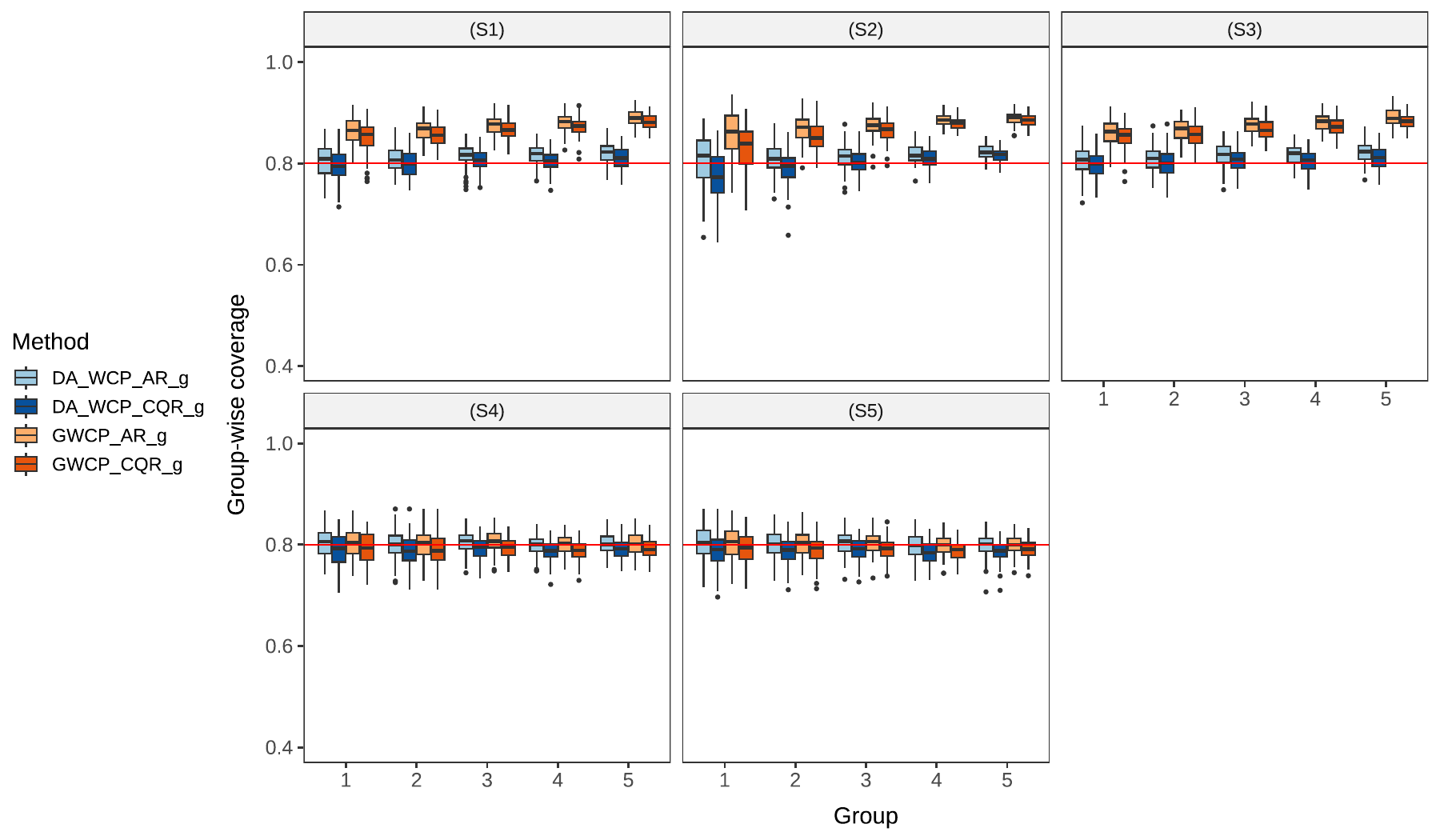}
    \vspace{0.5em} 
    \includegraphics[width=\textwidth]{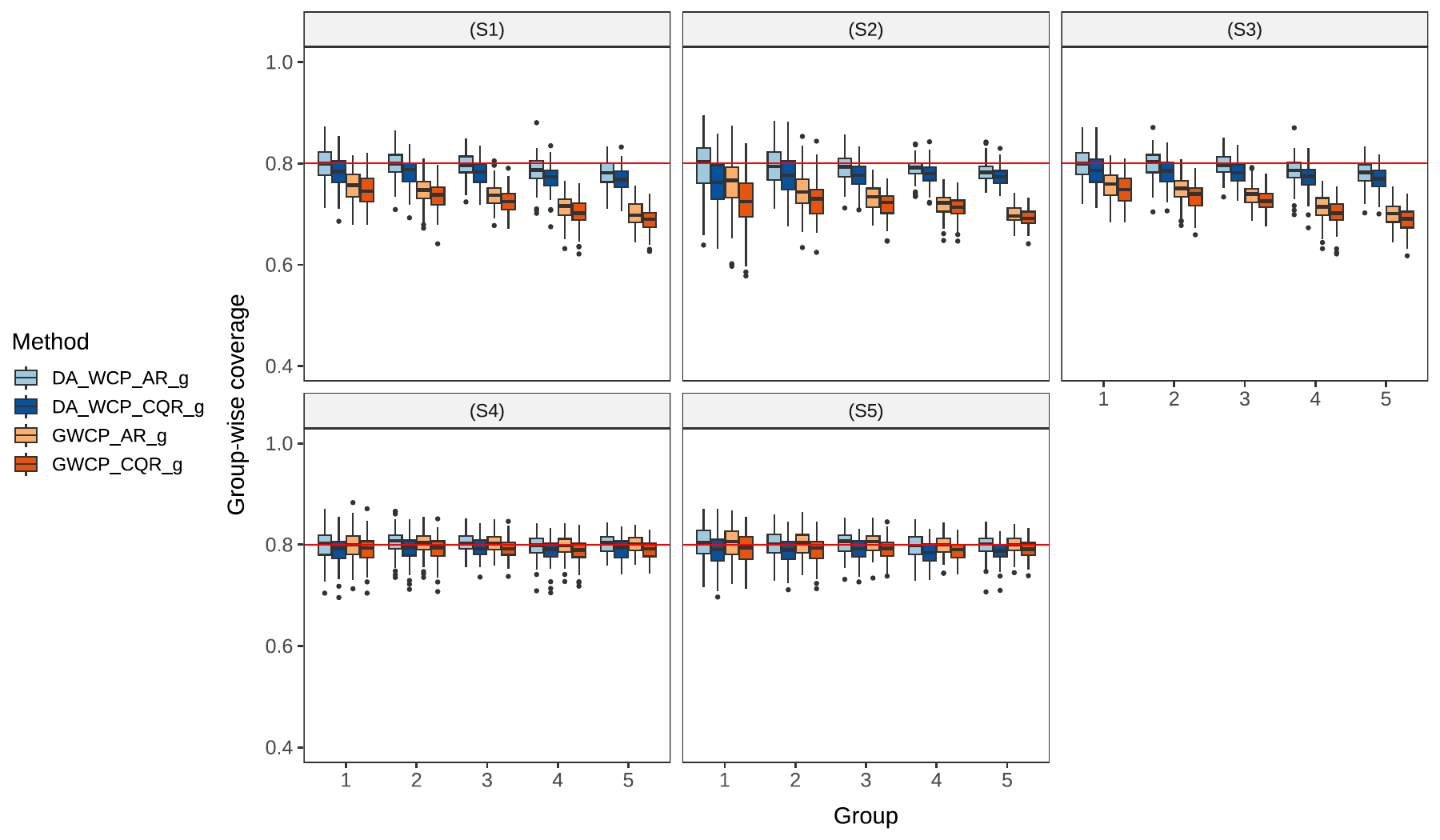}
    \caption{Group-wise empirical coverage of four group-conditional methods under the target level \(1-\alpha=0.8\), evaluated across scenarios (S1)--(S5) using 100 independently generated datasets.  
    The upper and lower panels correspond to cases (a) and (b), respectively.}
    \label{fig:simul-group-supp-2}
\end{figure}

\clearpage

\section{LDL-C Prediction via Group-Conditional Methods} \label{app:ldl-group} 

In Section~\ref{sec-ldl}, we used only the marginal methods for predicting LDL-C in the NHANES dataset. 
The group-conditional methods were omitted because their results were highly variable. 
In this section, we provide results obtained using the group-conditional methods.

For the continuous outcome (see Section~\ref{sec-ldl-conti}), Figure~\ref{fig:conti-group} presents the subgroup-wise empirical coverage of the four group-conditional methods listed in Table~\ref{table:methods}. For the categorical outcome (see Section~\ref{sec-ldl-categorical}), Figure~\ref{fig:categorical-group} presents the subgroup-wise empirical coverage of the six group-conditional methods for categorical outcomes. In both cases, empirical coverage is computed based on 100 random splits of the training and calibration datasets within each subgroup. All methods achieve empirical coverage roughly around the target level, but with considerable variability across subgroups. From the figures, DA-WCP methods appear to attain the desired coverage in more subgroups overall than the corresponding GWCP methods.

High variability in empirical coverage across subgroups is due to insufficient sample sizes within subgroups.
When subgroup sizes are small, (i) estimation of the subgroup-specific weight function, which involves density ratio estimation, and (ii) the approximation described in Appendix~\ref{app:eval-2} used to compute subgroup-wise empirical coverage, both become highly variable.
These sources of variability lead to unstable empirical coverage across subgroups.
However, when subgroup sizes are sufficiently large, the group-conditional methods yield empirical coverage as stable as that of the marginal methods.
In such cases, we can confirm that the proposed method generally outperforms GWCP, as demonstrated in the simulation study in Section~\ref{sec-simul} and Appendix~\ref{app:simul}.

\vspace{4mm}
\begin{figure}[!htbp]
    \centering
    \includegraphics[width=\linewidth]{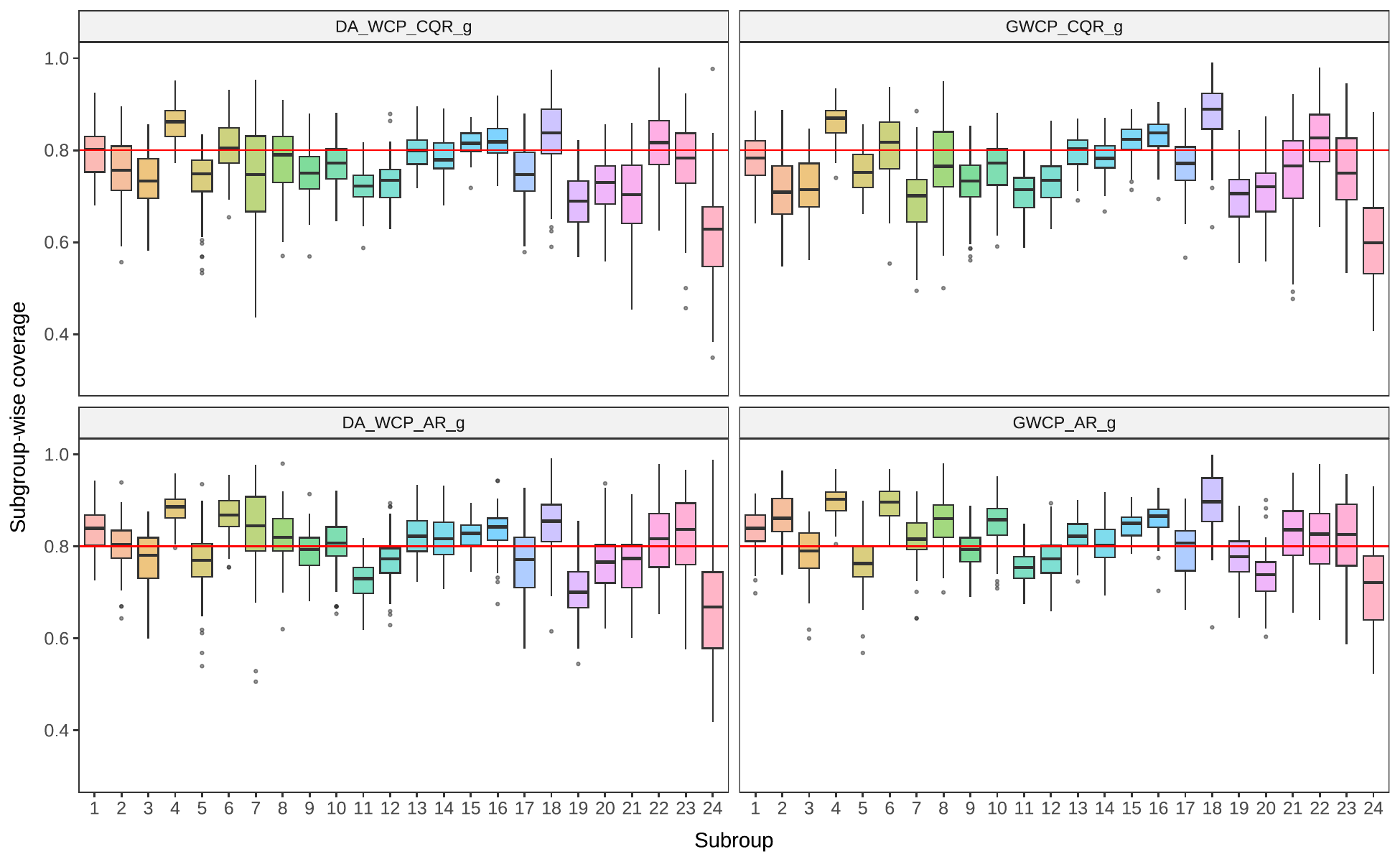}
    \caption{Subgroup-wise empirical coverage of the four group-conditional methods for continuous outcomes listed in Table~\ref{table:methods}, evaluated under the target coverage level $1 - \alpha = 0.8$ based on 100 random splits.}
    \label{fig:conti-group}
\end{figure}

\begin{figure}[!htbp]
    \centering
    \includegraphics[width=\linewidth]{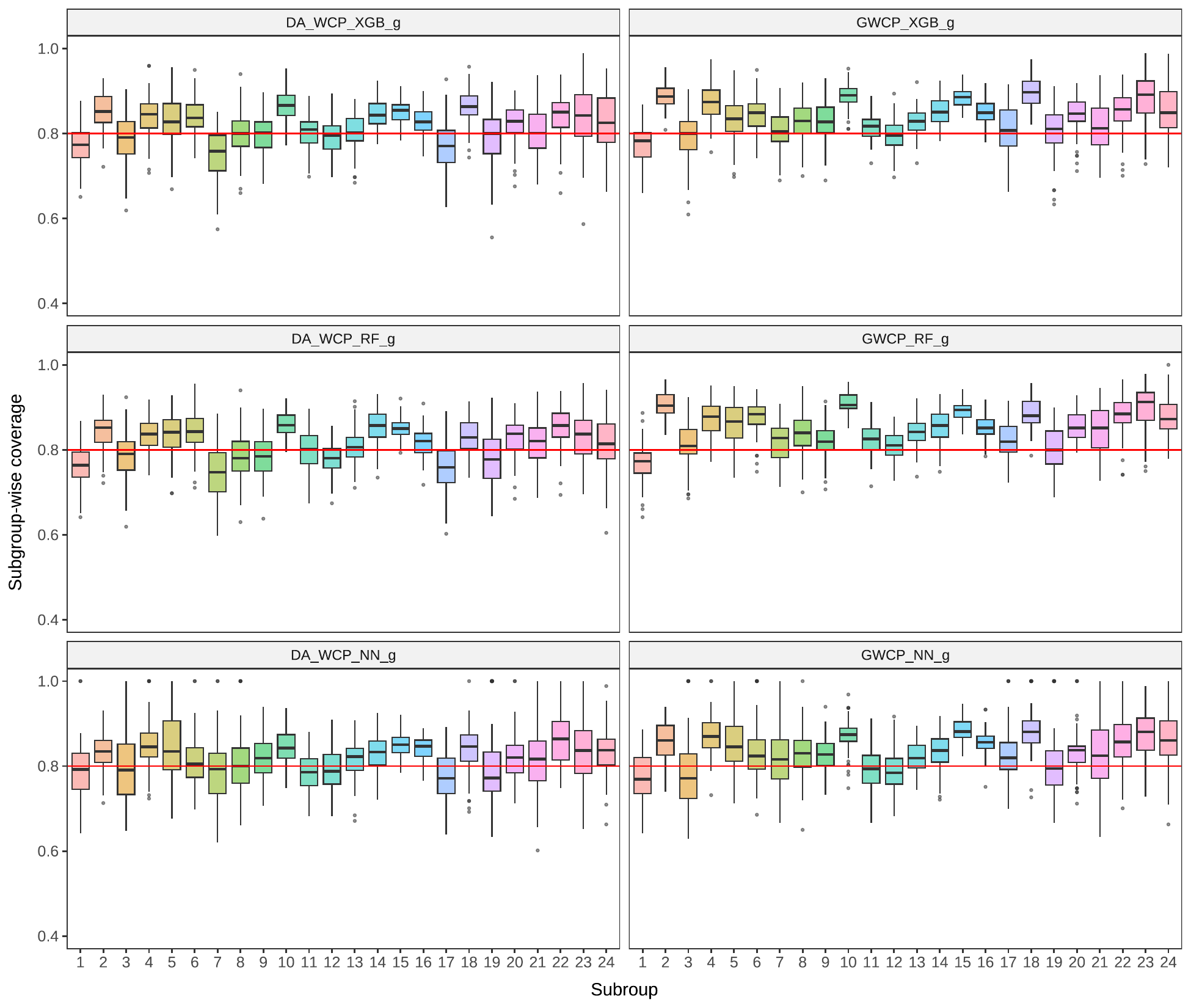}
    \caption{Subgroup-wise empirical coverage of the six group-conditional methods for categorical outcomes, evaluated under the target coverage level $1 - \alpha = 0.8$ based on 100 random splits.}
    \label{fig:categorical-group}
\end{figure}

\clearpage

\end{document}